\documentclass{amsart}%
\usepackage{amsfonts}
\usepackage{amsmath}
\usepackage{amssymb}
\usepackage{graphicx}%
\providecommand{\U}[1]{\protect\rule{.1in}{.1in}}
\newtheorem{theorem}{Theorem}
\theoremstyle{plain}

\newtheorem{conjecture}{Conjecture}
\newtheorem{corollary}{Corollary}

\newtheorem{definition}{Definition}
\newtheorem{example}{Example}

\newtheorem{lemma}{Lemma}

\newtheorem{proposition}{Proposition}
\newtheorem{remark}{Remark}

\numberwithin{equation}{section}
\begin{document}
\title[Quantum Knots and Lattices]{Quantum Knots and Lattices, \linebreak or a Blueprint for Quantum Systems that
Do Rope Tricks}
\author{Samuel J. Lomonaco}
\address{University of Maryland Baltimore County (UMBC)\\
Baltimore, MD \ 21250 \ \ USA}
\email{lomonaco@umbc.edu}
\urladdr{http://www.csee.umbc.edu/\symbol{126}lomonaco}
\author{Louis H. Kauffman}
\address{University of Illinois at Chicago\\
Chicago, IL \ 60607-7045 \ \ USA}
\email{kauffman@uic.edu}
\urladdr{http://www.math.uic.edu/\symbol{126}kauffman}
\date{February 24, 2008}

\begin{abstract}
In the paper "Quantum knots and mosaics," \cite{Lomonaco1}, the definition of
quantum knots was based on the planar projections of knots (i.e., on knot
diagrams) and the Reidemeister moves on these projections. \ In this paper, we
take a different tack by creating a definition of quantum knots based on the
cubic honeycomb decomposition of 3-space $\mathbb{R}^{3}$ (i.e., the cubic
tesselation $\mathcal{L}_{\ell}$ of $\mathbb{R}^{3}$ consisting of $2^{-\ell
}\times2^{-\ell}\times2^{-\ell}$ cubes) and a new set of knot moves, called
\textbf{wiggle}, \textbf{wag}, and \textbf{tug}, which unlike the two
dimensional Reidemeister moves are truly three dimensional moves. \ These new
moves have been so named because they mimic how a dog might wag its tail. \ 

\ We believe that these two different approaches to defining quantum knots are
essentially equivalent, but that the above three dimensional moves have a
definite advantage when it comes to the applications of knot theory to
physics. \ More specifically, we contend that the new moves wiggle, wag, and
tug are more "physics-friendly" than the Reidemeister ones. \ For unlike the
Reidemeister moves, the new moves are three dimensional moves that respect the
differential geometry of 3-space, which is indeed an essential component of
physics. \ And moreover, unlike the Reidemeister moves, they can be
transformed into infinitesimal moves and differential forms, which structures
can be seamlessly interwoven with the equations of physics.

\ Our basic building block for constructing a quantum knot is a
\textbf{lattice knot}, which is a knot in 3-space constructed from the edges
of the cubic honeycomb $\mathcal{L}_{\ell}$. \ We then create a Hilbert space
by identifying each edge of a bounded $n\times n\times n$ region of the cubic
honeycomb with a qubit. \ Lattice knots within this region then form the basis
of a sub-Hilbert space $\mathcal{K}^{\left(  \ell,n\right)  }$. \ The states
of $\mathcal{K}^{\left(  \ell,n\right)  }$ are called quantum knots. \ The
knot moves, wiggle, wag, and tug, are then naturally identified with the
generators of a unitary group $\Lambda_{\ell,n}$, called the \textbf{lattice
ambient group}, acting on the Hilbert space $\mathcal{K}^{\left(
\ell,n\right)  }$. \ 

\ This definition of a quantum knot can be viewed as a blueprint for the
construction of an actual physical quantum system that represents the "quantum
embodiment" of a closed knotted physical piece of rope. \ A quantum knot, as a
state of this quantum knot system, represents the state of such a knotted
closed piece of rope, i.e., the particular spacial configuration of the knot
tied in the rope. \ The lattice ambient group $\Lambda_{\ell,n}$ represents
all possible ways of moving the rope around (without cutting the rope, and
without letting the rope pass through itself.) \ Of course, unlike a classical
closed piece of rope, a quantum knot can exhibit non-classical behavior, such
as quantum superposition and quantum entanglement. \ 

\ After defining quantum knot type, we investigate quantum observables which
are invariants of quantum knot type. \ Moreover, we also study the
Hamiltonians associated with the generators of the lattice ambient group.

\end{abstract}
\maketitle
\tableofcontents

\section{Introduction}

\bigskip

Throughout this paper, the term "knot" means either a knot or a link. \ For
those unfamiliar with knot theory, we refer them to a quick overview of the
subject given in appendix A.

\bigskip

This paper is a sequel to the research program on quantum knots begun and
defined in \cite{Lomonaco1}. \ This sequel is motivated by the difficulties
encountered by the first author in applying knot theory to physics while
writing the paper \cite{Lomonaco2} on classical electromagnetic knots. \ 

\bigskip

The key difficulty encountered in writing \cite{Lomonaco2} is that physics
"lives" in geometric space and, on the other hand, knot theory "lives" in
topological space. \ As a consequence, in knot theory the inherent geometric
structure of 3-space is often ignored, or simply discarded. If one's objective
is to solve the central problem of knot theory, i.e., the placement problem,
then it is a sound strategy frequently to ignore the unneeded non-pertinent
geometric structure of 3-space. \ However, if one's objective is to use knot
theory as a tool for investigating problems in physics, then this may not be
the best strategy. \ 

\bigskip

Case in point is the set of the Reidemeister moves. \ These moves have become
one of the major corner stones of knot theory. \ They are two dimensional
moves which ignore much of the geometry that is naturally a part of geometric
3-space. \ They do so by focusing on the planar projections of knots. \ For
example, the Reidemeister moves inherently depend on the concept of a knot
crossing. \ However, knots do not have crossings! \ After all, a knot crossing
is simply a "figment" of one's chosen projection.

\bigskip

What is needed for applications to physics is another set of moves that is
more sensitive to the inherent differential geometry of 3-space. \ For that
reason (among others), we will introduce as a possible alternative to the
three Reidemeister moves, three moves called \textbf{wiggle}, \textbf{wag},
and \textbf{tug}.

\bigskip

\section{Part 0. The quest for a more "physics-friendly" set of knot moves}

\bigskip

However, before we can define the three moves, wiggle, wag, and tug, we first
need to gain a better understanding of how a dog wags it tail.

\bigskip

\section{How does a dog wag its tail?}

\bigskip

The first author's best friend Tazi certainly knew how to wag her tail.%
\begin{center}
\includegraphics[
height=1.9778in,
width=2.476in
]%
{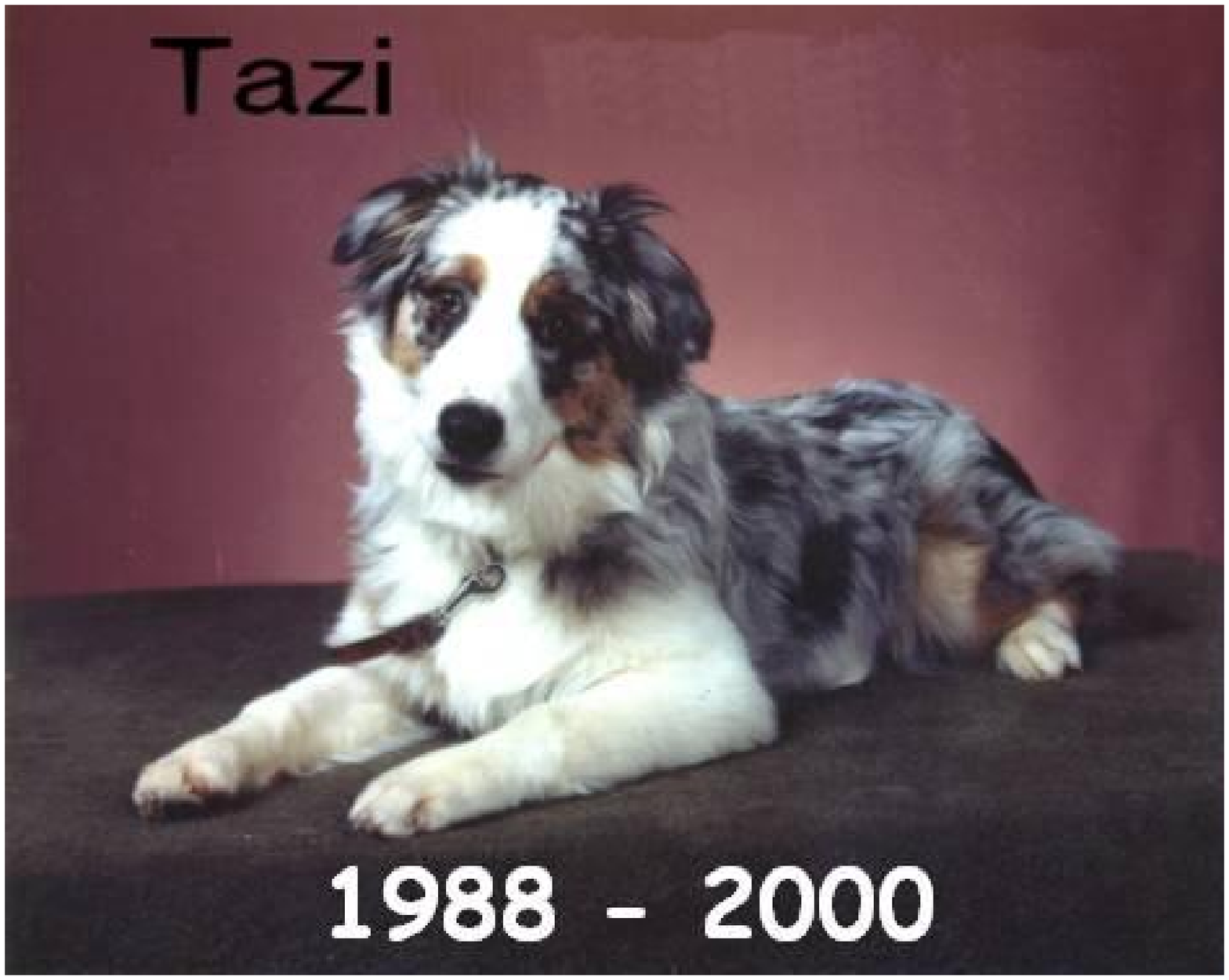}%
\\
\textbf{The first author's best friend Tazi knew the answer.}%
\end{center}

\bigskip

She would \textbf{wiggle} her tail, much as a creature would squirm on a flat
planar surface, such as for example:%
\begin{center}
\includegraphics[
height=0.6287in,
width=0.9487in
]%
{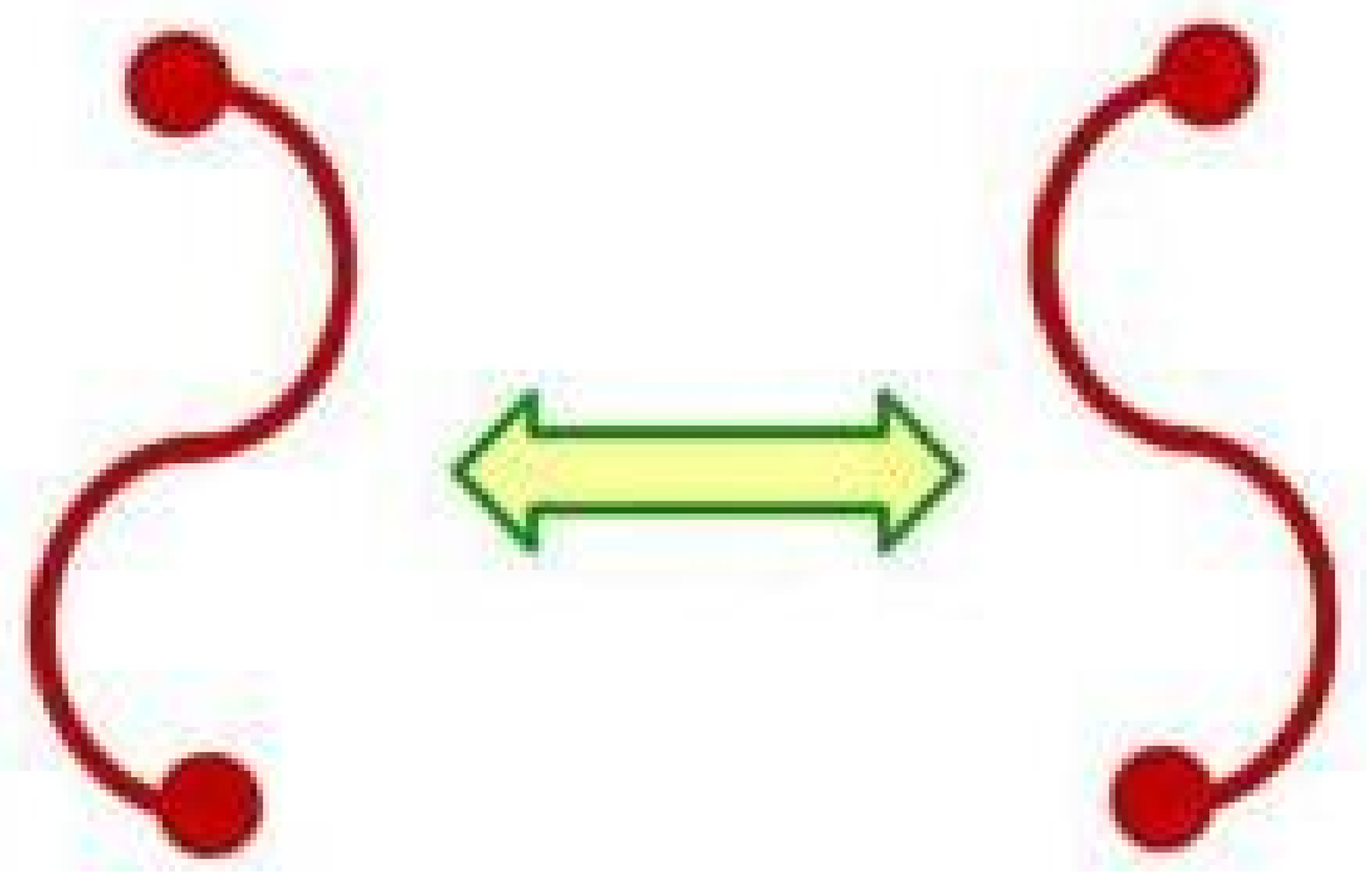}%
\end{center}

She would \textbf{wag} her tail in a twisting corkscrew motion, such as for
example:%
\begin{center}
\includegraphics[
height=0.6287in,
width=0.9063in
]%
{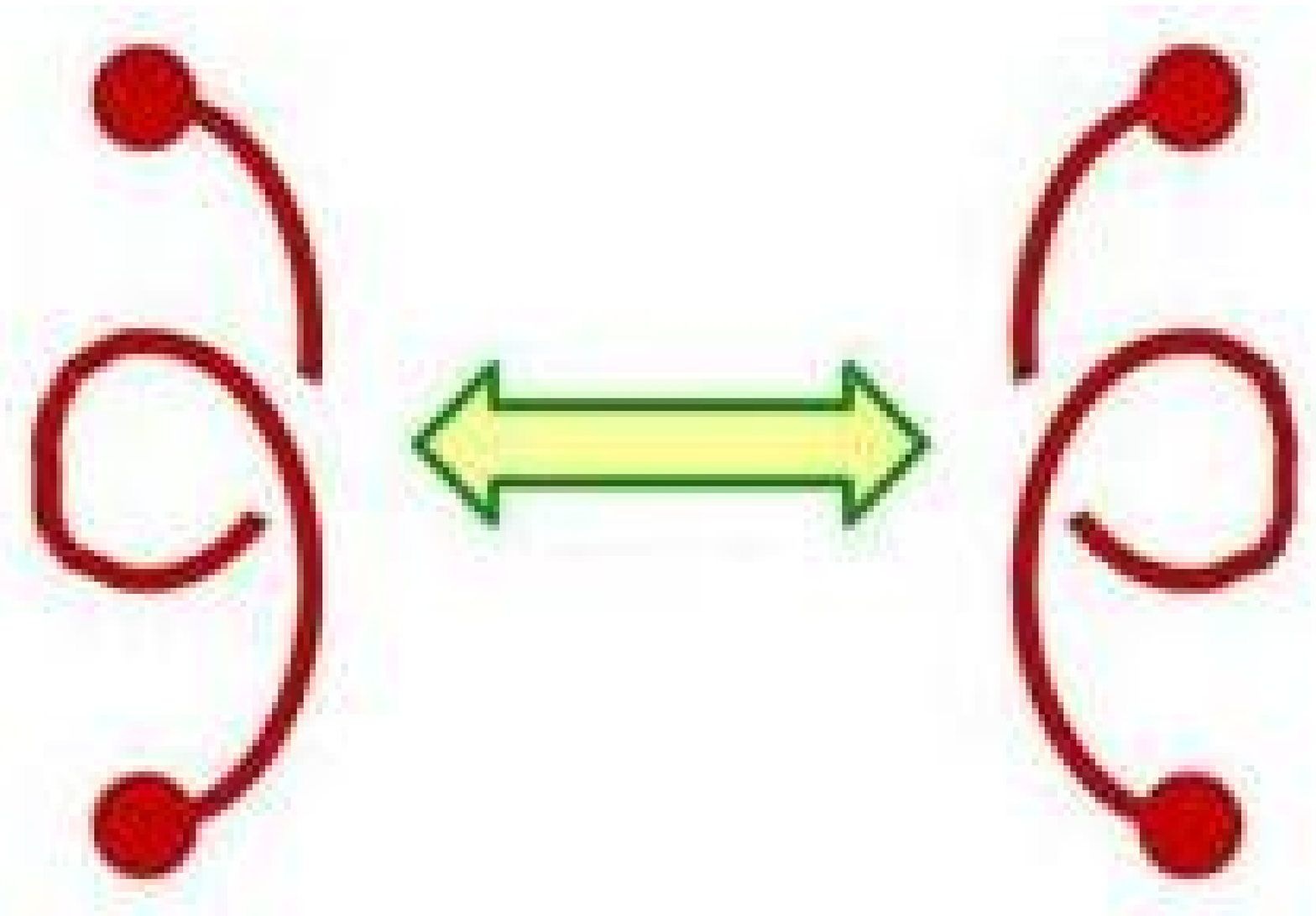}%
\end{center}

Her tail would also stretch or contract when an impolite child would
\textbf{tug} on it, such as for example:%
\begin{center}
\includegraphics[
height=0.5863in,
width=0.8285in
]%
{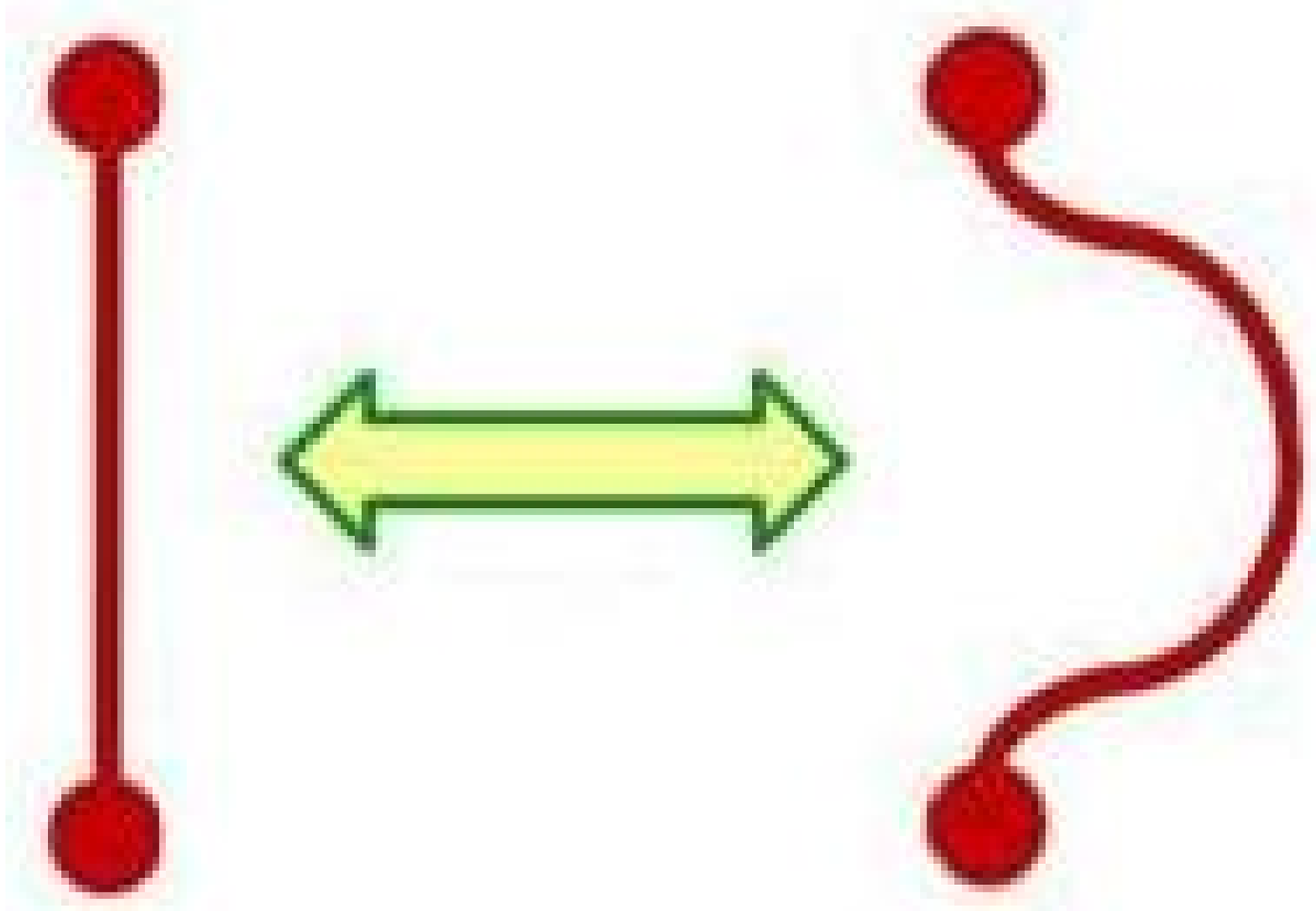}%
\end{center}

Yes, when Tazi moved her tail, she naturally understood how a curve can move
in 3-space. \ She had a keen understanding of the differential geometry of
curves. \ She instinctively understood that each point of a (sufficiently well
behaved) curve in 3-space naturally has associated with it a 3-frame, called
the Frenet frame\footnote{For readers unfamiliar with differential geometry,
please refer to, for example, \cite{DoCarmo1, O'Neill1, Spivak1}.}, consisting
of the unit tangent vector $T$, the unit normal vector $N$, and the unit
binormal vector $B$. \ She instinctively understood that

\begin{itemize}
\item A curve instantaneously bends by rotating about its binormal $B$, as
measured by its curvature $\kappa$,

\item A curve instantaneously twists by rotating about its normal $N$, as
measured by it torsion $\tau$, and

\item A curve instantaneously stretches or contracts along its tangent $T$.
\end{itemize}

%

\begin{center}
\includegraphics[
height=1.1485in,
width=1.6691in
]%
{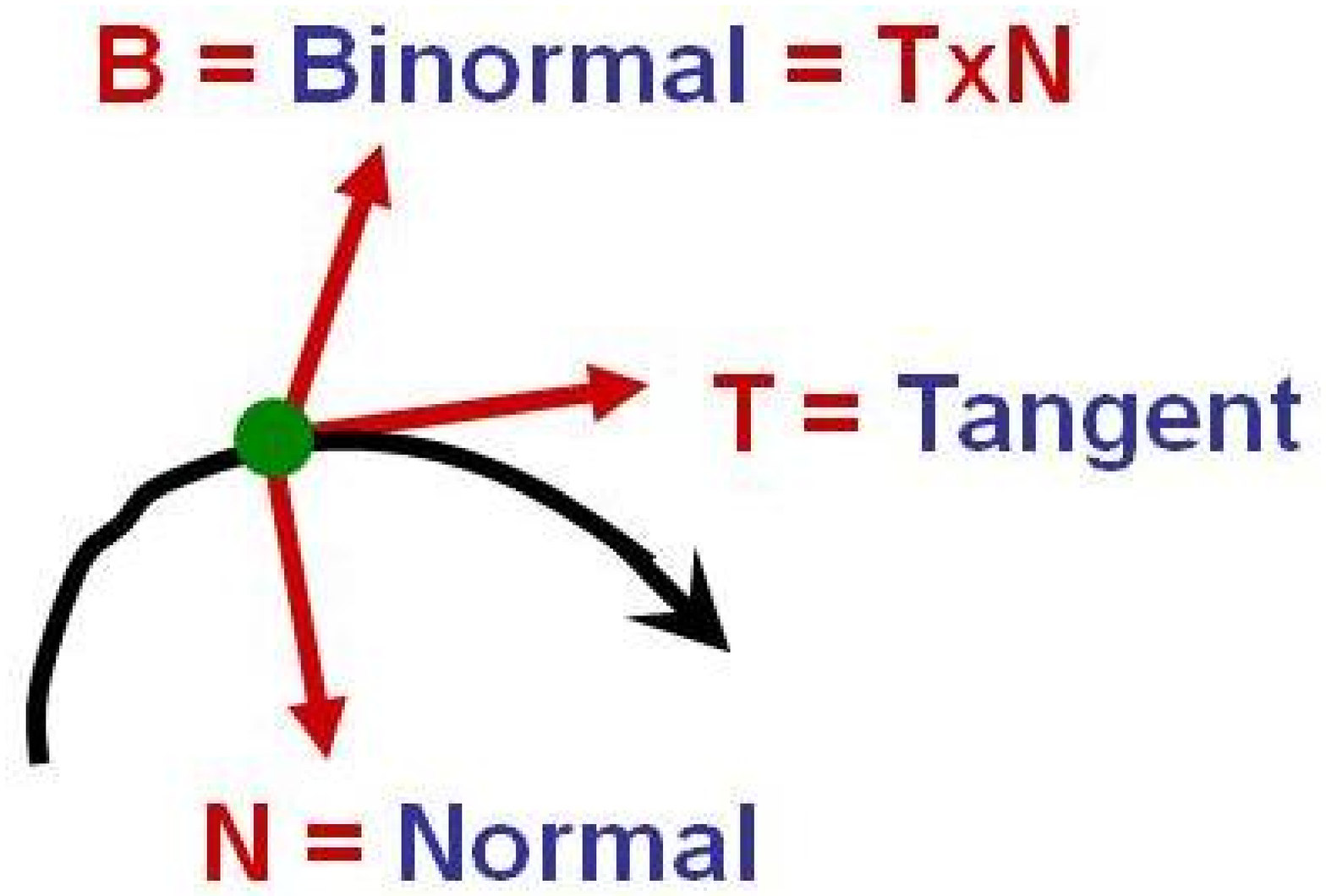}%
\\
\textbf{The Frenet Frame.}%
\end{center}
\bigskip

\bigskip

\noindent\textbf{Key Intuitive Idea:} \textit{A sufficiently well-behaved
curve in 3-space has at each point \textbf{three infinitesimal}
\textbf{degrees of freedom}.}

\bigskip

Tazi understood the key intuitive idea that a curve in 3-space has at each
point\textbf{ three local} (i.e., \textbf{infinitesimal}) \textbf{degrees of
freedom}. \ Can we use this intuition to create a usable well-defined set of
moves which can form a basis for knot theory, much as the Reidemeister moves
have filled that role? \ 

\bigskip

\section{Clues from mechanical engineering}

\bigskip

\noindent\textbf{Question:} Can we transform this intuition into a
mathematically rigorous definition? \ In particular, can we transform the
following intuitive moves into well-defined infinitesimal moves?

\bigskip

\hspace{-0.5in}%
\begin{tabular}
[c]{|c|c|c|}\hline%
{\includegraphics[
height=0.7325in,
width=0.3122in
]%
{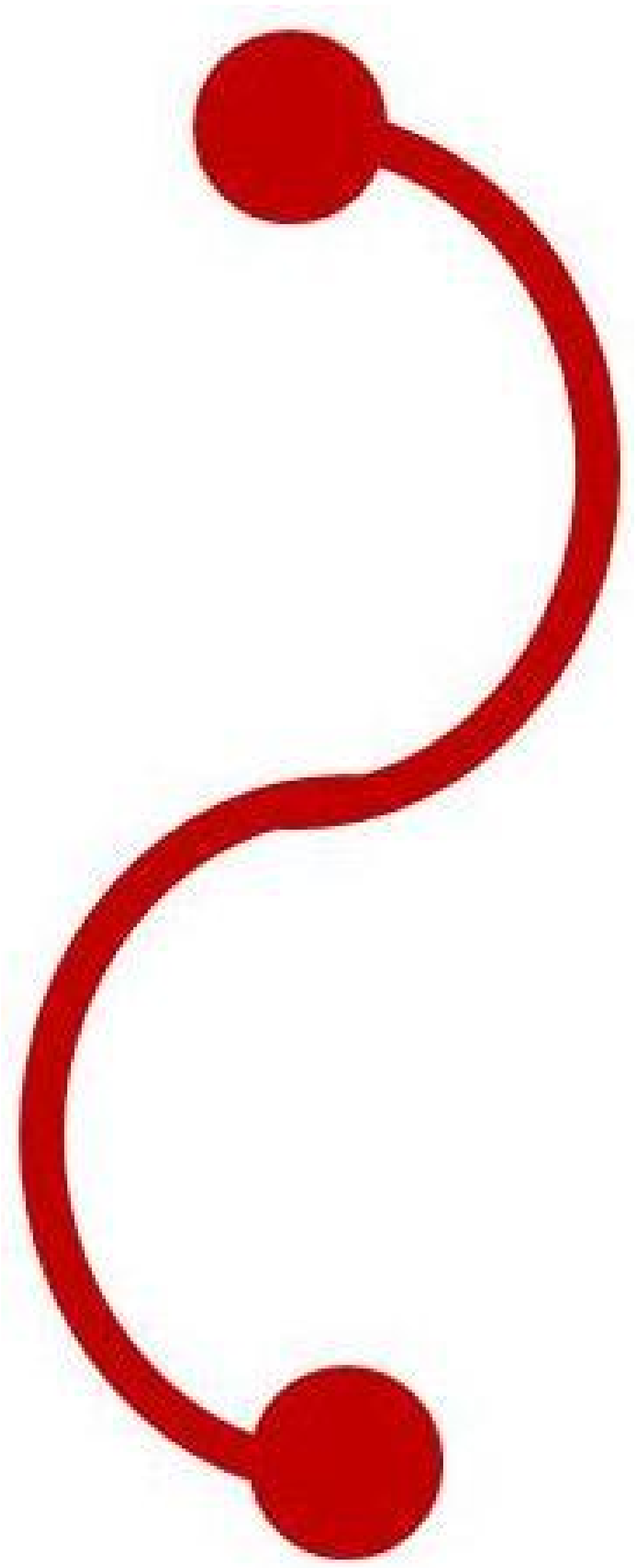}%
}%
& $%
\raisebox{0.2508in}{\includegraphics[
height=0.2361in,
width=0.5967in
]%
{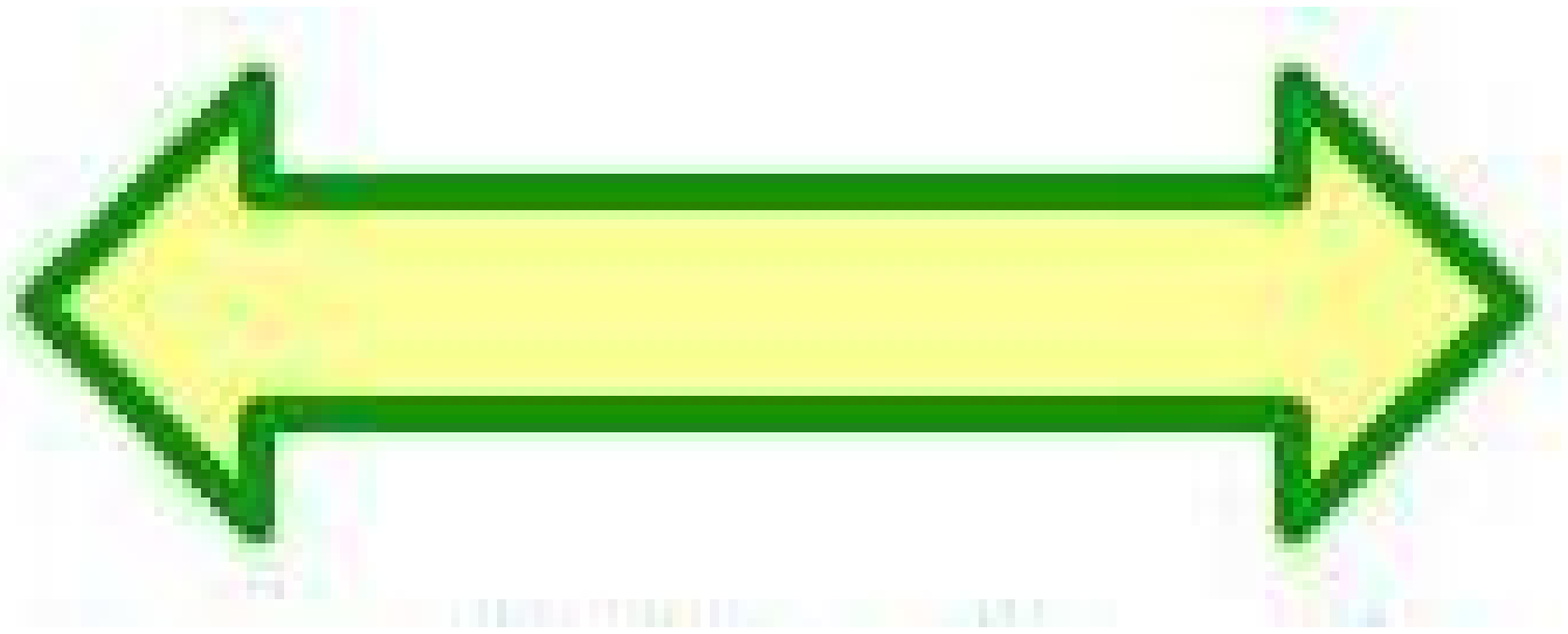}%
}%
$ &
{\includegraphics[
height=0.7247in,
width=0.3061in
]%
{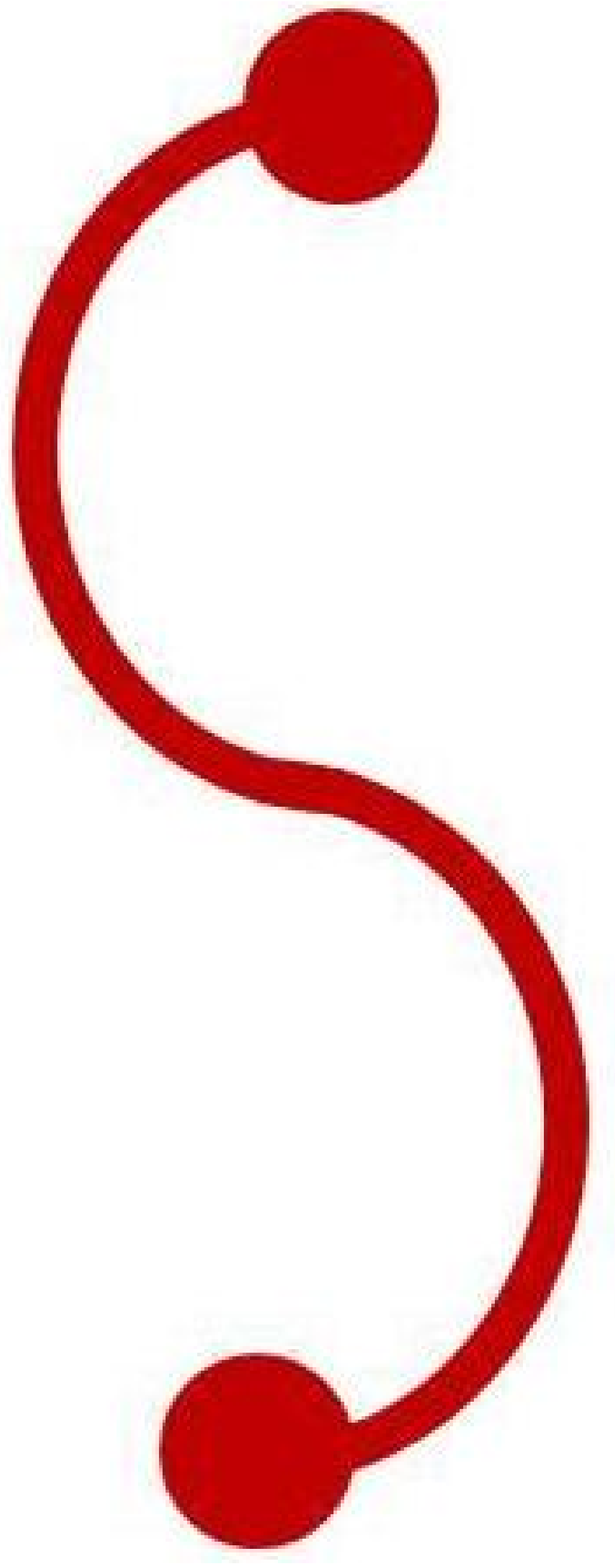}%
}%
\\\hline
\multicolumn{3}{|c|}{%
\begin{tabular}
[c]{c}%
\textbf{Wiggle}\\
\textbf{Curvature }$\kappa$\textbf{ Move}\\
\textbf{Inextensible}%
\end{tabular}
}\\\hline
\end{tabular}
\qquad%
\begin{tabular}
[c]{|c|c|c|}\hline%
{\includegraphics[
height=0.7334in,
width=0.2906in
]%
{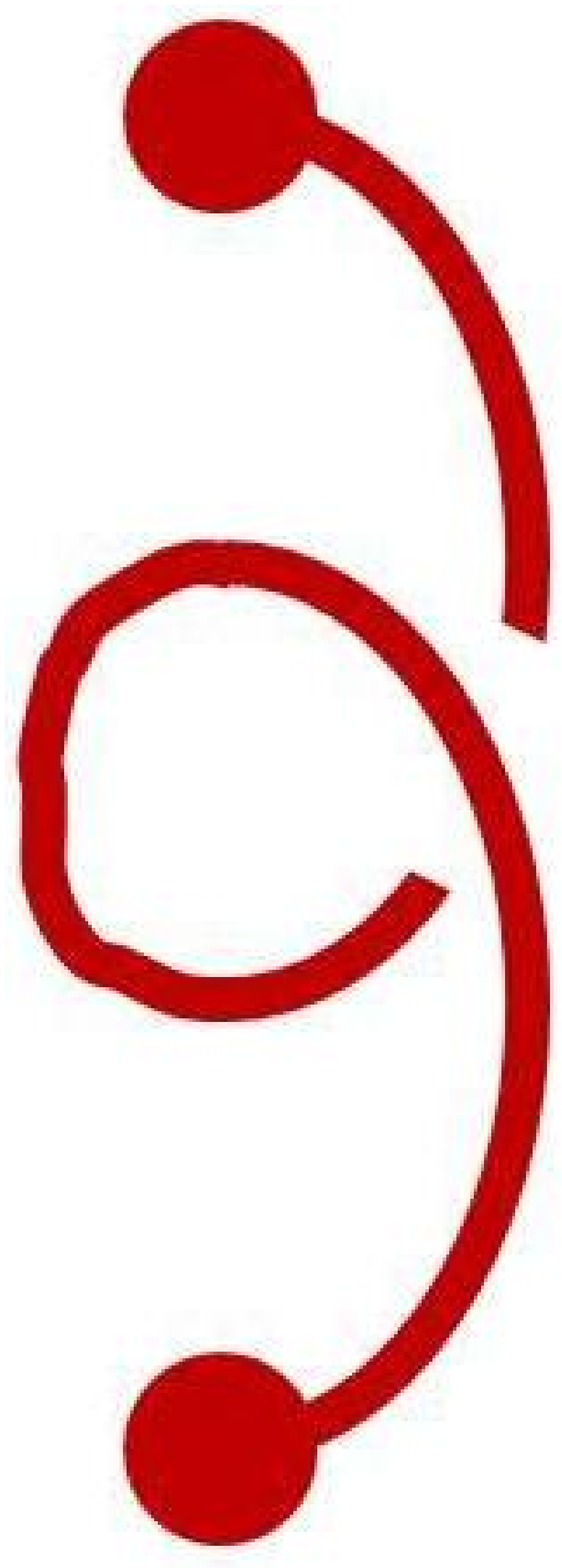}%
}%
& $%
\raisebox{0.2508in}{\includegraphics[
height=0.2361in,
width=0.5967in
]%
{arrow-ya.ps}%
}%
$ &
{\includegraphics[
height=0.7334in,
width=0.2906in
]%
{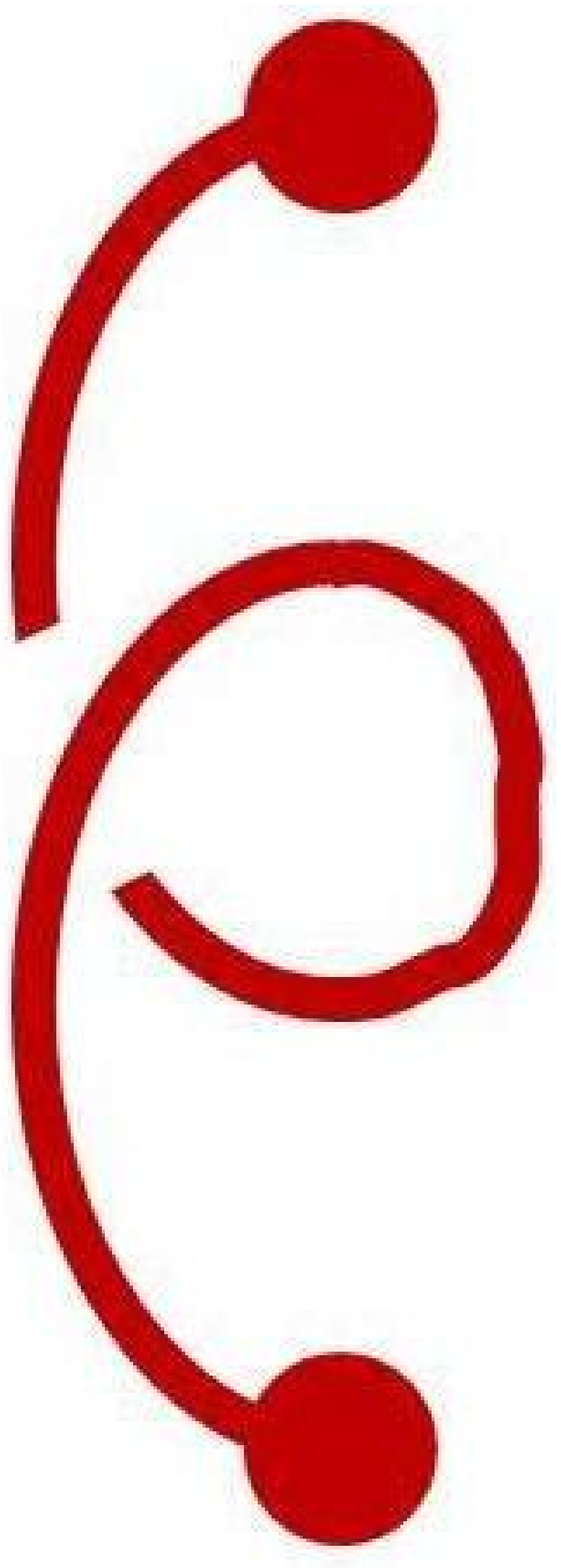}%
}%
\\\hline
\multicolumn{3}{|c|}{%
\begin{tabular}
[c]{c}%
\textbf{Wag}\\
\textbf{Torsion }$\tau$\textbf{ Move}\\
\textbf{Inextensible}%
\end{tabular}
}\\\hline
\end{tabular}
\qquad%
\begin{tabular}
[c]{|c|c|c|}\hline%
{\includegraphics[
height=0.704in,
width=0.1487in
]%
{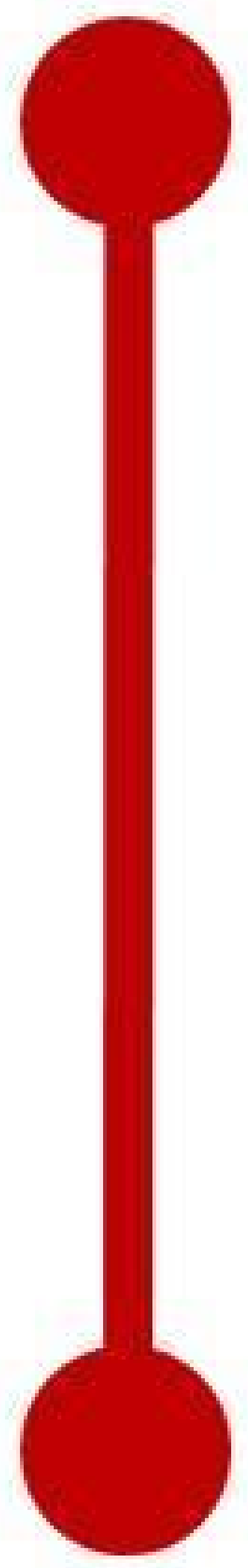}%
}%
& $%
\raisebox{0.2508in}{\includegraphics[
height=0.2361in,
width=0.5967in
]%
{arrow-ya.ps}%
}%
$ &
{\includegraphics[
height=0.6918in,
width=0.3269in
]%
{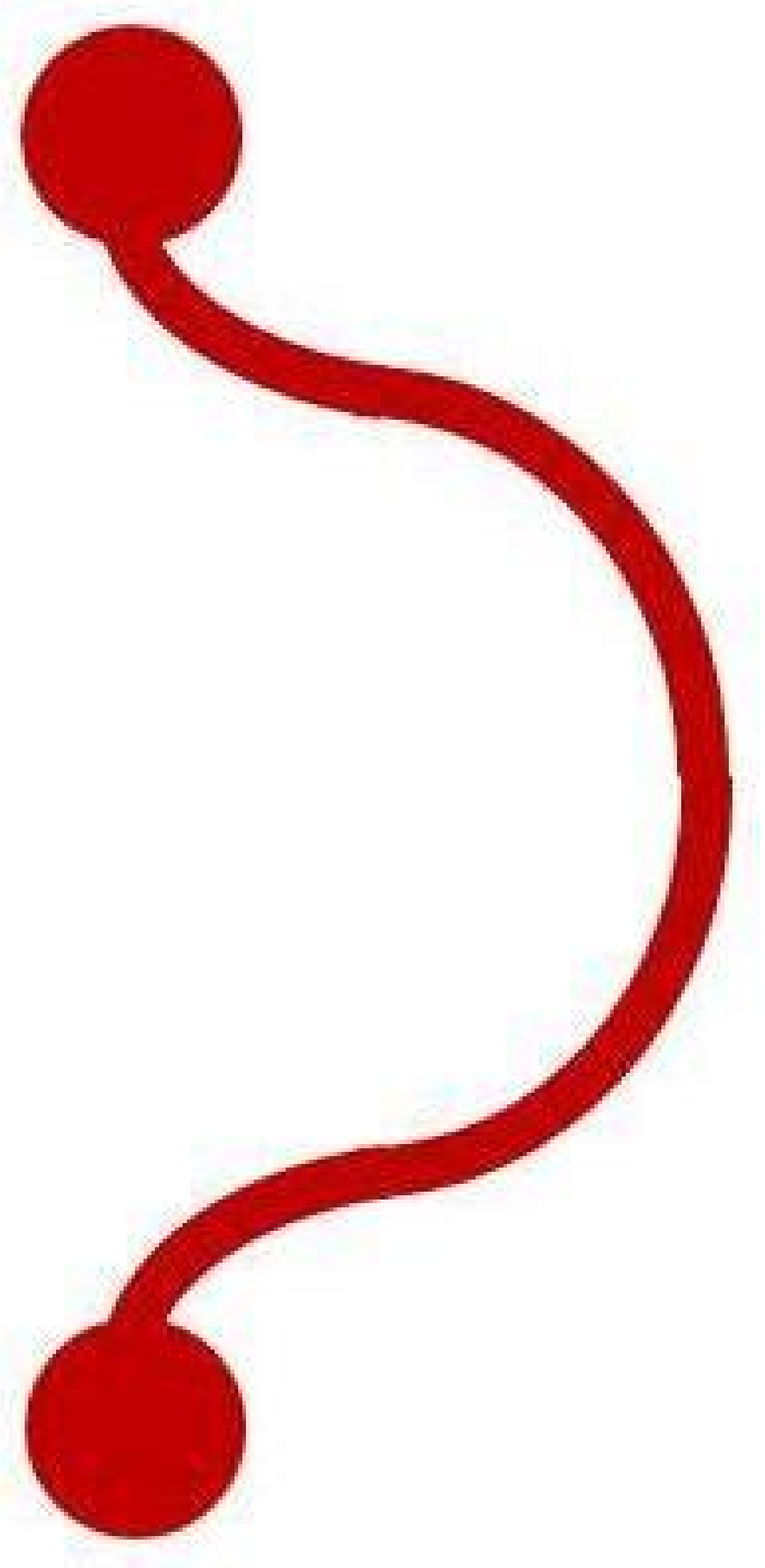}%
}%
\\\hline
\multicolumn{3}{|c|}{%
\begin{tabular}
[c]{c}%
\textbf{Tug}\\
\textbf{Elongation/Contraction Move}\\
\textbf{Extensible}%
\end{tabular}
}\\\hline
\end{tabular}

\bigskip

There are clues from mechanical engineering that suggest a possible approach
to creating a mathematically rigorous definition.

\bigskip

In mechanical engineering, a \textbf{linkage} is a sequence of inextensible
bars (i.e., rods) connected by \textbf{joints}
\begin{center}
\includegraphics[
height=0.5076in,
width=1.5082in
]%
{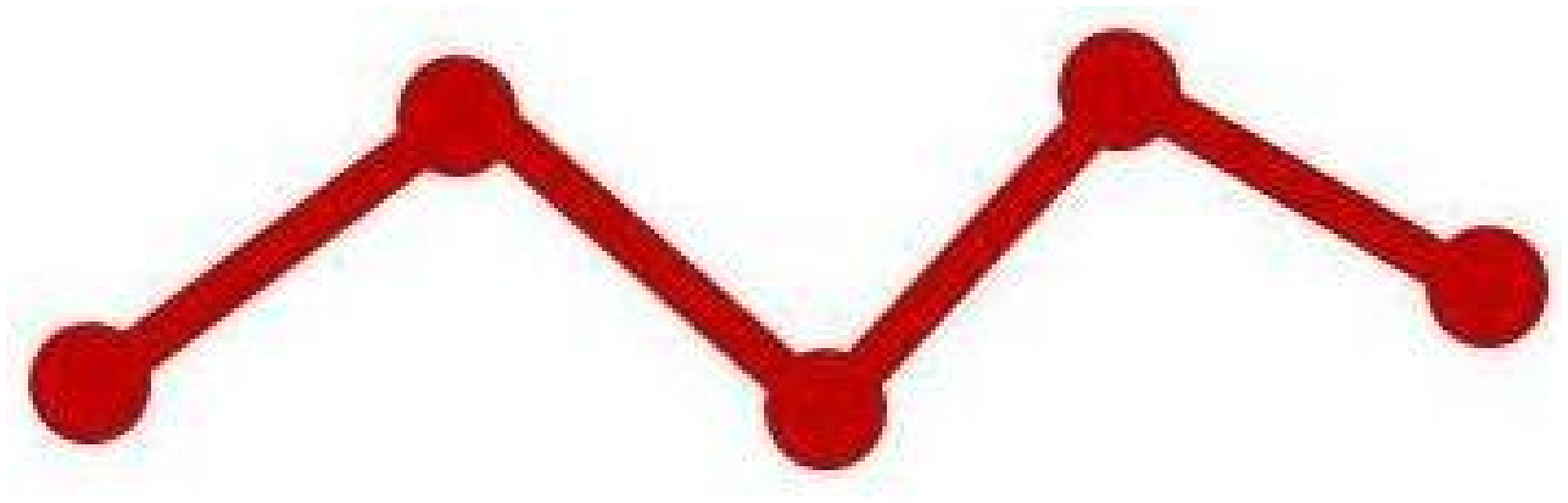}%
\\
\textbf{A mechanical linkage.}%
\end{center}

\bigskip

We will have need of the following three kinds of mechanical joints,
\textbf{planar} (a.k.a. \textbf{revolute}, \textbf{pin}, or \textbf{hinge}),
\textbf{spherical} (a.k.a., \textbf{ball} or \textbf{socket}), and
\textbf{prismatic (}a.k.a. \textbf{slider)}, which are illustrated below:%

\[%
\begin{tabular}
[c]{|c|c|}\hline
$%
\begin{array}
[c]{c}%
\mathstrut\\
\text{\textbf{Joints}}\\
\mathstrut
\end{array}
$ & \\\hline\hline
$%
\begin{array}
[c]{c}%
\mathstrut\\
\text{\textbf{Planar}}\\
\mathstrut
\end{array}
$ & $%
\begin{array}
[c]{c}%
\mathstrut\\%
{\includegraphics[
height=0.4281in,
width=0.5483in
]%
{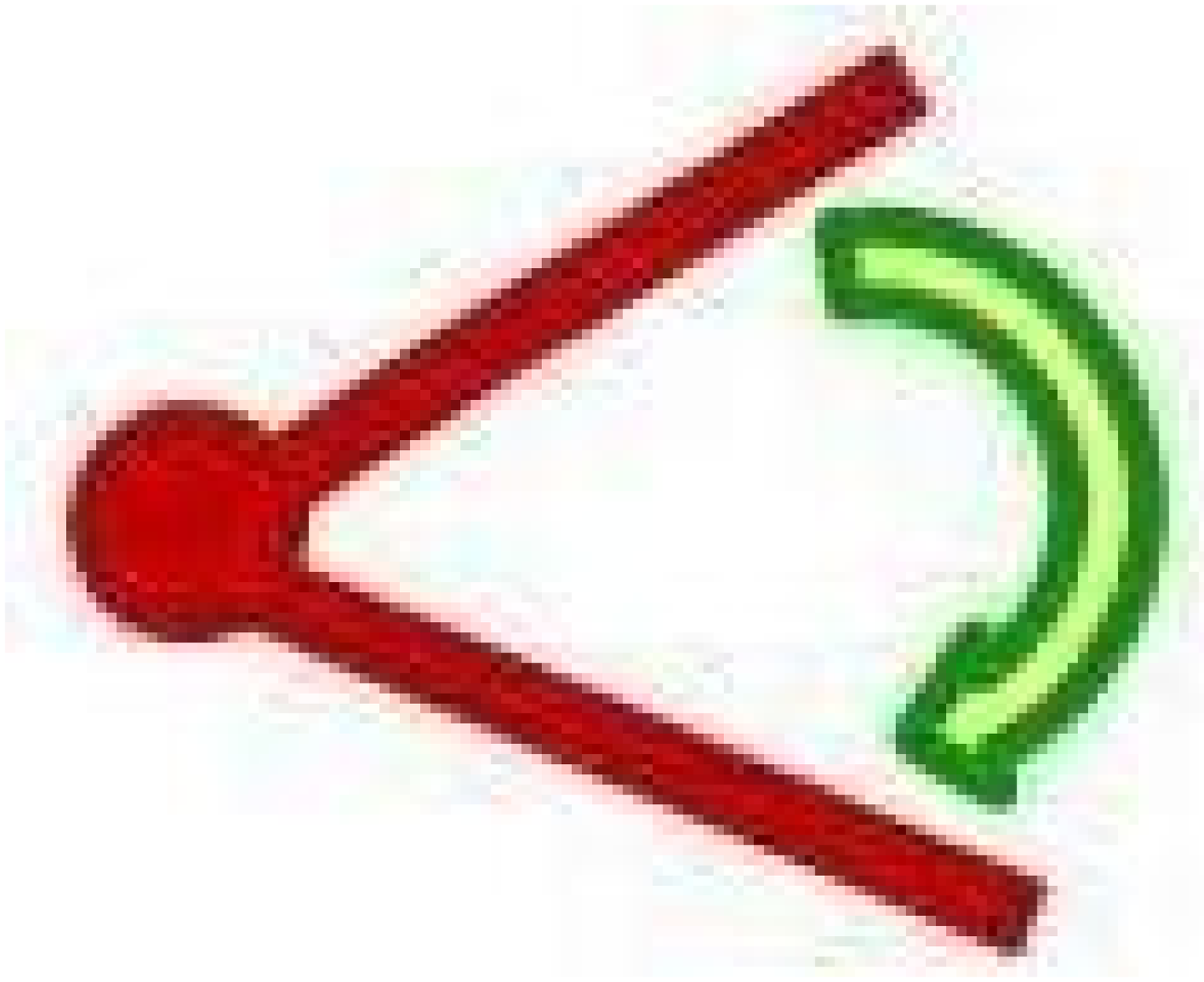}%
}%
\\
\mathstrut
\end{array}
$\\\hline
$%
\begin{array}
[c]{c}%
\\
\text{\textbf{Spherical}}\\
\end{array}
$ & $%
\begin{array}
[c]{c}%
\mathstrut\\%
{\includegraphics[
height=0.6287in,
width=0.8285in
]%
{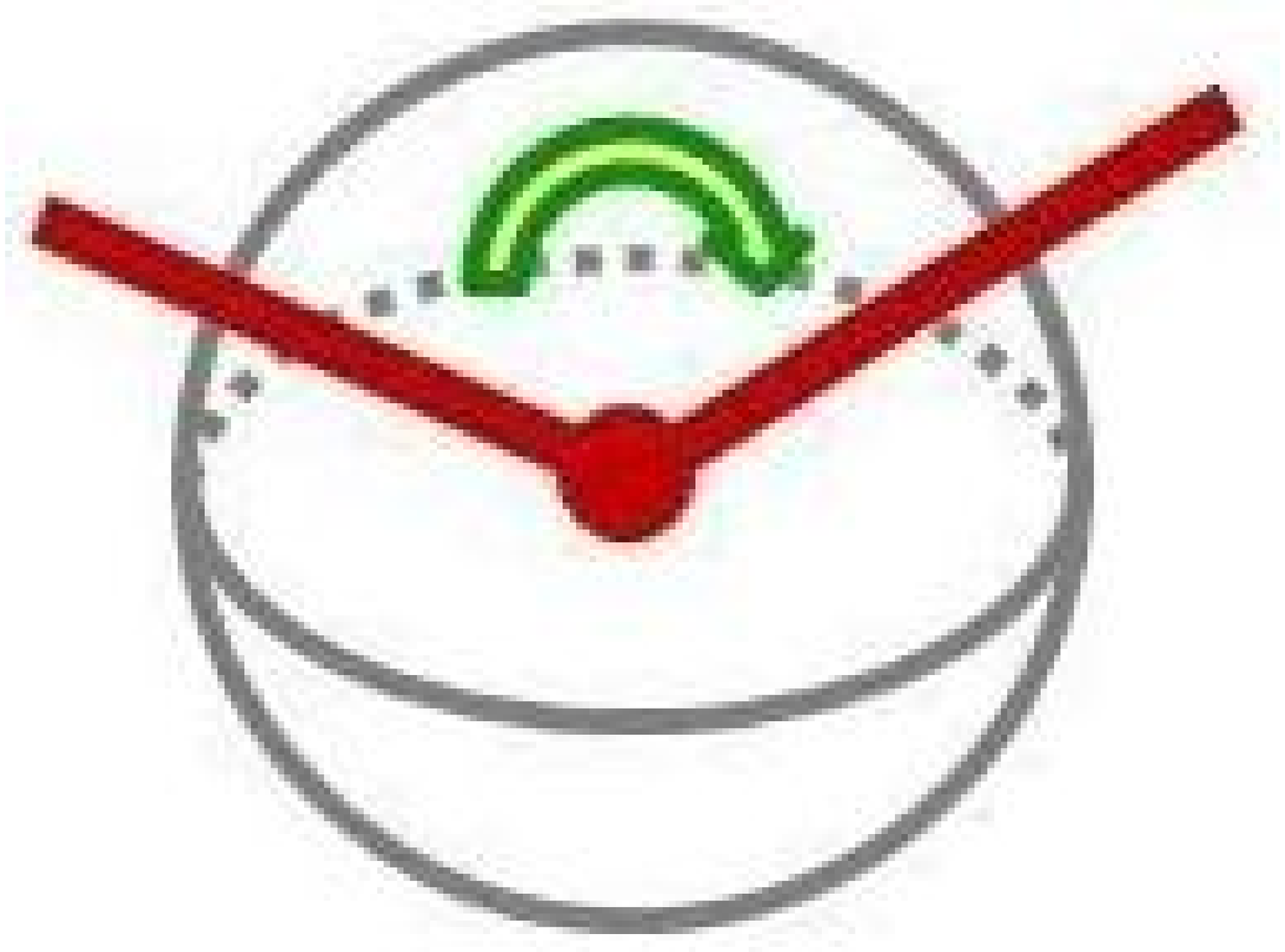}%
}%
\\
\mathstrut
\end{array}
$\\\hline
$%
\begin{array}
[c]{c}%
\\
\text{\textbf{prismatic}}\\
\end{array}
$ & $%
\begin{array}
[c]{c}%
\mathstrut\\%
{\includegraphics[
height=0.3061in,
width=2.7484in
]%
{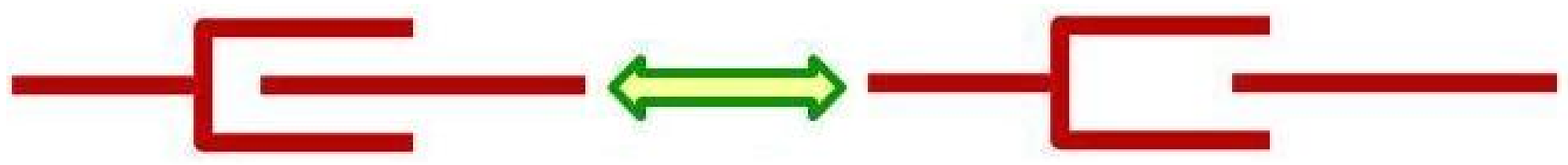}%
}%
\\
\mathstrut
\end{array}
$\\\hline
\end{tabular}
\ \ \
\]

\bigskip

In mechanical engineering, a \textbf{mechanism} is a linkage with one degree
of freedom. \ We will consider the following three mechanisms:\bigskip

\begin{itemize}
\item \textbf{The 4-Bar Linkage}: \ The 4-bar linkage, illustrated below, has
exactly one degree of freedom.%
\begin{center}
\fbox{\includegraphics[
height=1.2263in,
width=2.9075in
]%
{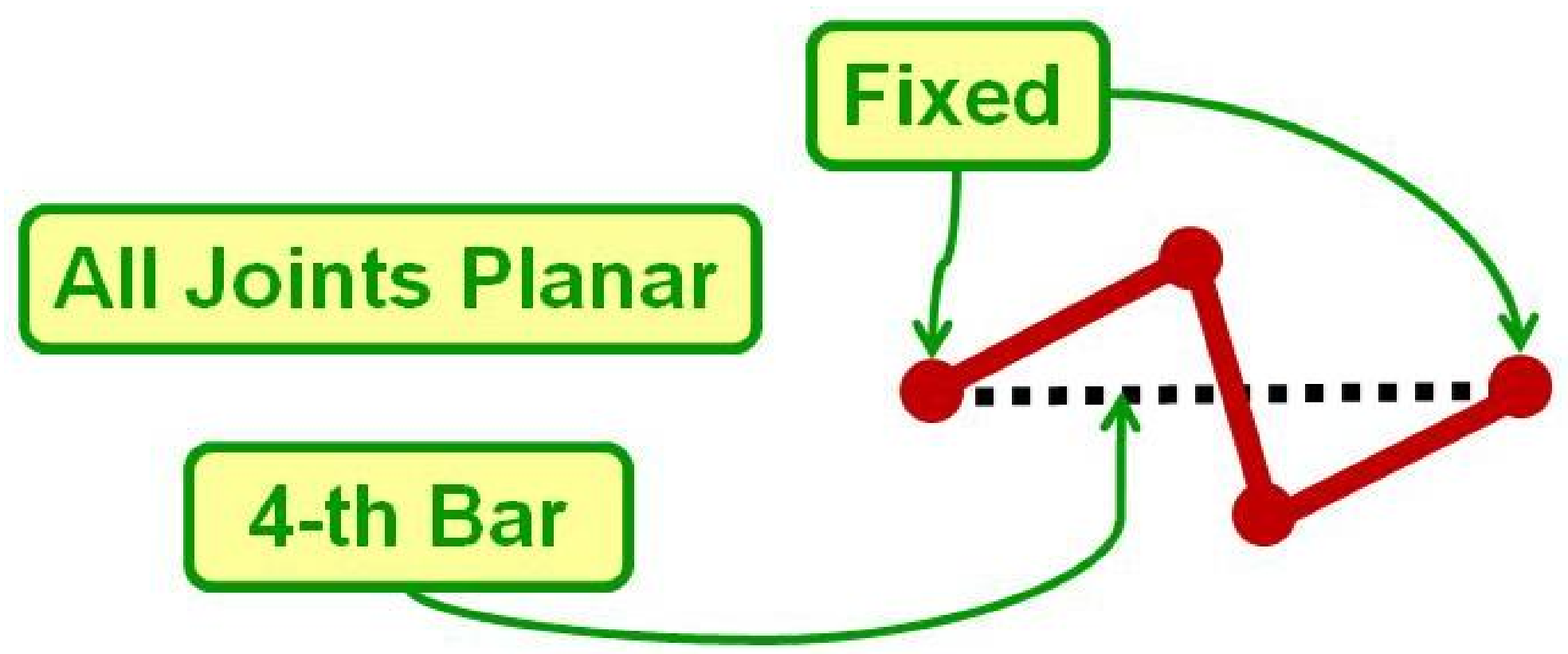}%
}
\end{center}
All joints in this linkage are planar. \ Since the leftmost and rightmost
joints are fixed (stationary), the missing fourth bar is effectively the
dotted line shown in the figure. \ If the leftmost and the rightmost joints
are connected to other linkages, then movement of the 4-bar linkage does not
effect any bars of the larger composite linkage other than the above three red
bars. \ In other words, the 4-bar linkage is a \textbf{local move} on a larger
linkage. \ \bigskip

In particular, the 4-bar linkage gives an illustration of a \textbf{local
curvature move}, taking place in a plane. \ We will call this local move a
\textbf{wiggle}. \bigskip

\item \textbf{The 3-Bar Linkage:} \ The 3-bar linkage, illustrated below, has
one degree of freedom.
\begin{center}
\fbox{\includegraphics[
height=1.3076in,
width=2.5858in
]%
{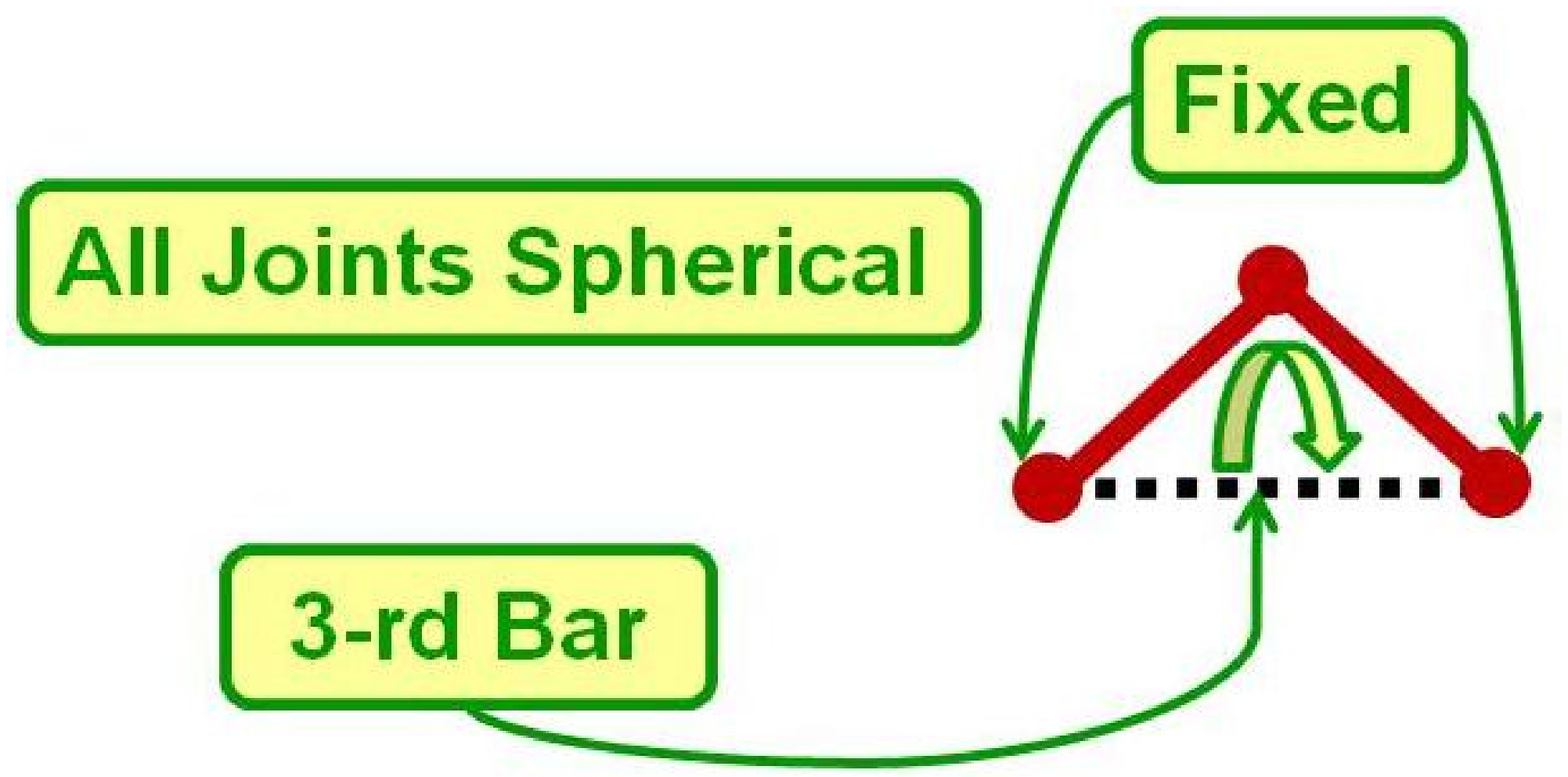}%
}
\end{center}
All joints in this linkage are spherical. \ Since its outermost joints are
fixed (stationary), the missing third bar is effectively the dotted line shown
in the figure. \ If the outermost joints are connected to other linkages, then
movement of the 3-bar linkage does not effect any bars of the larger composite
linkage other than the above two red bars. \ In other words, the 3-bar linkage
is a \textbf{local move} on a larger linkage. \ \bigskip

This is a \textbf{local torsion move}, locally twisting a portion of the
linkage into a new plane. \ We call it a \textbf{wag}.\bigskip

\item \textbf{The 4-Bar Slider:} \ The 4-bar slider, illustrated below, has
one degree of freedom.%
\begin{center}
\fbox{\includegraphics[
height=1.5074in,
width=3.5483in
]%
{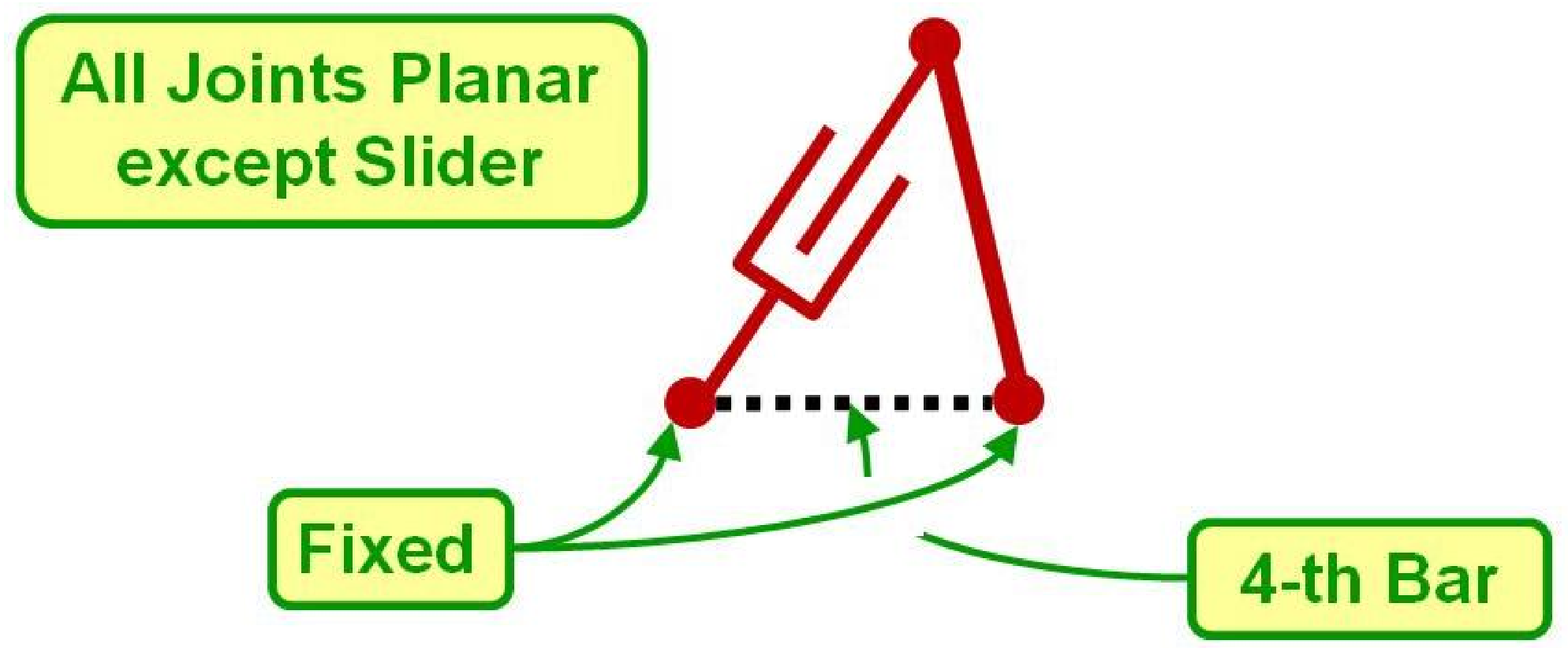}%
}
\end{center}
All joints in this linkage are planar except for the prismatic joint. \ Since
the outermost joints are fixed (stationary), the missing fourth bar is
effectively the dotted line shown in the figure. \ If the outermost joints are
connected to other linkages, then movement of the 4-bar slider does not effect
any bars of the larger composite linkage other than the above three red bars.
\ Thus, the 4-bar slider can be thought of as a \textbf{local move} on a
larger linkage. \ \bigskip

This is a \textbf{local expansion/contraction move}, taking place in a fixed
plane. \ We will call it a \textbf{tug}.
\end{itemize}

\bigskip

Before closing this section, we should mention three striking examples of
linkages. \ The first is the \textbf{Tangle} \cite{Zawitz1}, invented by
Richard E. Zawitz, and shown in the figure given below:\bigskip%
\begin{center}
\includegraphics[
height=1.5904in,
width=1.5904in
]%
{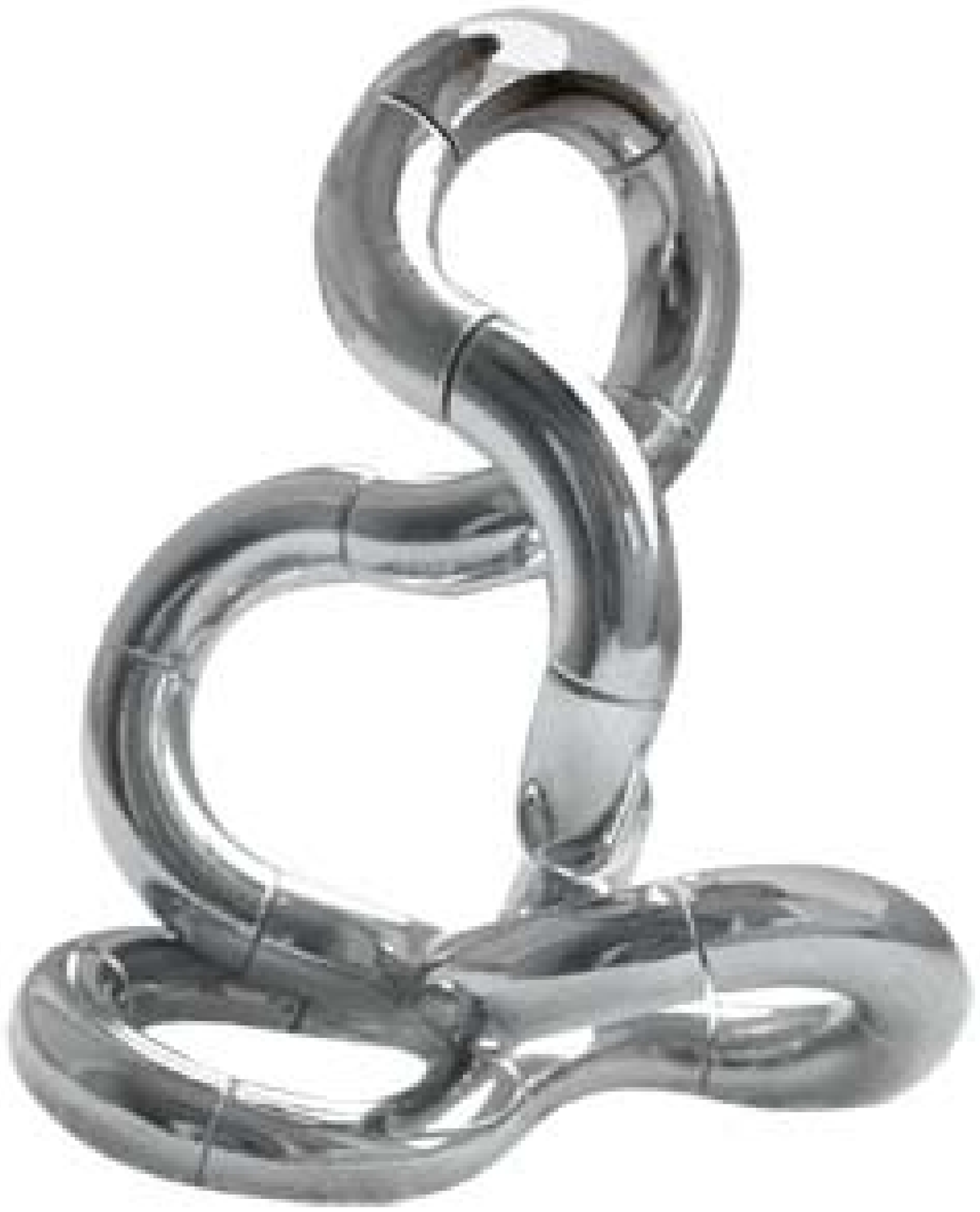}%
\\
\textbf{Zawitz's Tangle}$^{\textregistered}$\textbf{ moves only by wagging.}%
\end{center}
The Tangle is a linkage with only one local degree of freedom. \ It moves only
by wagging.

\bigskip

The second and third linkages are the \textbf{Bendandle}$^{\text{TM}}$ and the
\textbf{Universal Bendangle}$^{\text{TM}}$, invented by Samuel J. Lomonaco
(patents pending), and shown respectively in the two figures given below:

\bigskip

\quad%
{\parbox[b]{2.2044in}{\begin{center}
\includegraphics[
height=1.6596in,
width=2.2044in
]%
{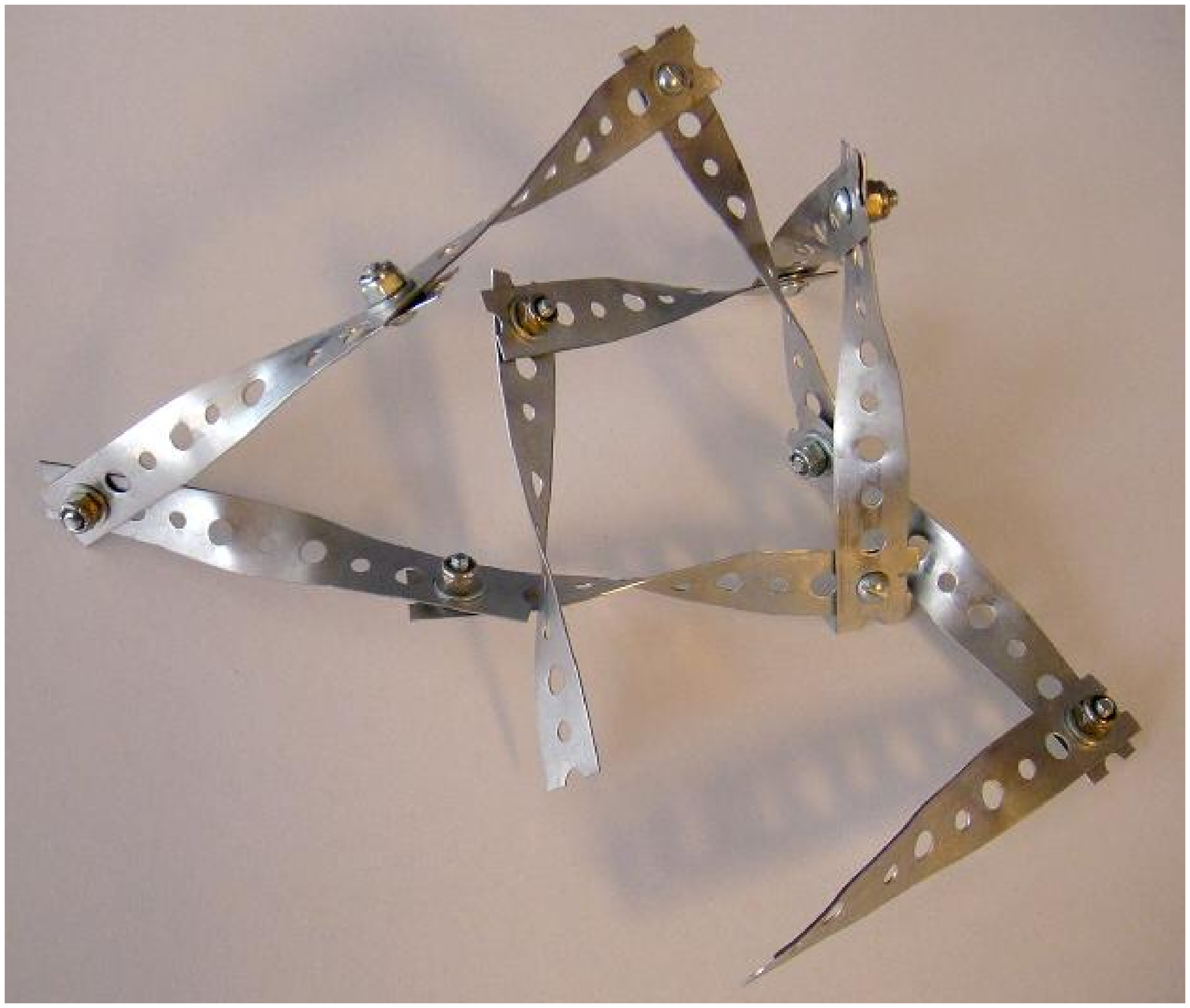}%
\\
\textbf{Lomonaco's Bendangle}$^{\text{\textbf{TM}}}$\textbf{ moves only by
wiggling. \ (Patent pending)}%
\end{center}}}%
\qquad\qquad%
{\parbox[b]{2.2044in}{\begin{center}
\includegraphics[
height=1.6596in,
width=2.2044in
]%
{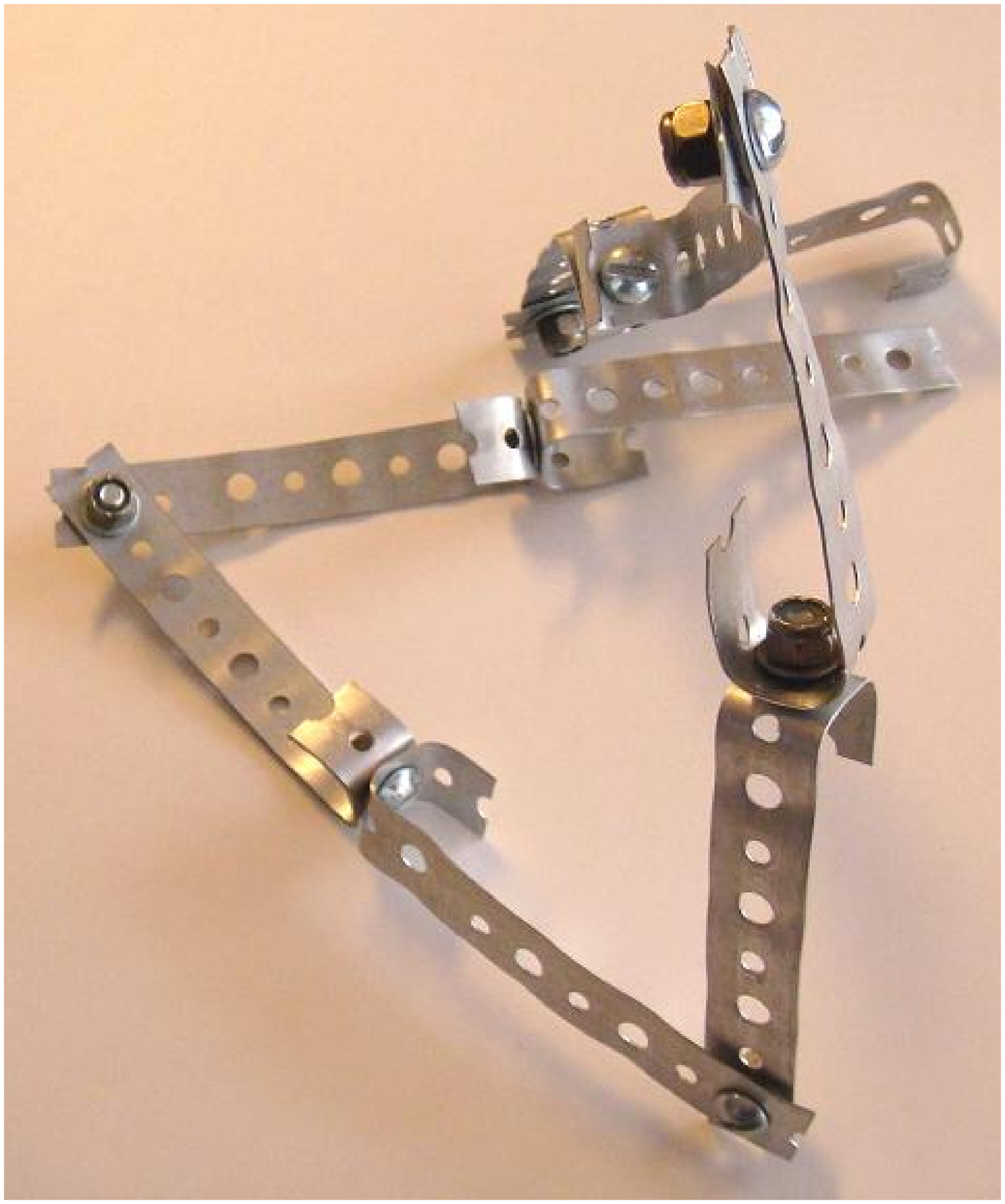}%
\\
\textbf{Lomonaco's Universal Bendangle}$^{\text{\textbf{TM}}}$\textbf{ moves
only by wiggling and wagging. (Patent pending.)}%
\end{center}}}%

\bigskip

The Bendangle also has only one local degree of freedom, but in this case,
moves only by bending. \ The Universal Bendangle, as its name suggests, has
two local degrees of freedom, moving only by bending and twisting.

\bigskip

In some sense, these three examples further support the key intuition that
curves in 3-space have in some sense three local degrees of freedom.

\bigskip

\section{A translation of mechanical engineering into knot theory}

\bigskip

Let us now translate mechanical engineering into knot theory.

\bigskip

\begin{definition}
Two piecewise linear (PL) knots $K_{1}$ and $K_{2}$ are said to be of the
\textbf{same knot type}, written%
\[
K_{1}\thicksim K_{2}\text{ ,}%
\]
provided there exist finite subdivisions $K_{1}^{\prime}$ and $K_{2}^{\prime}$
of $K_{1}$ and $K_{2}$ respectively such that one can be transformed into the
other by a finite sequence of the following three \textbf{local moves}:
\end{definition}

\begin{itemize}
\item[\textbf{1)}] A \textbf{tug}: \
\begin{center}
\includegraphics[
height=0.4281in,
width=2.0695in
]%
{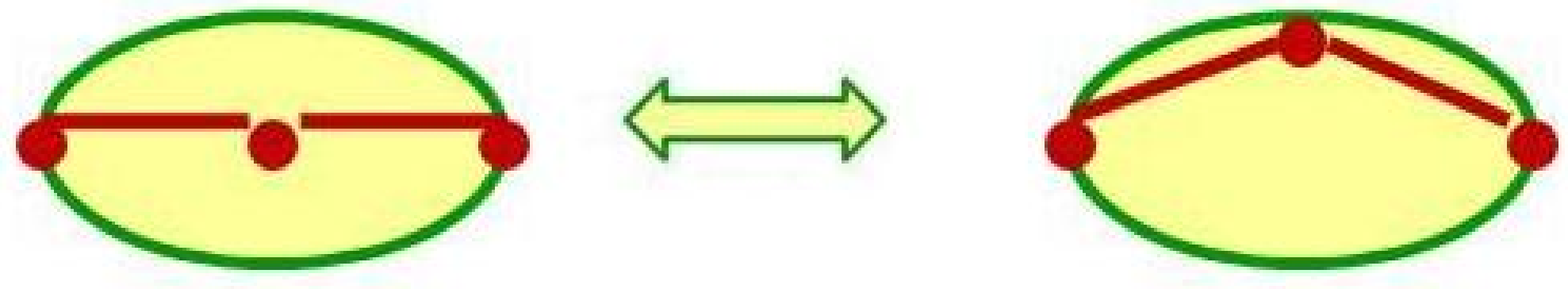}%
\end{center}
\bigskip

\item[\textbf{2)}] A \textbf{wiggle}:
\begin{center}
\includegraphics[
height=0.4281in,
width=2.028in
]%
{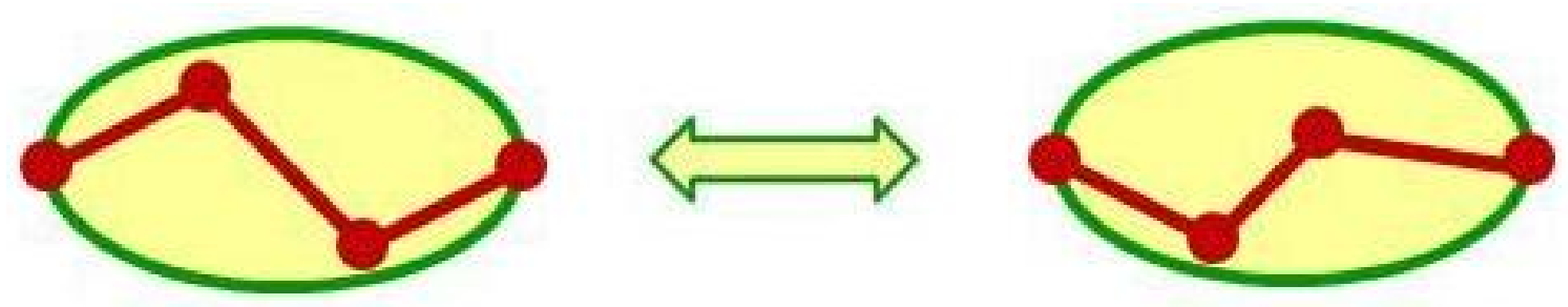}%
\end{center}
\bigskip

\item[\textbf{3)}] A \textbf{wag}:
\begin{center}
\includegraphics[
height=0.9063in,
width=2.348in
]%
{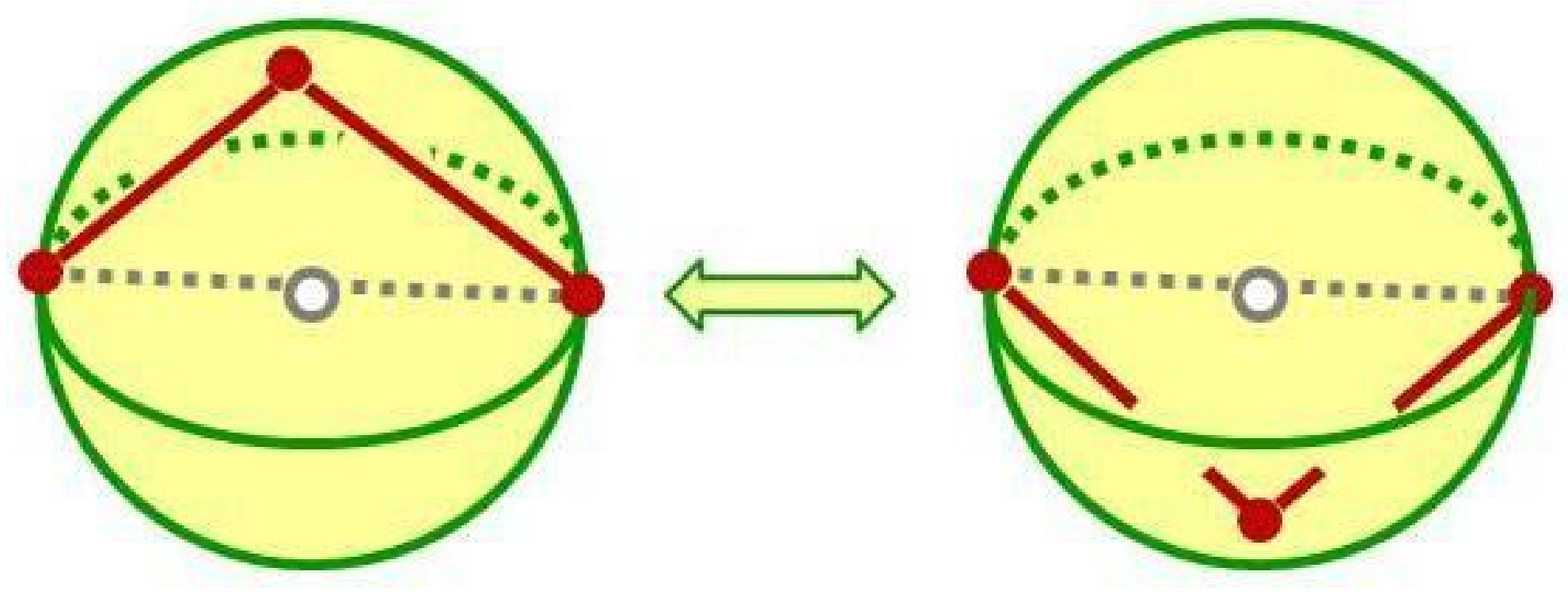}%
\end{center}

\end{itemize}

\bigskip

Using the methods found in Reidemeister's proof of the completeness of the
Reidemeister moves, we have:

\bigskip

\begin{theorem}
Wiggles and wags can be expressed as finite sequences of tugs.
\end{theorem}

\bigskip

\begin{remark}
In fact, Reidemeister's fundamental move, i.e., his \textbf{triangle move}, is
essentially a tug. \ (See \cite{Reidemeister1, Reidemeister2}.)
\end{remark}

\bigskip

So it would appear that we have accomplished nothing!

\bigskip

But on the contrary, we have indeed accomplished something after all. \ For we
are now in a position to alter knot theory in such a way as to incorporate
more of the geometry of 3-space. \ The telltale clue is that wiggle and wag
are \textbf{inextensible moves}, while tug is not. By an \textbf{inextensible
move}, we mean one that does not locally change the length of a curve (and
hence, preserves global length.)

\bigskip

\begin{definition}
Two piecewise linear (PL) knots $K_{1}$ and $K_{2}$ are said to be of the
\textbf{same inextensible knot type}, written
\[
K_{1}\approx K_{2}\text{ ,}%
\]
provided there exist finite subdivisions $K_{1}^{\prime}$ and $K_{2}^{\prime}$
of $K_{1}$ and $K_{2}$ respectively such that one can be transformed into the
other by applying a finite sequence of wiggles and wags.
\end{definition}

\bigskip

\begin{theorem}
Two PL knots $K_{1}$ and $K_{2}$ are of the same knot type if and only if they
have \bigskip

\begin{itemize}
\item[\textbf{1)}] The same inextensible knot type, i.e. $K_{1}\approx K_{2}$
, and\bigskip

\item[\textbf{2)}] The same length, i.e., $\left\vert K_{1}\right\vert
=\left\vert K_{2}\right\vert $ .
\end{itemize}
\end{theorem}

\bigskip

Thus, nothing from classical knot theory is lost with the above modified
definition of knot type. \ But on the other hand, with this modified
definition, we have succeeded in incorporating more of the geometry of 3-space
into knot theory!\footnote{For a more detailed justification of this
definition, we refer the reader to Section 16 of this paper.}

\bigskip

\section{Part 1: Lattice Knots}

\bigskip

\bigskip Lest we forget, one of our objectives is to create a firm
mathematical foundation for the intuition that sufficiently well-behaved
curves in 3-space have three local (i.e., infinitesimal) degrees of freedom.
\ We would like to answer the following question:

\bigskip

\noindent\textbf{Question:} Can we transform the following intuitive moves
into well-defined \textbf{infinitesimal} moves?

\bigskip

\hspace{-0.5in}%
\begin{tabular}
[c]{|c|c|c|}\hline%
{\includegraphics[
height=0.7325in,
width=0.3122in
]%
{i-wigglel.ps}%
}%
& $%
\raisebox{0.2508in}{\includegraphics[
height=0.2361in,
width=0.5967in
]%
{arrow-ya.ps}%
}%
$ &
{\includegraphics[
height=0.7247in,
width=0.3061in
]%
{i-wiggler.ps}%
}%
\\\hline
\multicolumn{3}{|c|}{%
\begin{tabular}
[c]{c}%
\textbf{Wiggle}\\
\textbf{Curvature }$\kappa$\textbf{ Move}\\
\textbf{Inextensible}%
\end{tabular}
}\\\hline
\end{tabular}
\qquad%
\begin{tabular}
[c]{|c|c|c|}\hline%
{\includegraphics[
height=0.7334in,
width=0.2906in
]%
{i-wagglel.ps}%
}%
& $%
\raisebox{0.2508in}{\includegraphics[
height=0.2361in,
width=0.5967in
]%
{arrow-ya.ps}%
}%
$ &
{\includegraphics[
height=0.7334in,
width=0.2906in
]%
{i-waggler.ps}%
}%
\\\hline
\multicolumn{3}{|c|}{%
\begin{tabular}
[c]{c}%
\textbf{Wag}\\
\textbf{Torsion }$\tau$\textbf{ Move}\\
\textbf{Inextensible}%
\end{tabular}
}\\\hline
\end{tabular}
\qquad%
\begin{tabular}
[c]{|c|c|c|}\hline%
{\includegraphics[
height=0.704in,
width=0.1487in
]%
{i-tugglel.ps}%
}%
& $%
\raisebox{0.2508in}{\includegraphics[
height=0.2361in,
width=0.5967in
]%
{arrow-ya.ps}%
}%
$ &
{\includegraphics[
height=0.6918in,
width=0.3269in
]%
{i-tuggler.ps}%
}%
\\\hline
\multicolumn{3}{|c|}{%
\begin{tabular}
[c]{c}%
\textbf{Tug}\\
\textbf{Elongation/Contraction Move}\\
\textbf{Extensible}%
\end{tabular}
}\\\hline
\end{tabular}

\bigskip

The easiest way to answer this question is to use a "scaffolding" for 3-space,
i.e., the so called cubic honeycomb.

\bigskip

\section{Lattice knots}

\bigskip

For each non-negative integer $\ell$, let $\mathcal{L}_{\ell}$ denote the
three dimensional \textbf{lattice} of points
\[
\mathcal{L}_{\ell}=\mathbb{\ }\frac{1}{2^{\ell}}\mathbb{Z}\times
\mathbb{Z}\times\mathbb{Z}=\left\{  \left(  \frac{m_{1}}{2^{\ell}},\frac
{m_{2}}{2^{\ell}},\frac{m_{3}}{2^{\ell}}\right)  :m_{1},m_{2},m_{3}%
\in\mathbb{Z}\right\}  \text{ ,}%
\]
lying in Euclidean 3-space $\mathbb{R}^{3}$, where $\mathbb{Z}$ denotes the
set of rational integers. \ This lattice determines a tiling of $\mathbb{R}%
^{3}$ by cubes of edge $2^{\ell}$, called the \textbf{cubic honeycomb}
(a.k.a., the \textbf{cubic tesselation}) \textbf{of order} $\ell$.%
\begin{center}
\includegraphics[
height=1.9061in,
width=1.9061in
]%
{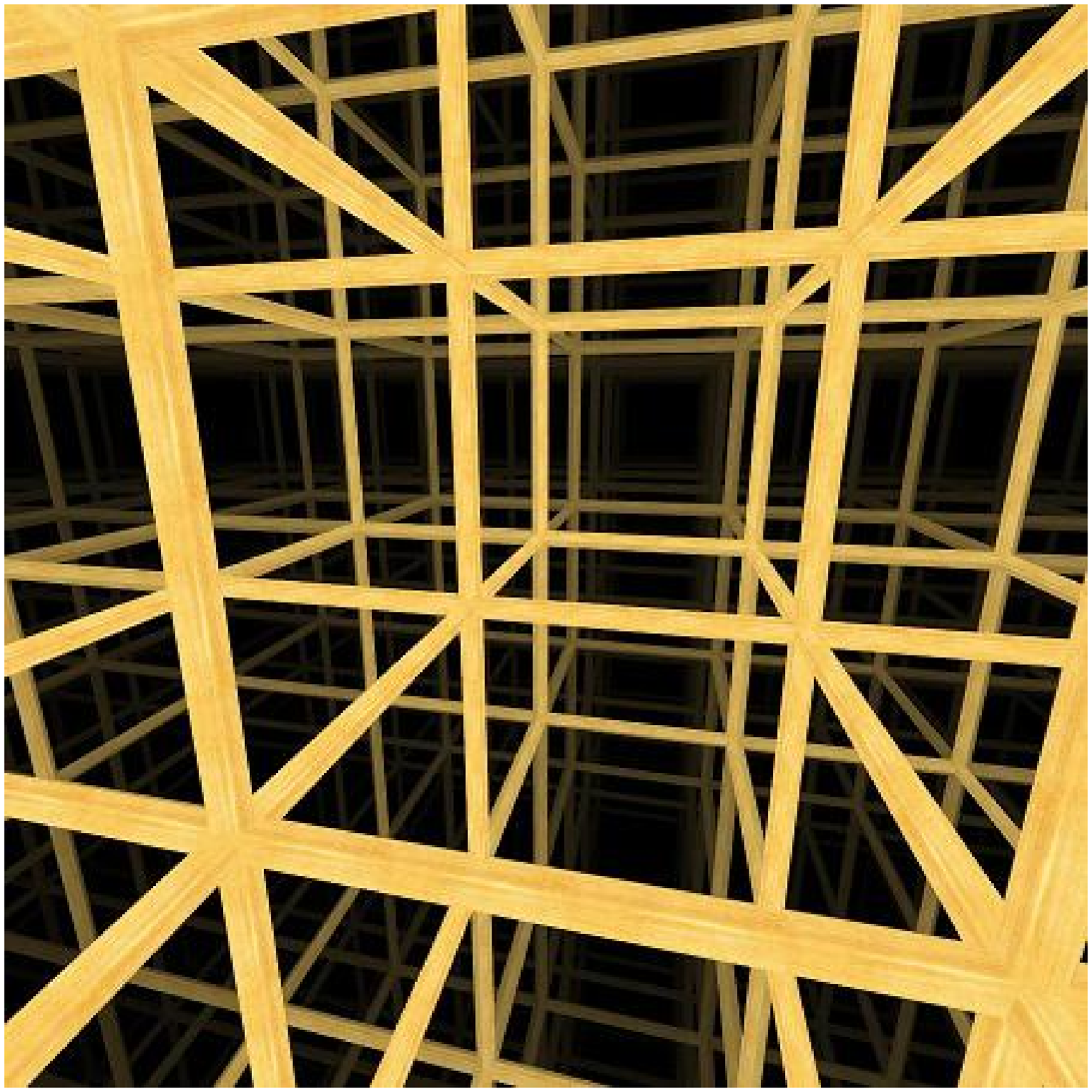}%
\\
\textbf{Cubic honeycomb of 3-space. [Figure taken from Wikipedia.]}%
\end{center}

We think of this honeycomb as a cell complex $\mathcal{C}_{\ell}$ for
$\mathbb{R}^{3}$ consisting of \textbf{0-cells, 1-cells, 2-cells, }and
\textbf{3-cells} called respectively \textbf{vertices, edges, faces, }and
\textbf{cubes}.%
\begin{center}
\fbox{\includegraphics[
height=1.2687in,
width=3.4281in
]%
{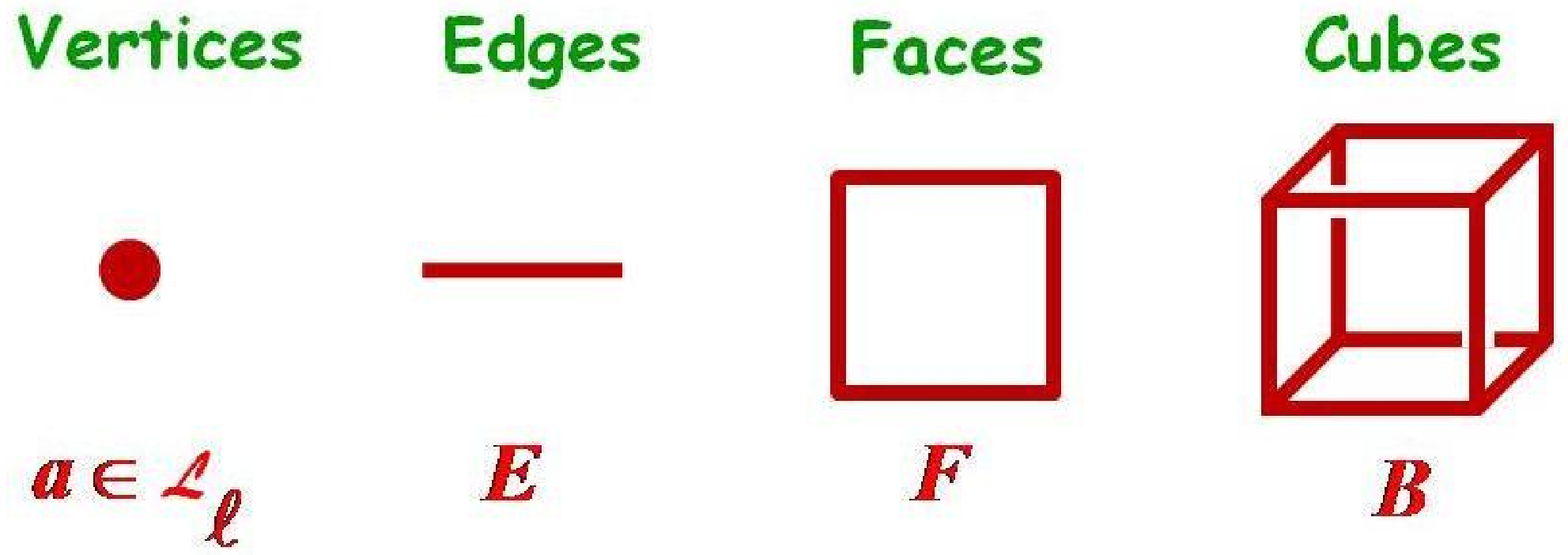}%
}
\end{center}
\textit{All cells of positive dimension are assumed to be open cells.
\ }Moreover, let $\mathcal{C}_{\ell}^{j}$ denote the $j$-\textbf{skeleton} of
the cell complex $\mathcal{C}_{\ell}$ for $j=0,1,2,3$.

\bigskip

\begin{definition}
A \textbf{lattice graph} $G$ (\textbf{of order} $\ell$) is a finite subset of
edges (together with their respective vertices) of the cubic honeycomb
$\mathcal{C}_{\ell}$. \ A \textbf{lattice knot} $K$ (\textbf{of order} $\ell$)
is a 2-valent lattice graph of order $\ell$. \ Moreover, let $\mathbb{G}%
^{(\ell)}$ and $\mathbb{K}^{(\ell)}$ respectively denote the \textbf{set of
all lattice graphs (of order }$\ell$\textbf{.)} and the \textbf{set of all
lattice knots} (\textbf{of order} $\ell$).
\end{definition}

\bigskip

\noindent\textbf{Reminder:} \ \textit{Throughout this paper, the term "knot"
will refer to both knots and links.}

\bigskip

Two examples of lattice knots are illustrated in the figure below:

\bigskip

\begin{center}%
{\parbox[b]{2.348in}{\begin{center}
\includegraphics[
height=1.6276in,
width=2.348in
]%
{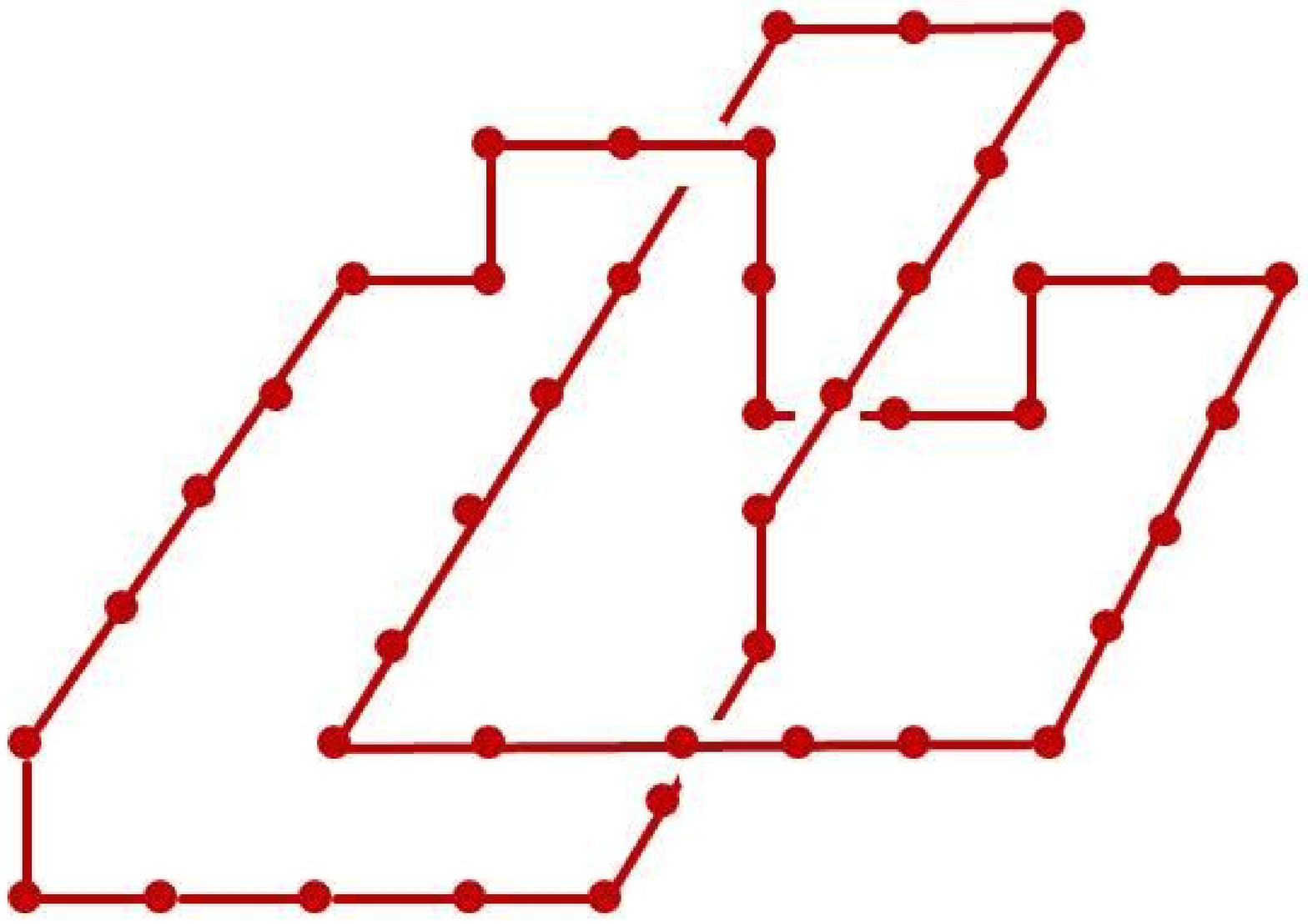}%
\\
\textbf{A lattice trefoil knot.}%
\end{center}}}%
\qquad%
{\parbox[b]{2.2771in}{\begin{center}
\includegraphics[
height=1.6786in,
width=2.2771in
]%
{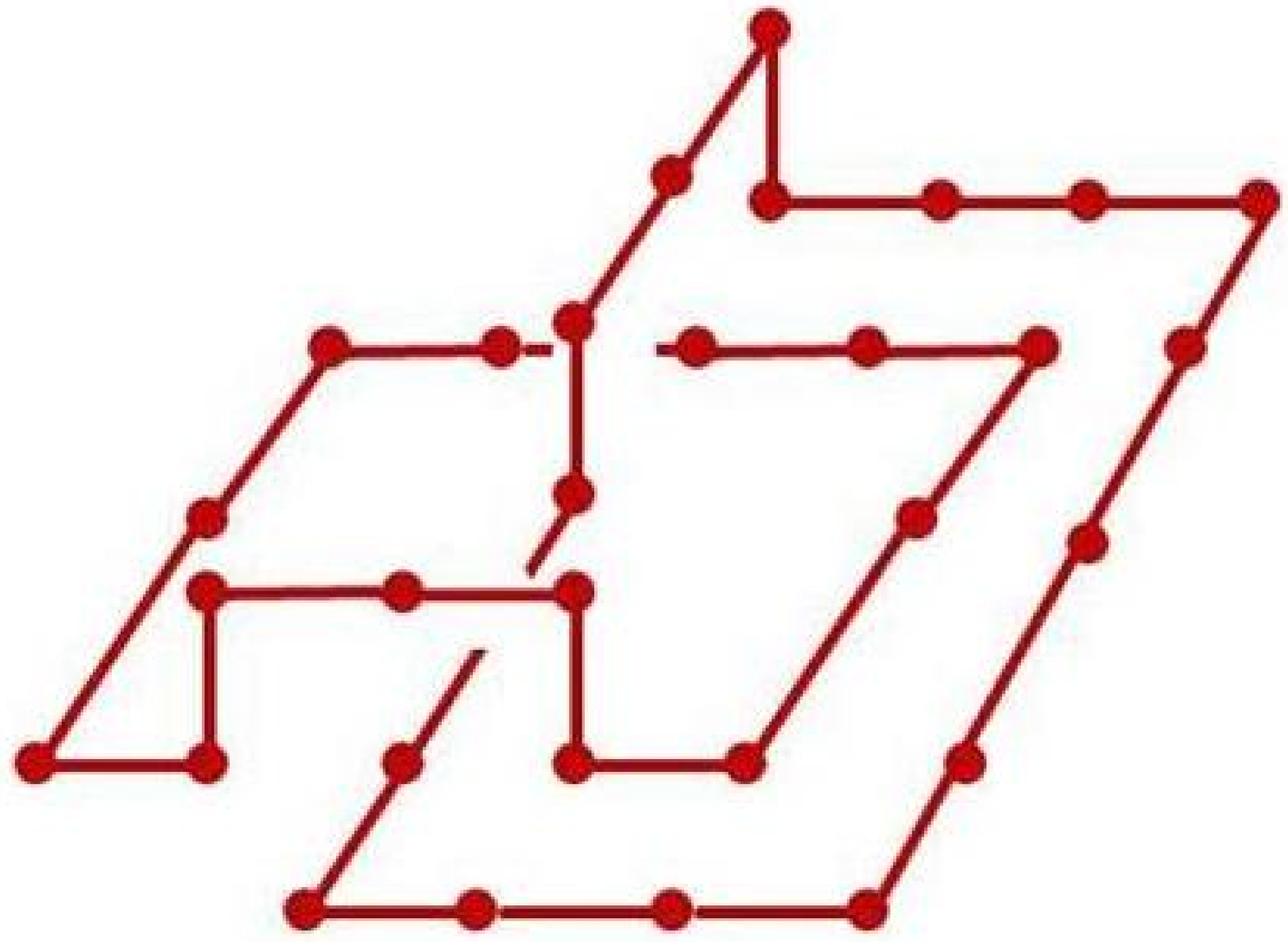}%
\\
\textbf{A lattice Hopf link.}%
\end{center}}}%

\end{center}

\bigskip

\section{Basic terminology and conventions}

\bigskip

Before we can proceed further, we will need to create an infrastructure and
nomenclature in which to work. \ 

\bigskip

\begin{remark}
The reader may find it convenient to quickly skim through this section, and
later to refer back to it as needed.
\end{remark}

\bigskip

We define an orientation of Euclidean 3-space $\mathbb{R}^{3}$ by selecting a
right handed frame
\[
e=\left\{  e_{1}=\left(
\begin{array}
[c]{c}%
1\\
0\\
0
\end{array}
\right)  ,e_{2}=\left(
\begin{array}
[c]{c}%
0\\
0\\
1
\end{array}
\right)  ,e_{3}=\left(
\begin{array}
[c]{c}%
0\\
0\\
1
\end{array}
\right)  \right\}
\]
at the origin, and by parallel transporting it to each vertex $a$ of the
honeycomb $\mathcal{C}_{\ell}$. \ We will refer to this frame as the
\textbf{preferred frame}.

\bigskip

\begin{definition}
A vertex $a$ of a cube $B$ is called the \textbf{preferred vertex of cube} $B$
if $B$ lies in the first octant of the preferred frame at $a$. \ Since $B$ is
uniquely determined by its preferred vertex, we use the following notation:%
\[
B=B^{(\ell)}(a)\text{ .}%
\]
The \textbf{preferred edges} and \textbf{the preferred faces of the cube}
$B^{(\ell)}(a)$ are respectively the edges and faces of $B^{(\ell)}(a)$ that
have $a$ as a vertex. \ We let
\[
E_{p}^{(\ell)}(a)\text{ and }F_{p}^{(\ell)}(a)
\]
denote respectively the\textbf{ preferred edge parallel to the frame vector}
$e_{p}$ and \textbf{the preferred face perpendicular to the frame vector
}$e_{p}$. The \textbf{preferred edges of} $F_{p}^{(\ell)}(a)$ are the edges of
$F_{p}^{(\ell)}(a)$\ that are preferred edges of the cube $B^{(\ell)}(a)$.
\ Finally, $a$ is called the \textbf{preferred vertex of the edge}
$E_{p}^{(\ell)}(a)$ \textbf{and of the face} $F_{p}^{(\ell)}(a)$
\end{definition}

\bigskip%

\begin{center}
\includegraphics[
height=3.2785in,
width=4.779in
]%
{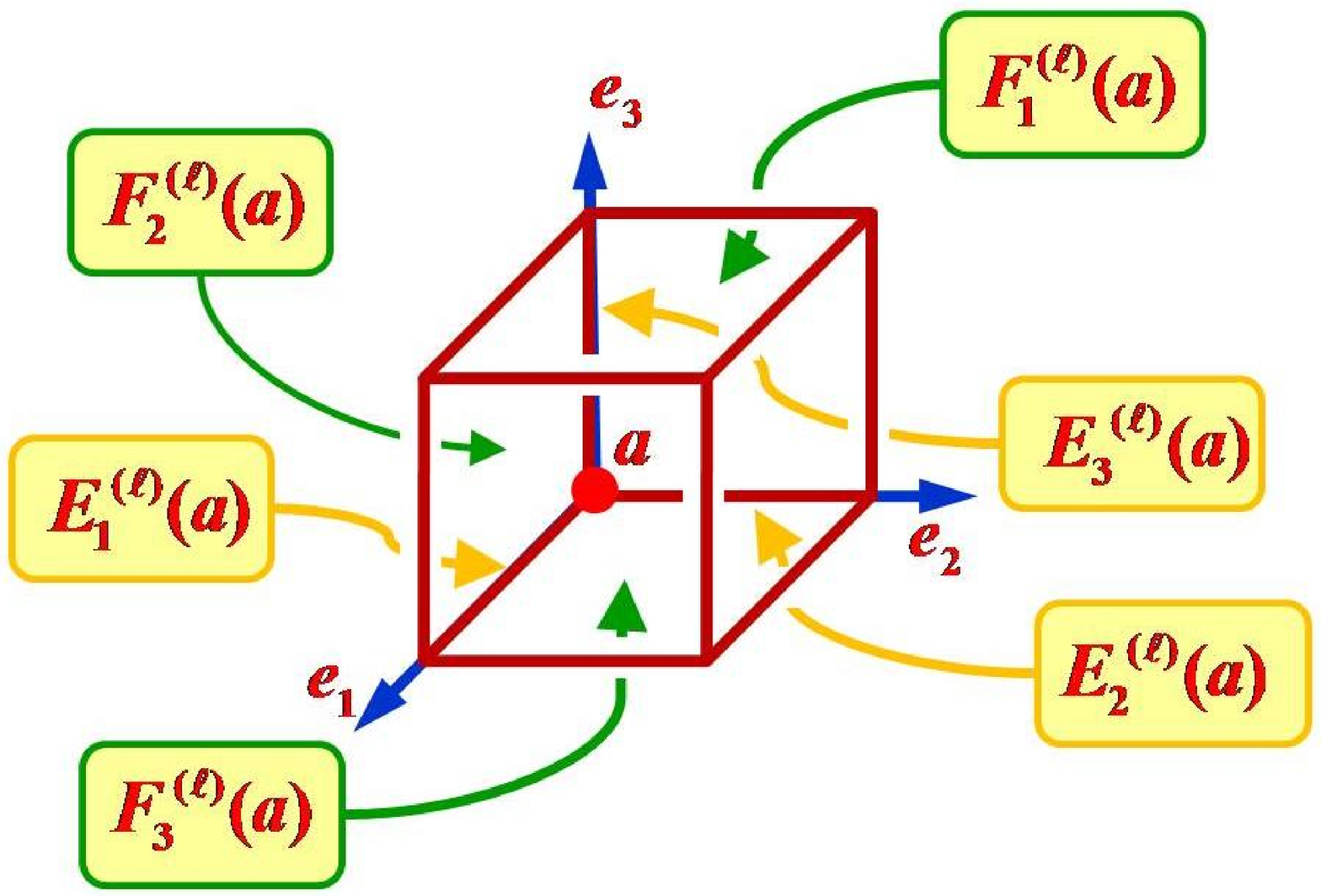}%
\\
\textbf{The preferred vertex, edges, and faces of the cube }$B^{(\ell)}%
(a)$\textbf{.}%
\end{center}

\bigskip

We will use the following \textbf{drawing conventions}:

\bigskip

\begin{itemize}
\item[\quad] \textbf{First drawing convention for cubes}: \textit{Each cube
}$B^{(\ell)}(a)$\textit{, when drawn in isolation, is drawn with edges
parallel to its preferred frame, and with its preferred vertex in the back
bottom left hand corner.}\bigskip

\item[\quad] \textbf{First drawing convention for faces}: \textit{Each face
}$F_{p}^{(\ell)}(a)$\textit{, when drawn in isolation, is always drawn with
its preferred vertex }$a$\textit{ in the upper left hand corner, and with the
frame vector }$e_{p}(a)$\textit{ pointing out of the page.} \ (Please refer to
the figure below.)%
\begin{center}
\includegraphics[
height=3.0277in,
width=4.0274in
]%
{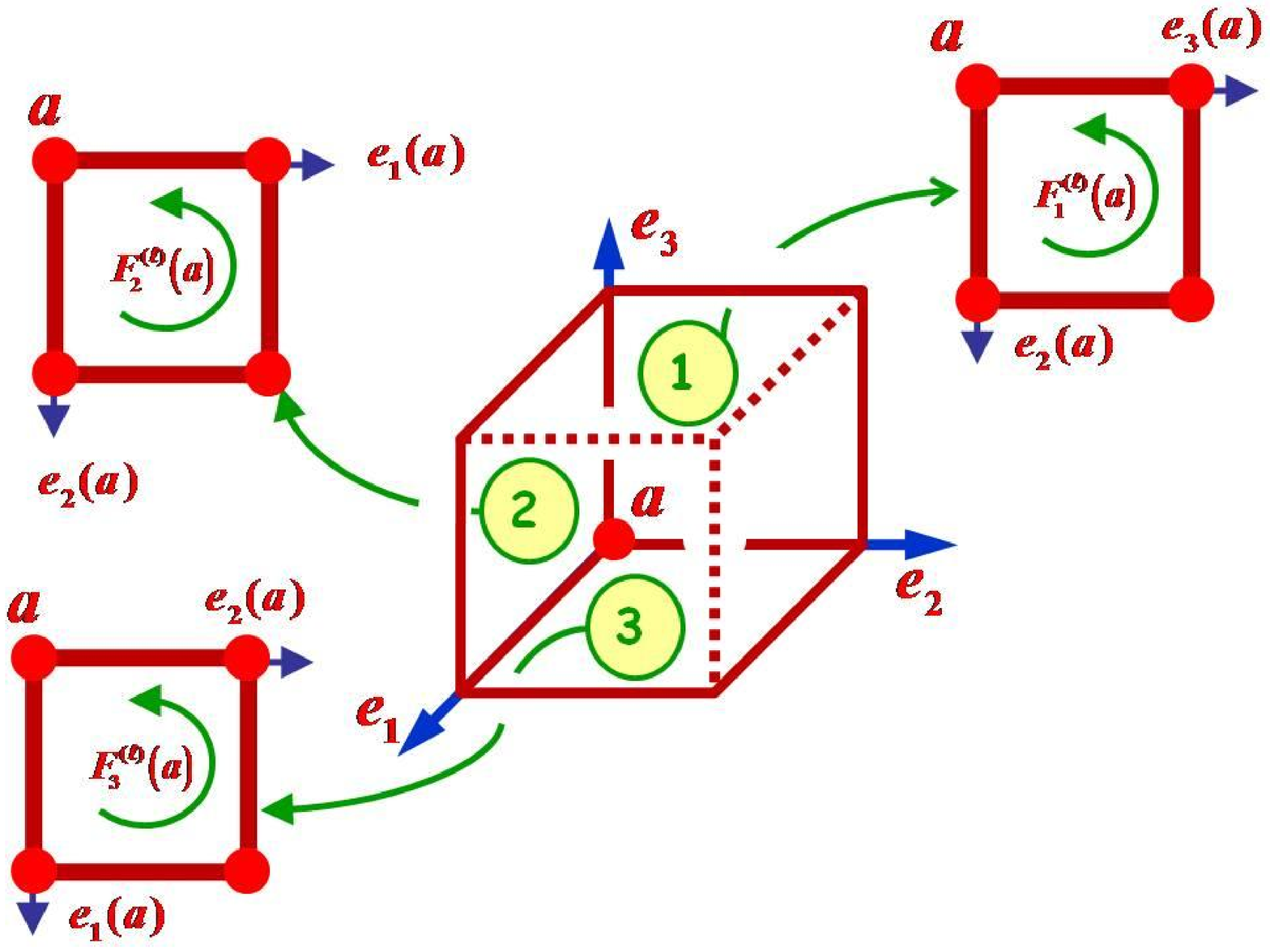}%
\\
\textbf{Drawing conventions for faces.}%
\end{center}

\end{itemize}

\bigskip

\begin{itemize}
\item[\quad] \textbf{Second drawing convention for faces and cubes}:
\textit{We will also make use of the \textbf{left} and \textbf{right
permutations} }$\lfloor$\textit{ and }$\rceil$\textit{ defined by}%
\[%
\begin{array}
[c]{ccc}%
\begin{array}
[c]{c}%
\left\lfloor \ \right.  :\left\{  1,2,3\right\}  \longrightarrow\left\{
1,2,3\right\}  \quad\quad\\
1\longmapsto2\\
2\longmapsto3\\
3\longmapsto1
\end{array}
&  &
\begin{array}
[c]{c}%
\left.  \ \right\rceil :\left\{  1,2,3\right\}  \longrightarrow\left\{
1,2,3\right\}  \quad\quad\\
1\longmapsto3\\
2\longmapsto1\\
3\longmapsto2
\end{array}
\end{array}
\]
These permutations have been defined so that
\[
e_{p}=e_{\left\lfloor p\right.  }\times e_{\left.  p\right\rceil }\text{ ,
\ \ }e_{\left\lfloor p\right.  }=e_{\left.  p\right\rceil }\times e_{p}\text{
, and \ \ }e_{\left.  p\right\rceil }=e_{p}\times e_{\left\lfloor p\right.  }%
\]
where `$\times$' denotes the \textbf{right handed vector cross product}.
\ \ With the left and right permutations, the first drawing conventions for
faces and cubes can now be more generally illustrated as shown below:
\begin{center}
\includegraphics[
height=1.7071in,
width=3.186in
]%
{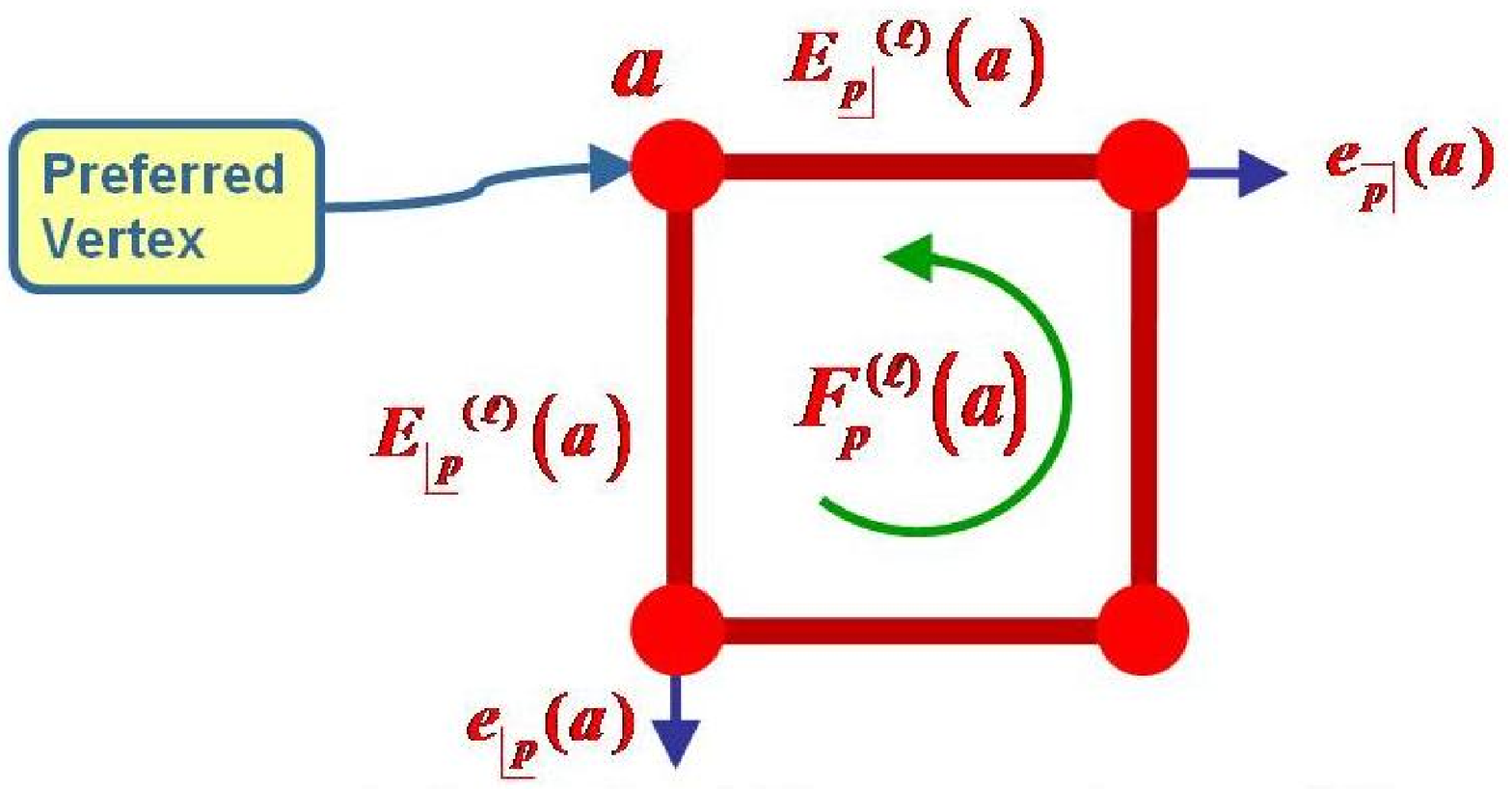}%
\\
\textbf{Face drawing conventions using the left and right permutations
"}$\lfloor$"\textbf{ and "}$\rceil$"\textbf{ . \ The frame vector }$e_{p}%
(a)$\textbf{ points out of the page toward the reader.}%
\end{center}
\begin{center}
\includegraphics[
height=2.7285in,
width=3.7689in
]%
{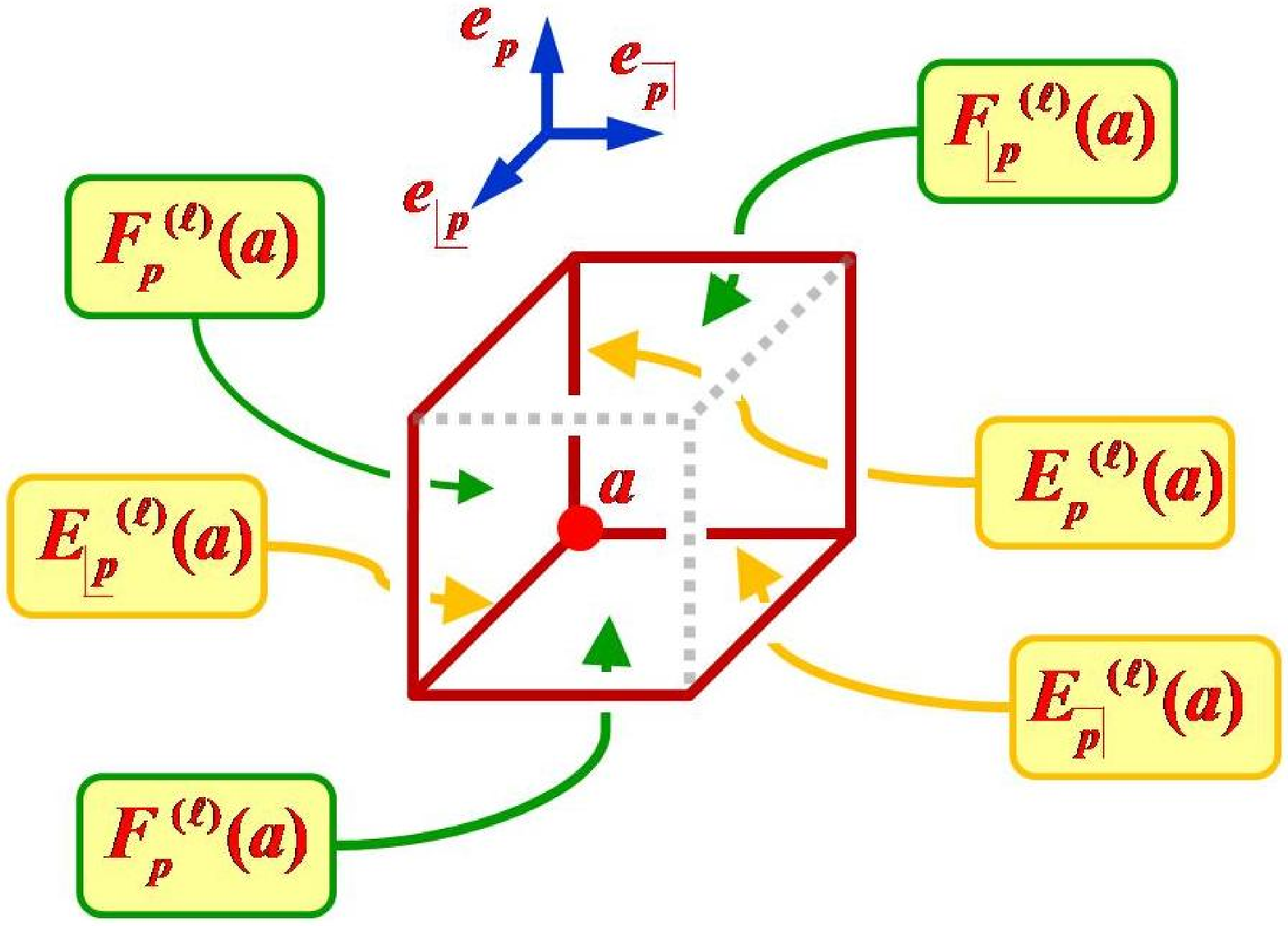}%
\\
\textbf{Cube drawing conventions using the left and right permutations
}$\lfloor$\textbf{ and }$\rceil$\textbf{ .}%
\end{center}

\end{itemize}

\bigskip

\bigskip

We will use the following \textbf{color coding scheme} for the vertices $a$
and edges $E$:%
\[%
\begin{tabular}
[c]{c|c|c}
& $\overset{\mathstrut}{\underset{\mathstrut}{\text{\textbf{Color}}}}$\textbf{
Coding Scheme} & \\\hline\hline
\multicolumn{1}{||c|}{$\overset{\mathstrut}{\underset{\mathstrut
}{\text{\textbf{Solid}}}}$\textbf{ Red}} & \textbf{"Hollow" Gray} &
\multicolumn{1}{|c||}{\textbf{Solid Gray}}\\\hline\hline
\multicolumn{1}{|c|}{$\overset{\mathstrut}{\underset{\mathstrut}{%
{\includegraphics[
height=0.2093in,
width=0.2093in
]%
{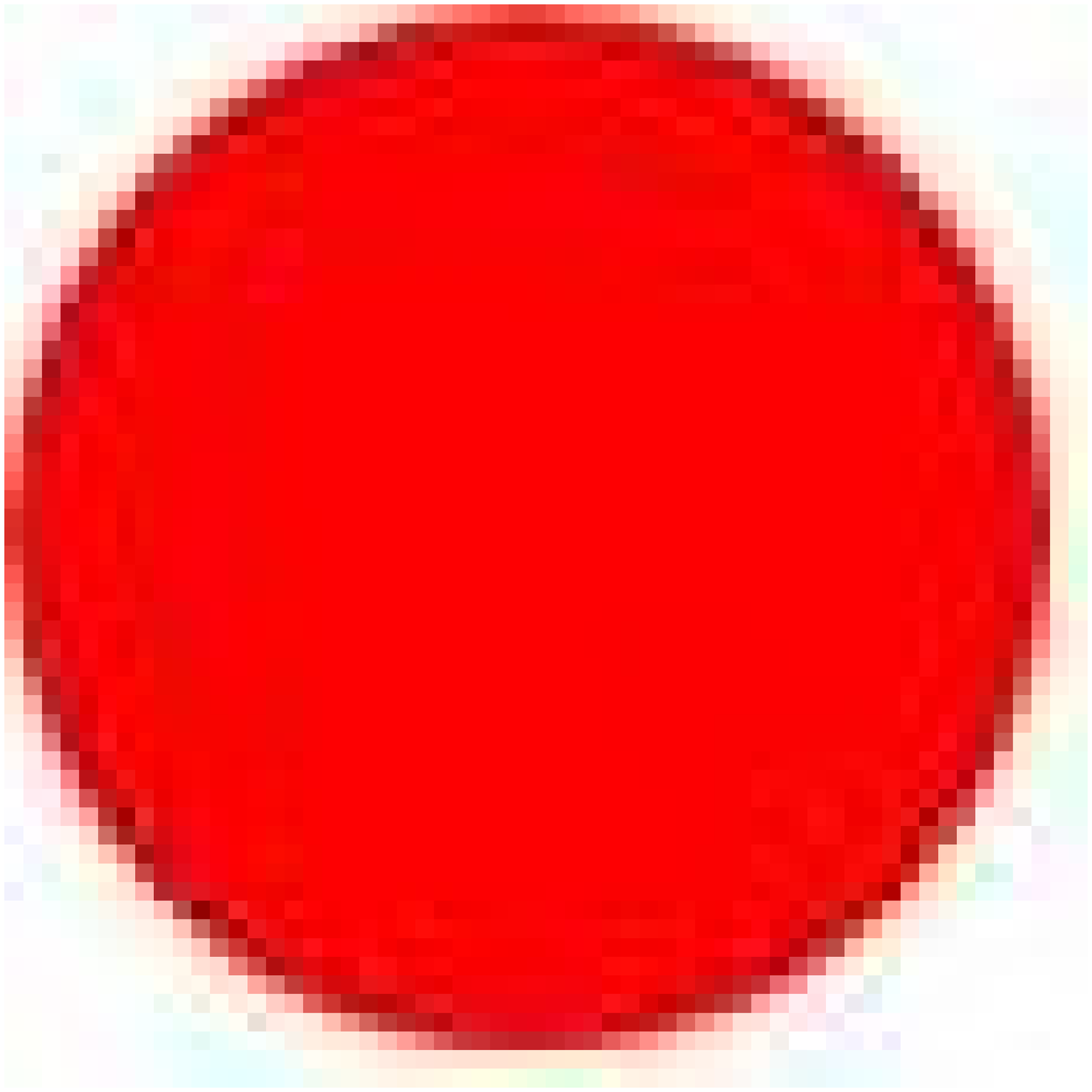}%
}%
}}$\quad%
{\includegraphics[
height=0.1781in,
width=0.6287in
]%
{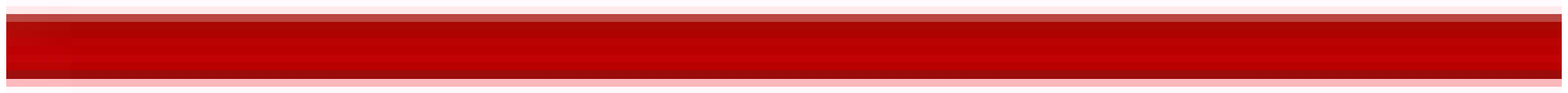}%
}%
} &
{\includegraphics[
height=0.2093in,
width=0.2093in
]%
{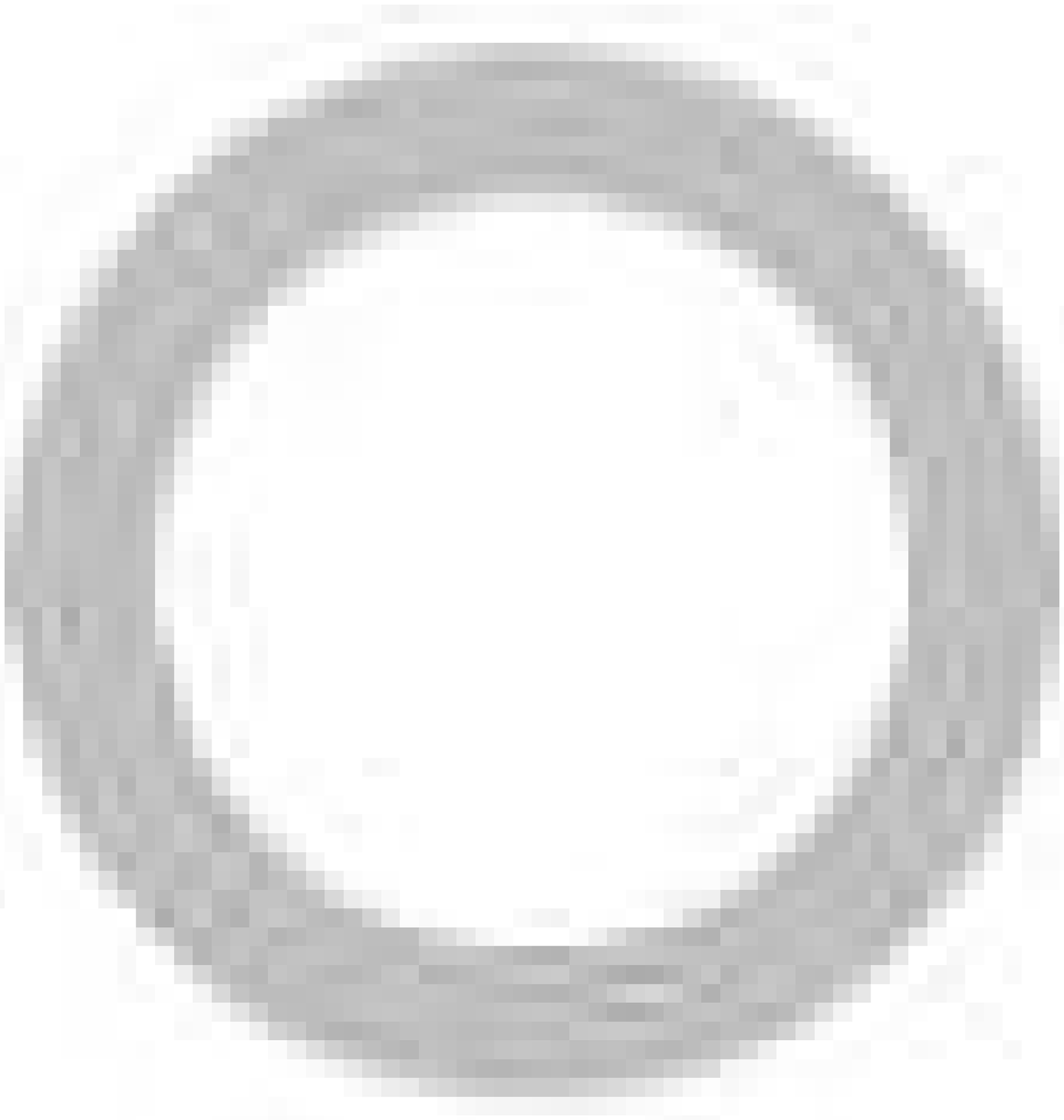}%
}%
\quad%
{\includegraphics[
height=0.1781in,
width=0.6287in
]%
{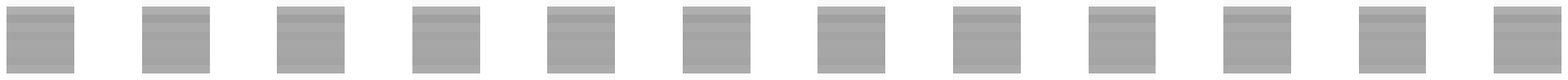}%
}%
& \multicolumn{1}{|c|}{%
{\includegraphics[
height=0.2093in,
width=0.2093in
]%
{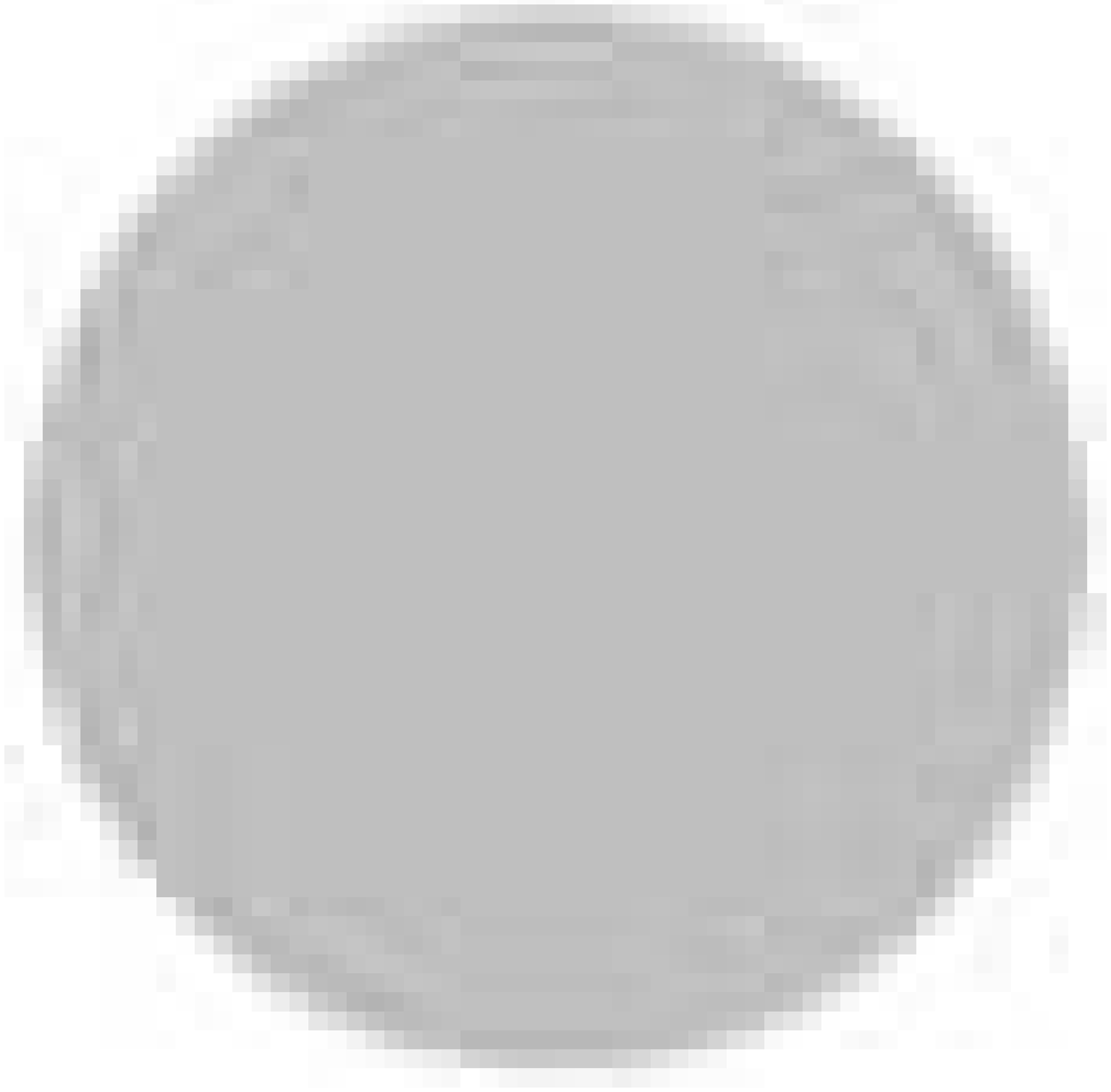}%
}%
\quad%
{\includegraphics[
height=0.1781in,
width=0.6287in
]%
{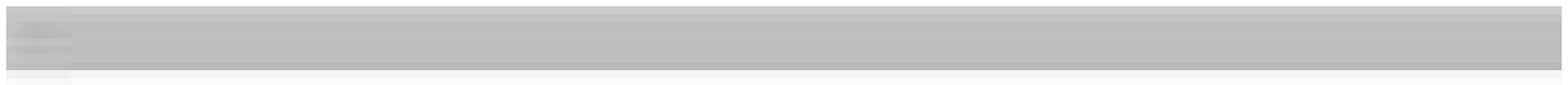}%
}%
}\\\hline
\multicolumn{1}{|c|}{Part of the lattice knot} & Not part of the lattice
knot & \multicolumn{1}{|l|}{%
\begin{tabular}
[c]{c}%
Indeterminate, may\\
or may not be part\\
of the lattice knot.
\end{tabular}
}\\\hline
\end{tabular}
\ \ \
\]

\bigskip%

\begin{center}
\fbox{\includegraphics[
height=3.7775in,
width=4.5282in
]%
{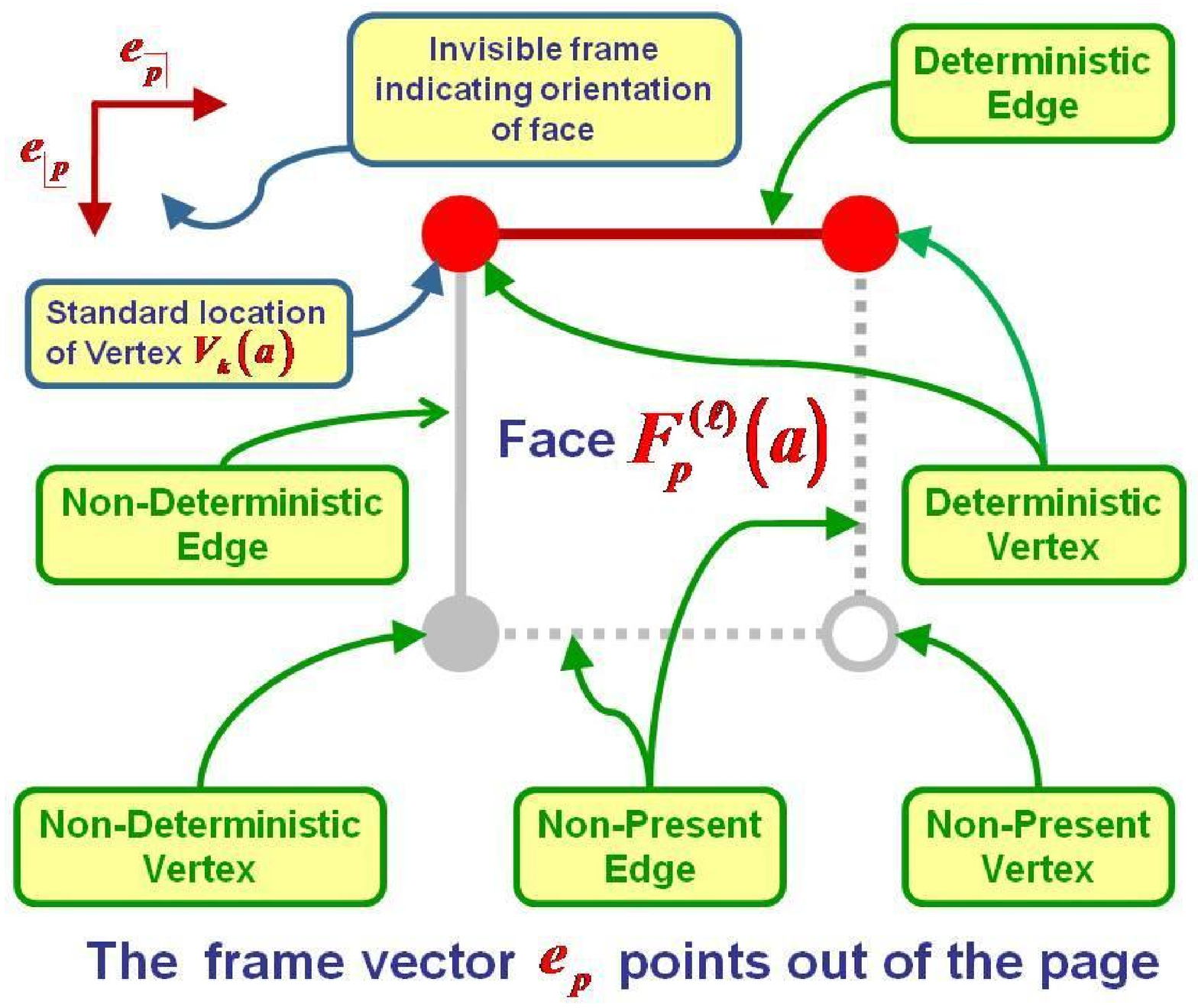}%
}\\
\textbf{An illustration of the color coding scheme..}%
\end{center}

\bigskip

Finally, we will have need of the following definition:

\bigskip

\begin{definition}
For each integer $p=1,2,3$, we define the \textbf{lattice translation map}
from the lattice $\mathcal{L}_{\ell}$ into itself as:
\[%
\begin{array}
[c]{ccc}%
\top_{p}:\mathcal{L}_{\ell} & \longrightarrow & \mathcal{L}_{\ell}\\
\quad\quad\quad a & \longmapsto & a+2^{-\ell}e_{p}%
\end{array}
\]
where $e_{p}$ denotes the $p$-th unit length vector of the preferred frame.
\ Moreover, we will often use the following more compact notation%
\[
\top_{p}a=a^{:p}\text{ .}%
\]
For example, $a^{:1^{2}\overline{2}^{3}3}$ denotes%
\[
a^{:1^{2}\overline{2}^{3}3}=\top_{1}^{2}\top_{2}^{-3}\top_{3}a=a+2\cdot
2^{-\ell}e_{1}-3\cdot2^{-\ell}e_{2}+2^{-\ell}e_{3}%
\]

\end{definition}

\bigskip

\begin{remark}
Throughout this paper, we have made an effort to devise a mathematical
notation that is intuitive as well as non-cumbersome. \ We hope the reader
will find that this is the case.
\end{remark}

\bigskip

\section{Lattice Knot Moves: Wiggle, Wag, and Tug}

\bigskip

Using the graphical conventions prescribed in the previous section, we now
define, for each non-negative integer $\ell$, three \textbf{lattice knot
moves} $L_{1}^{(\ell)}$,$L_{2}^{(\ell)}$, and $L_{3}^{(\ell)}$, called
respectively \textbf{tug}, \textbf{wiggle}, and \textbf{wag}. \ Each lattice
move is a bijection from the set of lattice knots $\mathbb{K}^{(\ell)}$ onto
itself, i.e., a permutation of $\mathbb{K}^{(\ell)}$.

\bigskip

While reading this section, the reader may find it helpful to refer to
notational summaries found in Appendix B.

\bigskip

\subsection{Definition of the move tug}

\bigskip

\bigskip

The first move, called a \textbf{tug}, and denoted by
\[
L_{1}^{(\ell)}\left(  a,p,q\right)  \text{ ,}%
\]
is defined for each of the four edges of each preferred face $F_{p}^{(\ell
)}\left(  a\right)  $ of each cube $B^{(\ell)}\left(  a\right)  $ in the cell
complex $\mathcal{C}_{\ell}$. \ As indicated in the figure given below, we
index the four edges of a preferred face\ $F_{p}^{(\ell)}(a)$, beginning with
the preferred edge $E_{\left\lfloor p\right.  }^{(\ell)}(a)$, with the
integers $q=0$, $1$, $2$, $3$ (also respectively by the symbols $q=$
\raisebox{-0.0199in}{\includegraphics[
height=0.1436in,
width=0.1505in
]%
{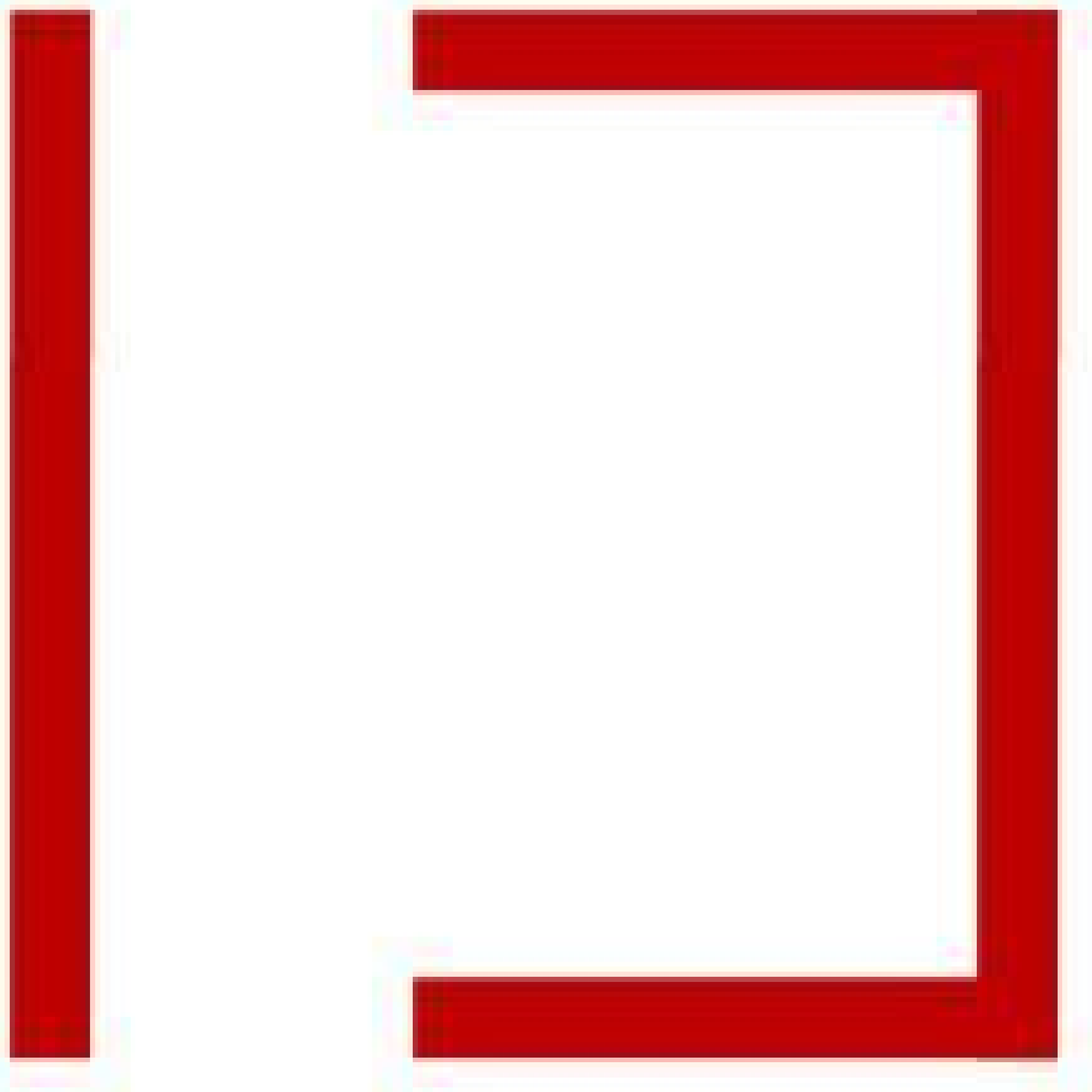}%
}%
,
\raisebox{-0.0199in}{\includegraphics[
height=0.1436in,
width=0.1436in
]%
{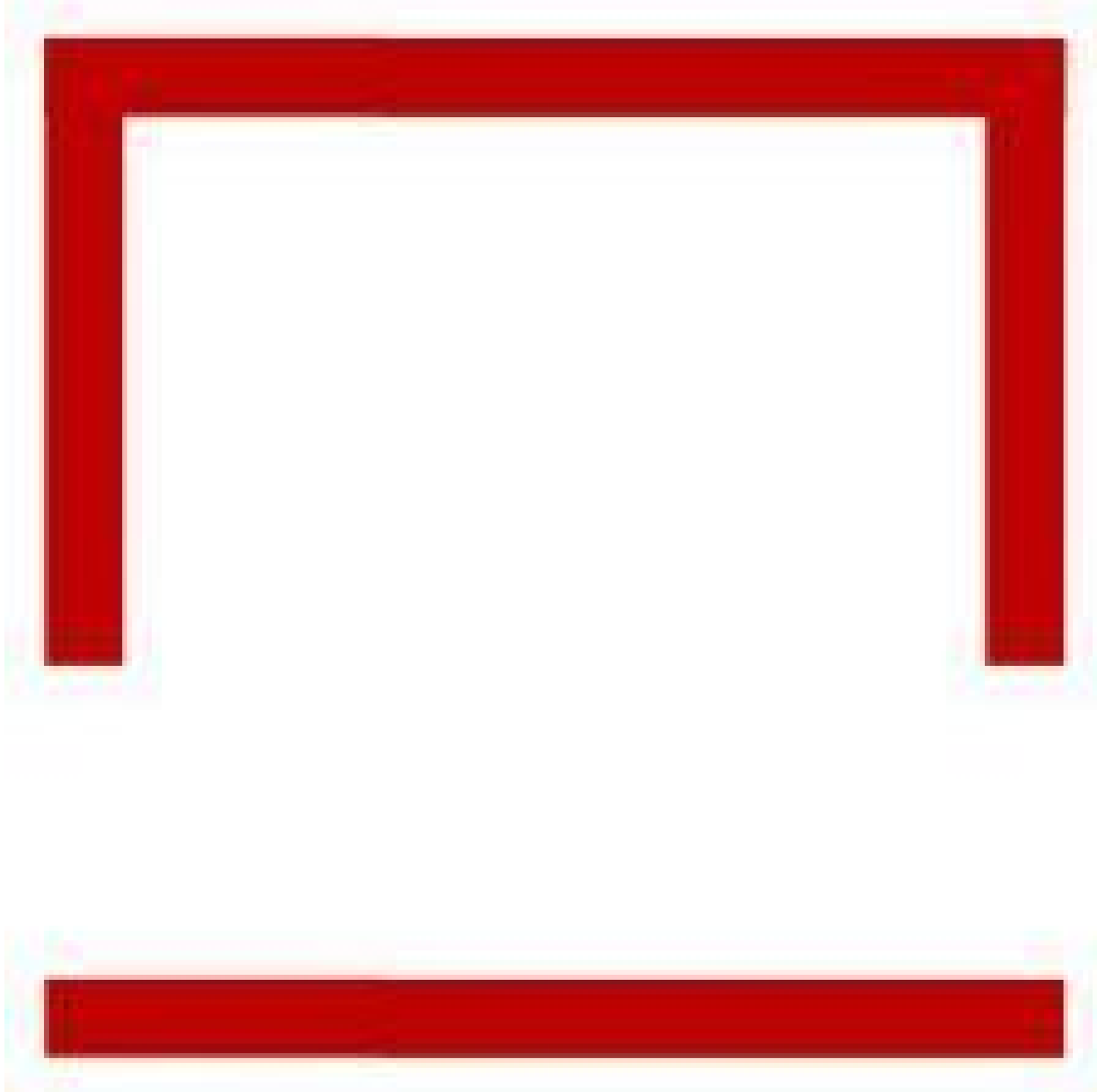}%
}%
,
\raisebox{-0.0199in}{\includegraphics[
height=0.1436in,
width=0.1436in
]%
{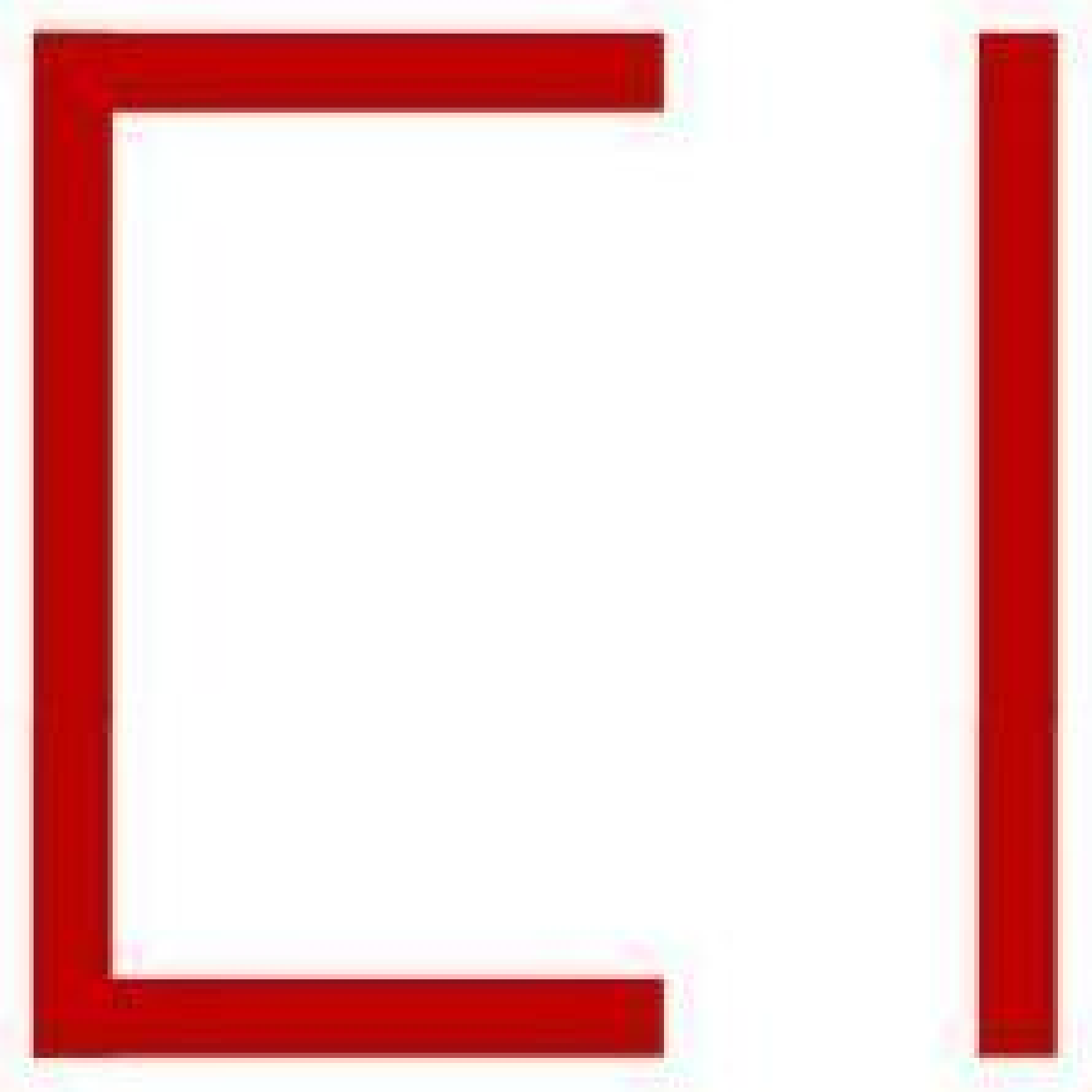}%
}%
,
\raisebox{-0.0199in}{\includegraphics[
height=0.1436in,
width=0.1436in
]%
{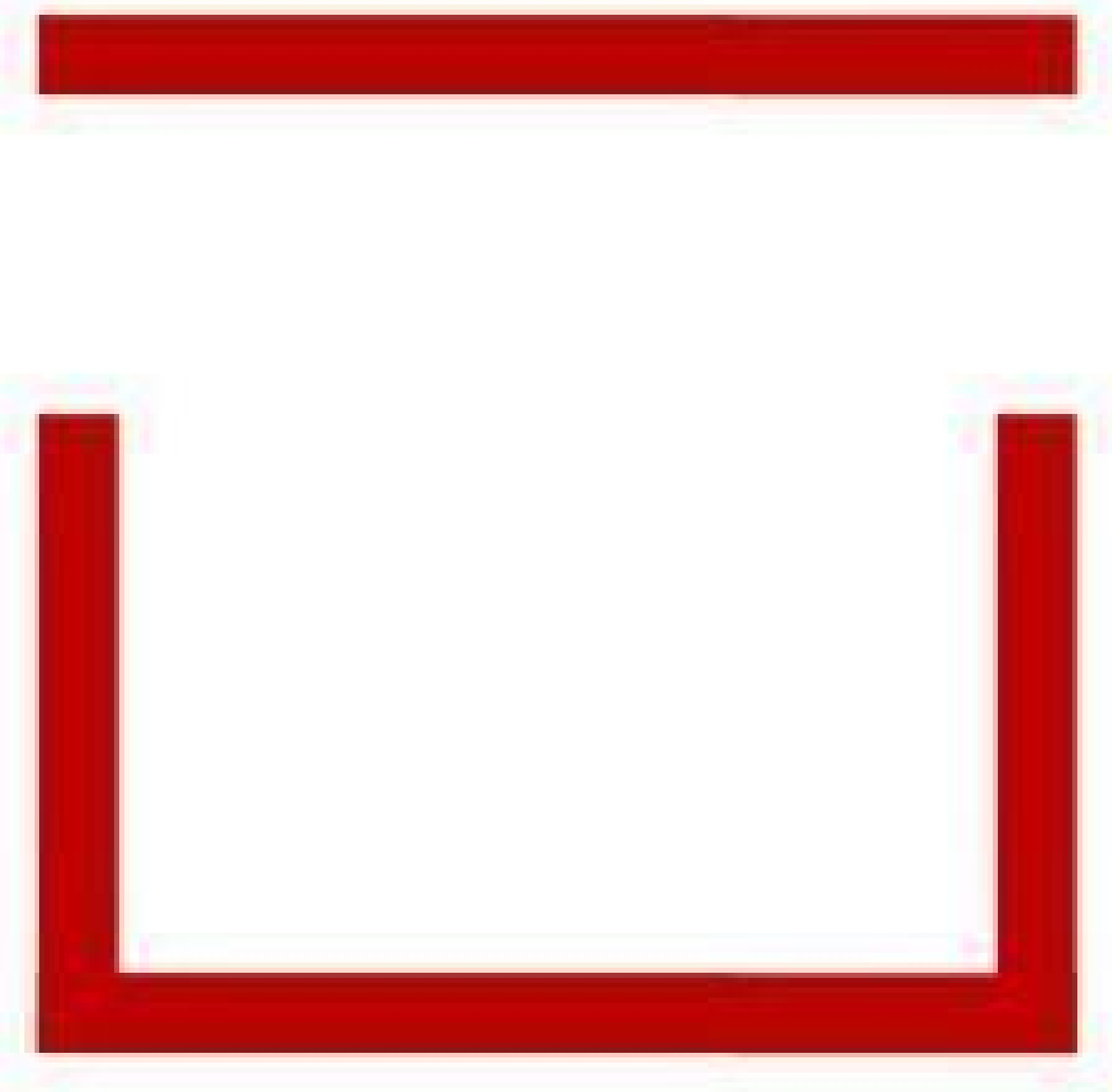}%
}%
), using the counterclockwise orientation induced on the face $F_{p}^{(\ell
)}(a)$ by the preferred frame $e$.

\bigskip%

\begin{center}
\includegraphics[
height=3.7775in,
width=5.028in
]%
{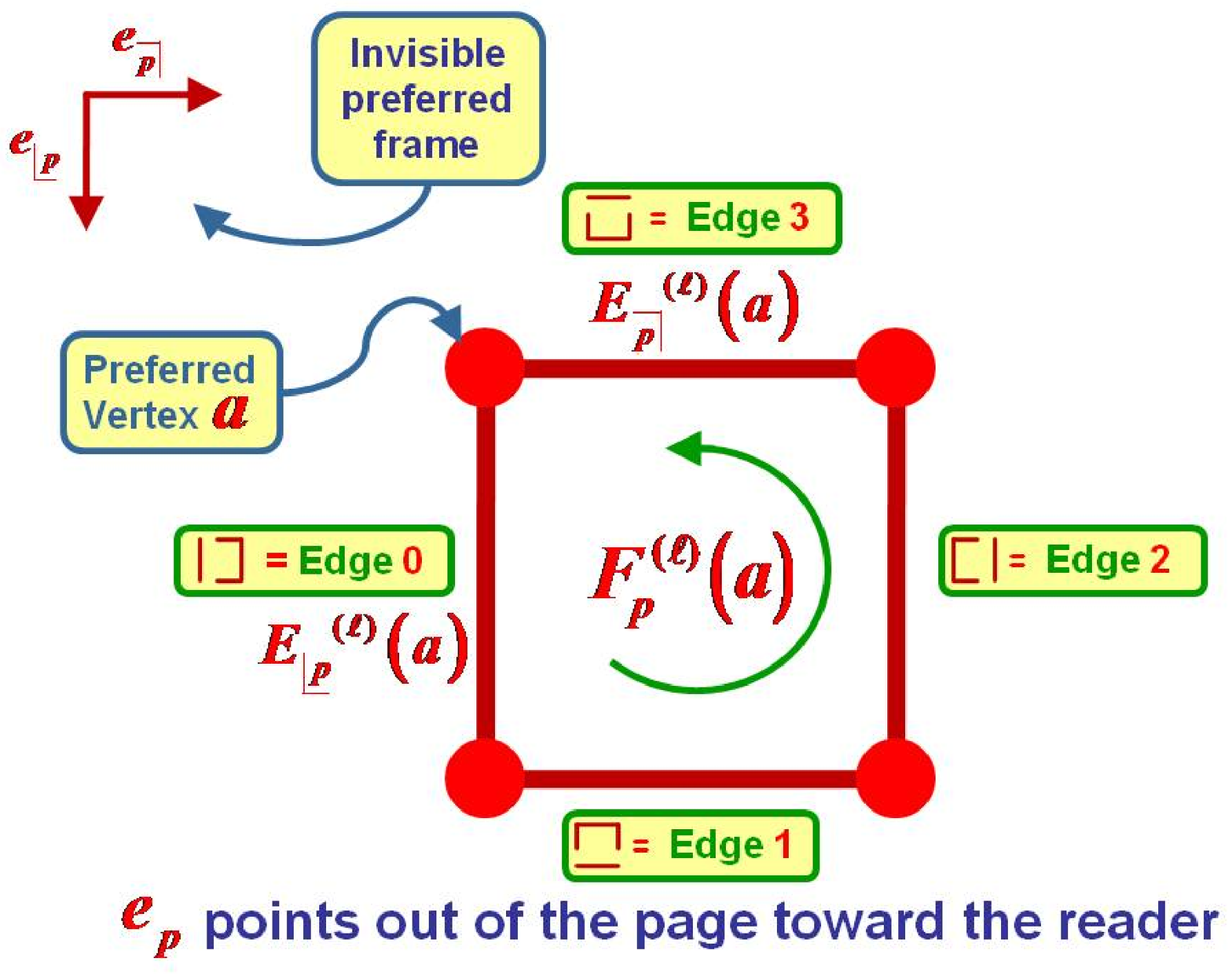}%
\\
\textbf{Edge ordering conventions for tug move }$L_{1}^{(\ell)}(a,p,q)$, for
edges $q=0,1,2,3$, which are also respectively denoted
by\ \raisebox{-0.0199in}{\includegraphics[
height=0.1436in,
width=0.1436in
]%
{icon10.ps}%
}%
$^{(\ell)}$,$\ $\raisebox{-0.0199in}{\includegraphics[
height=0.1436in,
width=0.1436in
]%
{icon11.ps}%
}%
$^{(\ell)}$, \raisebox{-0.0199in}{\includegraphics[
height=0.1436in,
width=0.1436in
]%
{icon12.ps}%
}%
$^{(\ell)}$, and$\ \ $\raisebox{-0.0199in}{\includegraphics[
height=0.1436in,
width=0.1436in
]%
{icon13.ps}%
}%
$^{(\ell)}$.
\end{center}

\bigskip

The tug $L_{1}^{(\ell)}\left(  a,p,q\right)  $ associated with the edge $q=0$
(also denoted by $q=$%
\raisebox{-0.0199in}{\includegraphics[
height=0.1436in,
width=0.1505in
]%
{icon10.ps}%
}%
) of the preferred face $F_{p}^{(\ell)}\left(  a\right)  $ of the cube
$B^{(\ell)}\left(  a\right)  $ will be denoted in anyone of the following
three ways%
\[
L_{1}^{(\ell)}\left(  a,p,%
\raisebox{-0.0199in}{\includegraphics[
height=0.1436in,
width=0.1505in
]%
{icon10.ps}%
}%
\right)  =L_{1}^{(\ell)}\left(  a,p,0\right)  =%
\raisebox{-0.0199in}{\includegraphics[
height=0.1436in,
width=0.1505in
]%
{icon10.ps}%
}%
^{(\ell)}\left(  a,p\right)  \text{ .}%
\]
The remaining tugs $L_{1}^{(\ell)}\left(  a,p,q\right)  $, for $q=1,2,3$ (also
indicated respectively by $q=%
\raisebox{-0.0199in}{\includegraphics[
height=0.1436in,
width=0.1436in
]%
{icon11.ps}%
}%
,%
\raisebox{-0.0199in}{\includegraphics[
height=0.1436in,
width=0.1436in
]%
{icon12.ps}%
}%
,%
\raisebox{-0.0199in}{\includegraphics[
height=0.1436in,
width=0.1436in
]%
{icon13.ps}%
}%
$ ), are denoted in like manner.

\bigskip

\begin{definition}
We define the \textbf{tug}, written $L_{1}^{(\ell)}\left(  a,p,0\right)  $
(also denoted by $%
\raisebox{-0.0406in}{\includegraphics[
height=0.1436in,
width=0.1436in
]%
{icon10.ps}%
}%
^{(\ell)}\left(  a,p\right)  $ and $L_{1}^{(\ell)}\left(  a,p,%
\raisebox{-0.0199in}{\includegraphics[
height=0.1436in,
width=0.1505in
]%
{icon10.ps}%
}%
\right)  $ ), associated with the $0$-th edge of the $p$-th preferred face
$F_{p}^{(\ell)}(a)$ of the cube $B_{\ell}(a)$ as the move

\bigskip

$\hspace{-1in}%
\raisebox{-0.0406in}{\includegraphics[
height=0.1436in,
width=0.1436in
]%
{icon10.ps}%
}%
^{(\ell)}\left(  a,p\right)  (K)=\left\{
\begin{array}
[c]{ll}%
\left(  K-%
\raisebox{-0.2508in}{\includegraphics[
height=0.6529in,
width=0.6529in
]%
{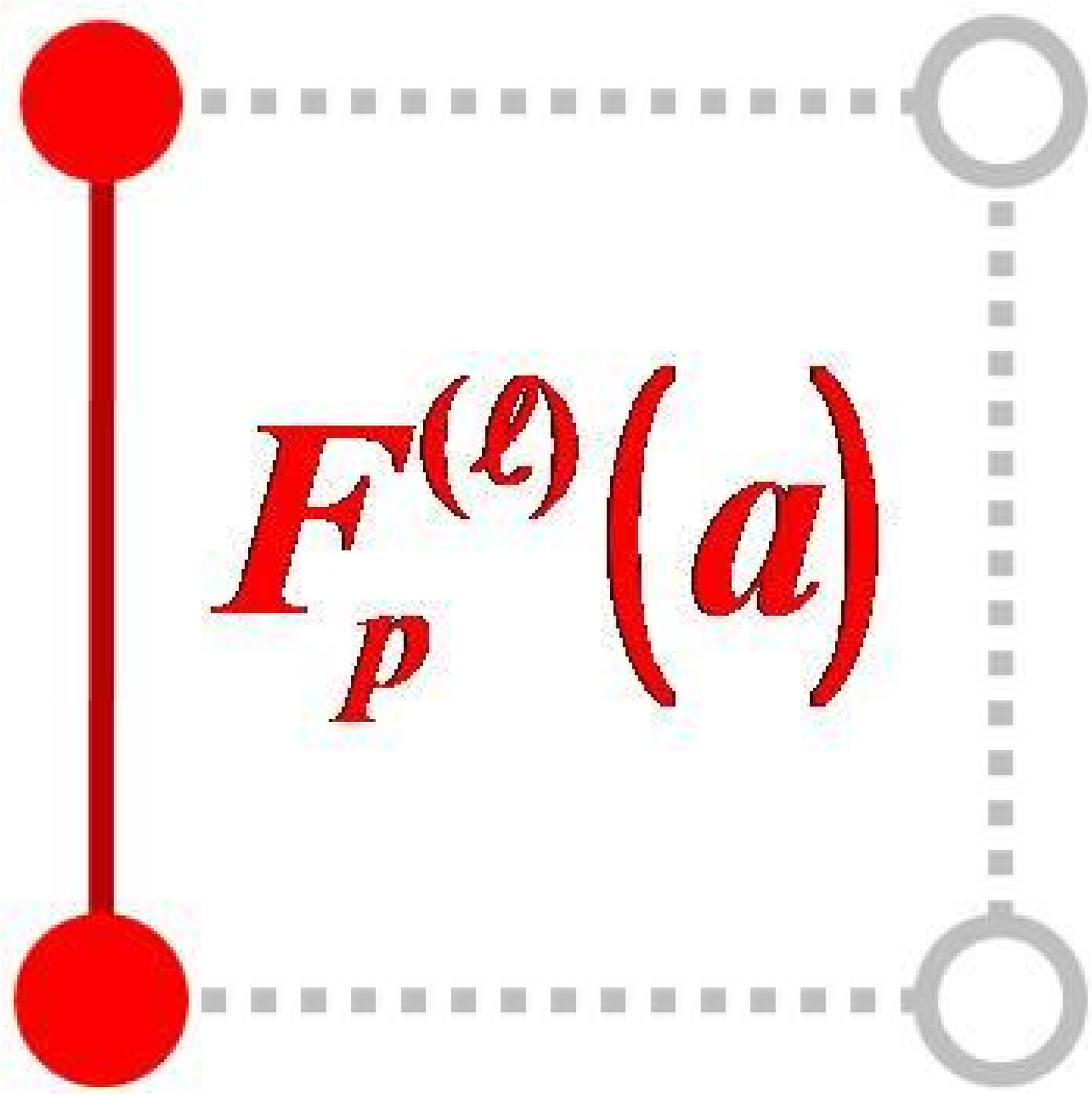}%
}%
\right)  \cup\left(  \mathcal{C}_{\ell}^{1}\cap%
\raisebox{-0.2508in}{\includegraphics[
height=0.6529in,
width=0.6529in
]%
{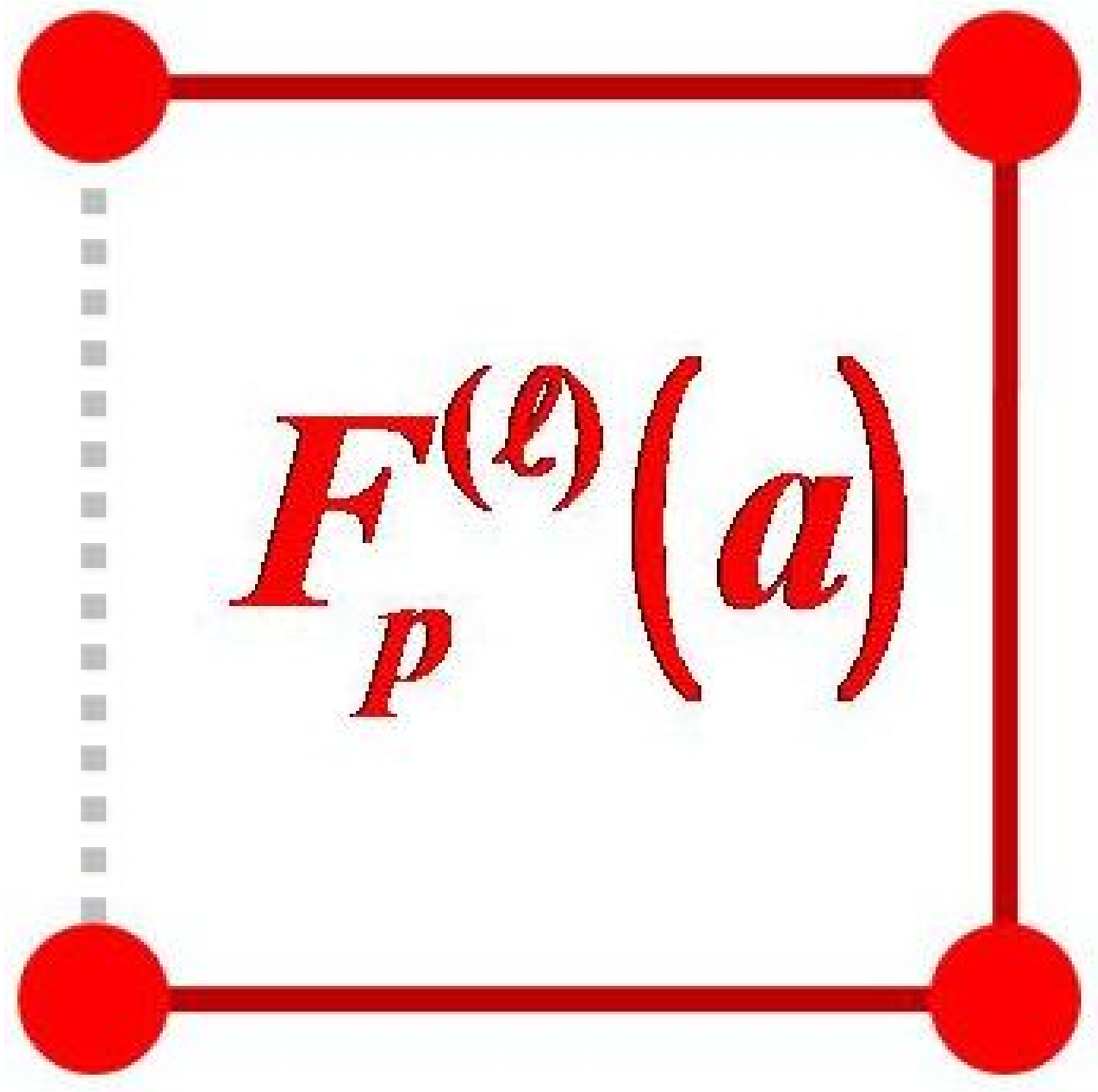}%
}%
\right)  & \text{if \ }K\cap%
\raisebox{-0.2508in}{\includegraphics[
height=0.6478in,
width=0.6478in
]%
{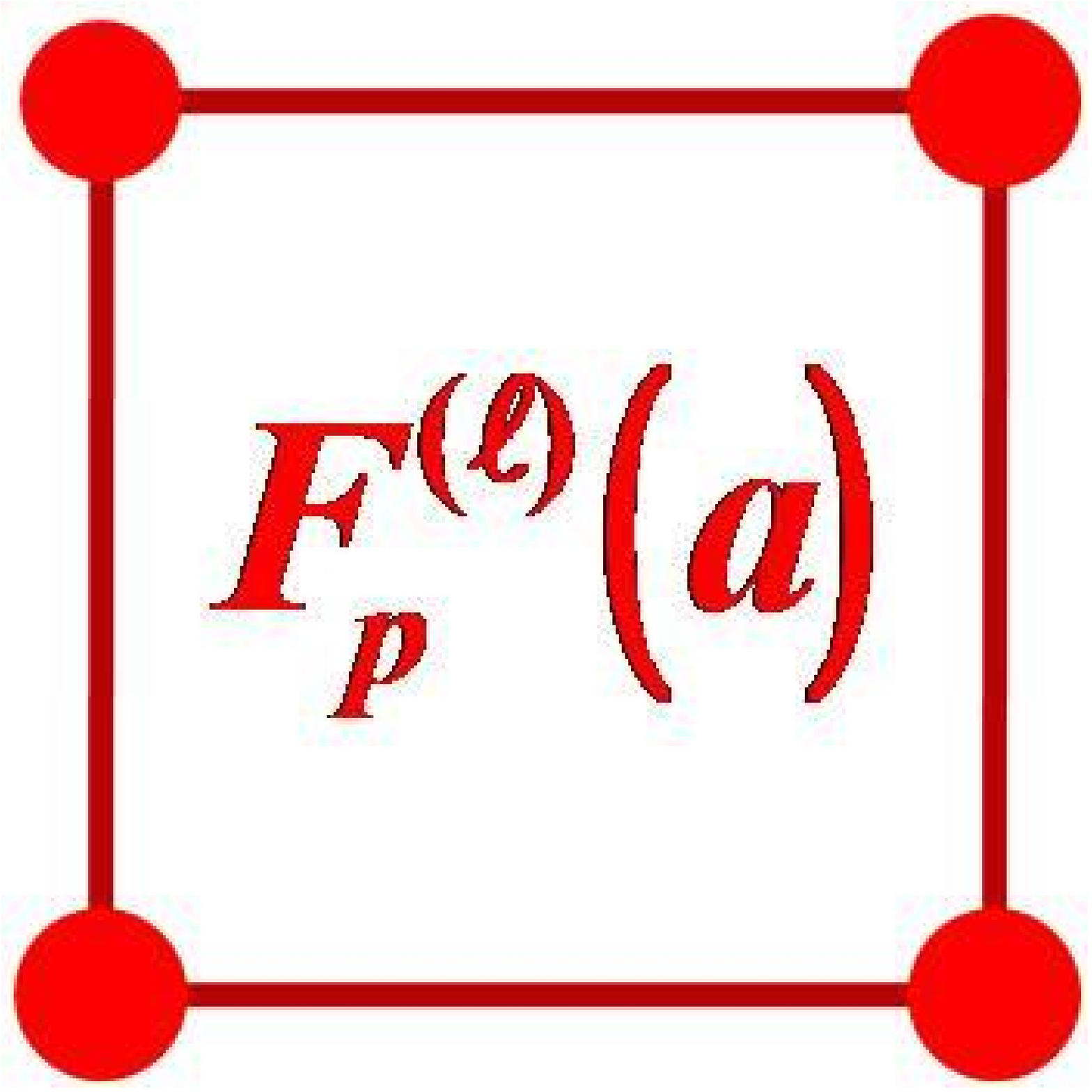}%
}%
=\mathcal{C}_{\ell}^{1}\cap%
\raisebox{-0.2508in}{\includegraphics[
height=0.6529in,
width=0.6529in
]%
{tug0l.ps}%
}%
\\
& \\
\left(  K-%
\raisebox{-0.2508in}{\includegraphics[
height=0.6529in,
width=0.6529in
]%
{tug0r.ps}%
}%
\right)  \cup\left(  \mathcal{C}_{\ell}^{1}\cap%
\raisebox{-0.2508in}{\includegraphics[
height=0.6529in,
width=0.6529in
]%
{tug0l.ps}%
}%
\right)  & \text{if \ }K\cap%
\raisebox{-0.2508in}{\includegraphics[
height=0.6478in,
width=0.6478in
]%
{face.ps}%
}%
=\mathcal{C}_{\ell}^{1}\cap%
\raisebox{-0.2508in}{\includegraphics[
height=0.6529in,
width=0.6529in
]%
{tug0r.ps}%
}%
\\
& \\
K & \text{otherwise}%
\end{array}
\right.  $\bigskip

\noindent where $%
\raisebox{-0.2508in}{\includegraphics[
height=0.6529in,
width=0.6529in
]%
{tug0l.ps}%
}%
$ , $%
\raisebox{-0.2508in}{\includegraphics[
height=0.6529in,
width=0.6529in
]%
{tug0r.ps}%
}%
$ , and
\raisebox{-0.2508in}{\includegraphics[
height=0.6478in,
width=0.6478in
]%
{face.ps}%
}%
denote the 2-subcomplexes of the cell complex $\mathcal{C}_{\ell}$, as defined
by the graphical conventions found in the previous section.

\bigskip

This \textbf{tug}, $L_{1}^{(\ell)}\left(  a,p,0\right)  =L_{1}^{(\ell)}\left(
a,p,%
\raisebox{-0.0406in}{\includegraphics[
height=0.1436in,
width=0.1436in
]%
{icon10.ps}%
}%
\right)  =%
\raisebox{-0.0406in}{\includegraphics[
height=0.1436in,
width=0.1436in
]%
{icon10.ps}%
}%
^{(\ell)}\left(  a,p\right)  $ is more succinctly illustrated in the figure
given below:

\bigskip%
\[%
\begin{array}
[c]{c}%
\begin{array}
[c]{ccc}%
\raisebox{-0.5016in}{\includegraphics[
height=1.1372in,
width=1.1372in
]%
{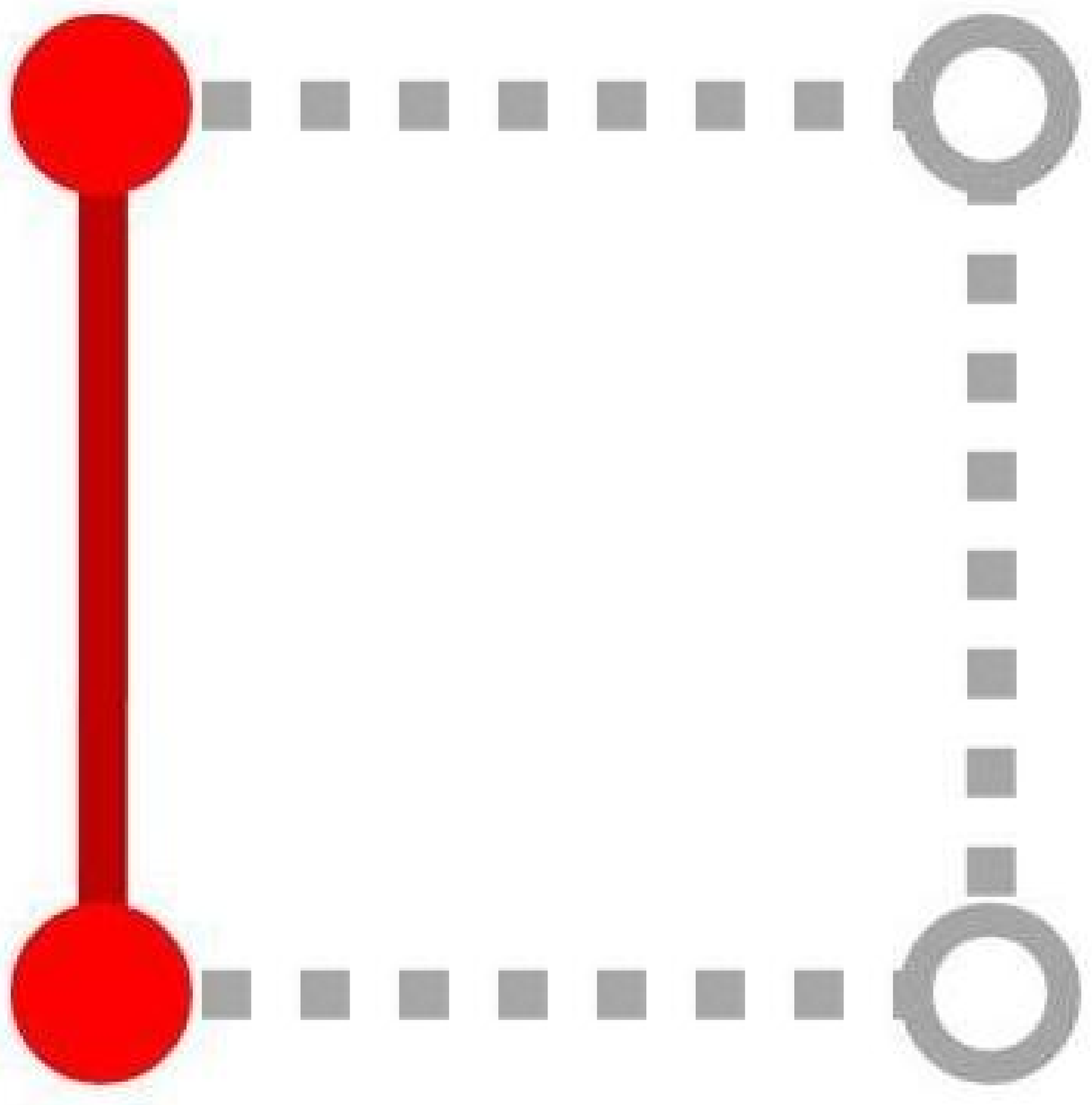}%
}%
&
\begin{array}
[c]{c}%
{\includegraphics[
height=0.237in,
width=0.5967in
]%
{arrow-ya.ps}%
}%
\\
F_{p}^{(\ell)}(a)
\end{array}
&
\raisebox{-0.5016in}{\includegraphics[
height=1.1372in,
width=1.1372in
]%
{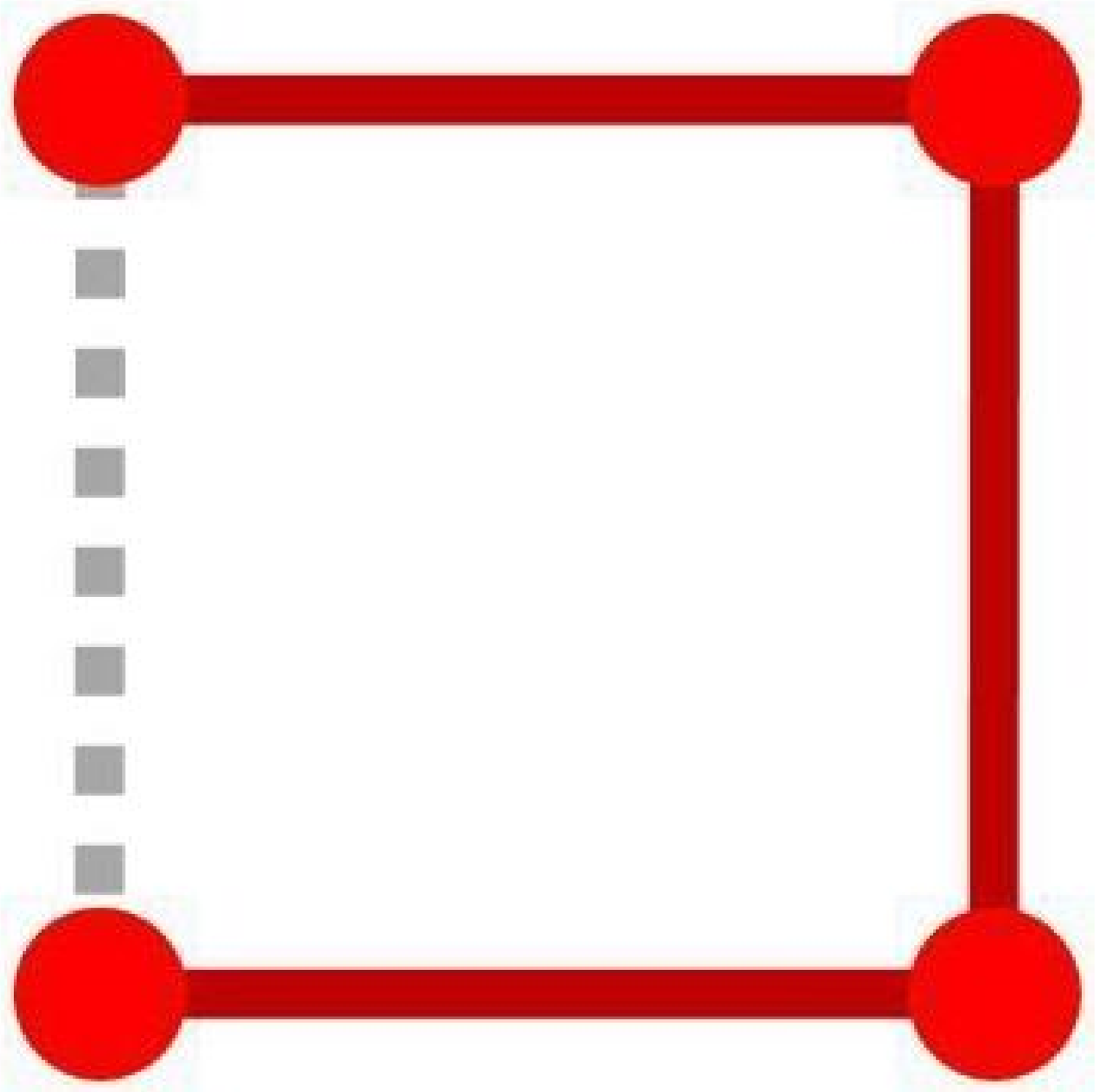}%
}%
\end{array}
\\
\text{\textbf{Lattice knot move }}L_{1}^{(\ell)}\left(  a,p,0\right)
=L_{1}^{(\ell)}\left(  a,p,%
\raisebox{-0.0406in}{\includegraphics[
height=0.1436in,
width=0.1436in
]%
{icon10.ps}%
}%
\right)  =%
\raisebox{-0.0406in}{\includegraphics[
height=0.1436in,
width=0.1436in
]%
{icon10.ps}%
}%
^{(\ell)}\left(  a,p\right)  \text{\textbf{, called a tug.}}%
\end{array}
\]

\bigskip

The remaining three tugs,
\[
\left\{
\begin{array}
[c]{l}%
L_{1}^{(\ell)}\left(  a,p,1\right)  =L_{1}^{(\ell)}\left(  a,p,%
\raisebox{-0.0406in}{\includegraphics[
height=0.1436in,
width=0.1436in
]%
{icon11.ps}%
}%
\right)  =%
\raisebox{-0.0406in}{\includegraphics[
height=0.1436in,
width=0.1436in
]%
{icon11.ps}%
}%
^{(\ell)}\left(  a,p\right)  \text{ ,}\\
\\
L_{1}^{(\ell)}\left(  a,p,2\right)  =L_{1}^{(\ell)}\left(  a,p,%
\raisebox{-0.0406in}{\includegraphics[
height=0.1436in,
width=0.1436in
]%
{icon12.ps}%
}%
\right)  =%
\raisebox{-0.0406in}{\includegraphics[
height=0.1436in,
width=0.1436in
]%
{icon12.ps}%
}%
^{(\ell)}\left(  a,p\right)  \text{ , and}\\
\\
L_{1}^{(\ell)}\left(  a,p,3\right)  =L_{1}^{(\ell)}\left(  a,p,%
\raisebox{-0.0406in}{\includegraphics[
height=0.1436in,
width=0.1436in
]%
{icon13.ps}%
}%
\right)  =%
\raisebox{-0.0406in}{\includegraphics[
height=0.1436in,
width=0.1436in
]%
{icon13.ps}%
}%
^{(\ell)}\left(  a,p\right)
\end{array}
\right.
\]
are defined in like manner, and illustrated in the three figures given below

\bigskip%

\[%
\begin{array}
[c]{c}%
\begin{array}
[c]{ccc}%
\raisebox{-0.5016in}{\includegraphics[
height=1.1372in,
width=1.1372in
]%
{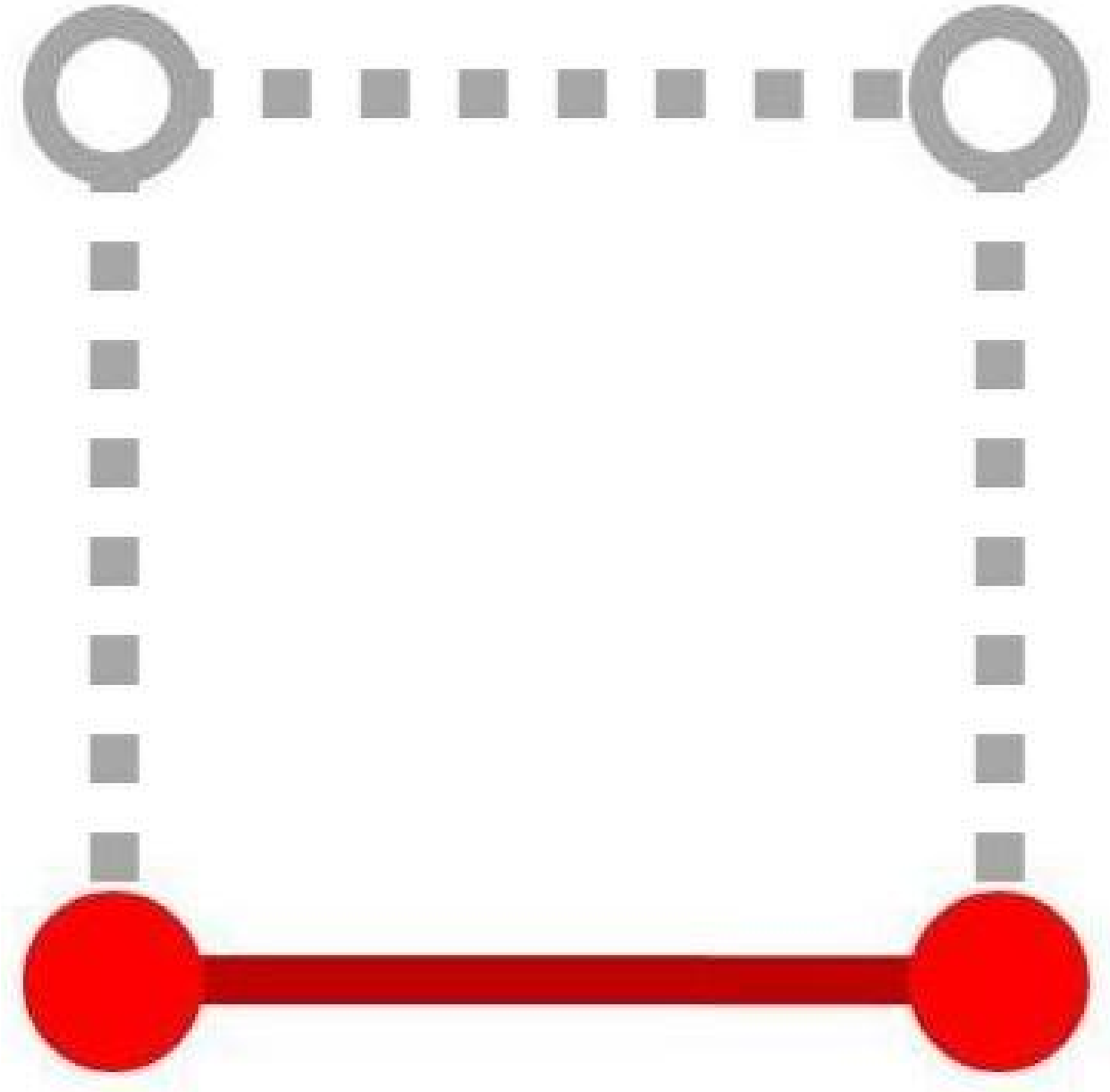}%
}%
&
\begin{array}
[c]{c}%
{\includegraphics[
height=0.237in,
width=0.5967in
]%
{arrow-ya.ps}%
}%
\\
F_{p}^{(\ell)}(a)
\end{array}
&
\raisebox{-0.5016in}{\includegraphics[
height=1.1372in,
width=1.1372in
]%
{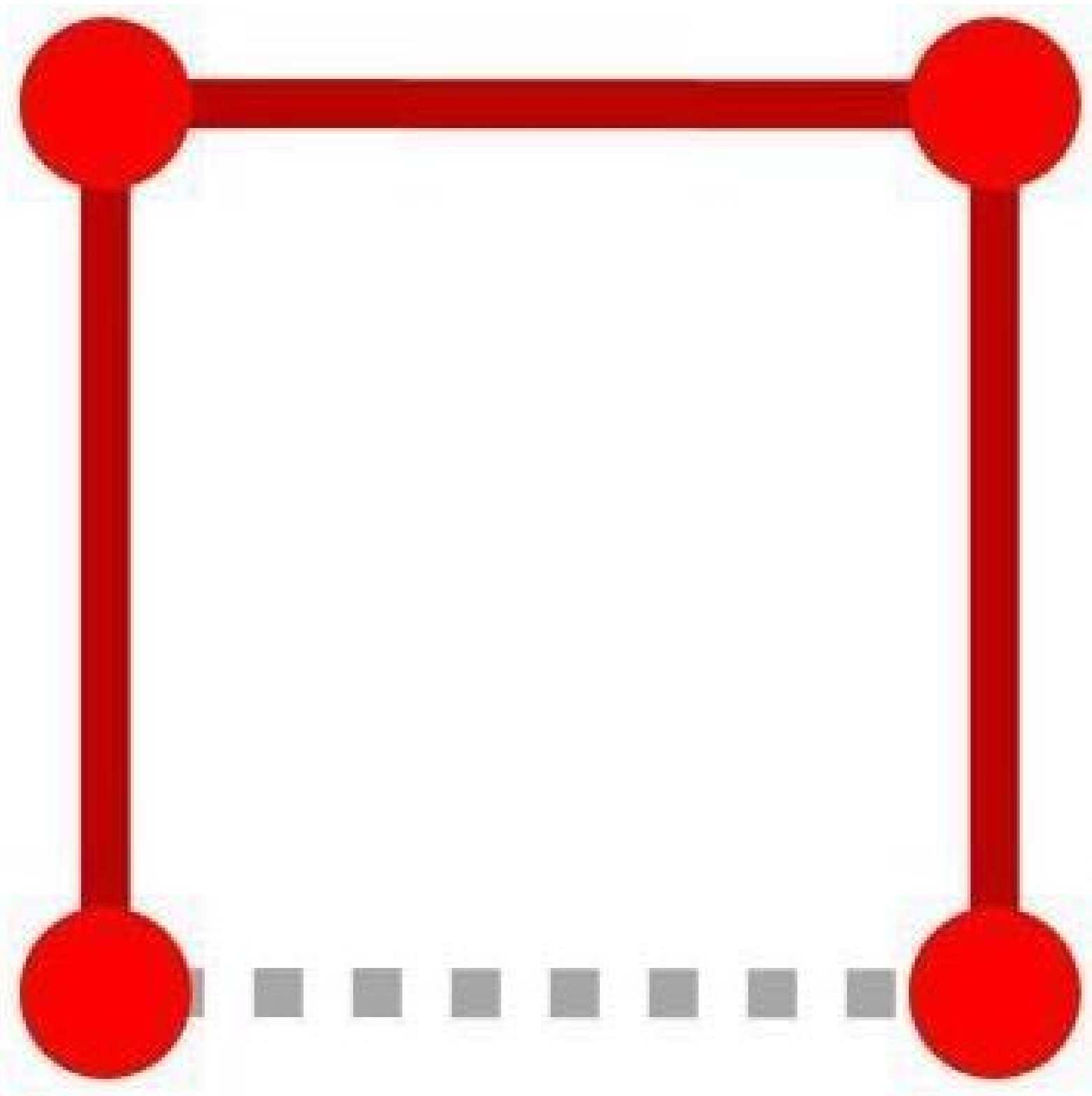}%
}%
\end{array}
\\
\text{\textbf{Lattice knot move }}L_{1}^{(\ell)}\left(  a,p,1\right)
=L_{1}^{(\ell)}\left(  a,p,%
\raisebox{-0.0406in}{\includegraphics[
height=0.1436in,
width=0.1436in
]%
{icon11.ps}%
}%
\right)  =%
\raisebox{-0.0406in}{\includegraphics[
height=0.1436in,
width=0.1436in
]%
{icon11.ps}%
}%
^{(\ell)}\left(  a,p\right)  \text{\textbf{, called a tug.}}%
\end{array}
\]

\bigskip%

\[%
\begin{array}
[c]{c}%
\begin{array}
[c]{ccc}%
\raisebox{-0.5016in}{\includegraphics[
height=1.1372in,
width=1.1372in
]%
{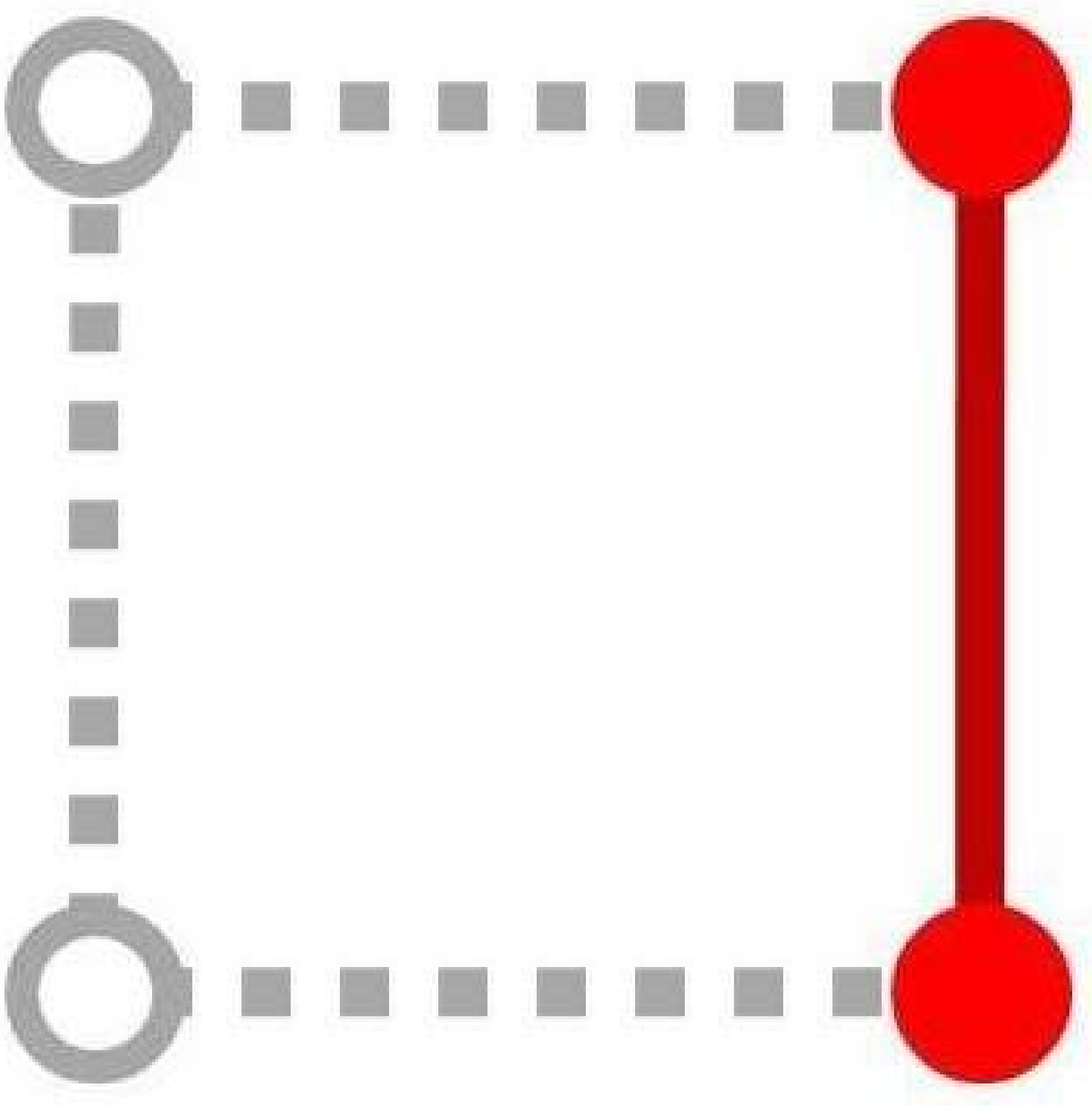}%
}%
&
\begin{array}
[c]{c}%
{\includegraphics[
height=0.237in,
width=0.5967in
]%
{arrow-ya.ps}%
}%
\\
F_{p}^{(\ell)}(a)
\end{array}
&
\raisebox{-0.5016in}{\includegraphics[
height=1.1372in,
width=1.1372in
]%
{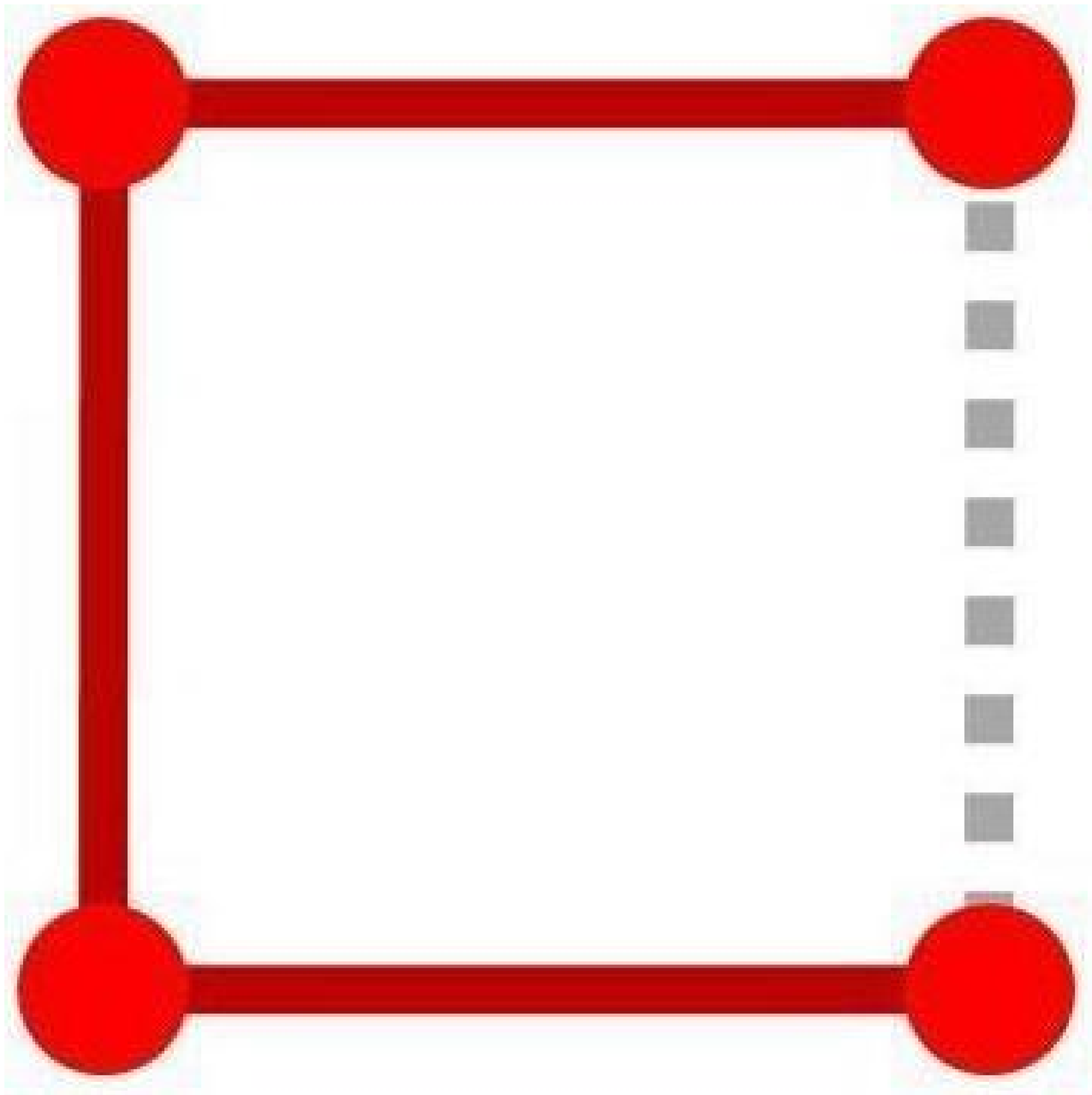}%
}%
\end{array}
\\
\text{\textbf{Lattice knot move }}L_{1}^{(\ell)}\left(  a,p,2\right)
=L_{1}^{(\ell)}\left(  a,p,%
\raisebox{-0.0406in}{\includegraphics[
height=0.1436in,
width=0.1436in
]%
{icon12.ps}%
}%
\right)  =%
\raisebox{-0.0406in}{\includegraphics[
height=0.1436in,
width=0.1436in
]%
{icon12.ps}%
}%
^{(\ell)}\left(  a,p\right)  \text{\textbf{, called a tug.}}%
\end{array}
\]

\bigskip%

\[%
\begin{array}
[c]{c}%
\begin{array}
[c]{ccc}%
\raisebox{-0.5016in}{\includegraphics[
height=1.1372in,
width=1.1372in
]%
{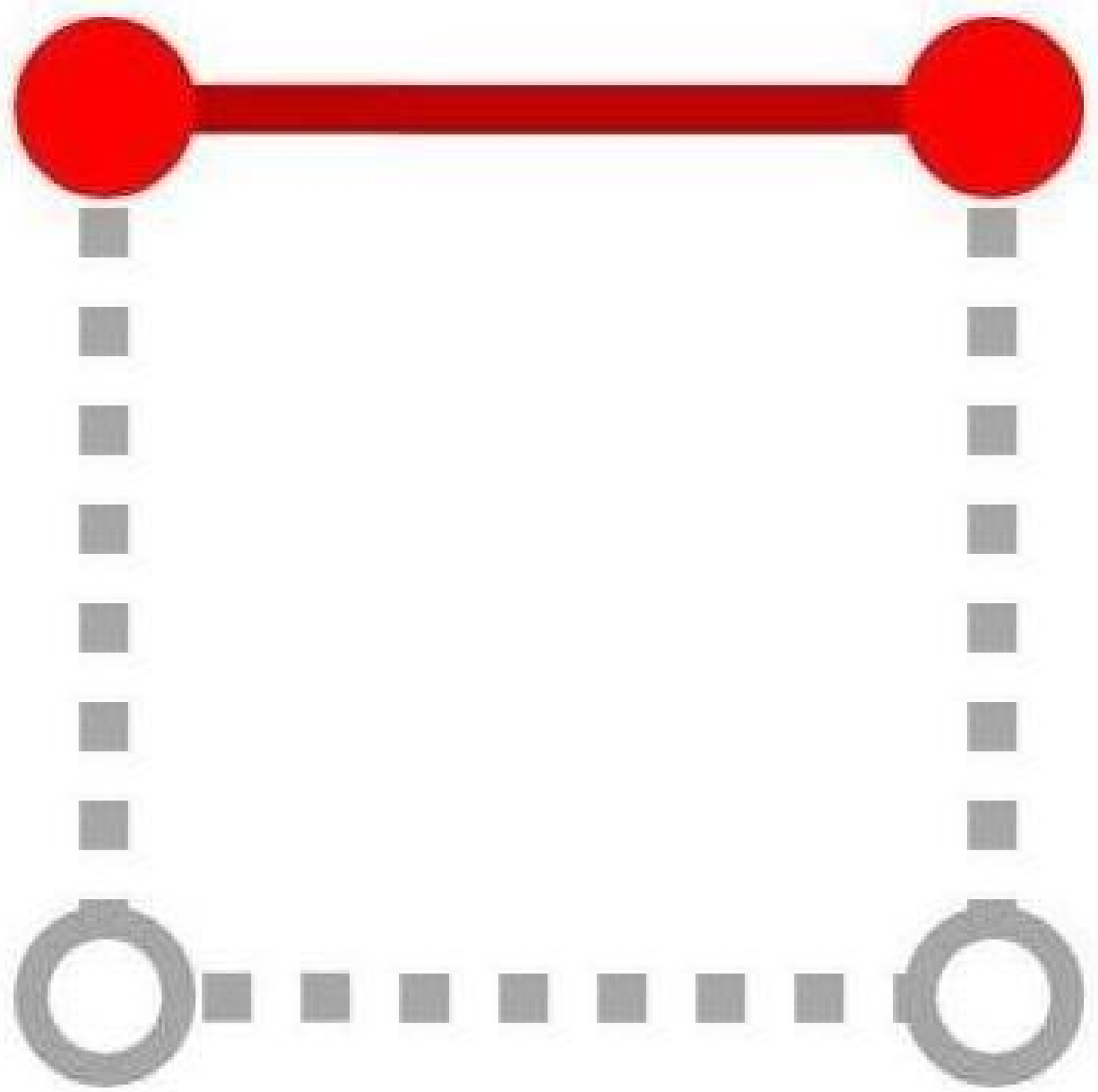}%
}%
&
\begin{array}
[c]{c}%
{\includegraphics[
height=0.237in,
width=0.5967in
]%
{arrow-ya.ps}%
}%
\\
F_{p}^{(\ell)}(a)
\end{array}
&
\raisebox{-0.5016in}{\includegraphics[
height=1.1372in,
width=1.1372in
]%
{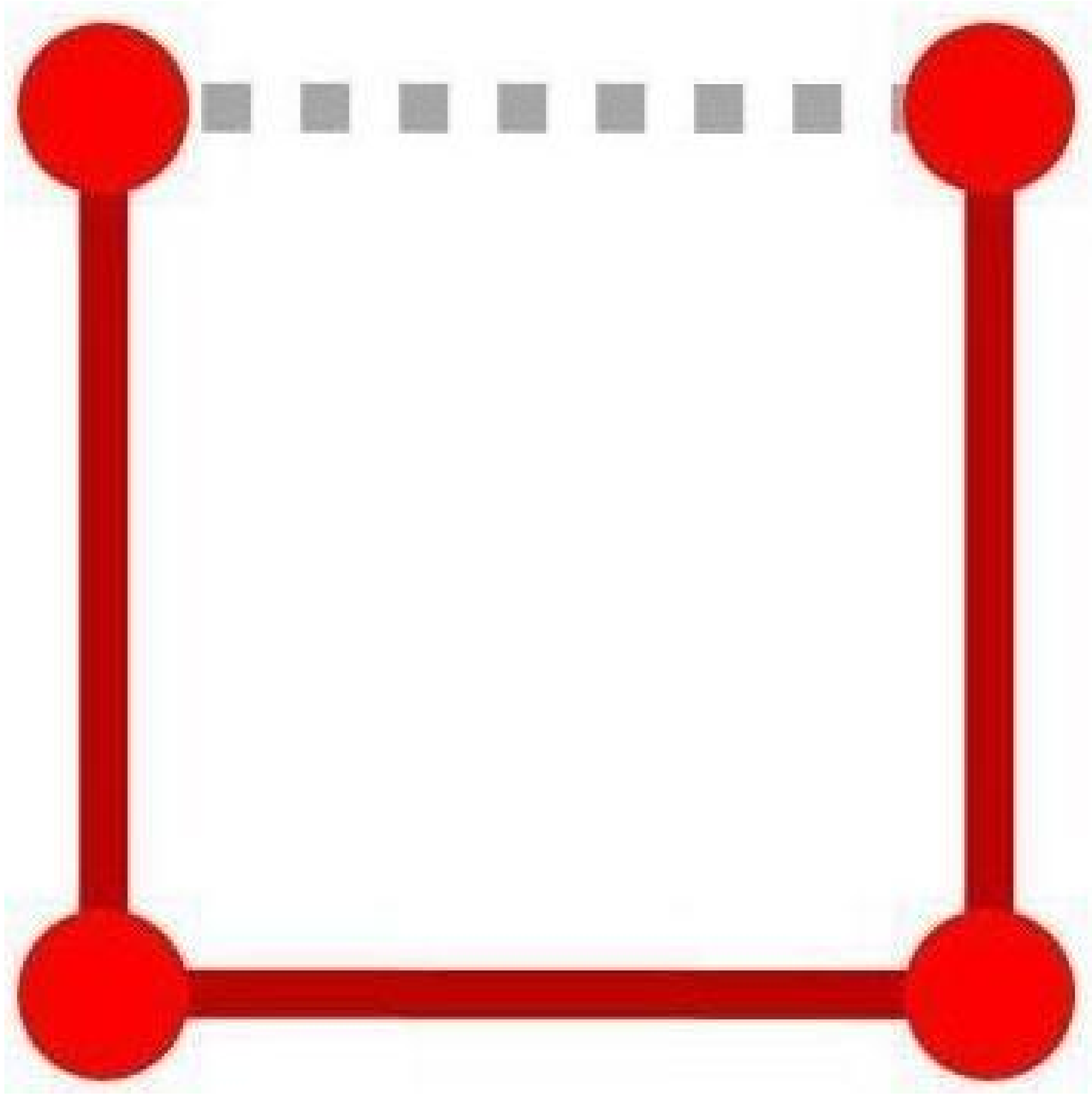}%
}%
\end{array}
\\
\text{\textbf{Lattice knot move }}L_{1}^{(\ell)}\left(  a,p,3\right)
=L_{1}^{(\ell)}\left(  a,p,%
\raisebox{-0.0406in}{\includegraphics[
height=0.1436in,
width=0.1436in
]%
{icon13.ps}%
}%
\right)  =%
\raisebox{-0.0406in}{\includegraphics[
height=0.1436in,
width=0.1436in
]%
{icon13.ps}%
}%
^{(\ell)}\left(  a,p\right)  \text{\textbf{, called a tug.}}%
\end{array}
\]

\end{definition}

\bigskip

\noindent\textbf{Terminology:} \ \textit{The designated edge of a tug move
will be frequently called the tug's }\textbf{extendable edge}.

\bigskip

\begin{remark}
For each cube $B_{\ell}(a)$, there are $12$ tug moves, i.e., $4$ for each of
the $3$ preferred faces.
\end{remark}

\bigskip

\subsection{Definition of the move wiggle}

\bigskip

The second move, called a \textbf{wiggle}, and denoted by
\[
L_{2}^{(\ell)}\left(  a,p,q\right)  \text{ ,}%
\]
is defined for each of the two diagonals $q=$`$/$' and $q=$`$\backslash$' of
each preferred face $F_{p}^{(\ell)}\left(  a\right)  $ of each cube
$B^{(\ell)}\left(  a\right)  $ in the cell complex $\mathcal{C}_{\ell}$.

\vspace{0.5in}

For reasons that will soon become apparent, we will denote the diagonal
$q=$`$/$' by either the symbol
\[
q=%
\raisebox{-0.0406in}{\includegraphics[
height=0.1436in,
width=0.1436in
]%
{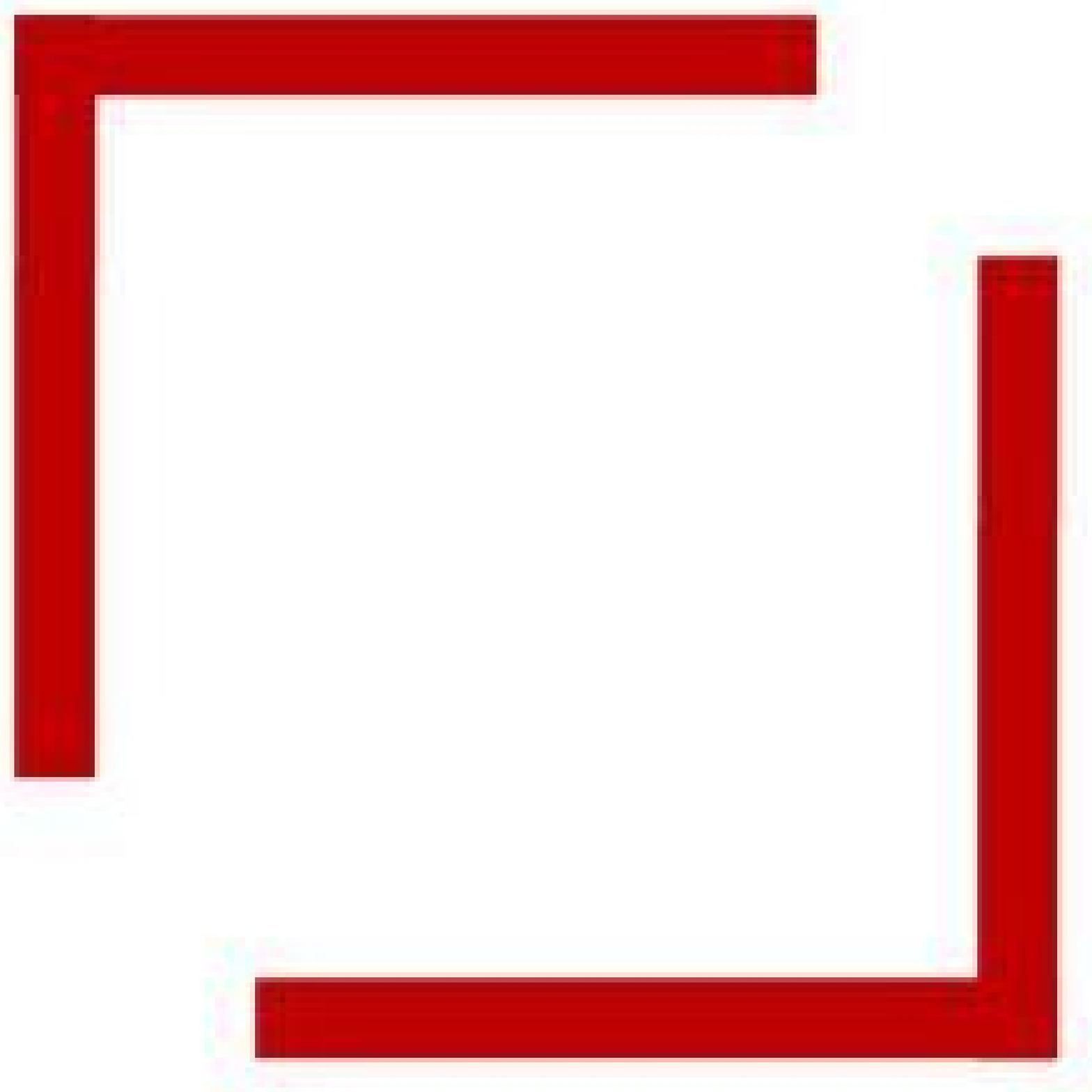}%
}%
\]
or by either one of the integers
\[
q=0\text{ \ \ or \ }q=2.
\]
Thus, the wiggle $L_{2}^{(\ell)}\left(  a,p,q\right)  $ with respect to
diagonal $q=$`$/$' of the preferred face $F_{p}^{(\ell)}\left(  a\right)  $ of
the cube $B^{(\ell)}\left(  a\right)  $ will be denoted in anyone of the
following three ways%
\[
L_{2}^{(\ell)}\left(  a,p,%
\raisebox{-0.0406in}{\includegraphics[
height=0.1436in,
width=0.1436in
]%
{icon21.ps}%
}%
\right)  =L_{2}^{(\ell)}\left(  a,p,0\right)  =L_{2}^{(\ell)}\left(
a,p,2\right)  \text{ ,}%
\]
or simply by
\[%
\raisebox{-0.0406in}{\includegraphics[
height=0.1436in,
width=0.1436in
]%
{icon21.ps}%
}%
^{(\ell)}\left(  a,p\right)  \text{ .}%
\]

\bigskip

In like manner, we will denote the diagonal $q=$`$/$' by either the symbol
\[
q=%
\raisebox{-0.0406in}{\includegraphics[
height=0.1436in,
width=0.1436in
]%
{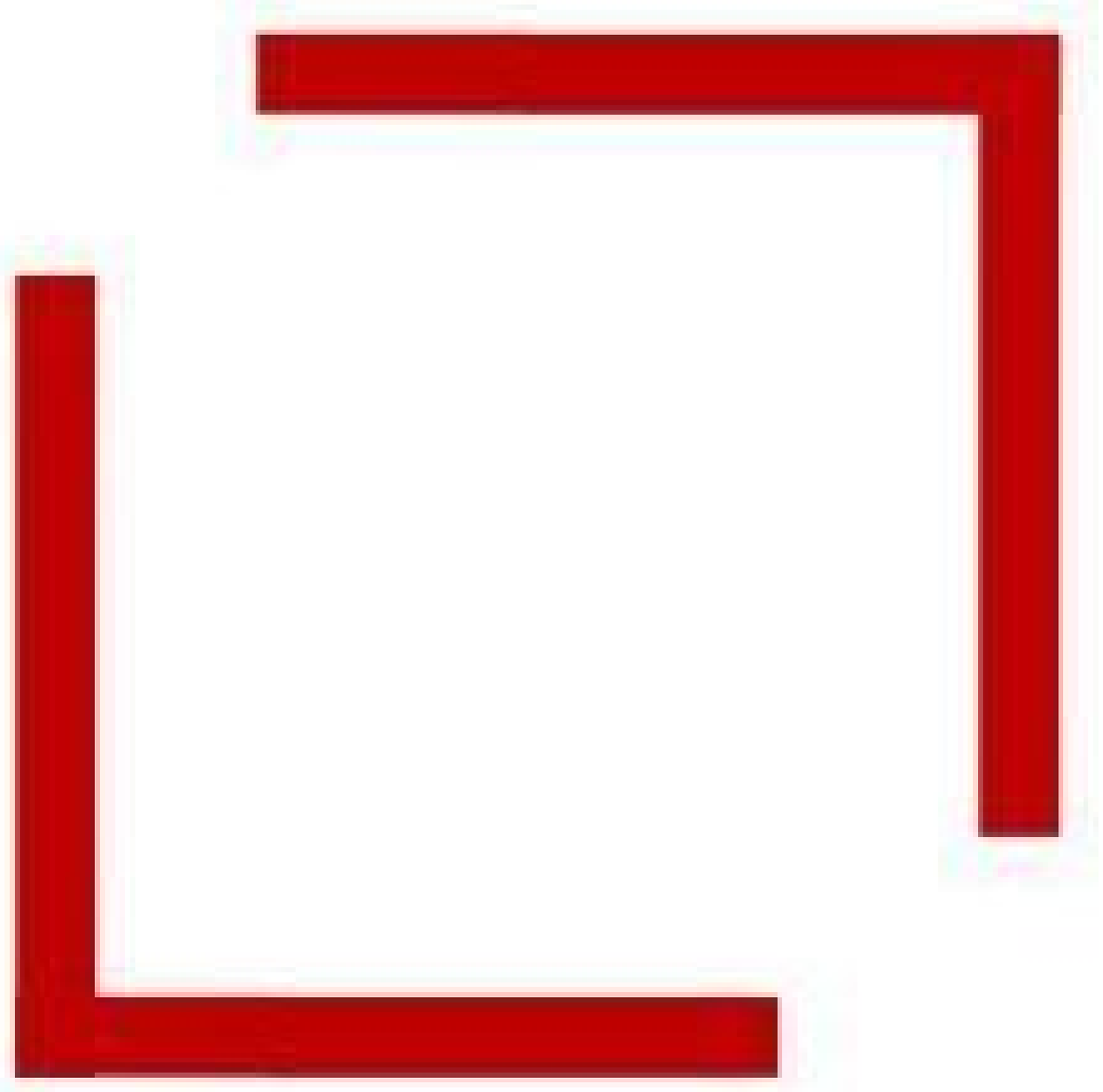}%
}%
\]
or by either one of the integers
\[
q=1\text{ \ \ or \ }q=3.
\]
Thus, the wiggle $L_{2}^{(\ell)}\left(  a,p,q\right)  $ with respect to
diagonal $q=$`$\backslash$' of the preferred face $F_{p}^{(\ell)}\left(
a\right)  $ of the cube $B^{(\ell)}\left(  a\right)  $ will be denoted in
anyone of the following three ways%
\[
L_{2}^{(\ell)}\left(  a,p,%
\raisebox{-0.0406in}{\includegraphics[
height=0.1436in,
width=0.1436in
]%
{icon20.ps}%
}%
\right)  =L_{2}^{(\ell)}\left(  a,p,0\right)  =L_{2}^{(\ell)}\left(
a,p,2\right)  \text{ ,}%
\]
or simply by
\[%
\raisebox{-0.0406in}{\includegraphics[
height=0.1436in,
width=0.1436in
]%
{icon20.ps}%
}%
^{(\ell)}\left(  a,p\right)  \text{ .}%
\]

\bigskip

\begin{definition}
The \textbf{wiggle} associated with the diagonal
\raisebox{-0.0406in}{\includegraphics[
height=0.1436in,
width=0.1436in
]%
{icon21.ps}%
}%
\ of the $p$-th preferred face $F_{p}^{(\ell)}(a)$ of the cube $B_{\ell}(a)$
on , written
\[
L_{2}^{(\ell)}\left(  a,p,%
\raisebox{-0.0406in}{\includegraphics[
height=0.1436in,
width=0.1436in
]%
{icon21.ps}%
}%
\right)  =L_{2}^{(\ell)}\left(  a,p,0\right)  =L_{2}^{(\ell)}\left(
a,p,2\right)  =%
\raisebox{-0.0406in}{\includegraphics[
height=0.1436in,
width=0.1436in
]%
{icon21.ps}%
}%
^{\!(\ell)}\left(  a,p\right)  \text{ ,}%
\]
is defined by

\bigskip%

\[
\hspace{-1in}%
\raisebox{-0.0406in}{\includegraphics[
height=0.1436in,
width=0.1436in
]%
{icon21.ps}%
}%
^{\!(\ell)}(a,p)(K)=\left\{
\begin{array}
[c]{ll}%
\left(  K-%
\raisebox{-0.2508in}{\includegraphics[
height=0.6564in,
width=0.6564in
]%
{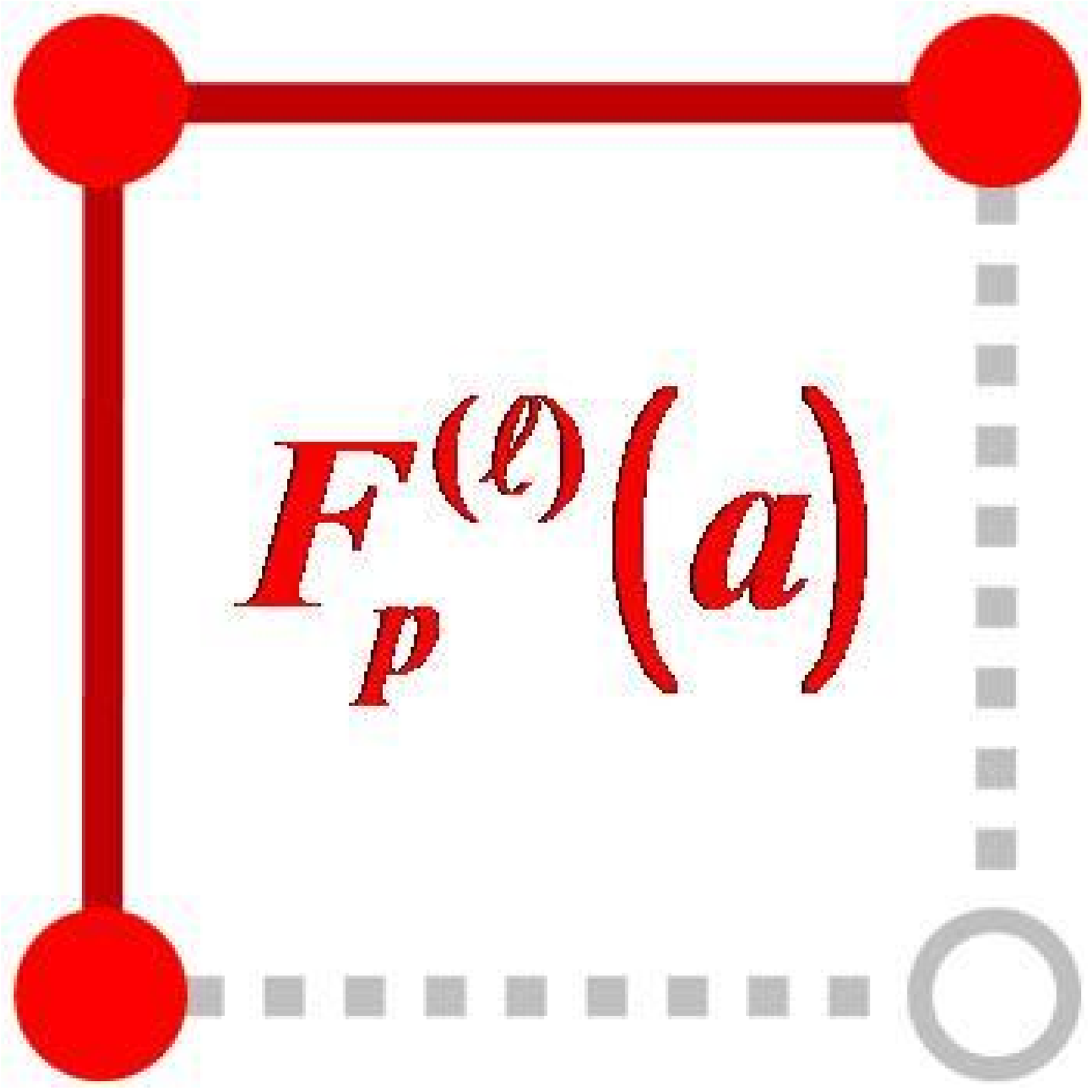}%
}%
\right)  \cup\left(  \mathcal{L}_{\ell}^{1}\cap%
\raisebox{-0.2508in}{\includegraphics[
height=0.6564in,
width=0.6564in
]%
{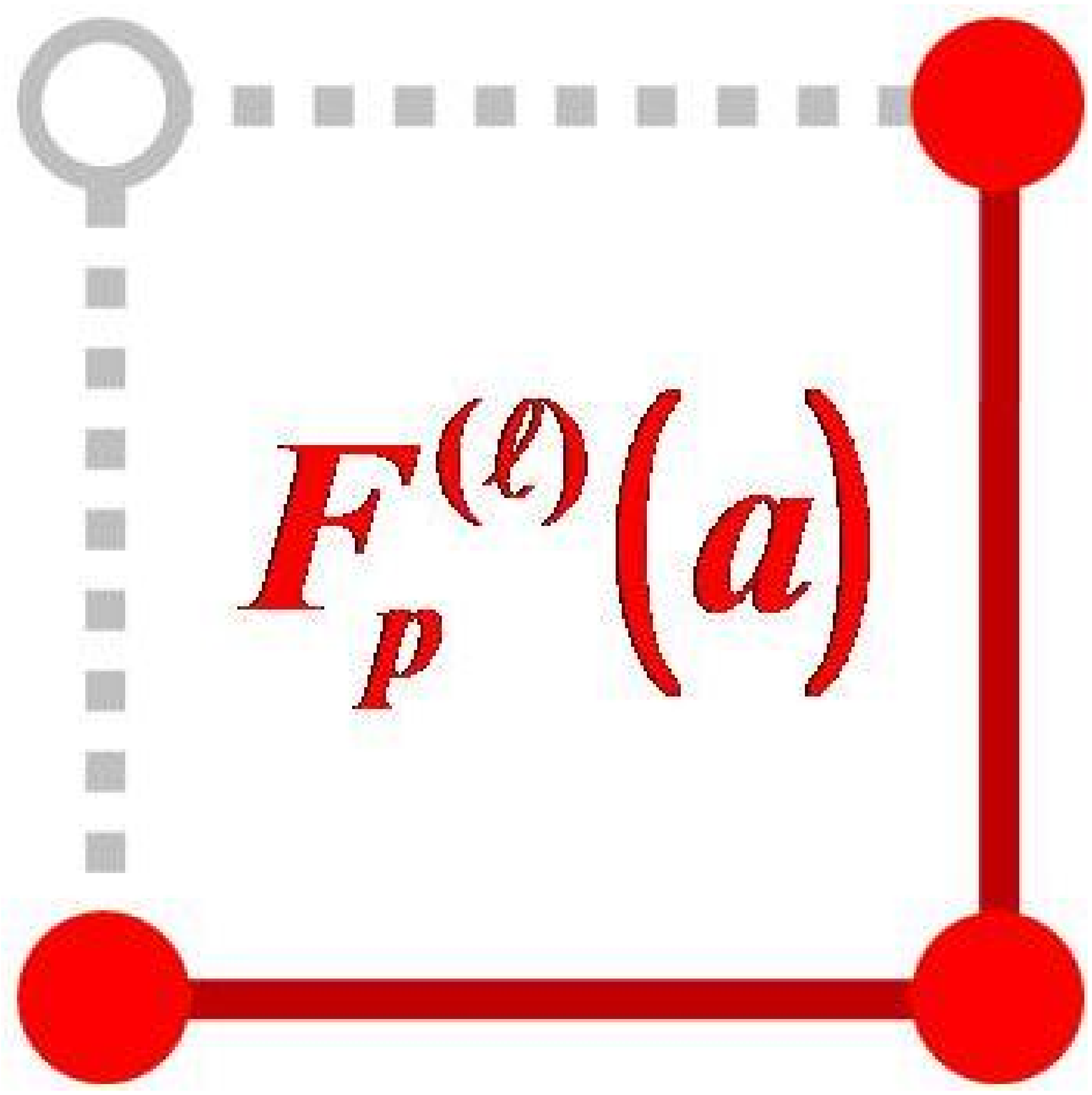}%
}%
\right)  & \text{if \ }K\cap%
\raisebox{-0.2508in}{\includegraphics[
height=0.6478in,
width=0.6478in
]%
{face.ps}%
}%
=\mathcal{L}_{\ell}^{1}\cap%
\raisebox{-0.2508in}{\includegraphics[
height=0.6564in,
width=0.6564in
]%
{wig0l.ps}%
}%
\\
& \\
\left(  K-%
\raisebox{-0.2508in}{\includegraphics[
height=0.6564in,
width=0.6564in
]%
{wig0r.ps}%
}%
\right)  \cup\left(  \mathcal{L}_{\ell}^{1}\cap%
\raisebox{-0.2508in}{\includegraphics[
height=0.6564in,
width=0.6564in
]%
{wig0l.ps}%
}%
\right)  & \text{if \ }K\cap%
\raisebox{-0.2508in}{\includegraphics[
height=0.6478in,
width=0.6478in
]%
{face.ps}%
}%
=\mathcal{L}_{\ell}^{1}\cap%
\raisebox{-0.2508in}{\includegraphics[
height=0.6564in,
width=0.6564in
]%
{wig0r.ps}%
}%
\\
& \\
K & \text{otherwise}%
\end{array}
\right.
\]
\bigskip It is more succinctly illustrated in the figure given below:

\bigskip%

\[%
\begin{array}
[c]{c}%
\begin{array}
[c]{ccc}%
\raisebox{-0.5016in}{\includegraphics[
height=1.1372in,
width=1.1372in
]%
{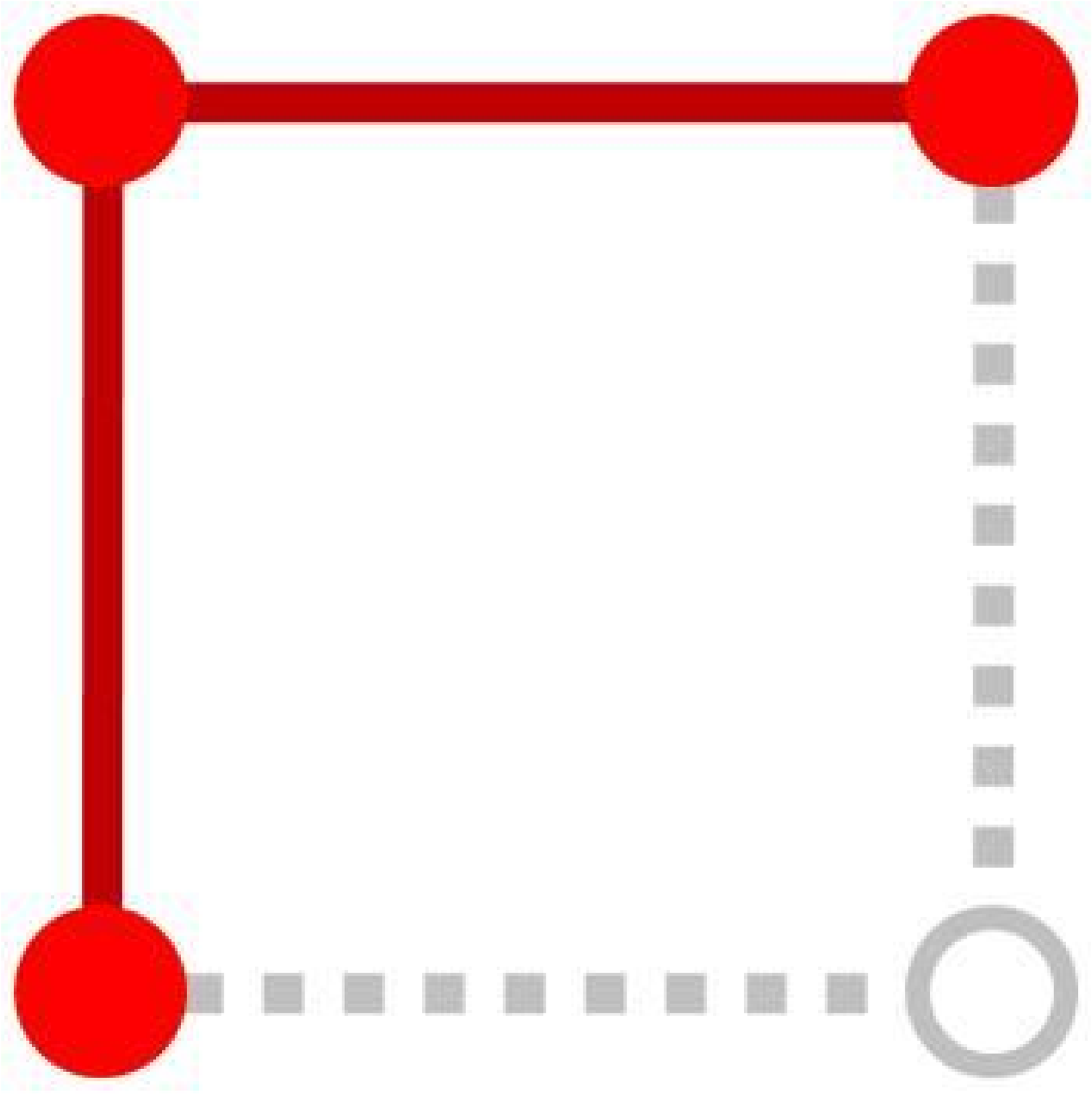}%
}%
&
\begin{array}
[c]{c}%
{\includegraphics[
height=0.237in,
width=0.5967in
]%
{arrow-ya.ps}%
}%
\\
F_{p}^{(\ell)}(a)
\end{array}
&
\raisebox{-0.5016in}{\includegraphics[
height=1.1372in,
width=1.1372in
]%
{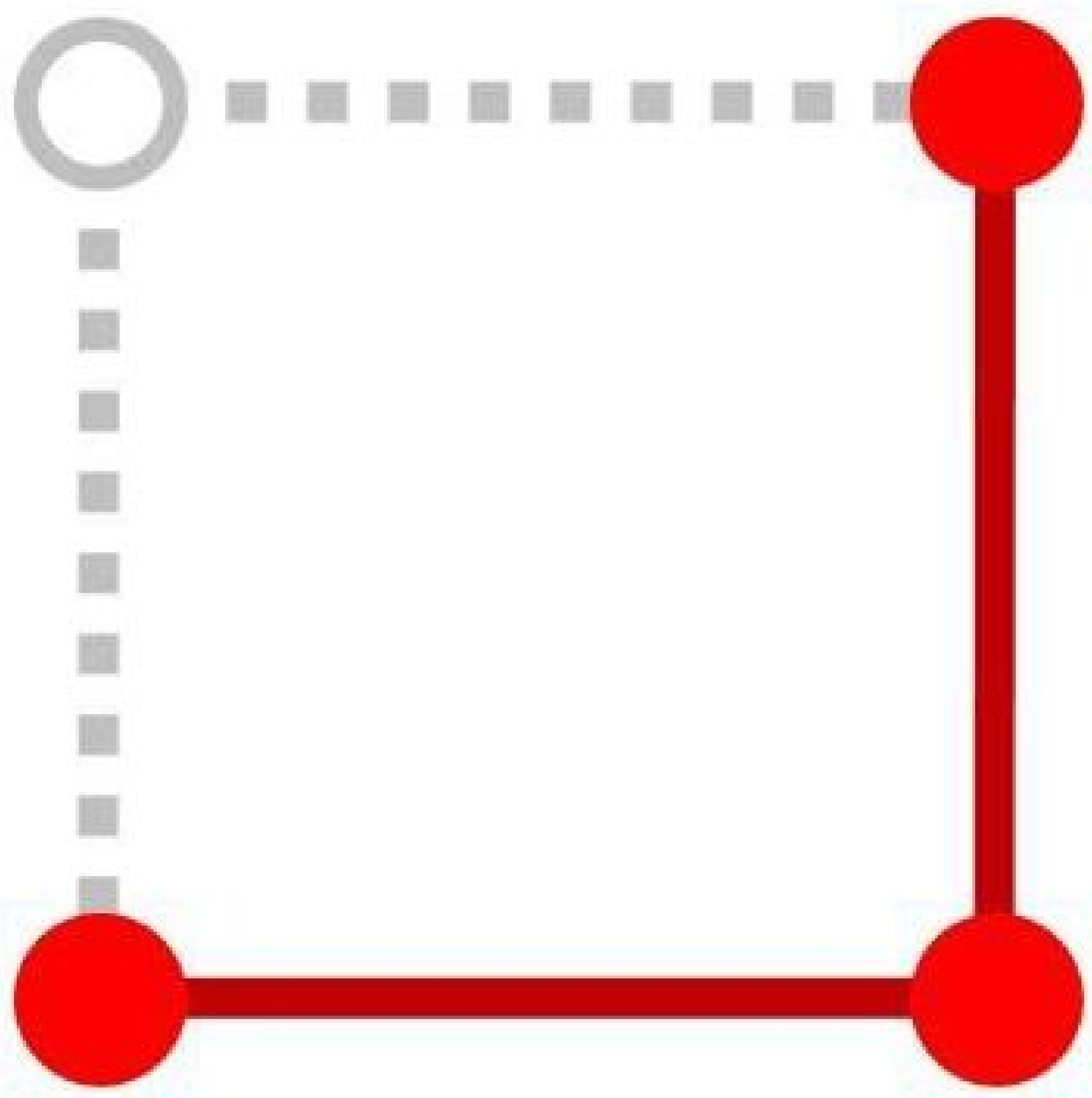}%
}%
\end{array}
\\
\text{\textbf{Lattice knot move }}L_{2}^{(\ell)}\left(  a,p,%
\raisebox{-0.0406in}{\includegraphics[
height=0.1436in,
width=0.1436in
]%
{icon21.ps}%
}%
\right)  =\text{\textbf{ }}L_{2}^{(\ell)}\left(  a,p,0\right)  =L_{2}^{(\ell
)}\left(  a,p,2\right)  =%
\raisebox{-0.0406in}{\includegraphics[
height=0.1436in,
width=0.1436in
]%
{icon21.ps}%
}%
^{\!(\ell)}(a,p)\text{\textbf{, called a wiggle.}}%
\end{array}
\]
\bigskip The remaining \textbf{wiggle}, $L_{2}^{(\ell)}\left(  a,p,%
\raisebox{-0.0406in}{\includegraphics[
height=0.1436in,
width=0.1436in
]%
{icon20.ps}%
}%
\right)  $ is defined in like manner as\bigskip

$\hspace{-1in}%
\raisebox{-0.0406in}{\includegraphics[
height=0.1436in,
width=0.1436in
]%
{icon20.ps}%
}%
^{\!(\ell)}(a,p)(K)=\left\{
\begin{array}
[c]{ll}%
\left(  K-%
\raisebox{-0.2508in}{\includegraphics[
height=0.6564in,
width=0.6564in
]%
{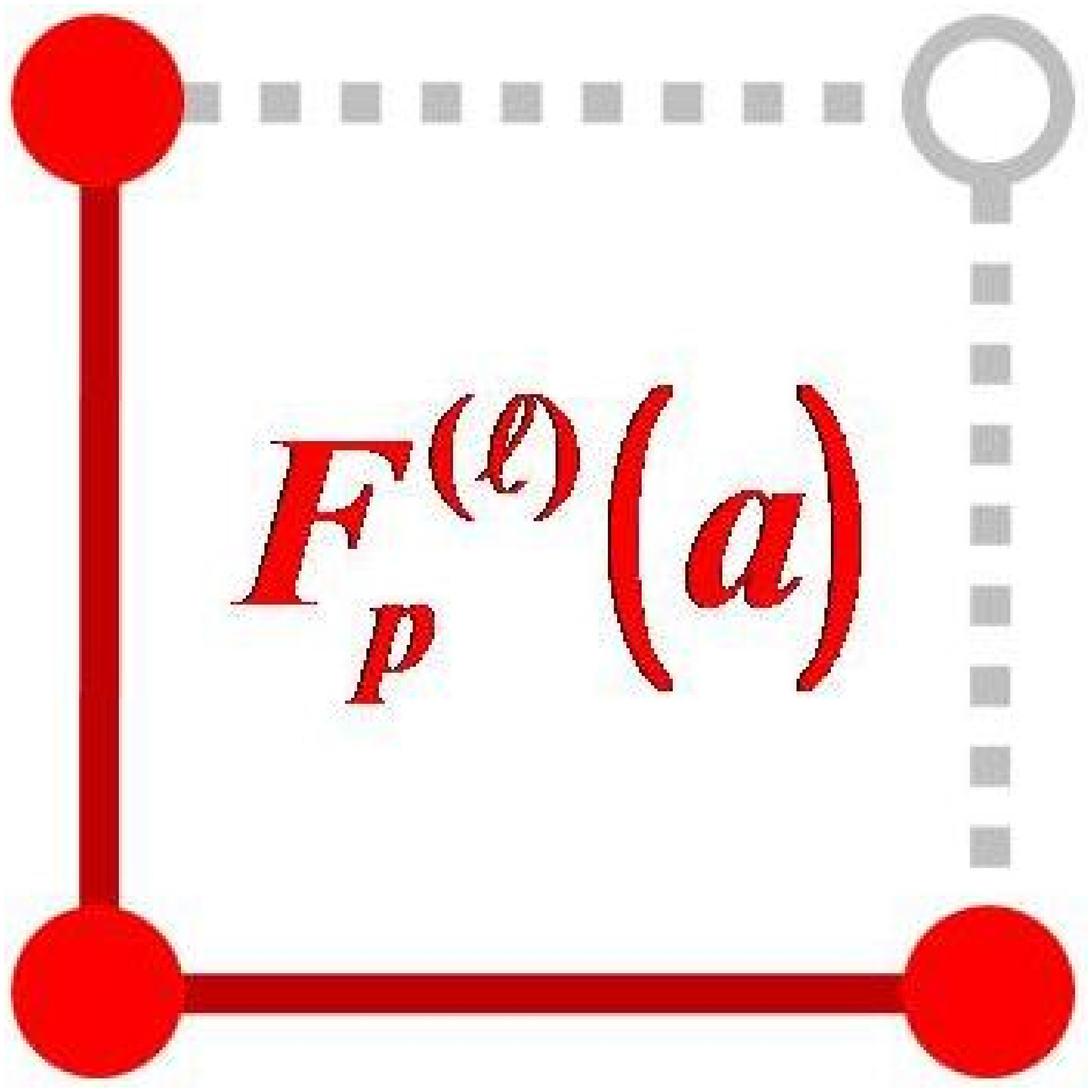}%
}%
\right)  \cup\left(  \mathcal{L}_{\ell}^{1}\cap%
\raisebox{-0.2508in}{\includegraphics[
height=0.6564in,
width=0.6564in
]%
{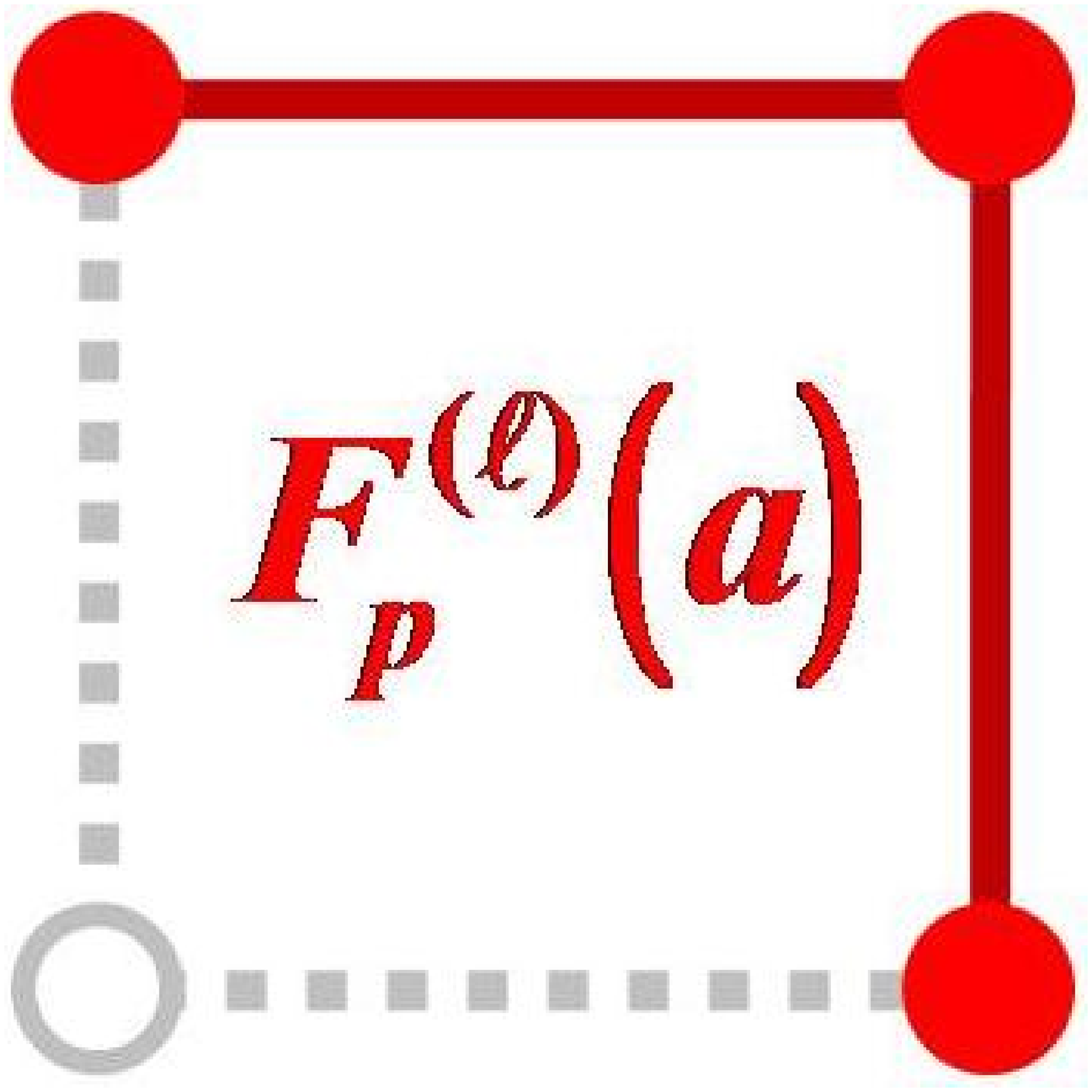}%
}%
\right)  & \text{if \ }K\cap%
\raisebox{-0.2508in}{\includegraphics[
height=0.6478in,
width=0.6478in
]%
{face.ps}%
}%
=\mathcal{L}_{\ell}^{1}\cap%
\raisebox{-0.2508in}{\includegraphics[
height=0.6564in,
width=0.6564in
]%
{wig1l.ps}%
}%
\\
& \\
\left(  K-%
\raisebox{-0.2508in}{\includegraphics[
height=0.6564in,
width=0.6564in
]%
{wig1r.ps}%
}%
\right)  \cup\left(  \mathcal{L}_{\ell}^{1}\cap%
\raisebox{-0.2508in}{\includegraphics[
height=0.6564in,
width=0.6564in
]%
{wig1l.ps}%
}%
\right)  & \text{if \ }K\cap%
\raisebox{-0.2508in}{\includegraphics[
height=0.6478in,
width=0.6478in
]%
{face.ps}%
}%
=\mathcal{L}_{\ell}^{1}\cap%
\raisebox{-0.2508in}{\includegraphics[
height=0.6564in,
width=0.6564in
]%
{wig1r.ps}%
}%
\\
& \\
K & \text{otherwise}%
\end{array}
\right.  $

\bigskip

\noindent And it is illustrated more succinctly in the figure given below

\bigskip%

\[%
\begin{array}
[c]{c}%
\begin{array}
[c]{ccc}%
\raisebox{-0.5016in}{\includegraphics[
height=1.1372in,
width=1.1372in
]%
{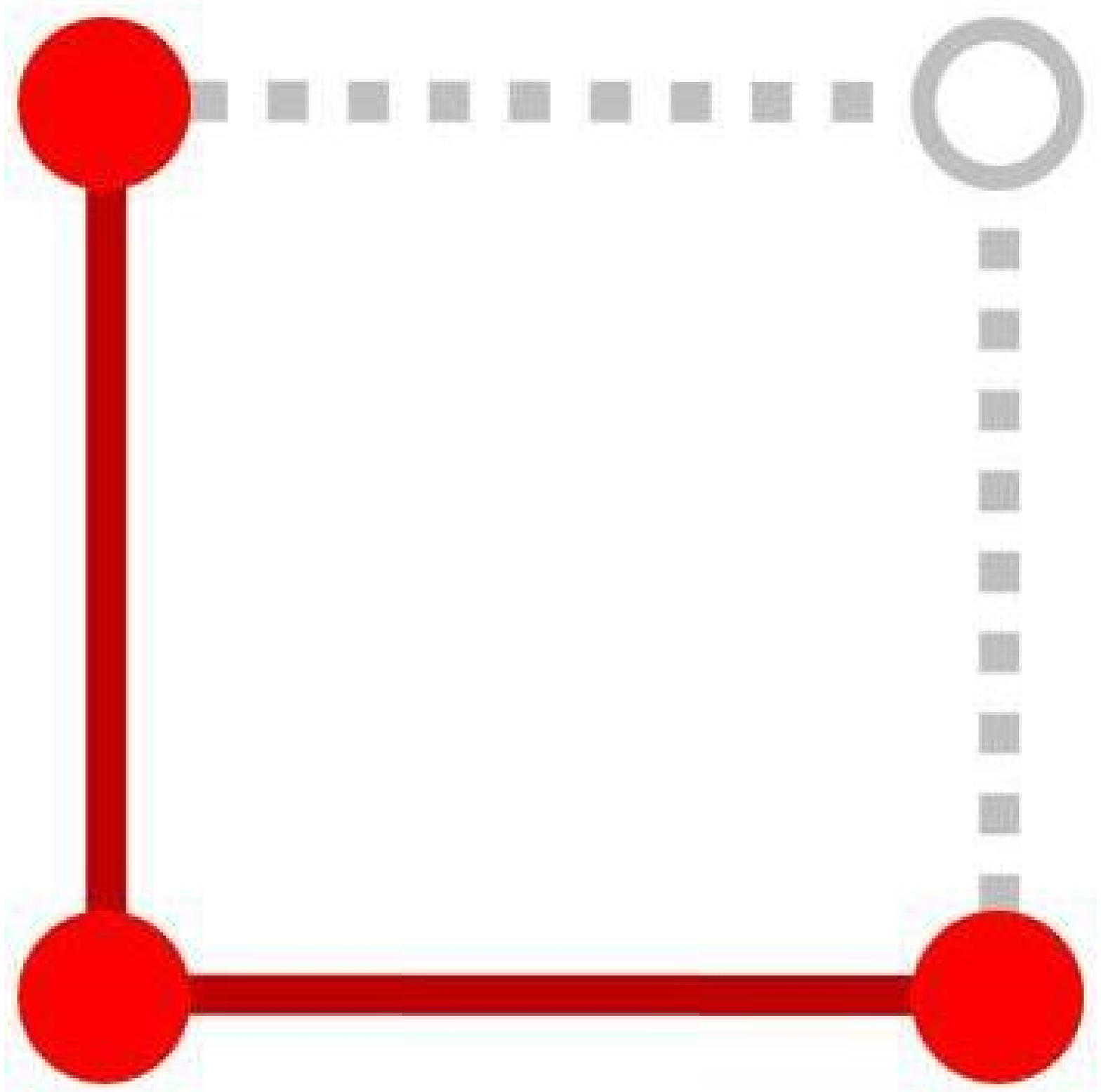}%
}%
&
\begin{array}
[c]{c}%
{\includegraphics[
height=0.237in,
width=0.5967in
]%
{arrow-ya.ps}%
}%
\\
F_{p}^{(\ell)}(a)
\end{array}
&
\raisebox{-0.5016in}{\includegraphics[
height=1.1372in,
width=1.1372in
]%
{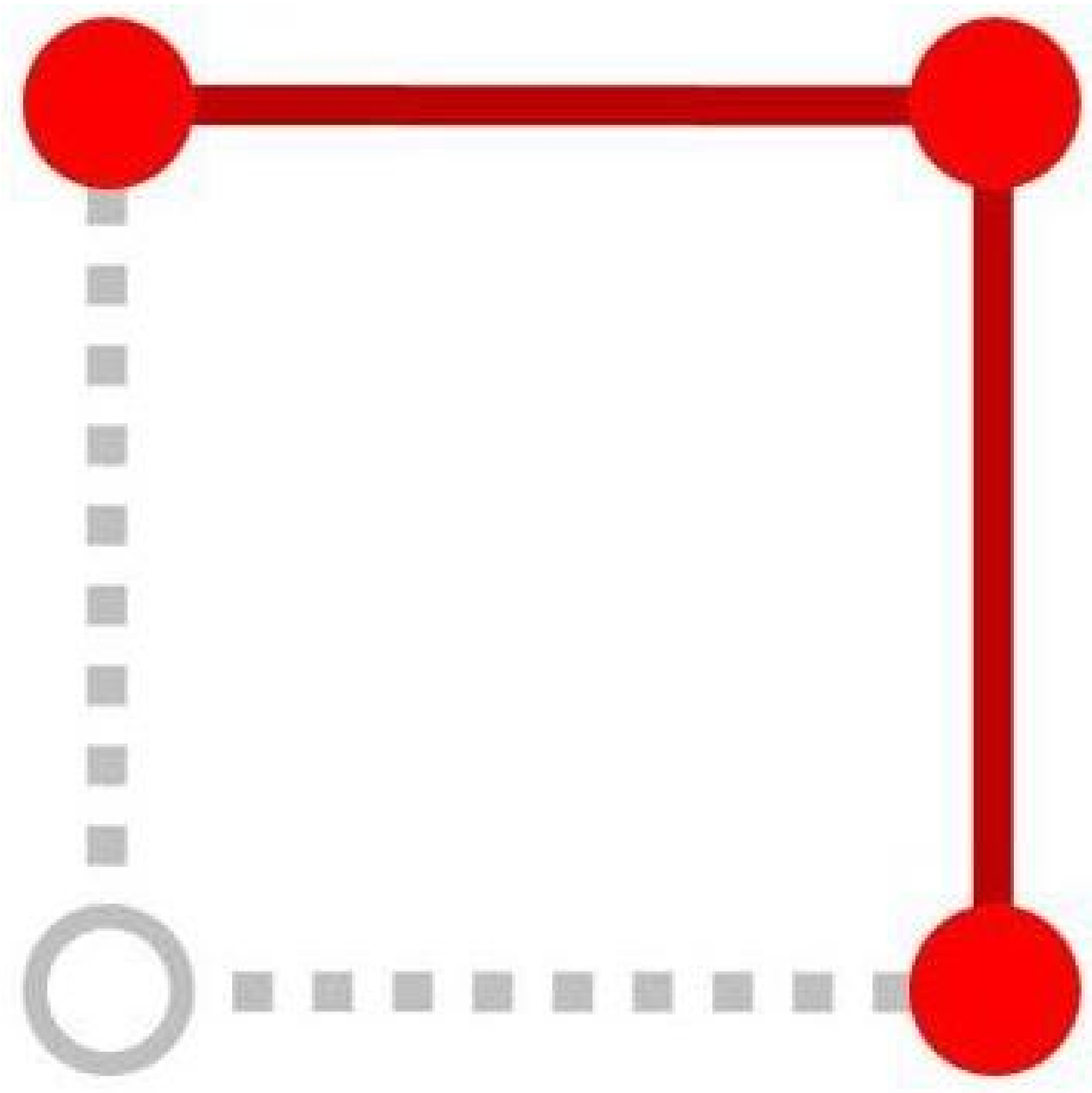}%
}%
\end{array}
\\
\text{\textbf{Lattice knot move }}L_{2}^{(\ell)}\left(  a,p,%
\raisebox{-0.0406in}{\includegraphics[
height=0.1436in,
width=0.1436in
]%
{icon20.ps}%
}%
\right)  =\text{\textbf{ }}L_{2}^{(\ell)}\left(  a,p,0\right)  =L_{2}^{(\ell
)}\left(  a,p,2\right)  =%
\raisebox{-0.0406in}{\includegraphics[
height=0.1436in,
width=0.1436in
]%
{icon20.ps}%
}%
^{\!(\ell)}(a,p)\text{\textbf{, called a wiggle.}}%
\end{array}
\text{ \ ,}%
\]

\end{definition}

\bigskip

\begin{remark}
For each cube $B^{(\ell)}(a)$, there are $6$ wiggle moves, $2$ for each of the
$3$ preferred faces.
\end{remark}

\bigskip

\subsection{Definition of the move wag}

\bigskip

The third move, called a \textbf{wag}, and denoted by%
\[
L_{3}^{(\ell)}\left(  a,p,q\right)  \text{ ,}%
\]
is defined for each of the four perpendicular edges of a preferred face
$F_{p}^{(\ell)}(a)$ of a cube $B^{(\ell)}(a)$ in the cell complex
$\mathcal{C}_{\ell}$. \ As indicated in the figure given below, we index the
four edges of the cube $B^{(\ell)}(a)$, which are perpendicular to a preferred
face $F_{p}^{(\ell)}(a)$, beginning with the preferred edge $E_{p}^{(\ell
)}(a)$ perpendicular to $F_{p}^{(\ell)}(a)$ at $a$, with the integers $0$,
$1$, $2$, $3$ (or respectively with the symbols
\raisebox{-0.0406in}{\includegraphics[
height=0.1531in,
width=0.1531in
]%
{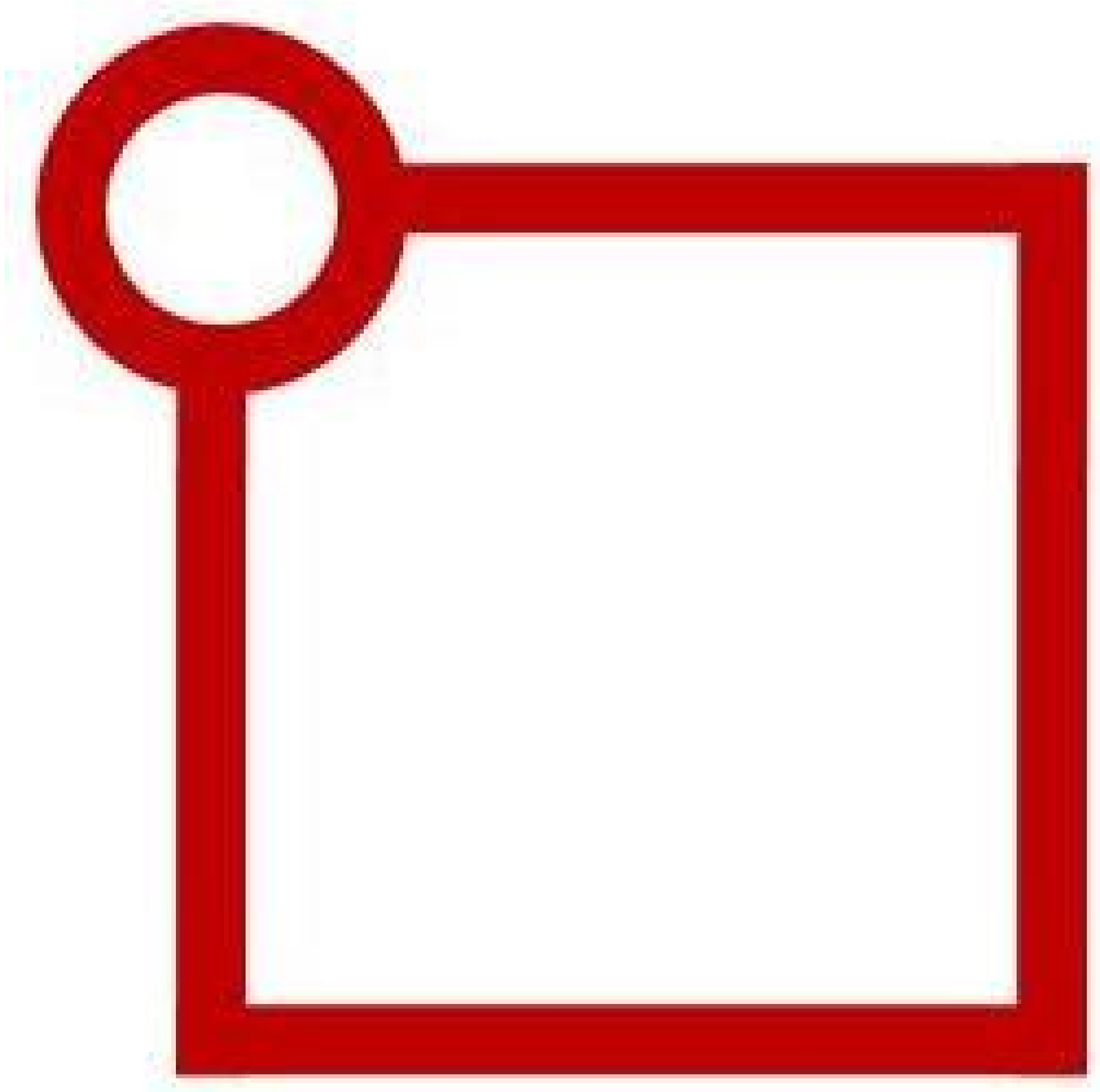}%
}%
,
\raisebox{-0.0406in}{\includegraphics[
height=0.1531in,
width=0.1531in
]%
{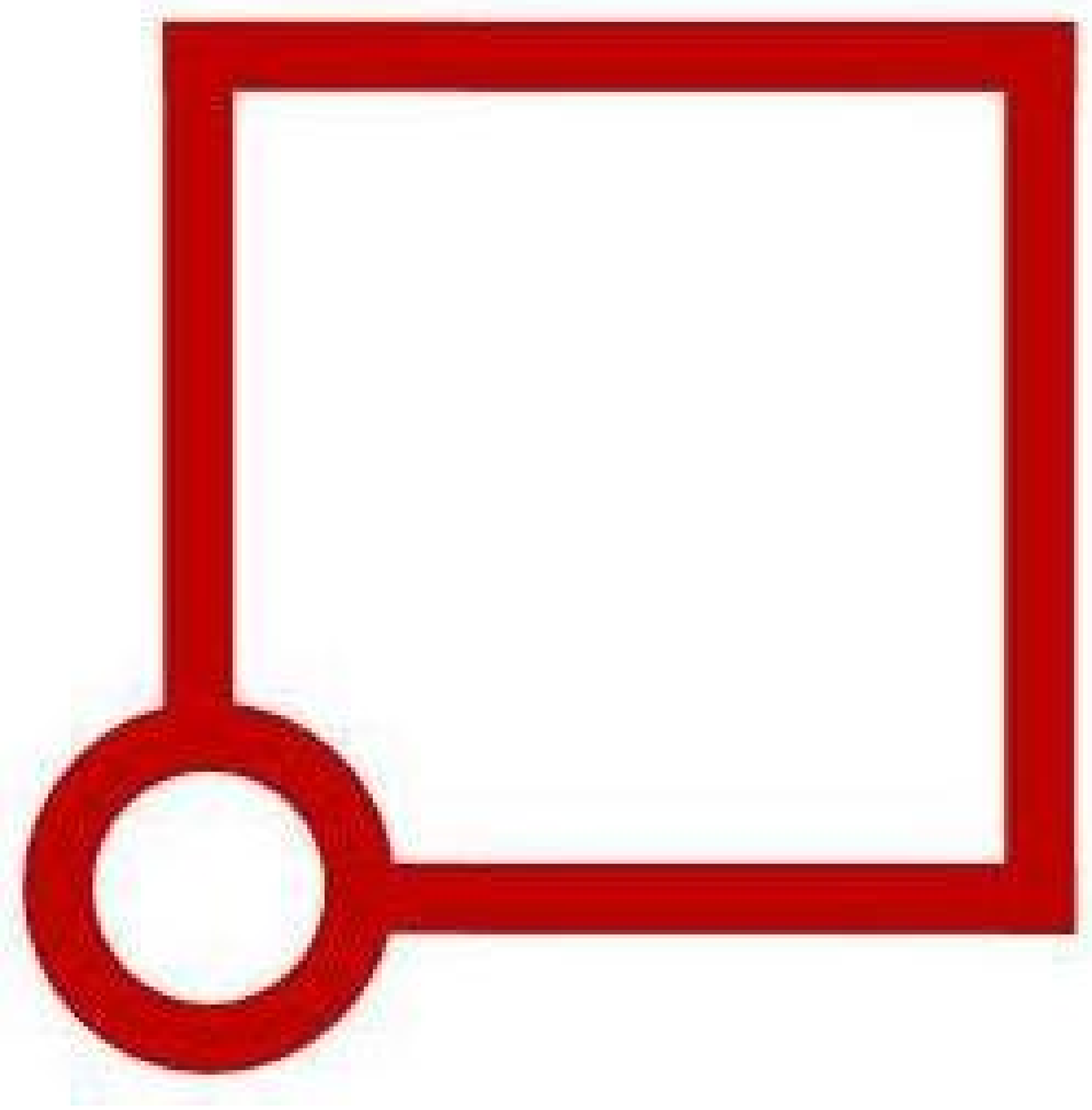}%
}%
,
\raisebox{-0.0406in}{\includegraphics[
height=0.1531in,
width=0.1531in
]%
{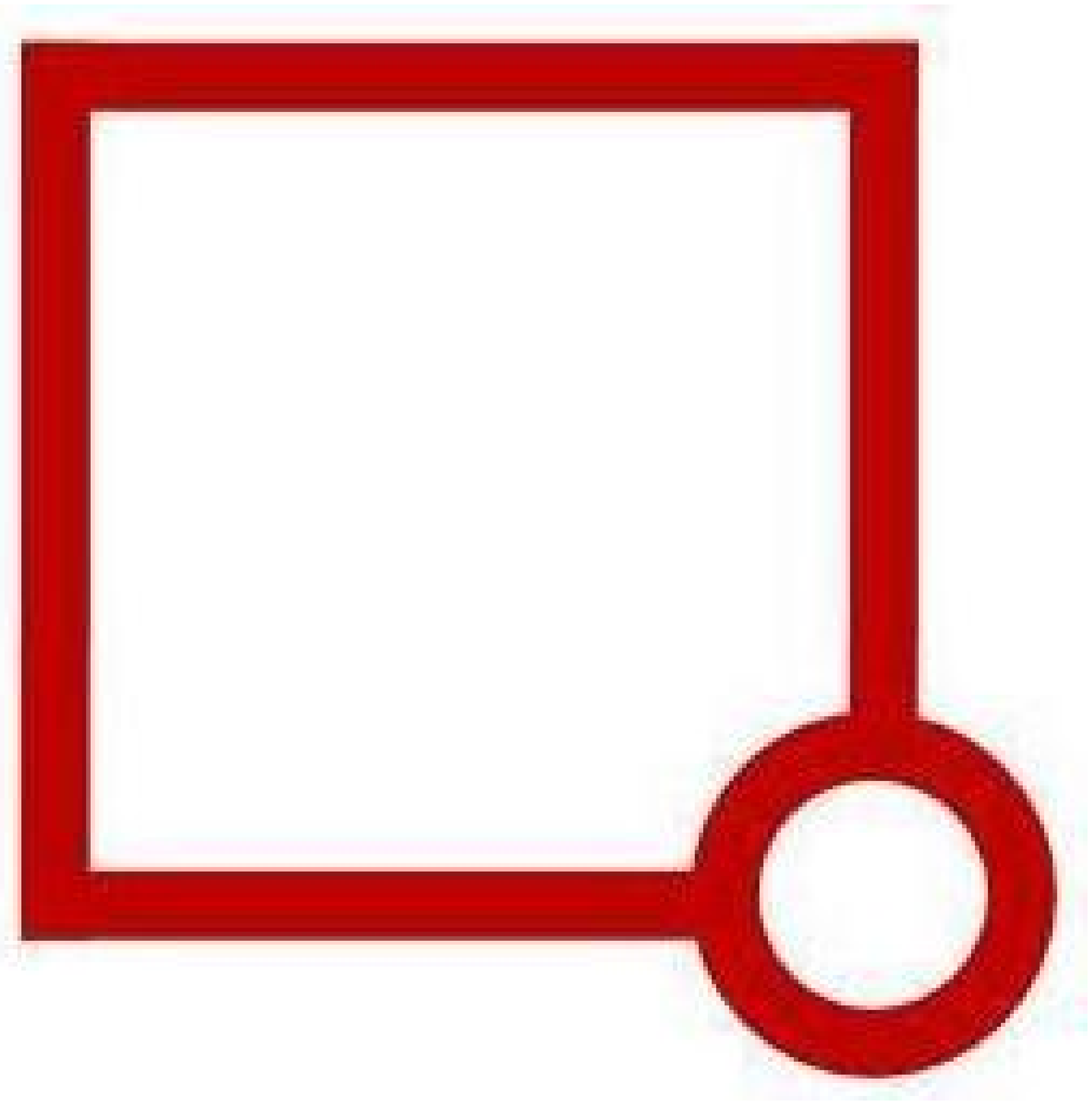}%
}%
,
\raisebox{-0.0406in}{\includegraphics[
height=0.1531in,
width=0.1531in
]%
{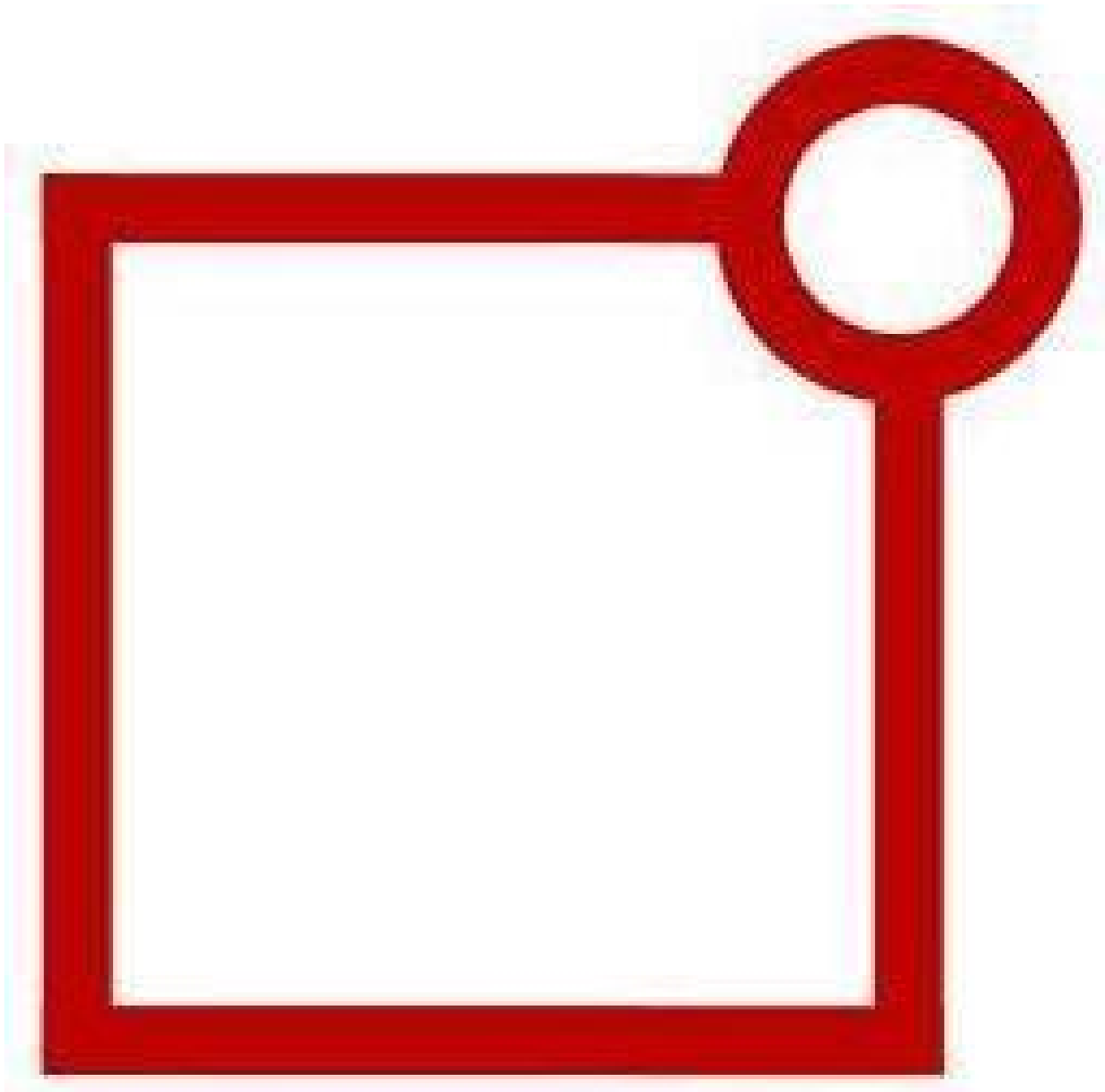}%
}%
), using the counterclockwise orientation induced on the face $F_{p}^{(\ell
)}(a)$ by the preferred frame $e$. \ The chosen perpendicular edge will be
called the \textbf{hinge} of the wag. \ 

\bigskip%
\begin{center}
\includegraphics[
height=3.5267in,
width=5.028in
]%
{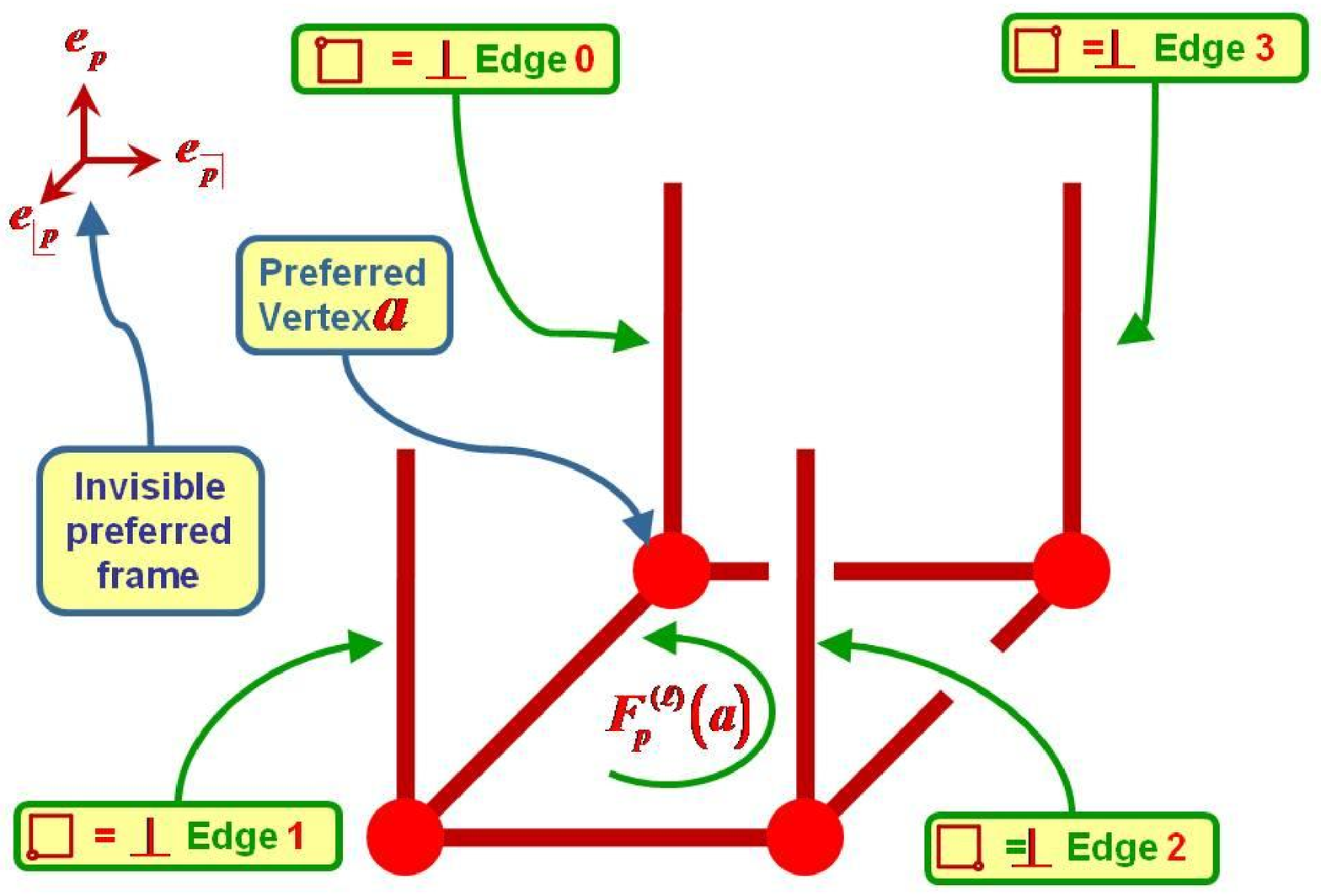}%
\\
\textbf{Edge ordering conventions for the }$L_{3}^{(\ell)}\left(
a,p,q\right)  $\textbf{ wag move, where }$q=0,1,2,3$ (or respectively by
$q=$\raisebox{-0.0406in}{\includegraphics[
height=0.1531in,
width=0.1531in
]%
{icon30.ps}%
}%
, \raisebox{-0.0406in}{\includegraphics[
height=0.1531in,
width=0.1531in
]%
{icon31.ps}%
}%
, \raisebox{-0.0406in}{\includegraphics[
height=0.1531in,
width=0.1531in
]%
{icon32.ps}%
}%
, \raisebox{-0.0406in}{\includegraphics[
height=0.1531in,
width=0.1531in
]%
{icon33.ps}%
}
).
\end{center}

\bigskip

We display the figure below as a cryptic reminder for the reader of the
notational conventions for the preferred edges and preferred faces of the cube
$B^{(\ell)}(a)$ which are defined in a previous section of this paper:

\bigskip%

\begin{center}
\fbox{\includegraphics[
height=2.6273in,
width=4.7305in
]%
{cubefaces.ps}%
}\\
\textbf{Preferred vertex }$a$\textbf{, preferred edges }$E_{p}^{(\ell)}%
(a)$\textbf{, }$E_{\left\lfloor p\right.  }^{(\ell)}(a)$\textbf{, }$E_{\left.
p\right\rceil }^{(\ell)}(a)$\textbf{, and preferred faces }$F_{p}^{(\ell)}%
(a)$\textbf{, }$F_{\left\lfloor p\right.  }^{(\ell)}(a)$\textbf{, }$F_{\left.
p\right\rceil }^{(\ell)}(a)$\textbf{ of cube }$B^{(\ell)}(a)$\textbf{.}%
\end{center}

\bigskip

The wag $L_{3}^{(\ell)}\left(  a,p,q\right)  $ associated with the hinge $q=0$
(also denoted by $q=%
\raisebox{-0.0406in}{\includegraphics[
height=0.1531in,
width=0.1531in
]%
{icon30.ps}%
}%
$ ) of the preferred face $F_{p}^{(\ell)}(a)$ of the cube $B^{(\ell)}(a)$ will
be denoted in any one of the following three ways%
\[
L_{3}^{(\ell)}\left(  a,p,%
\raisebox{-0.0406in}{\includegraphics[
height=0.1531in,
width=0.1531in
]%
{icon30.ps}%
}%
\right)  =L_{3}^{(\ell)}\left(  a,p,0\right)  =%
\raisebox{-0.0406in}{\includegraphics[
height=0.1531in,
width=0.1531in
]%
{icon30.ps}%
}%
^{(\ell)}\left(  a,p\right)  \text{ .}%
\]
The remaining three tugs $L_{3}^{(\ell)}\left(  a,p,q\right)  $, for $q=1,2,3$
(also indicated respectively by $q=$%
\raisebox{-0.0406in}{\includegraphics[
height=0.1531in,
width=0.1531in
]%
{icon31.ps}%
}%
,
\raisebox{-0.0406in}{\includegraphics[
height=0.1531in,
width=0.1531in
]%
{icon32.ps}%
}%
,
\raisebox{-0.0406in}{\includegraphics[
height=0.1531in,
width=0.1531in
]%
{icon33.ps}%
}%
), are denoted in like manner.

\bigskip

\begin{definition}
We define the \textbf{wag}, written $L_{3}^{(\ell)}\left(  a,p,0\right)  $
(also denoted by $L_{3}^{(\ell)}\left(  a,p,%
\raisebox{-0.0406in}{\includegraphics[
height=0.1531in,
width=0.1531in
]%
{icon30.ps}%
}%
\right)  $ and $%
\raisebox{-0.0406in}{\includegraphics[
height=0.1531in,
width=0.1531in
]%
{icon30.ps}%
}%
^{(\ell)}\left(  a,p\right)  $ .), associated with the $0$-th perpendicular
edge (called the $0$-th \textbf{hinge}) to the preferred face $F_{p}^{(\ell
)}(a)$ of the cube $B^{(\ell)}(a)$ as the move

\bigskip%
\[%
\raisebox{-0.0406in}{\includegraphics[
height=0.1531in,
width=0.1531in
]%
{icon30.ps}%
}%
^{\!(\ell)}\left(  a,p\right)  (K)=\left\{
\begin{array}
[c]{ll}%
\left(  K-%
\raisebox{-0.3009in}{\includegraphics[
height=0.7066in,
width=0.6694in
]%
{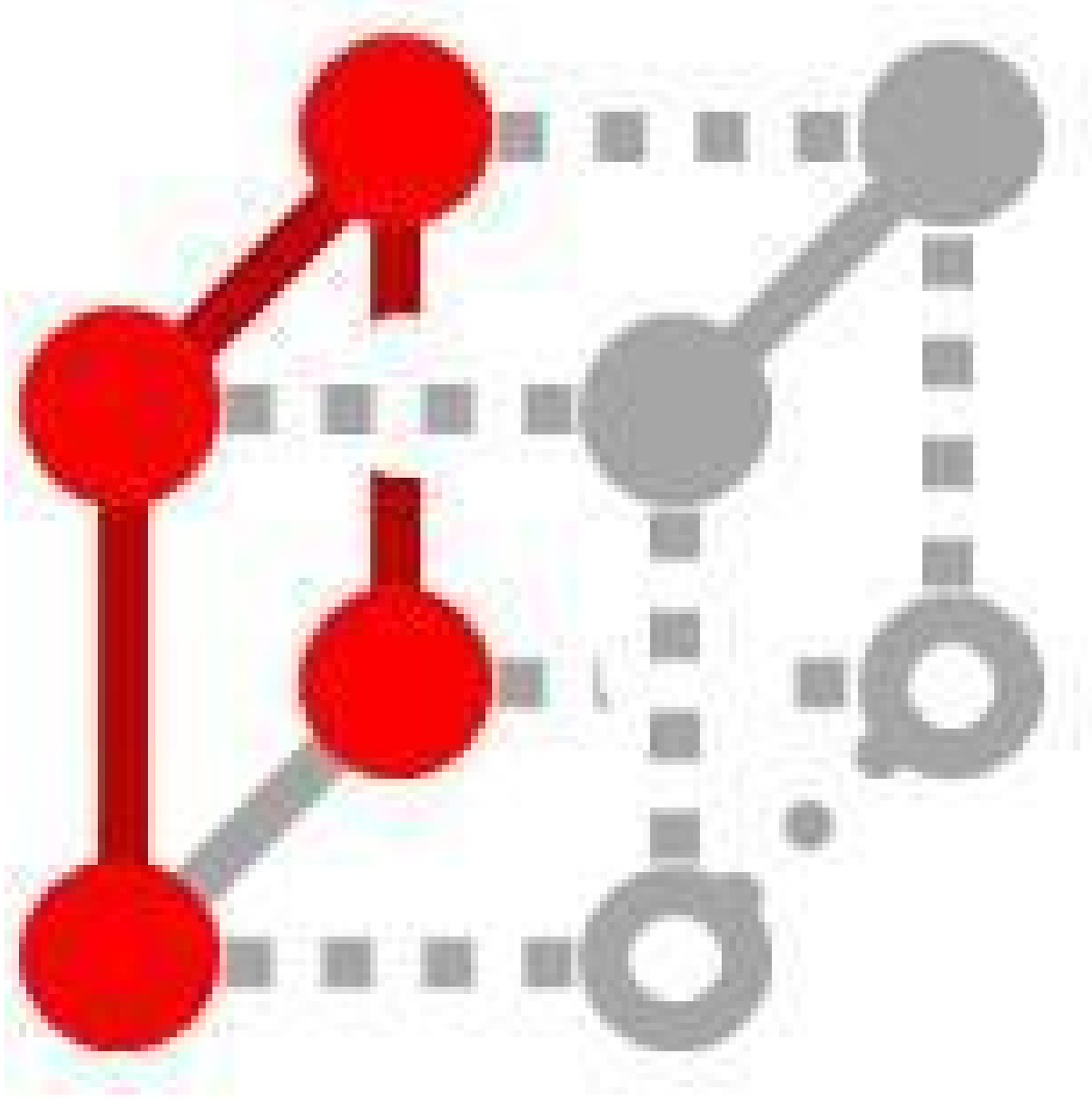}%
}%
\right)  \cup%
\raisebox{-0.3009in}{\includegraphics[
height=0.7066in,
width=0.6694in
]%
{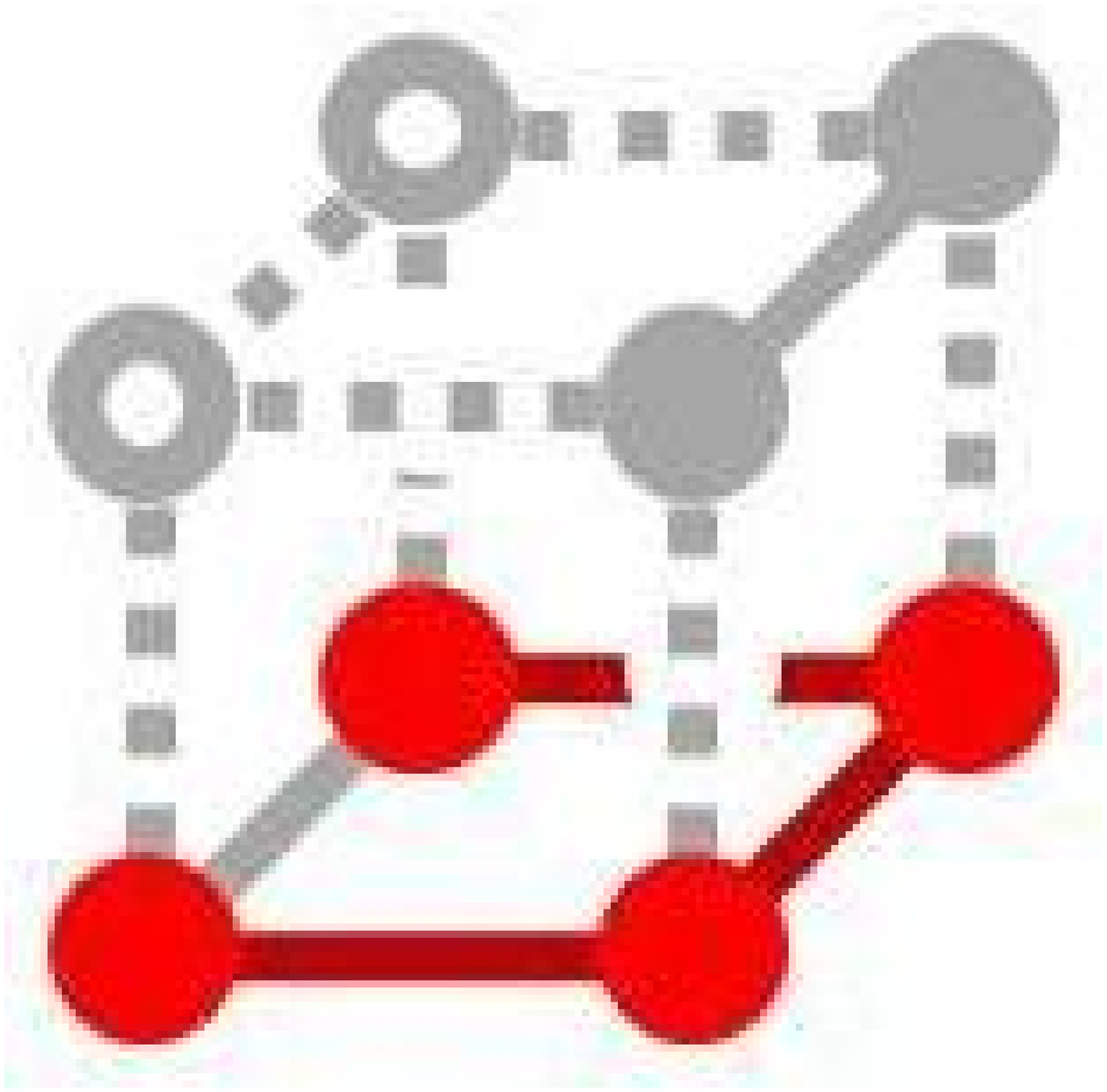}%
}%
& \text{if }\ K\cap%
\raisebox{-0.3009in}{\includegraphics[
height=0.7057in,
width=0.6694in
]%
{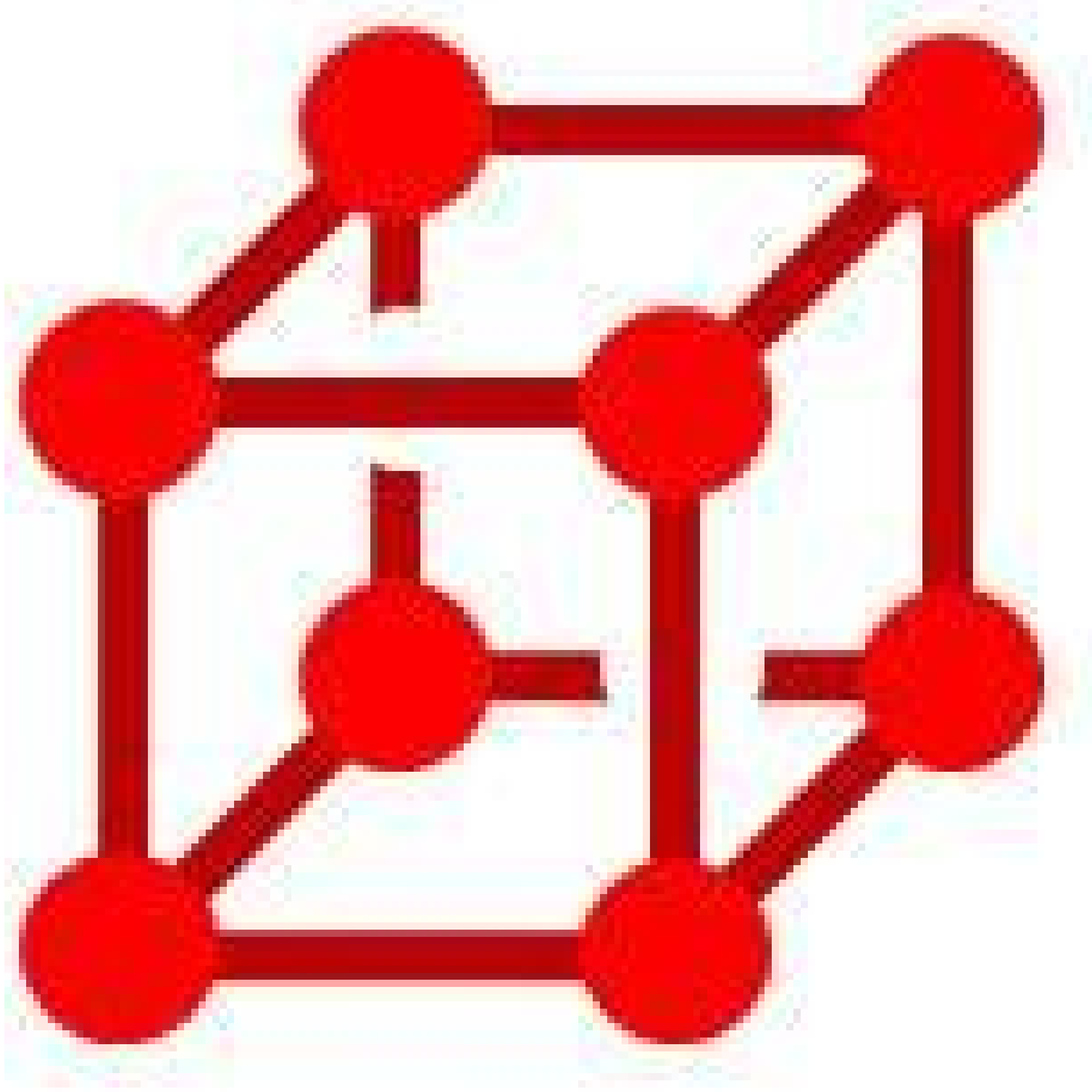}%
}%
=%
\raisebox{-0.3009in}{\includegraphics[
height=0.7066in,
width=0.6694in
]%
{wagle0f1.ps}%
}%
\\
& \\
\left(  K-%
\raisebox{-0.3009in}{\includegraphics[
height=0.7066in,
width=0.6694in
]%
{wagre0f1.ps}%
}%
\right)  \cup%
\raisebox{-0.3009in}{\includegraphics[
height=0.7066in,
width=0.6694in
]%
{wagle0f1.ps}%
}%
& \text{if \ }K\cap%
\raisebox{-0.3009in}{\includegraphics[
height=0.7057in,
width=0.6694in
]%
{cube.ps}%
}%
=%
\raisebox{-0.3009in}{\includegraphics[
height=0.7066in,
width=0.6694in
]%
{wagre0f1.ps}%
}%
\\
& \\
K & \text{otherwise}%
\end{array}
\right.
\]

\bigskip

\noindent where, in each of the above graphics, the preferred face
$F_{p}^{(\ell)}(a)$ is assumed to be the back face, and where%
\[%
\raisebox{-0.3009in}{\includegraphics[
height=0.7066in,
width=0.6694in
]%
{wagle0f1.ps}%
}%
\text{, \ \ \ \ \ }%
\raisebox{-0.3009in}{\includegraphics[
height=0.7066in,
width=0.6694in
]%
{wagre0f1.ps}%
}%
\text{, \ \ \ \ \ }%
\raisebox{-0.3009in}{\includegraphics[
height=0.7057in,
width=0.6694in
]%
{cube.ps}%
}%
\]
denote subcomplexes of the 1-skeleton of the boundary of the cube $B^{(\ell
)}(a)$, as defined by the notational conventions found in the previous
section, and where we have drawn the preferred face $F_{p}^{(\ell)}(a)$ as the
back face in each of the above drawings.

\bigskip

This wag move $L_{3}^{(\ell)}\left(  a,p,0\right)  =L_{3}^{(\ell)}\left(  a,p,%
\raisebox{-0.0406in}{\includegraphics[
height=0.1531in,
width=0.1531in
]%
{icon30.ps}%
}%
\right)  =%
\raisebox{-0.0406in}{\includegraphics[
height=0.1531in,
width=0.1531in
]%
{icon30.ps}%
}%
^{(\ell)}\left(  a,p\right)  $ is illustrated more succinctly in the figure
given below

\bigskip%

\[%
\begin{array}
[c]{c}%
\begin{array}
[c]{ccc}%
\raisebox{-0.3009in}{\includegraphics[
height=0.7066in,
width=0.6694in
]%
{wagle0f1.ps}%
}%
&
\begin{array}
[c]{c}%
{\includegraphics[
height=0.237in,
width=0.5967in
]%
{arrow-ya.ps}%
}%
\\
F_{p}^{(\ell)}(a)
\end{array}
&
\raisebox{-0.3009in}{\includegraphics[
height=0.7066in,
width=0.6694in
]%
{wagre0f1.ps}%
}%
\end{array}
\\
\text{\textbf{Lattice knot move }}L_{3}^{(\ell)}\left(  a,p,0\right)
=L_{3}^{(\ell)}\left(  a,p,%
\raisebox{-0.0406in}{\includegraphics[
height=0.1531in,
width=0.1531in
]%
{icon30.ps}%
}%
\right)  =%
\raisebox{-0.0406in}{\includegraphics[
height=0.1531in,
width=0.1531in
]%
{icon30.ps}%
}%
^{\!(\ell)}\left(  a,p\right)  \text{\textbf{, called a wag.}}%
\end{array}
\]

\bigskip

The other three tugs for face $F_{p}^{(\ell)}(a)$ are given below:%
\[%
\begin{array}
[c]{c}%
\begin{array}
[c]{ccc}%
\raisebox{-0.3009in}{\includegraphics[
height=0.7057in,
width=0.6694in
]%
{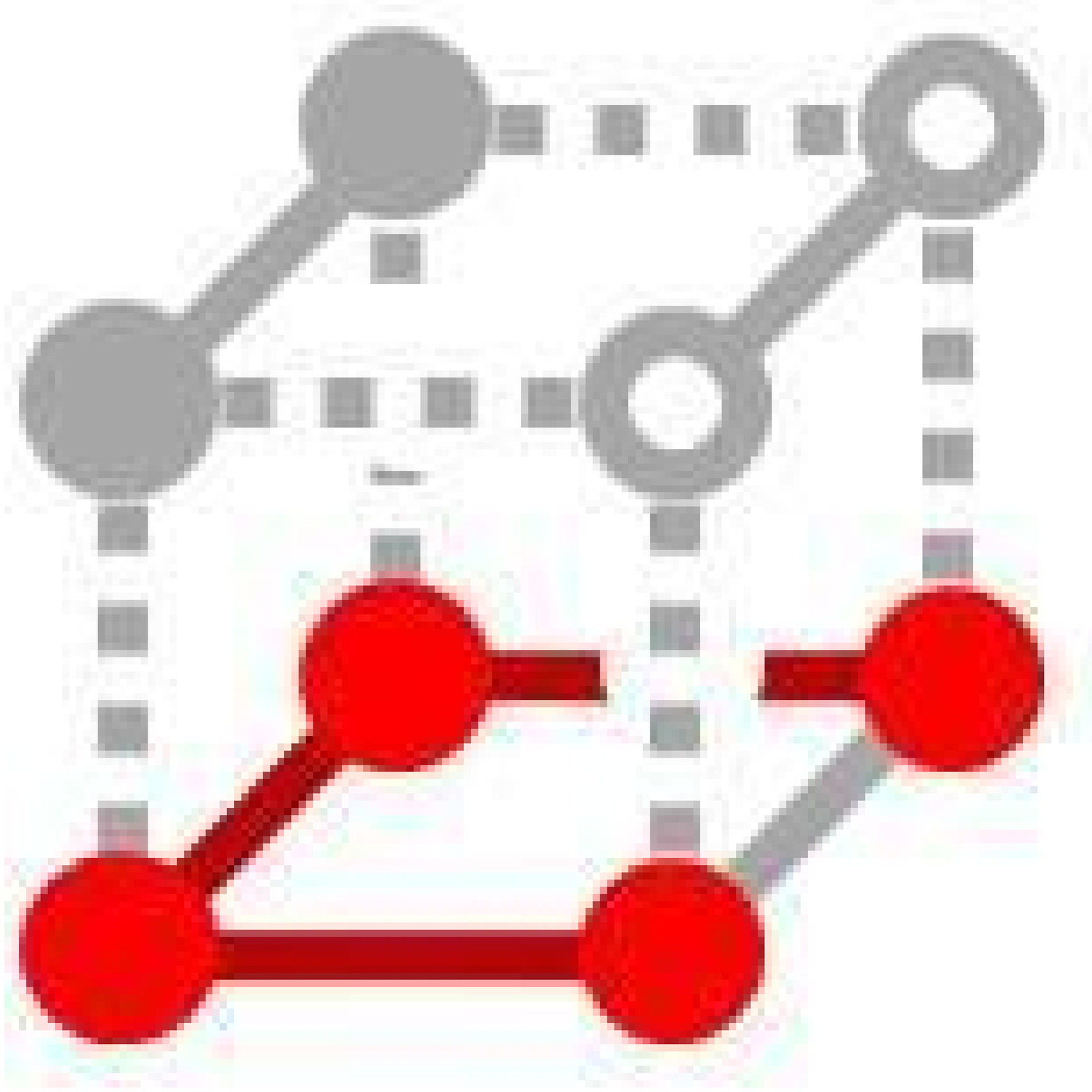}%
}%
&
\begin{array}
[c]{c}%
{\includegraphics[
height=0.237in,
width=0.5967in
]%
{arrow-ya.ps}%
}%
\\
F_{p}^{(\ell)}(a)
\end{array}
&
\raisebox{-0.3009in}{\includegraphics[
height=0.7057in,
width=0.6694in
]%
{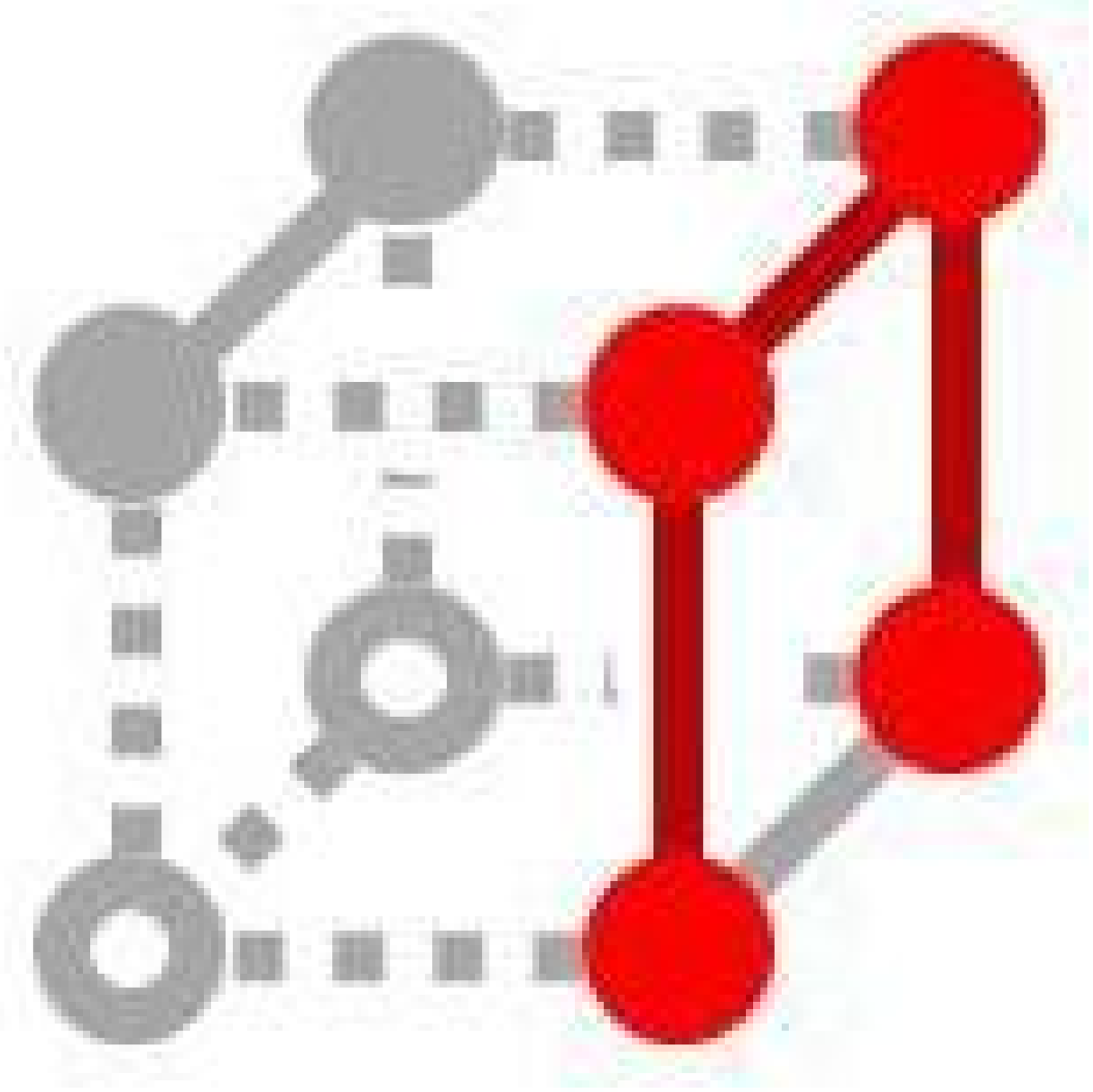}%
}%
\end{array}
\\
\text{\textbf{Lattice knot move }}L_{3}^{(\ell)}\left(  a,p,1\right)
=L_{3}^{(\ell)}\left(  a,p,%
\raisebox{-0.0406in}{\includegraphics[
height=0.1531in,
width=0.1531in
]%
{icon31.ps}%
}%
\right)  =%
\raisebox{-0.0406in}{\includegraphics[
height=0.1531in,
width=0.1531in
]%
{icon31.ps}%
}%
^{\!(\ell)}\left(  a,p\right)  \text{\textbf{, called a wag.}}%
\end{array}
\]%
\[%
\begin{array}
[c]{c}%
\begin{array}
[c]{ccc}%
\raisebox{-0.3009in}{\includegraphics[
height=0.7057in,
width=0.6694in
]%
{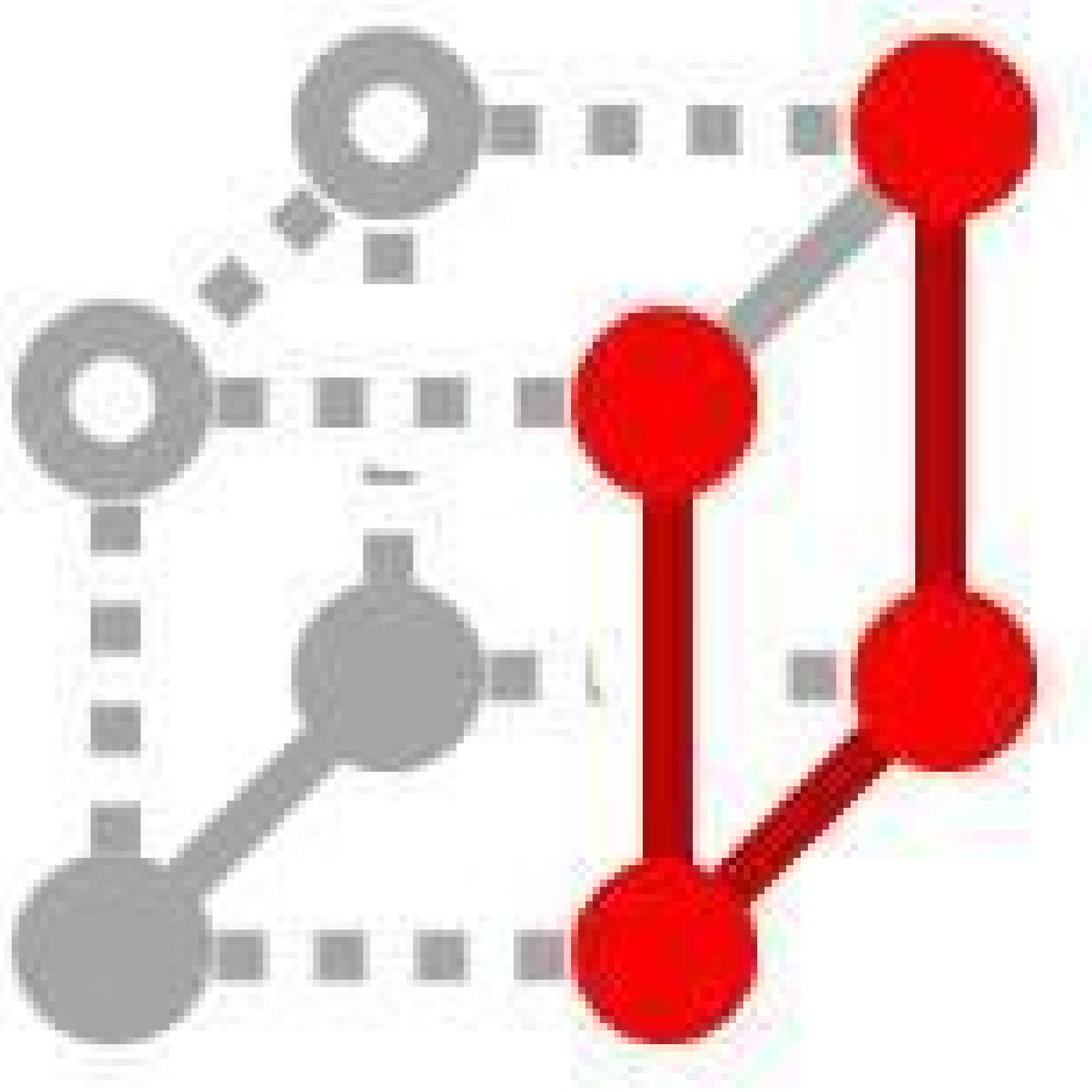}%
}%
&
\begin{array}
[c]{c}%
{\includegraphics[
height=0.237in,
width=0.5967in
]%
{arrow-ya.ps}%
}%
\\
F_{p}^{(\ell)}(a)
\end{array}
&
\raisebox{-0.3009in}{\includegraphics[
height=0.7057in,
width=0.6694in
]%
{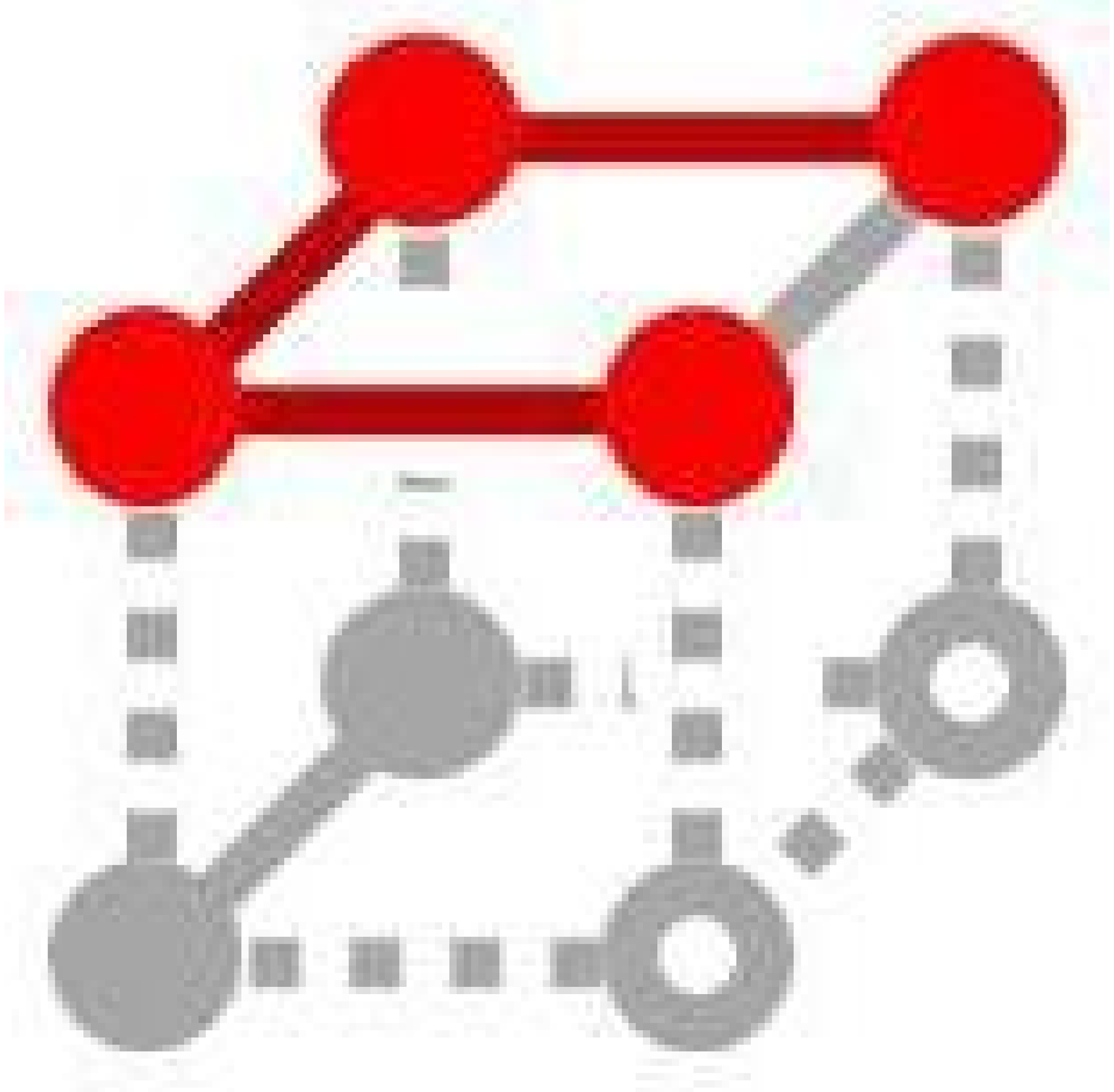}%
}%
\end{array}
\\
\text{\textbf{Lattice knot move }}L_{3}^{(\ell)}\left(  a,p,2\right)
=L_{3}^{(\ell)}\left(  a,p,%
\raisebox{-0.0406in}{\includegraphics[
height=0.1531in,
width=0.1531in
]%
{icon32.ps}%
}%
\right)  =%
\raisebox{-0.0406in}{\includegraphics[
height=0.1531in,
width=0.1531in
]%
{icon32.ps}%
}%
^{\!(\ell)}\left(  a,p\right)  \text{\textbf{, called a wag.}}%
\end{array}
\]%
\[%
\begin{array}
[c]{c}%
\begin{array}
[c]{ccc}%
\raisebox{-0.3009in}{\includegraphics[
height=0.7057in,
width=0.6694in
]%
{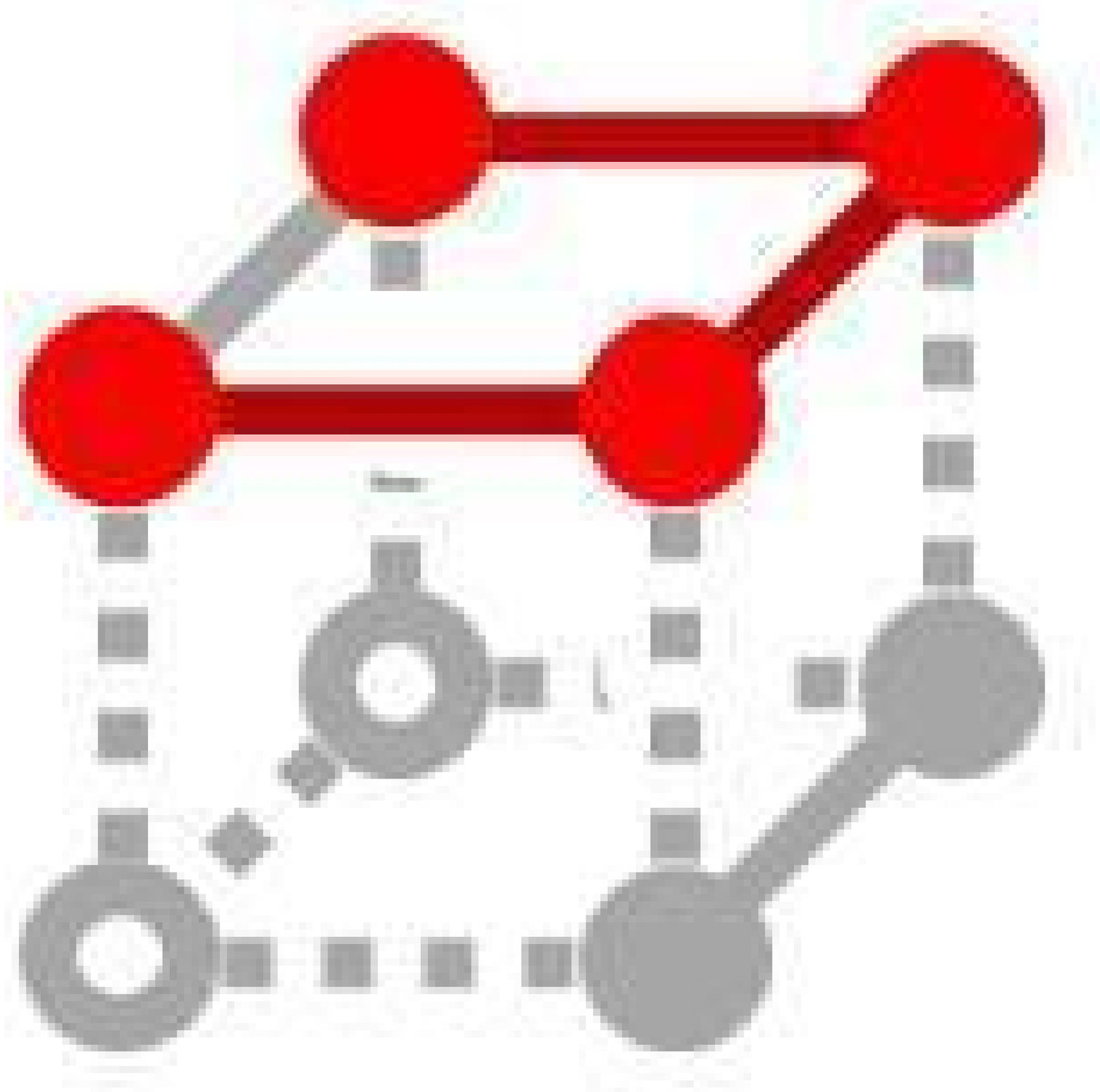}%
}%
&
\begin{array}
[c]{c}%
{\includegraphics[
height=0.237in,
width=0.5967in
]%
{arrow-ya.ps}%
}%
\\
F_{p}^{(\ell)}(a)
\end{array}
&
\raisebox{-0.3009in}{\includegraphics[
height=0.7057in,
width=0.6694in
]%
{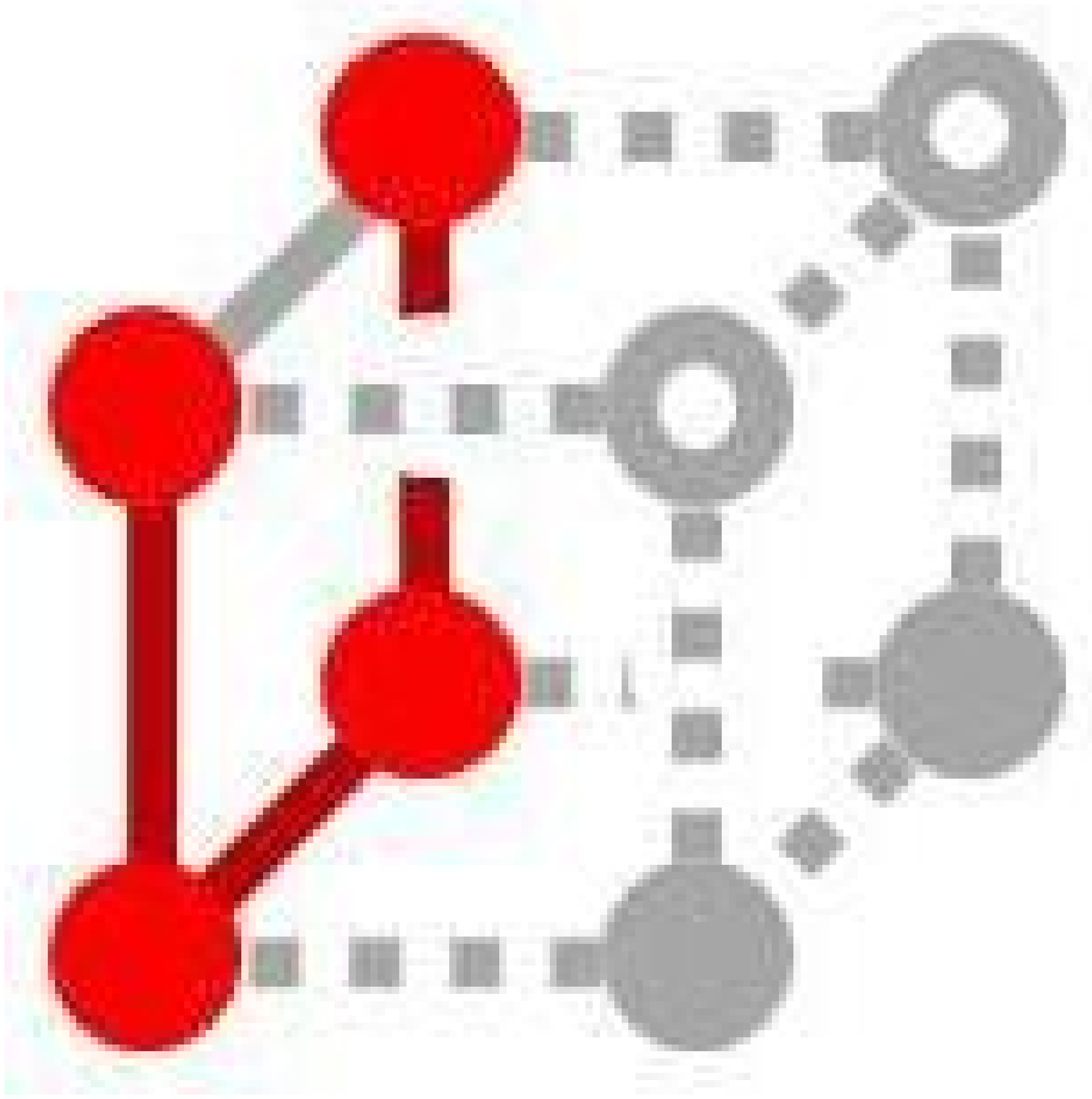}%
}%
\end{array}
\\
\text{\textbf{Lattice knot move }}L_{3}^{(\ell)}\left(  a,p,3\right)
=L_{3}^{(\ell)}\left(  a,p,%
\raisebox{-0.0406in}{\includegraphics[
height=0.1531in,
width=0.1531in
]%
{icon33.ps}%
}%
\right)  =%
\raisebox{-0.0406in}{\includegraphics[
height=0.1531in,
width=0.1531in
]%
{icon33.ps}%
}%
^{\!(\ell)}\left(  a,p\right)  \text{\textbf{, called a wag.}}%
\end{array}
\]

\end{definition}

\bigskip

\begin{remark}
For each cube $B^{(\ell)}(a)$, there are $12$ wag moves, i.e., $4$ for each of
the $3$ preferred faces.
\end{remark}

\bigskip

\subsection{Historical perspective}

\bigskip

We should mention that the lattice moves tug and wiggle were first used toward
the end of the ninetieth century by Dehn and Heegard. \ For more information,
we refer the reader to \cite{Dehn1} and \cite{Przytycki1}.

\bigskip

\section{The ambient groups $\Lambda_{\ell}$ and $\widetilde{\Lambda}_{\ell}$}

\bigskip

The following proposition is an almost immediate consequence of the
definitions of lattice knot moves given in the previous section.

\bigskip

\begin{proposition}
For each non-negative integer $\ell$, each lattice knot move $L_{m}^{(\ell
)}(a,p,q)$ is a permutation of the set of all lattice knots $\mathbb{K}%
^{(\ell)}$ of order $\ell$. \ In fact, each is a permutation which is the
product of disjoint transpositions.
\end{proposition}

\begin{proof}
Let $L_{1}^{(\ell)}(a,p,q)$ be an arbitrary tug move, and let
\begin{align*}
G_{L}  &  :\left\{  0,1,2,3\right\}  \longrightarrow\left\{
\raisebox{-0.2508in}{\includegraphics[
height=0.5829in,
width=0.5829in
]%
{tugl0.ps}%
}%
,%
\raisebox{-0.2508in}{\includegraphics[
height=0.5829in,
width=0.5829in
]%
{tugl1.ps}%
}%
,%
\raisebox{-0.2508in}{\includegraphics[
height=0.5829in,
width=0.5829in
]%
{tugl2.ps}%
}%
,%
\raisebox{-0.2508in}{\includegraphics[
height=0.5838in,
width=0.5838in
]%
{tugl3.ps}%
}%
\right\} \\
&  \text{ \ \ and \ \ }\\
G_{R}  &  :\left\{  0,1,2,3\right\}  \longrightarrow\left\{
\raisebox{-0.2508in}{\includegraphics[
height=0.5829in,
width=0.5829in
]%
{tugr0.ps}%
}%
,%
\raisebox{-0.2508in}{\includegraphics[
height=0.5829in,
width=0.5829in
]%
{tugr1.ps}%
}%
,%
\raisebox{-0.2508in}{\includegraphics[
height=0.5829in,
width=0.5829in
]%
{tugr2.ps}%
}%
,%
\raisebox{-0.2508in}{\includegraphics[
height=0.5838in,
width=0.5838in
]%
{tugr3.ps}%
}%
\right\}
\end{align*}
be the functions, from the set of integers $\left\{  0,1,2,3\right\}  $ into
the above indicated set of symbols, defined by%
\[
G_{L}\left(  q\right)  =\left\{
\begin{array}
[c]{cc}%
\raisebox{-0.2508in}{\includegraphics[
height=0.5829in,
width=0.5829in
]%
{tugl0.ps}%
}%
& \text{if }q=0\\
& \\%
\raisebox{-0.2508in}{\includegraphics[
height=0.5829in,
width=0.5829in
]%
{tugl1.ps}%
}%
& \text{if }q=1\\
& \\%
\raisebox{-0.2508in}{\includegraphics[
height=0.5829in,
width=0.5829in
]%
{tugl2.ps}%
}%
& \text{if }q=2\\
& \\%
\raisebox{-0.2508in}{\includegraphics[
height=0.5838in,
width=0.5838in
]%
{tugl3.ps}%
}%
& \text{if }q=3
\end{array}
\right.  \text{ \ \ and \ \ }G_{R}\left(  q\right)  =\left\{
\begin{array}
[c]{cc}%
\raisebox{-0.2508in}{\includegraphics[
height=0.5829in,
width=0.5829in
]%
{tugr0.ps}%
}%
& \text{if }q=0\\
& \\%
\raisebox{-0.2508in}{\includegraphics[
height=0.5829in,
width=0.5829in
]%
{tugr1.ps}%
}%
& \text{if }q=1\\
& \\%
\raisebox{-0.2508in}{\includegraphics[
height=0.5829in,
width=0.5829in
]%
{tugr2.ps}%
}%
& \text{if }q=2\\
& \\%
\raisebox{-0.2508in}{\includegraphics[
height=0.5838in,
width=0.5838in
]%
{tugr3.ps}%
}%
& \text{if }q=3
\end{array}
\right.
\]

\noindent Then the definition of the tug move, which has been given in a
previous section of this paper, can more succinctly be written as%
\[
L_{1}^{(\ell)}(a,p,q)=G_{L}(q)\underset{F_{p}(a)}{\longleftrightarrow}%
G_{R}(q)\text{ .}%
\]

\vspace{0.3in}

Now let $\mathbb{K}_{L}^{(\ell)}(q)$ and $\mathbb{K}_{R}^{(\ell)}(q)$ be sets
of order $\ell$ lattice knots respectively defined by\bigskip

$\hspace{-0.5in}\mathbb{K}_{L}^{(\ell)}(q)=\left\{  K\in\mathbb{K}^{(\ell
)}:K\cap%
\raisebox{-0.2508in}{\includegraphics[
height=0.6478in,
width=0.6478in
]%
{face.ps}%
}%
=G_{L}(q)\right\}  $ \ \ and \ \ $\mathbb{K}_{R}^{(\ell)}(q)=\left\{
K\in\mathbb{K}^{(\ell)}:K\cap%
\raisebox{-0.2508in}{\includegraphics[
height=0.6478in,
width=0.6478in
]%
{face.ps}%
}%
=G_{R}(q)\right\}  $.\bigskip

\noindent Finally, let $\mathbb{K}_{R}^{(\ell)}(q)$ be the set of order $\ell$
lattice graphs defined by%
\[
\mathbb{K}_{\ast}^{(\ell)}(q)=\mathbb{K}_{L}^{(\ell)}(q)-G_{L}(q)=\mathbb{K}%
_{R}^{(\ell)}(q)-G_{R}(q)\text{ .}%
\]

\noindent Then it immediately follows that $L_{1}^{(\ell)}\left(
a,p,q\right)  $ is the permutation%
\[
L_{1}^{(\ell)}(a,p,q=%
{\displaystyle\prod\limits_{\alpha\in\mathbb{K}_{\ast}^{(\ell)}(q)}}
\left(  \alpha\cup G_{L}(q),\alpha\cup G_{R}(q)\right)  \text{ ,}%
\]
where $\left(  \alpha\cup G_{L}(q),\alpha\cup G_{R}(q)\right)  $ is the
transposition that interchanges the lattice knots $\alpha\cup G_{L}(q)$ and
$\alpha\cup G_{R}(q)$.

\bigskip

For the remaining moves, i.e., for wiggles and wags, the proof is similar.
\end{proof}

\bigskip

Since we have shown that tug, wiggle, and wag are permutations of the set of
lattice knots, we can now give the following definition:

\bigskip

\begin{definition}
For each non-negative integer $\ell$, we define the \textbf{(lattice) ambient
group}
\[
\Lambda_{\ell}%
\]
as the group of all permutations of the set $\mathbb{K}^{(\ell)}$ of lattice
knots of order $\ell$ generated by the lattice knot moves tug, wiggle, and
wag. \ Moreover, we define the \textbf{inextensible (lattice) ambient group}%
\[
\widetilde{\Lambda}_{\ell}%
\]
as the group of all permutations of the set $\mathbb{K}^{(\ell)}$ of lattice
knots of order $\ell$ generated by the lattice knot moves wiggle and wag.
\end{definition}

\bigskip

\begin{theorem}
As abstract groups, all ambient groups $\Lambda_{\ell}$ are isomorphic, i.e.,
$\Lambda_{\ell}\simeq\Lambda_{\ell+1}$, for $\ell\geq0$. \ More specifically,
\[
L_{m}^{(\ell)}\left(  a,p,q\right)  \longmapsto L_{m}^{(\ell)}\left(  \frac
{a}{2},p,q\right)
\]
uniquely defines an isomorphism from $\Lambda_{\ell}$ onto $\Lambda_{\ell+1}$,
for $\ell\geq0$. \ The same is true for all inextensible ambient groups
$\widetilde{\Lambda}_{\ell}$, i.e., $\widetilde{\Lambda}_{\ell}\simeq
\widetilde{\Lambda}_{\ell+1}$, for $\ell\geq0$.
\end{theorem}

\bigskip

\begin{remark}
Thus, as an abstract group, the ambient groups do not "see" the metric
structure of Euclidean 3-space $\mathbb{R}^{3}$. \ However, as we will later
see, the metric structure can be found in the sequences of actions
$\Lambda_{\ell}\times\mathbb{K}^{(\ell)}\longrightarrow\mathbb{K}^{(\ell)}$
and $\widetilde{\Lambda}_{\ell}\times\mathbb{K}^{(\ell)}\longrightarrow
\mathbb{K}^{(\ell)}$.
\end{remark}

\bigskip

At first, one might think that each wag can simply be written as a product of
tugs. \ Surprisingly, the following theorem states that this is only
conditionally true.

\bigskip

\begin{lemma}
Let $a$ be a vertex in the lattice $\mathcal{L}_{\ell}$. \ Then%
\[%
\raisebox{-0.0406in}{\includegraphics[
height=0.1531in,
width=0.1531in
]%
{icon30.ps}%
}%
^{(\ell)}\left(  a,1\right)  \left(  K\right)  =\left(
\raisebox{-0.0406in}{\includegraphics[
height=0.1436in,
width=0.1436in
]%
{icon13.ps}%
}%
^{(\ell)}\left(  a,2\right)
\raisebox{-0.0406in}{\includegraphics[
height=0.1436in,
width=0.1436in
]%
{icon10.ps}%
}%
^{(\ell)}\left(  a,3\right)
\raisebox{-0.0406in}{\includegraphics[
height=0.1436in,
width=0.1436in
]%
{icon13.ps}%
}%
^{(\ell)}\left(  a,2\right)  \right)  \left(  K\right)
\]
if and only if either
\[
\text{either }K\cap\overline{F_{3}^{(\ell)}(a)}\notin\left\{
\raisebox{-0.1505in}{\includegraphics[
height=0.4091in,
width=0.4091in
]%
{tugl0.ps}%
}%
\text{, }%
\raisebox{-0.1505in}{\includegraphics[
height=0.3987in,
width=0.3987in
]%
{tugr0.ps}%
}%
\right\}  \text{ or \ }K\cap\overline{F_{2}^{(\ell)}(a)}\in\left\{
\raisebox{-0.1505in}{\includegraphics[
height=0.4091in,
width=0.4091in
]%
{tugl3.ps}%
}%
\text{, }%
\raisebox{-0.1505in}{\includegraphics[
height=0.3987in,
width=0.3987in
]%
{tugr3.ps}%
}%
\right\}  \text{, or both}%
\]
and
\[%
\raisebox{-0.0406in}{\includegraphics[
height=0.1531in,
width=0.1531in
]%
{icon30.ps}%
}%
^{(\ell)}\left(  a,1\right)  \left(  K\right)  =\left(
\raisebox{-0.0406in}{\includegraphics[
height=0.1436in,
width=0.1436in
]%
{icon10.ps}%
}%
^{(\ell)}\left(  a,3\right)
\raisebox{-0.0406in}{\includegraphics[
height=0.1436in,
width=0.1436in
]%
{icon13.ps}%
}%
^{(\ell)}\left(  a,2\right)
\raisebox{-0.0406in}{\includegraphics[
height=0.1436in,
width=0.1436in
]%
{icon10.ps}%
}%
^{(\ell)}\left(  a,3\right)  \right)  \left(  K\right)
\]
if and only if%
\[
\text{either }K\cap\overline{F_{2}^{(\ell)}(a)}\notin\left\{
\raisebox{-0.1505in}{\includegraphics[
height=0.4091in,
width=0.4091in
]%
{tugl3.ps}%
}%
\text{, }%
\raisebox{-0.1505in}{\includegraphics[
height=0.3987in,
width=0.3987in
]%
{tugr3.ps}%
}%
\right\}  \text{ or \ }K\cap\overline{F_{3}^{(\ell)}(a)}\in\left\{
\raisebox{-0.1505in}{\includegraphics[
height=0.4091in,
width=0.4091in
]%
{tugl0.ps}%
}%
\text{, }%
\raisebox{-0.1505in}{\includegraphics[
height=0.3987in,
width=0.3987in
]%
{tugr0.ps}%
}%
\right\}  \text{, or both}%
\]

Similar statements can be made for the remaining 11 tugs.
\end{lemma}

\bigskip

\begin{corollary}
For every wag $L_{3}^{(\ell)}(a,p,q)$ of order $\ell$ and for every lattice
knot $K\in\mathbb{K}^{(\ell)}$, there is a finite sequence of tugs that
transform $K$ into $L_{3}^{(\ell)}(a,p,q)(K)$. \ In this sense, every wag can
be written as a finite product of tugs.
\end{corollary}

\bigskip

\begin{example}
The following is an example of a lattice knot $K$ where
\begin{align*}
L_{3}(a,1,%
\raisebox{-0.0406in}{\includegraphics[
height=0.1531in,
width=0.1531in
]%
{icon30.ps}%
}%
)K  &  \neq L_{1}\left(  a,3,%
\raisebox{-0.0406in}{\includegraphics[
height=0.1436in,
width=0.1436in
]%
{icon10.ps}%
}%
\right)  L_{1}\left(  a,2,%
\raisebox{-0.0406in}{\includegraphics[
height=0.1436in,
width=0.1436in
]%
{icon13.ps}%
}%
\right)  L_{1}\left(  a,3,%
\raisebox{-0.0406in}{\includegraphics[
height=0.1436in,
width=0.1436in
]%
{icon10.ps}%
}%
\right)  K\\
& \\
L_{3}(a,1,%
\raisebox{-0.0406in}{\includegraphics[
height=0.1531in,
width=0.1531in
]%
{icon30.ps}%
}%
)K  &  =L_{1}\left(  a,2,%
\raisebox{-0.0406in}{\includegraphics[
height=0.1436in,
width=0.1436in
]%
{icon13.ps}%
}%
\right)  L_{1}\left(  a,3,%
\raisebox{-0.0406in}{\includegraphics[
height=0.1436in,
width=0.1436in
]%
{icon10.ps}%
}%
\right)  L_{1}\left(  a,2,%
\raisebox{-0.0406in}{\includegraphics[
height=0.1436in,
width=0.1436in
]%
{icon13.ps}%
}%
\right)  K
\end{align*}%
\begin{center}
\fbox{\includegraphics[
height=1.1969in,
width=2.4267in
]%
{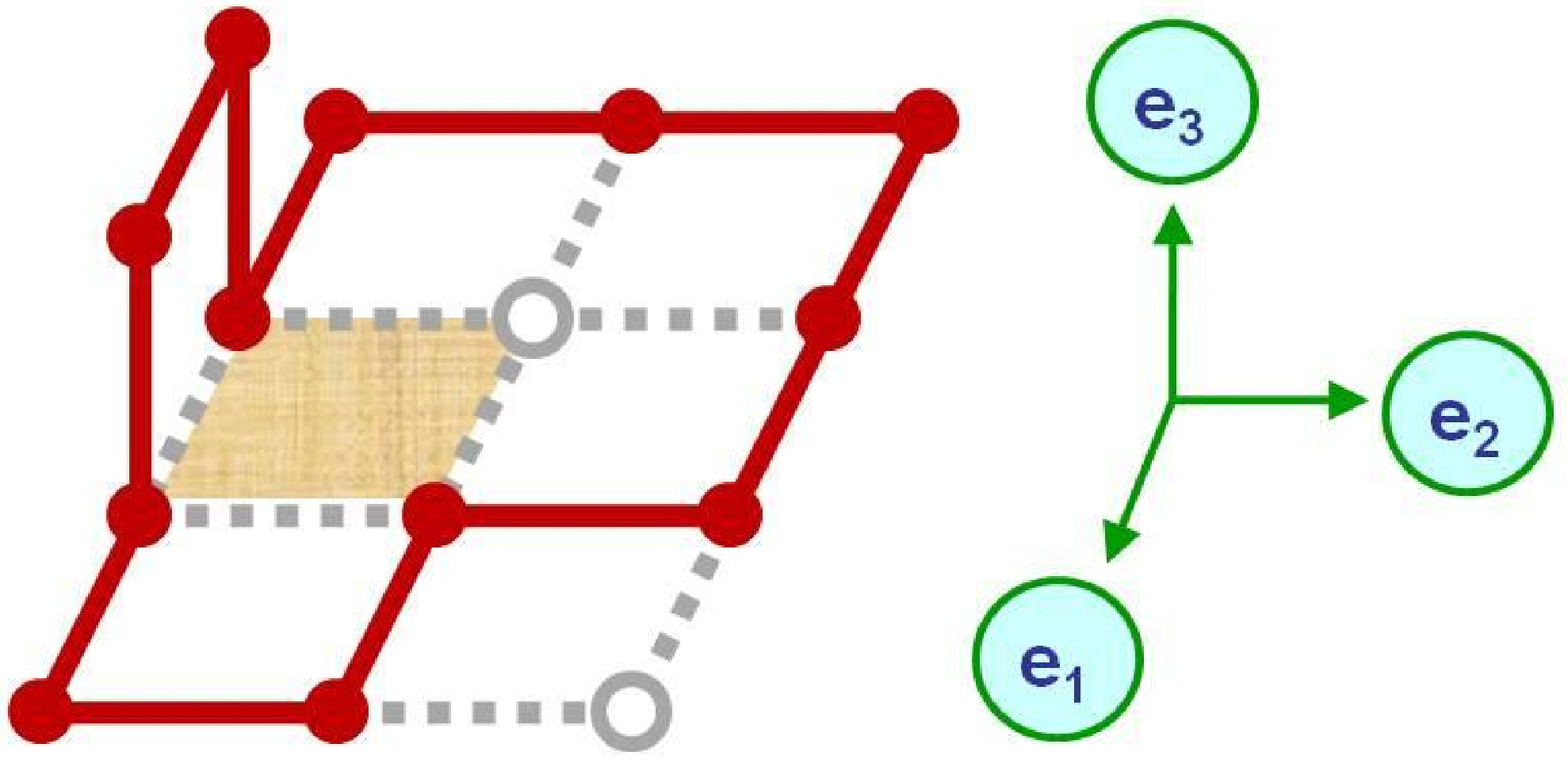}%
}
\end{center}

\end{example}

\bigskip

\begin{remark}
The alert reader may also ask if there is an analogous lemma and corollary for
wiggles. \ Unfortunately, this is not true because, unlike the more general
wiggle move found in Section \ 5 of this paper, the lattice wiggle is a wiggle
confined to a move only in a cubic lattice. \ 
\end{remark}

\bigskip

\section{Conditional auto-homeomorphism representations of $\Lambda_{\ell}$}

\bigskip

\noindent\textbf{Question.} \ \textbf{What is the intuitive meaning of the
ambient group?}

\bigskip

We begin a search for an answer to this question by noting that the moves
wiggle, wag, and tug are \textbf{conditional symbolic moves}, as are the
Reidemeister moves. \ For example, the tug

\bigskip

$\hspace{-1in}%
\raisebox{-0.0406in}{\includegraphics[
height=0.1436in,
width=0.1436in
]%
{icon10.ps}%
}%
^{(\ell)}\left(  a,p\right)  (K)=\left\{
\begin{array}
[c]{ll}%
\left(  K-%
\raisebox{-0.2508in}{\includegraphics[
height=0.6529in,
width=0.6529in
]%
{tug0l.ps}%
}%
\right)  \cup\left(  \mathcal{L}_{\ell}^{1}\cap%
\raisebox{-0.2508in}{\includegraphics[
height=0.6529in,
width=0.6529in
]%
{tug0r.ps}%
}%
\right)  & \text{if \ }K\cap%
\raisebox{-0.2508in}{\includegraphics[
height=0.6478in,
width=0.6478in
]%
{face.ps}%
}%
=\mathcal{L}_{\ell}^{1}\cap%
\raisebox{-0.2508in}{\includegraphics[
height=0.6529in,
width=0.6529in
]%
{tug0l.ps}%
}%
\\
& \\
\left(  K-%
\raisebox{-0.2508in}{\includegraphics[
height=0.6529in,
width=0.6529in
]%
{tug0r.ps}%
}%
\right)  \cup\left(  \mathcal{L}_{\ell}^{1}\cap%
\raisebox{-0.2508in}{\includegraphics[
height=0.6529in,
width=0.6529in
]%
{tug0l.ps}%
}%
\right)  & \text{if \ }K\cap%
\raisebox{-0.2508in}{\includegraphics[
height=0.6478in,
width=0.6478in
]%
{face.ps}%
}%
=\mathcal{L}_{\ell}^{1}\cap%
\raisebox{-0.2508in}{\includegraphics[
height=0.6529in,
width=0.6529in
]%
{tug0r.ps}%
}%
\\
& \\
K & \text{otherwise}%
\end{array}
\right.  $

\bigskip

\noindent is a symbolic move based on a complete set of three mutually
independent conditions. \ 

\bigskip

Each such move is a symbolic representation of a \textbf{conditional authentic
move}, i.e., a conditional orientation preserving (OP) auto-homeomorphism of
$\mathbb{R}^{3}$. \ Moreover, each involved OP auto-homeomorphism
\[
h:\mathbb{R}^{3}\longrightarrow\mathbb{R}^{3}%
\]
is\textbf{ local} if there exists an closed 3-ball $D$ such that $h$ is the
identity homeomorphism $id$ on the complement of $int\left(  D\right)  $,
i.e., such that%
\[
\left.  h\right\vert _{\mathbb{R}^{3}-int(D)}=id:\mathbb{R}^{3}-int\left(
D\right)  \longrightarrow\mathbb{R}^{3}-int\left(  D\right)  \text{ ,}%
\]
where $int\left(  D\right)  $ denotes the interior of the closed 3-ball $D$.

\bigskip

To complete the answer to our question, we will need the following definition.\ 

\bigskip

\begin{definition}
Let $LAH_{OP}\left(  \mathbb{R}^{3}\right)  $ be the \textbf{group of
orientation preserving (OP) local auto-homeomorphisms of 3-space}
$\mathbb{R}^{3}$. \ A \textbf{conditional authentic move} $\Phi$ for a family
$\mathcal{F}$ of knots in 3-space $\mathbb{R}^{3}$is a map%
\[%
\begin{array}
[c]{rcl}%
\Phi:\mathcal{F} & \longrightarrow & AH_{OP}\left(  \mathbb{R}^{3}\right) \\
K & \longmapsto & \left(  \Phi_{K}:\mathbb{R}^{3}\longrightarrow\mathbb{R}%
^{3}\right)
\end{array}
\]
such that
\[
\Phi_{K}\left(  K\right)  \in\mathcal{F}%
\]
for all $K\in\mathcal{F}$.
\end{definition}

\bigskip

\begin{remark}
It readily follows that the elements of the ambient groups $\Lambda_{\ell}$
and $\widetilde{\Lambda}_{\ell}$ are all conditional authentic moves for the
family $\mathbb{K}^{(\ell)}$ of lattice knots of order $\ell$.
\end{remark}

\bigskip

\begin{definition}
Let $LAH_{OP}\left(  \mathbb{R}^{3}\right)  ^{\mathcal{F}}$ denote the
\textbf{space of all conditional OP local auto-homeomorphisms} for a family of
knots $\mathcal{F}$, together with the binary operation%
\[%
\begin{array}
[c]{ccc}%
AH_{OP}\left(  \mathbb{R}^{3}\right)  ^{\mathcal{F}}\times AH_{OP}\left(
\mathbb{R}^{3}\right)  ^{\mathcal{K}} & \longrightarrow & AH_{OP}\left(
\mathbb{R}^{3}\right)  ^{\mathcal{F}}\\
\left(  \Phi^{\prime},\Phi\right)  & \longmapsto & \Phi^{\prime}\cdot\Phi
\end{array}
\]
defined by
\[
\left(  \Phi^{\prime}\cdot\Phi\right)  =\Phi_{\Phi_{K}(K)}^{\prime}\circ
\Phi_{K}\text{ ,}%
\]
where `$\circ$' denotes the composition of functions. \ This is readily seen
to be a well defined binary operation on $LAH_{OP}\left(  \mathbb{R}%
^{3}\right)  ^{\mathcal{F}}$.
\end{definition}

\bigskip

\begin{remark}
We should remind the reader that a knot $K$ is an imbedding $K:\bigsqcup
_{j=1}^{s}S^{1}\longrightarrow\mathbb{R}^{3}$ of a disjoint union of finitely
many circles into 3-space $\mathbb{R}^{3}$. \ Hence, $\Phi_{K}(K)$ is the
imbedding $\Phi_{K}\circ K:\bigsqcup_{j=1}^{s}S^{1}\longrightarrow
\mathbb{R}^{3}$, where '$\circ$' denotes the composition of functions.
\end{remark}

\bigskip

\begin{proposition}
The space $LAH_{OP}\left(  \mathbb{R}^{3}\right)  ^{\mathcal{F}}$ together
with the above binary operation `$\cdot$' is a monoid.
\end{proposition}

\begin{proof}
Let $\Phi$, $\Phi^{\prime}$, $\Phi^{\prime\prime}$ be three arbitrary
conditional authentic moves. \ Then%
\[
\left(  \Phi\cdot\left(  \Phi^{\prime}\cdot\Phi^{\prime\prime}\right)
\right)  _{K}=\Phi_{\left(  \Phi^{\prime}\cdot\Phi^{\prime\prime}\right)
_{K}\left(  K\right)  }\circ\left(  \Phi^{\prime}\cdot\Phi^{\prime\prime
}\right)  _{K}=\Phi_{\left(  \Phi_{\Phi_{K}^{\prime\prime}\left(  K\right)
}^{\prime}\circ\Phi_{K}^{\prime\prime}\right)  (K})\circ\Phi_{\Phi_{K}%
^{\prime\prime}\left(  K\right)  }^{\prime}\circ\Phi_{K}^{\prime\prime}%
\]
On the other hand,%
\[
\left(  \left(  \Phi\cdot\Phi^{\prime}\right)  \cdot\Phi^{\prime\prime
}\right)  _{K}=\left(  \Phi\cdot\Phi^{\prime}\right)  _{\Phi_{K}^{\prime
\prime}\left(  K\right)  }\circ\Phi_{K}^{\prime\prime}=\Phi_{\left(
\Phi_{\Phi_{K}^{\prime\prime}\left(  K\right)  }^{\prime}\circ\Phi_{K}%
^{\prime\prime}\right)  (K)}\circ\Phi_{\Phi_{K}^{\prime\prime}\left(
K\right)  }^{\prime}\circ\Phi_{K}^{\prime\prime}%
\]
Hence, `$\cdot$' is associative. \ 

Let $id:\mathbb{R}^{3}\longrightarrow\mathbb{R}^{3}$ be the identity
homeomorphism. \ Then it easily follows that
\[
K\longmapsto id:\mathbb{R}^{3}\longrightarrow\mathbb{R}^{3}%
\]
is a conditional OP auto-homeomorphism which is an identity with respect to
the binary operation `$\cdot$'.
\end{proof}

\bigskip

We will now construct a faithful representation%
\[
\Gamma:\Lambda_{\ell}\longrightarrow LAH_{OP}\left(  \mathbb{R}^{3}\right)
^{\mathbb{K}^{(\ell)}}%
\]
of the ambient group $\Lambda_{\ell}$ onto a subgroup of the monoid $\left(
LAH_{OP}\left(  \mathbb{R}^{3}\right)  ^{\mathbb{K}^{(\ell)}},\cdot\right)  $
by mapping each of the generators of $\Lambda_{\ell}$ onto a conditional OP
local auto-homeomorpism of $\mathbb{R}^{3}$. \ 

\bigskip

To define this representation, we need to construct for each of the generator
$L_{m}^{(\ell)}\left(  a,p,q\right)  $ of $\Lambda_{\ell}$ a conditional OP
local auto-homeomorphism $\Phi_{m,q}^{(\ell)}\left(  a,p\right)
:\mathbb{R}^{3}\longrightarrow\mathbb{R}^{3}$ such that\medskip%
\[
L_{m}^{(\ell)}\left(  a,p,q\right)  K_{1}=L_{m}^{(\ell)}\left(  a,p,q\right)
K_{2}\text{ \ \ if and only if \ \ }\Phi_{m,q}^{(\ell)}\left(  a,p\right)
K_{1}=\Phi_{m,q}^{(\ell)}\left(  a,p\right)  K_{2}\text{ .}%
\]

\bigskip

\bigskip

We now do so for the generators $L_{1}^{(\ell)}\left(  a,1,0\right)  $,
$L_{2}^{(\ell)}\left(  a,1,0\right)  $, $L_{3}^{(\ell)}\left(  a,1,0\right)
$. \ The construction is similar for the remaining generators.

\bigskip

\subsection{Construction for tugs}

\bigskip

For the tug $L_{1}^{(\ell)}\left(  a,1,0\right)  =%
\raisebox{-0.0406in}{\includegraphics[
height=0.1436in,
width=0.1436in
]%
{icon10.ps}%
}%
^{(\ell)}\left(  a,1\right)  $, we construct a conditional OP
auto-homeomorphism $\Phi_{1,0}^{(\ell)}\left(  a,1\right)  :\mathbb{R}%
^{3}\longrightarrow\mathbb{R}^{3}$ as follows:

\bigskip

Let $\epsilon$ be a sufficiently small positive real number, let $c=\left(
a+a^{:23}\right)  /2$ be the center of the face $F_{1}^{(\ell)}(a)$, and let
$D_{\epsilon}$ be the closed 3-cell bounded by the sphere%
\[
\left\vert x-c-\epsilon e_{3}\right\vert ^{2}=\left\vert a-c-\epsilon
e_{3}\right\vert ^{2}%
\]
where $x=\left(  x_{1},x_{2},x_{3}\right)  $. \ Then
\[
\overline{F_{1}^{(\ell)}(a)}\subset D_{\epsilon}\text{ \ \ and \ \ }\partial
F_{1}^{(\ell)}(a)\cap\partial D_{\epsilon}=\left\{  a,a^{:2}\right\}  \text{
,}%
\]
where $\overline{C}$ and $\partial C$ respectively denote the closure and the
boundary of a cell $C$. \ Now let $h:D_{\epsilon}\longrightarrow D_{\epsilon}$
be an OP auto-homeomorpism of the 3-cell $D_{\epsilon}$ such that%
\[
h\left(  \overline{E_{2}^{(\ell)}(a)}\right)  =\overline{E_{3}^{(\ell)}%
(a)}\cup\overline{E_{2}^{(\ell)}(a^{:3})}\cup\overline{E_{3}^{(\ell)}(a^{:2}%
)}\text{ \ \ \ and \ \ }\left.  h\right\vert _{\partial D_{\epsilon}}=\left.
id\right\vert _{\partial D_{\epsilon}}%
\]
where $id$ is the identity auto-homeomorphism of $\mathbb{R}^{3}$.%
\begin{center}
\includegraphics[
height=2.2701in,
width=3.6694in
]%
{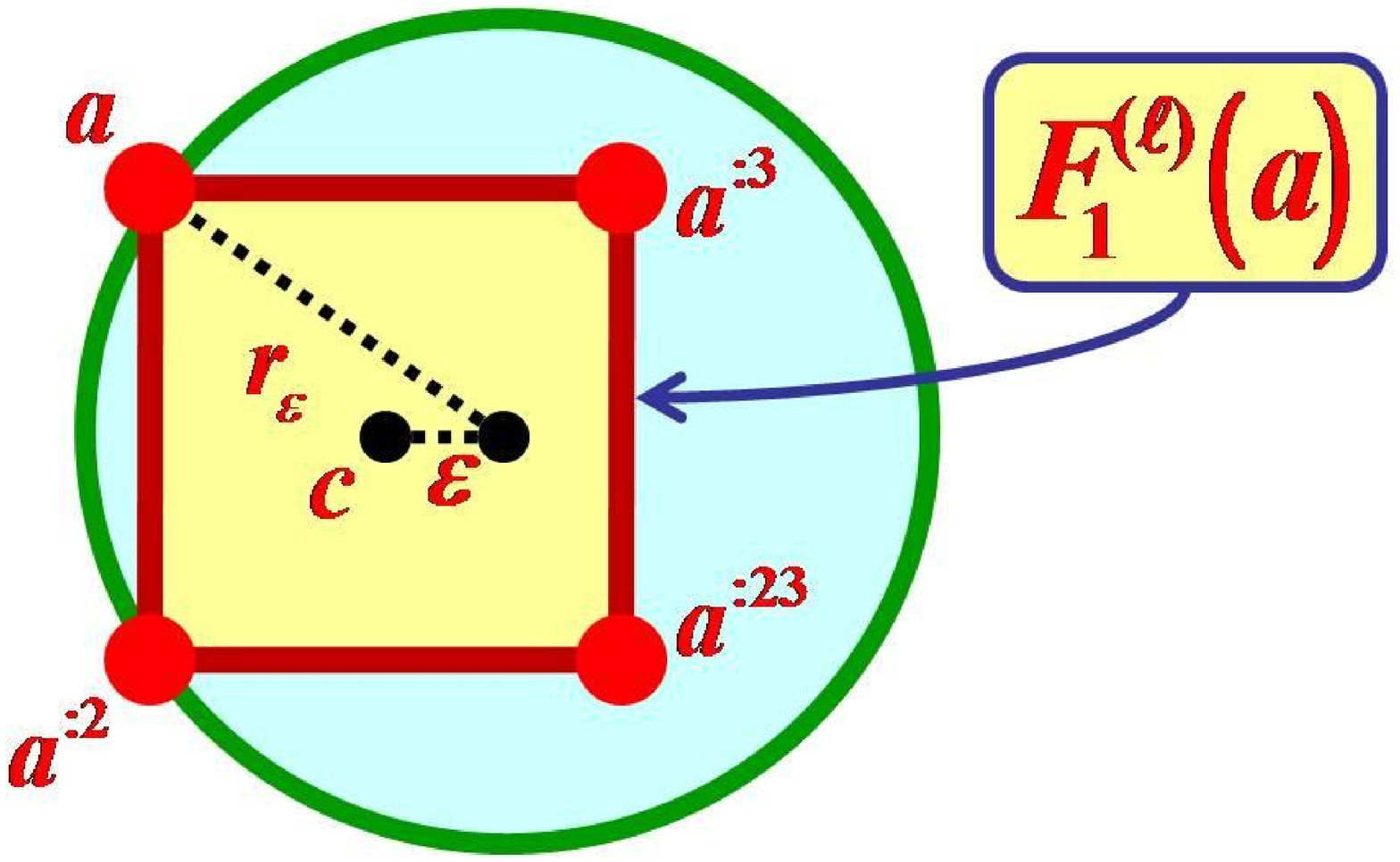}%
\\
\textbf{Cross Section of the 3-cell }$D_{\varepsilon}$\textbf{, where
}$r_{\varepsilon}=\left\vert a-c-\epsilon e_{3}\right\vert $\textbf{.}%
\end{center}

Then we define $\Phi_{1,0}^{(\ell)}\left(  a,1\right)  $ as%

\[
\Phi_{1,0}^{(\ell)}\left(  a,1\right)  _{K}=\left\{
\begin{array}
[c]{ll}%
h & \text{if }K\cap F_{1}^{(\ell)}(a)=\overline{E}_{2}^{(\ell)}\left(
a\right)  \text{ and }x\in D_{\epsilon}\\
& \\
h^{-1} & \text{if }K\cap F_{1}^{(\ell)}(a)=\overline{E_{3}^{(\ell)}(a)}%
\cup\overline{E_{2}^{(\ell)}(a^{:3})}\cup\overline{E_{3}^{(\ell)}(a^{:2}%
)}\text{ and }x\in D_{\epsilon}\\
& \\
id & \text{otherwise}%
\end{array}
\right.
\]

\bigskip

\subsection{Construction for wiggles}

\bigskip

For the wiggle $L_{2}^{(\ell)}\left(  a,1,0\right)  =%
\raisebox{-0.0406in}{\includegraphics[
height=0.1436in,
width=0.1436in
]%
{icon21.ps}%
}%
^{\!(\ell)}\left(  a,1\right)  $, we construct a conditional OP
auto-homeomorphism $\Phi_{2,0}^{(\ell)}\left(  a,1\right)  :\mathbb{R}%
^{3}\longrightarrow\mathbb{R}^{3}$ as follows:

\bigskip

\bigskip

Let $\epsilon$ be a sufficiently small positive real number, and let
$D_{\epsilon}$ be the closed 3-cell bounded by the ellipsoid with axes%
\[
Seg\left(  a^{:2},a^{:3}\right)  \text{, \ \ }Seg\left(  c-\sqrt{2}%
\cdot2^{-\ell-1}e_{1},c+\sqrt{2}\cdot2^{-\ell-1}e_{1}\right)  \text{,
\ }Seg\left(  a-\epsilon(e_{2}+e_{3}),a^{:23}+\epsilon(e_{2}+e_{3})\right)
\]
where $c=\left(  a^{:2}+a^{:3}\right)  /2$ is the center of the face
$F_{1}^{(\ell)}(a)$, and where $Seg\left(  b,b^{\prime}\right)  $ denotes the
line segment connecting points $b$ and $b^{\prime}$. \ Then
\[
\overline{F_{1}^{(\ell)}(a)}\subset D_{\epsilon}\text{ \ \ and \ \ }\partial
F_{1}^{(\ell)}(a)\cap\partial D_{\epsilon}=\left\{  a^{:2},a^{:3}\right\}
\text{ ,}%
\]
where $\overline{C}$ and $\partial C$ respectively denote the closure and the
boundary of a cell $C$. \
\begin{center}
\includegraphics[
height=2.4275in,
width=3.269in
]%
{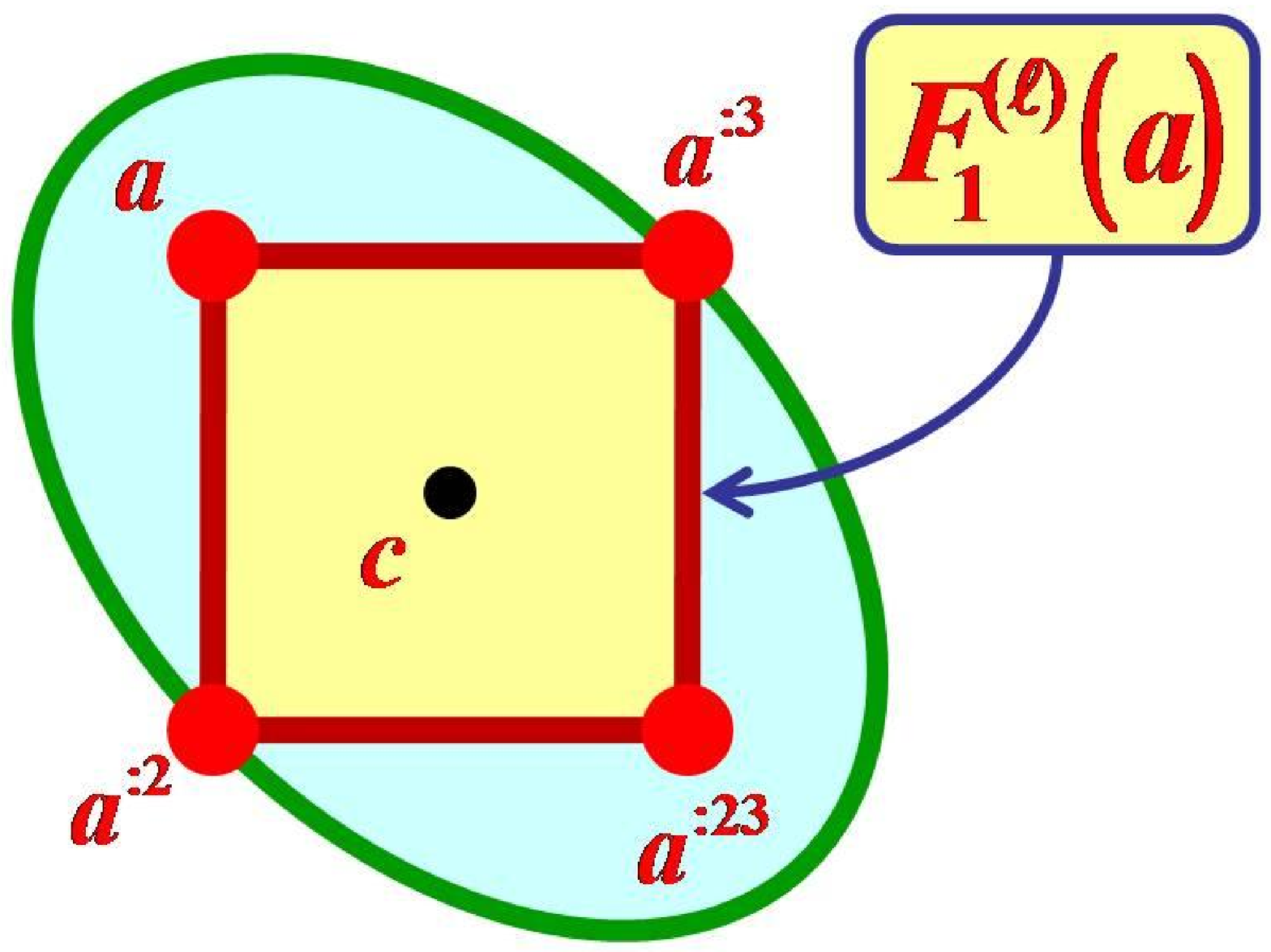}%
\\
\textbf{Cross section of 3-cell }$D_{\varepsilon}$.
\end{center}
Now let $h:D_{\epsilon}\longrightarrow D_{\epsilon}$ be an OP
auto-homeomorpism of the 3-cell $D_{\epsilon}$ such that%
\[
h\left(  \overline{E_{2}^{(\ell)}(a)}\cup\overline{E_{3}^{(\ell)}(a)}\right)
=\overline{E_{2}^{(\ell)}(a^{:3})}\cup\overline{E_{3}^{(\ell)}(a^{:2})}\text{
\ \ \ and \ \ }\left.  h\right\vert _{\partial D_{\epsilon}}=\left.
id\right\vert _{\partial D_{\epsilon}}\text{ \ ,}%
\]
where $id$ is the identity auto-homeomorphism of $\mathbb{R}^{3}$.

\bigskip

Then we define $\Phi_{2,0}^{(\ell)}\left(  a,1\right)  $ as%

\[
\Phi_{2,0}^{(\ell)}\left(  a,1\right)  _{K}=\left\{
\begin{array}
[c]{ll}%
h & \text{if }K\cap F_{1}^{(\ell)}(a)=\overline{E_{2}^{(\ell)}(a)}%
\cup\overline{E_{3}^{(\ell)}(a)}\text{ and }x\in D_{\epsilon}\\
& \\
h^{-1} & \text{if }K\cap F_{1}^{(\ell)}(a)=\overline{E_{2}^{(\ell)}(a^{:3}%
)}\cup\overline{E_{3}^{(\ell)}(a^{:2})}\text{ and }x\in D_{\epsilon}\\
& \\
id & \text{otherwise}%
\end{array}
\right.
\]

\bigskip

\subsection{Construction for wags}

\bigskip

For the wag $L_{3}^{(\ell)}\left(  a,1,0\right)  =%
\raisebox{-0.0406in}{\includegraphics[
height=0.1531in,
width=0.1531in
]%
{icon30.ps}%
}%
^{\!(\ell)}\left(  a,1\right)  $, we construct a conditional OP
auto-homeomorphism $\Phi_{3,0}^{(\ell)}\left(  a,1\right)  :\mathbb{R}%
^{3}\longrightarrow\mathbb{R}^{3}$ as follows:

\bigskip

We wish to construct a closed 3-cell $D_{\epsilon}$ such that%
\[
\overline{F_{2}^{(\ell)}(a)}\cup\overline{F_{3}^{(\ell)}(a)}\subset
D_{\epsilon}\text{, \ }\left(  \partial F_{2}^{(\ell)}(a)\cup\partial
F_{3}^{(\ell)}(a)\right)  \cap\partial D_{\epsilon}=\left\{  a,a^{:1}\right\}
\text{ , \ and \ }\overline{E_{1}^{(\ell)}\left(  a^{:23}\right)  }\cap
D_{\epsilon}=\varnothing\text{ ,}%
\]
where $\overline{C}$ and $\partial C$ respectively denote the closure and the
boundary of a cell $C$. \ We do so as follows:

\bigskip

Let $\epsilon$ be a sufficiently small positive real number, let $c$ be the
center of the cube $B^{(\ell)}(a)$, and let $D_{\epsilon}^{\prime}$ be the
closed 3-cell bounded by the sphere with center%
\[
c+\epsilon\left(  \frac{e_{2}+e_{3}}{\sqrt{2}}\right)
\]
and of radius
\[
r=\left\vert a-c-\epsilon\left(  \frac{e_{2}+e_{3}}{\sqrt{2}}\right)
\right\vert
\]
Let $M$ be the plane passing through the point
\[
c+\frac{e_{2}+e_{3}}{3\sqrt{2}}%
\]
with normal $\left(  e_{2}+e_{3}\right)  /\sqrt{2}$. \ Then $\mathbb{R}^{3}-M$
consists of two disjoint open components $\mathbb{R}_{+}^{3}$ and
$\mathbb{R}_{-}^{3}$, where $\mathbb{R}_{+}^{3}$ is that component into which
the normal $\left(  e_{2}+e_{3}\right)  /\sqrt{2}$ points. \ Let $D_{\epsilon
}$ be the closed 3-cell
\[
D_{\epsilon}=D_{\epsilon}^{\prime}-\mathbb{R}_{-}^{3}\text{ .}%
\]

\bigskip

Now let $h:D_{\epsilon}\longrightarrow D_{\epsilon}$ be an OP
auto-homeomorpism of the 3-cell $D_{\epsilon}$ such that%
\[
h\left(  \overline{E_{3}^{(\ell)}(a)}\cup\overline{E_{1}^{(\ell)}(a^{:3})}%
\cup\overline{E_{3}^{(\ell)}(a^{:1})}\right)  =\overline{E_{2}^{(\ell)}%
(a)}\cup\overline{E_{1}^{(\ell)}(a^{:2})}\cup\overline{E_{2}^{(\ell)}(a^{:1}%
)}\text{ \ \ \ and \ \ }\left.  h\right\vert _{\partial D_{\epsilon}}=\left.
id\right\vert _{\partial D_{\epsilon}}\text{ \ ,}%
\]
where $id$ is the identity auto-homeomorphism of $\mathbb{R}^{3}$.

\bigskip

We can now define $\Phi_{3,0}^{(\ell)}\left(  a,1\right)  $ as%

\[
\Phi_{3,0}^{(\ell)}\left(  a,1\right)  _{K}=\left\{
\begin{array}
[c]{ll}%
h & \text{if }K\cap\left(  \overline{F_{2}^{(\ell)}(a)}\cup\overline
{F_{3}^{(\ell)}(a)}\right)  =\overline{E_{3}^{(\ell)}(a)}\cup\overline
{E_{1}^{(\ell)}(a^{:3})}\cup\overline{E_{3}^{(\ell)}(a^{:1})}\text{ and }x\in
D_{\epsilon}\\
& \\
h^{-1} & \text{if }K\cap\left(  \overline{F_{2}^{(\ell)}(a)}\cup
\overline{F_{3}^{(\ell)}(a)}\right)  =\overline{E_{3}^{(\ell)}(a)}%
\cup\overline{E_{1}^{(\ell)}(a^{:3})}\cup\overline{E_{3}^{(\ell)}(a^{:1}%
)}\text{ and }x\in D_{\epsilon}\\
& \\
id & \text{otherwise}%
\end{array}
\right.
\]

\bigskip

We leave it as an exercise for the reader to verify that the above constructed
elements $\Phi_{m,q}^{(\ell)}\left(  a,p\right)  $ of $LAH_{OP}\left(
\mathbb{R}^{3}\right)  ^{\mathbb{K}^{(\ell)}}$ are invertible in the monoid
$LAH_{OP}\left(  \mathbb{R}^{3}\right)  ^{\mathbb{K}^{(\ell)}}$. \ Hence we have:

\bigskip

\bigskip

\begin{proposition}
A faithful representation
\[
\Gamma:\Lambda_{\ell}\longrightarrow LAH_{OP}\left(  \mathbb{R}^{3}\right)
^{\mathbb{K}^{(\ell)}}%
\]
of the ambient group $\Lambda_{\ell}$ onto a subgroup of the monoid $\left(
LAH_{OP}\left(  \mathbb{R}^{3}\right)  ^{\mathbb{K}^{(\ell)}},\cdot\right)  $
is uniquely determined by
\[
\Gamma:L_{1}^{(\ell)}\left(  a,p,q\right)  \longmapsto\Phi_{m,q}^{(\ell
)}\left(  a,p\right)
\]

\end{proposition}

\bigskip

\begin{remark}
Please note that the above construction of the faithful representation
$\Gamma$ is far from unique.
\end{remark}

\bigskip

\section{The refinement injection and the conjectured refinement morphism}

\bigskip

\begin{definition}
We define the refinement injection $%
\raisebox{-0.0303in}{\includegraphics[
height=0.1332in,
width=0.1193in
]%
{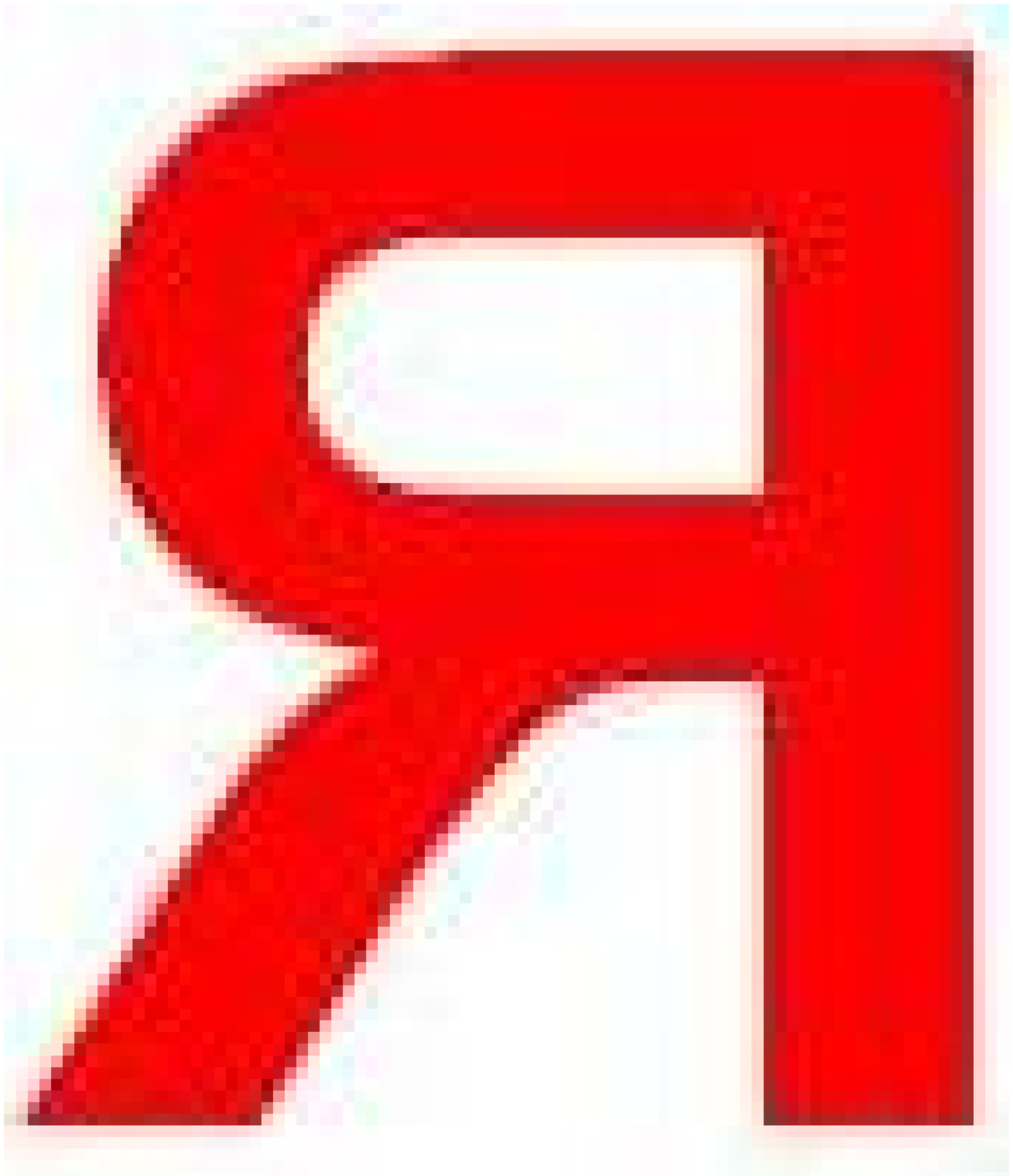}%
}%
:\mathbb{K}^{(\ell)}\longrightarrow\mathbb{K}^{(\ell+1)}$ from the set of
lattice knots $\mathbb{K}^{(\ell)}$ of order $\ell$ to the set $\mathbb{K}%
^{(\ell+1)}$ of lattice knots of order $\ell+1$ as%
\[%
\begin{array}
[c]{rccl}%
\raisebox{-0.0303in}{\includegraphics[
height=0.1332in,
width=0.1193in
]%
{red-refinement.ps}%
}%
: & \mathbb{K}^{(\ell)} & \longrightarrow & \mathbb{K}^{(\ell+1)}\\
& K & \longmapsto &
{\displaystyle\bigcup\limits_{a\in\mathcal{L}_{\ell}}}
\
{\displaystyle\bigcup\limits_{p=1}^{3}}
\
{\displaystyle\bigcup\limits_{E_{p}^{(\ell)}(a)\in K}}
\left\{  \overline{E_{p}^{(\ell+1)}}(a),\overline{E_{p}^{(\ell+1)}}%
(a^{:p})\right\}
\end{array}
\]
where $\overline{E_{p}^{(\ell+1)}}$ denotes the closure of the open edge
$E_{p}^{(\ell+1)}$.%
\[%
\begin{tabular}
[c]{c}%
$%
\begin{array}
[c]{ccc}%
\raisebox{-0.3009in}{\includegraphics[
height=0.7463in,
width=0.7463in
]%
{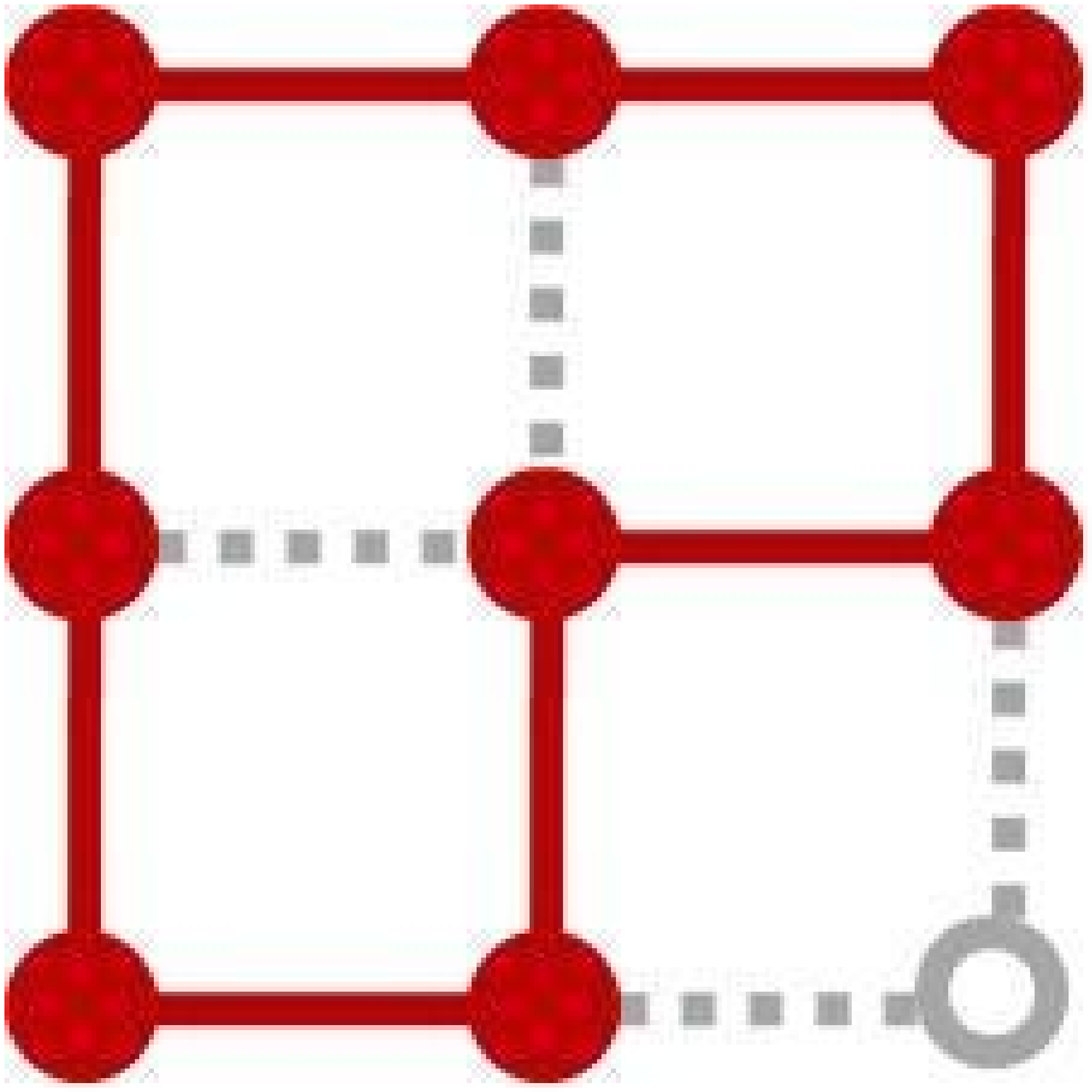}%
}%
&
\begin{array}
[c]{c}%
\\
\overset{%
\raisebox{-0.0303in}{\includegraphics[
height=0.1332in,
width=0.1193in
]%
{red-refinement.ps}%
}%
}{\longrightarrow}\\
\end{array}
&
\raisebox{-0.3009in}{\includegraphics[
height=0.7463in,
width=0.7619in
]%
{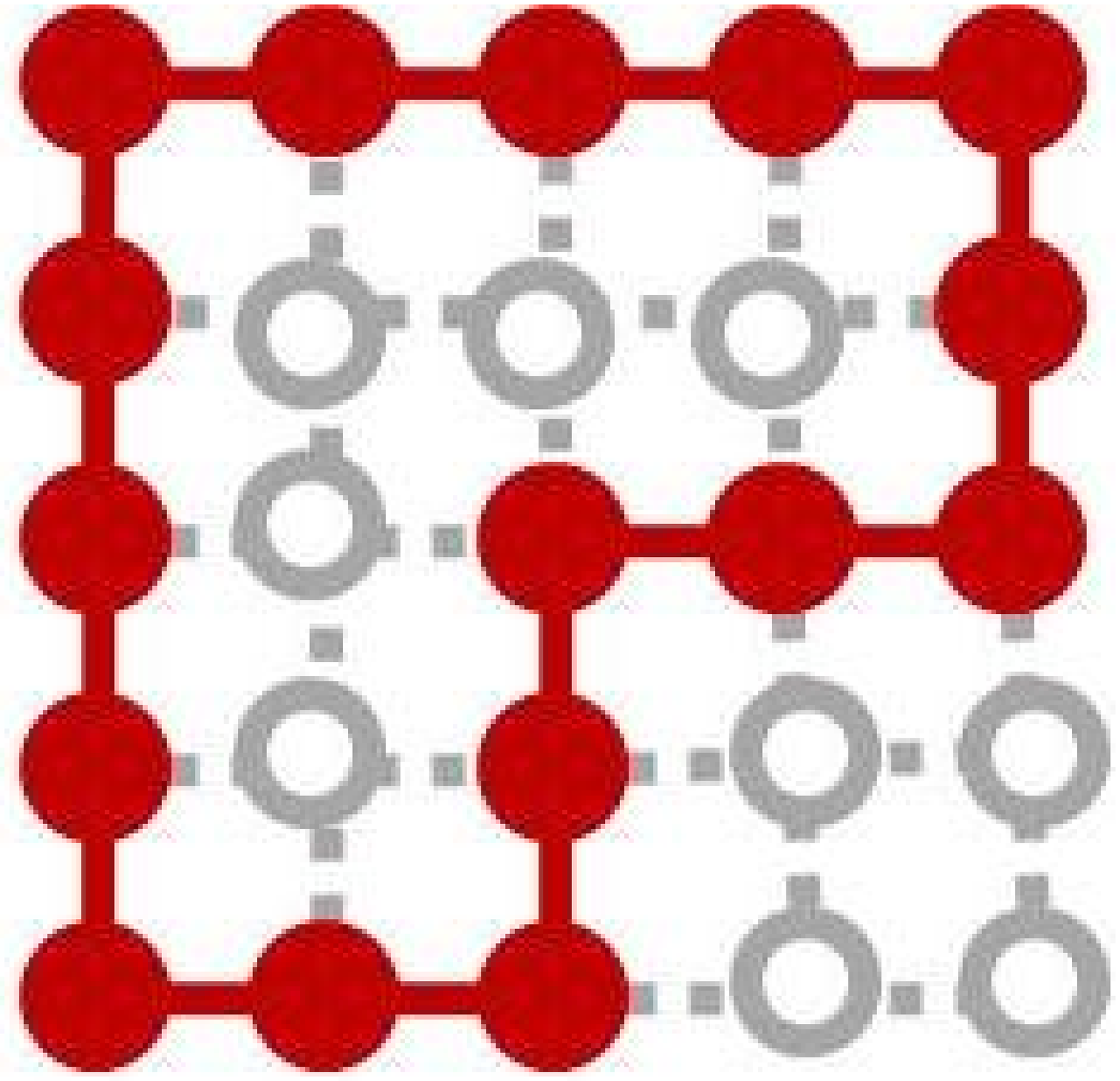}%
}%
\end{array}
$\\
\textbf{An example of the refinement injection }$%
\raisebox{-0.0303in}{\includegraphics[
height=0.1332in,
width=0.1193in
]%
{red-refinement.ps}%
}%
.$%
\end{tabular}
\]

\end{definition}

\bigskip

We would also like to construct a \textbf{refinement map}
\[%
\raisebox{-0.0303in}{\includegraphics[
height=0.1332in,
width=0.1193in
]%
{red-refinement.ps}%
}%
:\Lambda_{\ell}\longrightarrow\Lambda_{\ell+1}%
\]
that

\bigskip

\begin{itemize}
\item[\textbf{(a)}] Is a group monomorphism, and\bigskip

\item[\textbf{(b)}] Preserves the actions of the ambient groups $\Lambda
_{\ell}$ and $\Lambda_{\ell+1}$ on the sets $\mathbb{K}^{(\ell)}$ and
$\mathbb{K}^{(\ell+1)}$, respectively. \ In other words, we would like the
following diagram to be commutative:
\[%
\begin{array}
[c]{ccc}%
\Lambda_{\ell}\times\mathbb{K}^{(\ell)} & \longrightarrow & \mathbb{K}%
^{(\ell)}\\%
\raisebox{-0.0303in}{\includegraphics[
height=0.1332in,
width=0.1193in
]%
{red-refinement.ps}%
}%
\times%
\raisebox{-0.0303in}{\includegraphics[
height=0.1332in,
width=0.1193in
]%
{red-refinement.ps}%
}%
\downarrow\qquad\qquad\quad &  &
\raisebox{-0.0303in}{\includegraphics[
height=0.1332in,
width=0.1193in
]%
{red-refinement.ps}%
}%
\downarrow\qquad\\
\Lambda_{\ell+1}\times\mathbb{K}^{(\ell+1)} & \longrightarrow & \mathbb{K}%
^{(\ell+1)}%
\end{array}
\text{ ,}%
\]

\end{itemize}

\bigskip

Appendix C gives a suggested construction for such a refinement map $%
\raisebox{-0.0303in}{\includegraphics[
height=0.1332in,
width=0.1193in
]%
{red-refinement.ps}%
}%
:\Lambda_{\ell}\longrightarrow\Lambda_{\ell+1}$. \ It also gives a rationale
for the following conjectures:

\bigskip

\noindent\textbf{Conjecture 1A.} \ \textit{The Appendix C construction
produces a well-defined map }$%
\raisebox{-0.0303in}{\includegraphics[
height=0.1332in,
width=0.1193in
]%
{red-refinement.ps}%
}%
:\Lambda_{\ell}\longrightarrow\Lambda_{\ell+1}$\textit{.}\bigskip

\noindent\textbf{Conjecture 1B. \ }\textit{The Appendix C construction
produces an injection }$%
\raisebox{-0.0303in}{\includegraphics[
height=0.1332in,
width=0.1193in
]%
{red-refinement.ps}%
}%
:\Lambda_{\ell}\longrightarrow\Lambda_{\ell+1}$\textit{.}\bigskip

\noindent\textbf{Conjecture 1C. \ }\textit{The Appendix C construction
produces a monomorphism }$%
\raisebox{-0.0303in}{\includegraphics[
height=0.1332in,
width=0.1193in
]%
{red-refinement.ps}%
}%
:\Lambda_{\ell}\longrightarrow\Lambda_{\ell+1}$\textit{ that respects the
action of the ambient groups }$\Lambda_{\ell}$\textit{ and }$\Lambda_{\ell+1}%
$\textit{ on the sets n the sets }$\mathbb{K}^{(\ell)}$\textit{ and
}$\mathbb{K}^{(\ell+1)}$\textit{.}

\bigskip

\section{Lattice knot type}

\bigskip

We are now finally at a point where we can define lattice knot type.

\bigskip

\begin{definition}
Two lattice knots $K_{1}$ and $K_{2}$ of order $\ell$ are said to be of the
\textbf{same knot }$\ell$-\textbf{type}, written%
\[
K_{1}\underset{\ell}{\thicksim}K_{2}\text{ \ ,}%
\]
provided there exists an element $g$ of the lattice ambient group
$\Lambda_{\ell}$ that transforms one into the other. \ \ They are said to be
of the \textbf{same} \textbf{knot type}, written%
\[
K_{1}\thicksim K_{2}\text{ ,}%
\]
provided there exists a non-negative integer $m$ such that
\[%
\raisebox{-0.0303in}{\includegraphics[
height=0.1332in,
width=0.1193in
]%
{red-refinement.ps}%
}%
^{m}K_{1}\underset{\ell+m}{\thicksim}%
\raisebox{-0.0303in}{\includegraphics[
height=0.1332in,
width=0.1193in
]%
{red-refinement.ps}%
}%
^{m}K_{2}%
\]

\end{definition}

\bigskip

In like manner, we can define inextensible knot $\ell$-type and inextensible
knot type as follows:

\bigskip

\begin{definition}
Two lattice knots $K_{1}$ and $K_{2}$ of order $\ell$ are said to be of the
\textbf{same inextensible knot }$\ell$-\textbf{type}, written%
\[
K_{1}\underset{\ell}{\thickapprox}K_{2}\text{ \ ,}%
\]
provided there exists an element $g$ of the inextensible lattice ambient group
$\widetilde{\Lambda}_{\ell}$ that transforms one into the other. \ They are
said to be of the \textbf{same inextensible} \textbf{knot type}, written%
\[
K_{1}\approx K_{2}\text{ ,}%
\]
provided there exists a non-negative integer $m$ such that
\[%
\raisebox{-0.0303in}{\includegraphics[
height=0.1332in,
width=0.1193in
]%
{red-refinement.ps}%
}%
^{m}K_{1}\underset{\ell+m}{\thickapprox}%
\raisebox{-0.0303in}{\includegraphics[
height=0.1332in,
width=0.1193in
]%
{red-refinement.ps}%
}%
^{m}K_{2}%
\]
\ 
\end{definition}

\bigskip

Because the move tug is part of the definition of knot $\ell$-type, but not
part of the definition of inextensible knot $\ell$-type, we have the following proposition:

\bigskip

\section{The preferred vertex (PV) approximation for knots}

\bigskip

\begin{definition}
A knot $x$ in Euclidean 3-space $\mathbb{R}^{3}$ is said to be
\textbf{finitely piecewise smooth (FPS)} provided the knot $x$ consists of
finitely many smooth ($C^{\infty}$) segments $x_{1}$, $x_{2}$, $\ldots$,
$x_{r}$ with no two segments tangent to each other at their
endpoints\footnote{We should also mention that we really only need for the
curve $x$ to be piecewise $C^{3}$.}. \ Such a knot $x$ will be called an
\textbf{FPS knot}. \ Obviously, every lattice knot is an FPS knot.
\end{definition}

\bigskip

In this section, we will create, for each non-negative integer $\ell$, a
procedure for finding an $\ell$-th order lattice graph $PV^{(\ell)}\left(
x\right)  $, called the $\ell$\textbf{-th preferred vertex approximation},
that approximates an arbitrary FPS knot $x$ in $\mathbb{R}^{3}$.

\bigskip

We begin by noting that, for each non-negative integer $\ell$, the set%
\[
\left\{  B^{(\ell)}(a_{\beta})\cup\overset{}{\left(
{\displaystyle\bigcup_{p=1}^{3}}
F_{p}^{(\ell)}(a_{\beta})\right)  }\cup\left(
{\displaystyle\bigcup_{p=1}^{3}}
E_{p}^{(\ell)}(a_{\beta})\right)  \cup\left\{  a\right\}  =\overline
{B^{(\ell)}(a_{\beta})}-%
{\displaystyle\bigcup_{p=1}^{3}}
\overline{F_{p}^{(\ell)}(a_{\beta}^{:p})}\quad:\quad a\in\mathcal{L}_{\ell
}\right\}
\]
of half closed cubes is a partition of 3-space $\mathbb{R}^{3}$. (We have used
$\overline{C}$ to denote the closure of a cell $C$.) \ So we can create a map
of space $\mathbb{R}^{3}$ into the lattice $\mathcal{L}_{\ell}$ simply by
mapping each point of a half closed cube to the preferred vertex of that half
closed cube\footnote{We remind the reader that cells $E_{p}^{(\ell)}(a_{\beta
})$, $F_{p}^{(\ell)}(a_{\beta})$, $B^{(\ell)}(a_{\beta})$ are open cells,
unless stated otherwise.}:%
\begin{center}
\fbox{\includegraphics[
height=1.5074in,
width=3.7905in
]%
{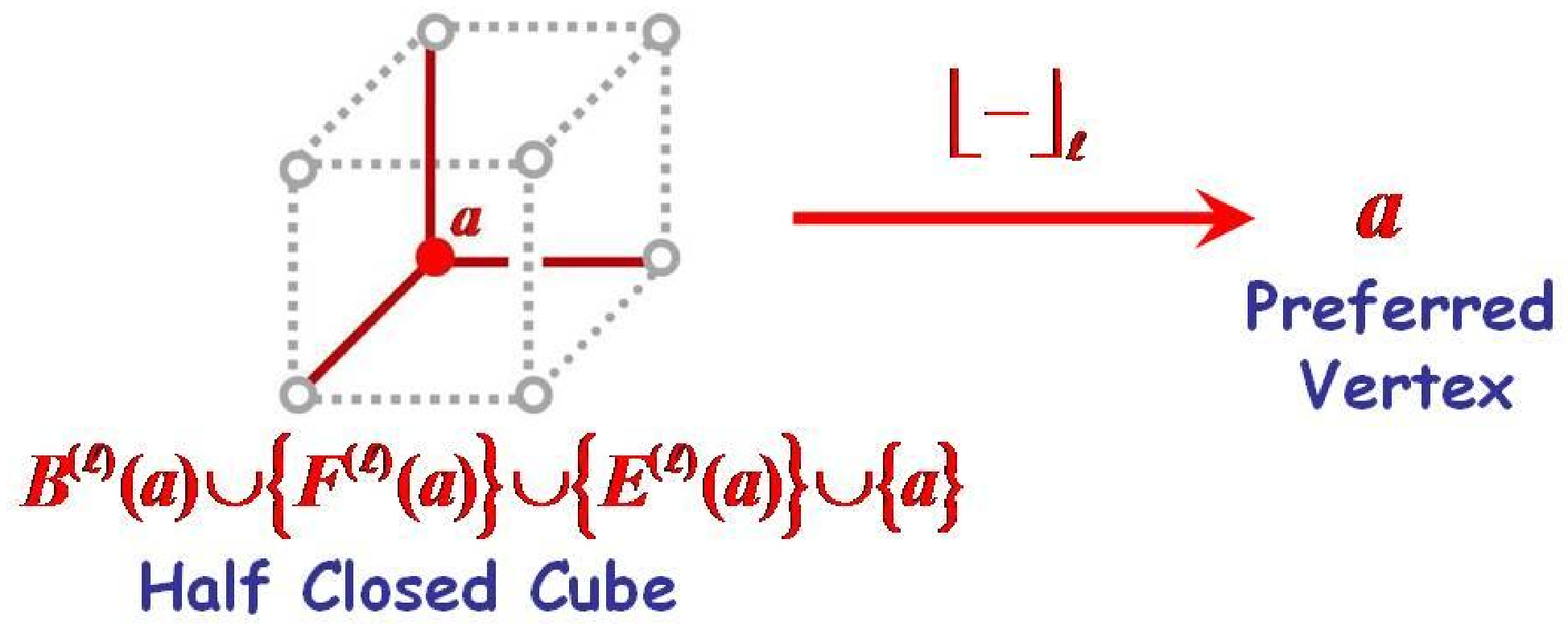}%
}\\
\textbf{The preferred vertex map which maps each half closed cube to its
preferred vertex.}%
\end{center}

\bigskip

\bigskip

\begin{definition}
We define the $\mathbf{\ell}$\textbf{-th preferred vertex map} $\left\lfloor
-\right\rfloor _{\ell}:\mathbb{R}^{3}\longrightarrow\mathcal{L}_{\ell}$ from
Euclidean 3-space $\mathbb{R}^{3}$ to the lattice $\mathcal{L}_{\ell}$ as%
\[%
\begin{array}
[c]{rccc}%
\left\lfloor -\right\rfloor _{\ell}: & \mathbb{R}^{3} & \longrightarrow &
\mathcal{L}_{\ell}\\
& x=\left(  x_{1},x_{2},x_{3}\right)  & \longmapsto & \left\lfloor
x\right\rfloor _{\ell}=2^{-\ell}\left\lfloor 2^{\ell}x\right\rfloor =\left(
2^{-\ell}\underset{}{\overset{}{\left\lfloor 2^{\ell}x_{1}\right\rfloor }%
},2^{-\ell}\left\lfloor 2^{\ell}x_{2}\right\rfloor ,2^{-\ell}\left\lfloor
2^{\ell}x_{3}\right\rfloor \right)
\end{array}
\]
where $\left\lfloor -\right\rfloor $ denotes the floor function.
\end{definition}

\bigskip

\begin{remark}
It immediately follows that the inverse image under $\left\lfloor
-\right\rfloor _{\ell}$ of each vertex $a$ in the lattice $\mathcal{L}_{\ell}$
is the half closed cube $\overline{B^{(\ell)}}(a)-%
{\displaystyle\bigcup_{p=1}^{3}}
\overline{F_{p}^{(\ell)}}(a^{:p})$.
\end{remark}

\bigskip

\begin{remark}
Let $r$ be a real number. $\ $Then $\left\lfloor r\right\rfloor _{\ell}$ is
the rational number resulting from deleting, in the binary expansion of $r$,
all bits to the right of the $\ell$-th bit after the decimal point. \ For
example, $\left\lfloor \ \left(  1101.10110100\ldots\right)  _{2}%
\ \right\rfloor _{5}=\left(  1101.10110\right)  _{2}$, where $\left(
-\right)  _{2}$ denotes the binary expansion of a real number.
\end{remark}

\bigskip

Let $\ell$ be an arbitrarily chosen non-negative integer. Let $x$ be an
arbitrary FPS knot in $\mathbb{R}^{3}$, and let $x_{1}$, $x_{2}$, $\ldots$,
$x_{r}$ denote the simple closed FPS curves which are the components of $x$.
\ (Thus each component $x_{\beta}$ consists of finitely many piecewise smooth segments.)

\bigskip

For each component $x_{\beta}$, choose any one of the two possible
orientations. \ Also select an arbitrary point $P_{0,\beta}$ of $x_{\beta}$.
\ Starting at the point $P_{0,\beta}$, traverse the curve $x_{\beta}$ in the
direction of its orientation all the way around from $P_{0,\beta}$ back to
$P_{0,\beta}$. \ As the simple closed curve is traversed, the preferred vertex
map $\left\lfloor -\right\rfloor _{\ell}$ will produce a finite sequence of
lattice points%
\[
a_{\beta}^{(0)},a_{\beta}^{(1)},a_{\beta}^{(2)},\ldots,a_{\beta}^{(u_{\beta}%
)}=a_{\beta}^{(0)}%
\]

\bigskip

For each pair of two consecutive vertices $\left(  a_{\beta}^{(j)},a_{\beta
}^{(j+1)}\right)  $ of this sequence, select a shortest path from vertex
$a_{\beta}^{(j)}$ to vertex $a_{\beta}^{(j+1)}$ that lies in the 1-skeleton
$\mathcal{C}_{\ell}^{1}$ of the cell complex $\mathcal{C}_{\ell}$, and let
$S_{j,\beta}$ denote the set of edges and vertices in the selected shortest
path. \ An $\ell$\textbf{-th order preferred vertex approximation of }the knot
component $x_{\beta}$, written $PV^{(\ell)}(x_{\beta})$, is defined as the
lattice graph
\[
PV^{(\ell)}(x_{\beta})=%
{\displaystyle\bigcup_{j=0}^{u_{\beta}-1}}
S_{j,\beta}%
\]

\bigskip

Finally, an $\ell$\textbf{-th order preferred vertex approximation of }the
link $x$, written $PV^{(\ell)}(x)$, is defined as the lattice graph
\[
PV^{(\ell)}(x)=%
{\displaystyle\bigcup_{\beta=0}^{r}}
PV^{(\ell)}(x_{\beta})
\]

\bigskip

The following lemma is a consequence of the observation that, since the knot
$x$ is piecewise smooth and compact, its curvature and torsion are bounded.

\bigskip

\begin{lemma}
Let $x$ be a FPS knot in Euclidean 3-space $\mathbb{R}^{3}$. \ Then there
exists a non-negative integer $\ell_{0}^{\prime}=\ell_{0}^{\prime}(x)$ such
that, for $\ell\geq\ell_{0}^{\prime}$, every $\ell$-th order preferred vertex
approximation $PV^{(\ell)}(x)$ of $x$ is a lattice knot of order $\ell$. \ In
other words, for sufficiently large $\ell$, $PV^{(\ell)}(x)$ of $x$ is an
$\ell$-th order lattice knot.
\end{lemma}

\bigskip

\begin{corollary}
Let $x$ be a FPS knot in Euclidean 3-space $\mathbb{R}^{3}$, and let
$\epsilon$ be an arbitrary positive real number. \ Then there exists a
non-negative integer $\ell_{0}=\ell_{0}(x,\epsilon)$ such that, for $\ell
\geq\ell_{0}$, every $\ell$-th order preferred vertex approximation
$PV^{(\ell)}(x)$ of $x$ is such that

\begin{itemize}
\item[\textbf{1)}] $PV^{(\ell)}(x)$ is a lattice knot of order $\ell$,

\item[\textbf{2)}] $PV^{(\ell)}(x)$ lies inside the open tubular neighborhood
of $x$ of radius $\epsilon$, and

\item[\textbf{3)}] $PV^{(\ell)}(x)$ is of the same knot type as the knot $x$.
\end{itemize}
\end{corollary}

\bigskip

\begin{remark}
The key feature of the preferred vertex approximation is that there exists a
non-negative integer $\ell_{0}$ such that $PV^{(\ell)}(x)$ is of the same knot
knot type as the knot $x$. \ The reader should note, however, that the length
of $PV^{(\ell)}(x)$ is is the same for all $\ell\geq0$. \ Hence, the limit
$\lim_{\ell\rightarrow\infty}PV^{(\ell)}(x)$, if it exists, is not necessarily
the knot $x$. \ An illustration of this is the unit circle
\[
\left\{  x\in\mathbb{R}^{3}:x_{1}^{2}+x_{2}^{2}=1,x_{3}=0\right\}  \text{ .}%
\]
While the circle $x$ is of total length $2\pi$, the total length of its
preferred vertex approximation $PV^{(\ell)}(x)$ is $8$ for all $\ell\geq0$.
\ Thus, for sufficiently large $\ell$, the preferred vertex approximation
$PV^{(\ell)}(x)$ can be thought of as a "wrinkled up" version of the original
knot $x$.
\end{remark}

\bigskip

\section{Wiggle, Wag, and Tug variational derivatives, and knot invariants}

\bigskip

In this section, we will define wiggle, wag, and tug variational derivatives,
and discuss how they are connected to knot invariants. \ We begin by defining
a domain for these variational derivatives. \ We will be making extensive use
of the calculus of variation. For readers unfamiliar with this subject, we
refer them, for example, to \cite{Gelfand1}.

\bigskip

\begin{definition}
Let $\mathfrak{K}$ be the \textbf{topological space\footnote{Actualy, this
space has much more mathematical structure. \ But to use this additional
structure would take us far beyond the scope of this paper.} of all finitely
piecewise smooth (FPS)\ knots}\footnote{We remind the reader that in this
paper "knot" means either a knot or a link.} with the compact-open topology.
\ A subspace $\mathfrak{K}^{\prime}$ of $\mathfrak{K}$ is said to be
\textbf{FPS-complete subspace of }$\mathfrak{K}$ provided for every knot
$x^{\prime}$ in $\mathfrak{K}^{\prime}$ and for every knot $x$ in
$\mathfrak{K}$, $x^{\prime}$ is of the same knot type as $x$ (written
$x^{\prime}\sim x$) implies that $x$ also lies in $\mathfrak{K}^{\prime}$.
\end{definition}

\bigskip

\begin{remark}
We should mention that the \textbf{set of all lattice knots}%
\[
\mathbb{K}=%
{\displaystyle\bigcup\limits_{\ell=0}^{\infty}}
\mathbb{K}^{(\ell)}\subset\mathfrak{K}%
\]
is a proper subset of the set $\mathfrak{K}$ of FPS knots, but is by no means
an FPS-complete subspace of $\mathfrak{K}$.
\end{remark}

\bigskip

\begin{example}
An example of a FPS-complete subspace $\mathfrak{K}^{\prime}$ of
$\mathfrak{K}$ would be the topological space $\mathcal{LINK}_{2}$ of all FPS
knots which are two component links, i.e.,
\[
\mathcal{LINK}_{2}=\left\{  x\in\mathfrak{K}:x\text{ is a two component
link}\right\}
\]

\end{example}

\bigskip

\begin{definition}
A real valued map $F:\mathfrak{K}^{\prime}\longrightarrow\mathbb{R}$ on a
FPS-complete subspace $\mathfrak{K}^{\prime}$ of $\mathfrak{K}$ in will be
called a \textbf{knot functional}.
\end{definition}

\bigskip

\begin{example}
Let $x=\sqcup_{\beta=1}^{r}K_{\beta}$ be an $r$-component knot in
$\mathfrak{K}$, i.e., $x$ is the disjoint union of $r$ simple closed FPS
curves $x_{1}$, $x_{2}$, \ $\ldots$ , $x_{r}$. \ For each component $x_{\beta
}$, let $x_{\beta}=x_{\beta}\left(  s_{\beta}\right)  =\left(  x_{\beta
,1},x_{\beta,2},x_{\beta,3}\right)  $ be a parameterization by arclength
$s_{\beta}$, $0\leq s_{\beta}\leq L_{\beta}$, where $L_{\beta}$ denotes the
length of $x_{\beta}$. \ Then the following is an example of a knot functional
which is an invariant of inextensible knot type.%
\[%
\begin{array}
[c]{rcl}%
Length:\mathfrak{K} & \longrightarrow & \mathbb{R}\\
x & \longmapsto &
{\displaystyle\sum\limits_{\beta=1}^{r}}
{\displaystyle\oint}
\sqrt{\left(  \frac{dx_{\beta,1}}{ds_{\beta}}\right)  ^{2}+\left(
\frac{dx_{\beta,2}}{ds_{\beta}}\right)  ^{2}+\left(  \frac{dx_{\beta,3}%
}{ds_{\beta}}\right)  ^{2}}ds_{\beta}%
\end{array}
\]

\end{example}

\bigskip

\begin{example}
An example of a knot functional, which is an invariant of knot type, is the
magnitude of the Gauss integral. \ Let $x$ be an arbitrary knot in
$\mathcal{LINK}_{\mathbf{2}}$, and let $\xi$ and $\zeta$ denote its two
components. \ Then the magnitude
\[
Link[x]=\frac{1}{4\pi}\left\vert
{\displaystyle\oint\nolimits_{\xi}}
{\displaystyle\oint\nolimits_{\zeta}}
\frac{\left(  \zeta-\xi\right)  \times d\zeta}{\left\vert \zeta-\xi\right\vert
^{3}}\cdot d\xi\right\vert \text{ ,}%
\]
of the Gauss integral of $x$ is an invariant of inextensible knot type,\ where
$\xi$ and $\zeta$ denote its two components\footnote{We use the absolute value
of the Gauss integral since only unoriented knots and links are discussed in
this paper. \ Everything in this paper can easily be extended to oriented
knots. \ With that generalization, there is no longer a need to take the
magnitude of the Gauss integral.}. \ 
\end{example}

\bigskip

We will need the following elements of the ambient group $\Lambda_{\ell}%
$:\bigskip

\begin{itemize}
\item \fbox{$1$}$^{(\ell)}(a,p)=%
{\displaystyle\prod\limits_{q=0}^{3}}
L_{1}^{(\ell)}(a,p,q)$, called a \textbf{total tug} on the face $F_{p}%
^{(\ell)}(a)$, $p=1,2,3$\bigskip

\item \fbox{$2$}$^{(\ell)}(a,p)=%
{\displaystyle\prod\limits_{q=0}^{1}}
L_{2}^{(\ell)}(a,p,q)$, called a \textbf{total wiggle} on the face
$F_{p}^{(\ell)}(a)$, $p=1,2,3$\bigskip

\item \fbox{$3$}$^{(\ell)}(a,p)=%
{\displaystyle\prod\limits_{q=0}^{3}}
L_{3}^{(\ell)}(a,p,q)$, called a \textbf{total wag} on the face $F_{p}%
^{(\ell)}(a)$, $p=1,2,3$
\end{itemize}

\bigskip

\begin{remark}
The reader should note that all the elements in each of the above products
commute with one another. Moreover, at most one element in each product can be
different from the identity transformation when the total move is applied to a
specific lattice knot $K$. \ For example, in the product for the total tug
\fbox{$1$}$^{(\ell)}(a,p)$, the tugs $L_{1}^{(\ell)}(a,p,0)$, $L_{1}^{(\ell
)}(a,p,1)$, $L_{1}^{(\ell)}(a,p,2)$, $L_{1}^{(\ell)}(a,p,3)$ all commute with
one another, and moreover, when they are applied to a specific lattice knot,
at most one of these tugs is different from the identity $1$.
\end{remark}

\bigskip

We are now in a position to define wiggle, wag, and tug variational derivatives.

\bigskip

\begin{definition}
Let $\mathfrak{K}^{\prime}$ be a FPS-complete subspace of the space
$\mathfrak{K}$ of all FPS knots, and let $x$ be a FPS knot in $\mathfrak{K}%
^{\prime}$. \ Moreover, let
\[
F:\mathfrak{K}^{\prime}\longrightarrow\mathbb{R}%
\]
be a knot functional. \ At each point $a\in x$, we define the \textbf{tug,
wiggle, and wag variational derivatives} respectively as follows:\bigskip

\begin{itemize}
\item $\frac{\delta F\left[  x\right]  }{\delta\fbox{$1$}(a,p)}=\lim
\limits_{\ell\rightarrow\infty}\frac{F\left[  \fbox{$1$}^{(\ell)}%
(a,p)PV^{(\ell)}\left(  x\right)  \right]  -F\left[  \overset{\mathstrut
}{PV^{(\ell)}\left(  x\right)  }\right]  }{\left(  2^{-\ell}\right)  ^{2}}$,
whenever the limit exists\bigskip

\item $\frac{\delta F\left[  x\right]  }{\delta\fbox{$2$}(a,p)}=\lim
\limits_{\ell\rightarrow\infty}\frac{F\left[  \fbox{$2$}^{(\ell)}%
(a,p)PV^{(\ell)}\left(  x\right)  \right]  -F\left[  \overset{\mathstrut
}{PV^{(\ell)}\left(  x\right)  }\right]  }{\left(  2^{-\ell}\right)  ^{2}}$,
whenever the limit exists\bigskip

\item $\frac{\delta F\left[  x\right]  }{\delta\fbox{$3$}(a,p)}=\lim
\limits_{\ell\rightarrow\infty}\frac{F\left[  \fbox{$3$}^{(\ell)}%
(a,p)PV^{(\ell)}\left(  x\right)  \right]  -F\left[  \overset{\mathstrut
}{PV^{(\ell)}\left(  x\right)  }\right]  }{\left(  2^{-\ell}\right)  ^{2}}$,
whenever the limit exists
\end{itemize}

\noindent for $p=1,2,3$, where $\left(  2^{-\ell}\right)  ^{2}$ is the area of
the face $F_{p}^{(\ell)}(a)$. \ Moreover, we define the \textbf{tug, wiggle,
and wag variational gradients} as%
\[
\frac{\delta F\left[  K\right]  }{\delta\fbox{$m$}\left(  a\right)  }=\left(
\frac{\delta F\left[  K\right]  }{\delta\fbox{$m$}(a,1)},\frac{\delta F\left[
K\right]  }{\delta\fbox{$m$}(a,2)},\frac{\delta F\left[  K\right]  }%
{\delta\fbox{$m$}(a,3)}\right)  \text{ .}%
\]

\end{definition}

\bigskip

\begin{conjecture}
A functional $F:\mathfrak{K}^{\prime}\longrightarrow\mathbb{R}$ is a knot
invariant if all its tug, wiggle, and wag variational gradients exist and are
equal to $\left(  0,0,0\right)  $. \ The functional $F$ is an inextensible
knot invariant if all its wiggle, and wag variational gradients exist and are
equal to $\left(  0,0,0\right)  $.
\end{conjecture}

\bigskip

We leave the following two exercises for the reader:

\bigskip

\noindent\textbf{Exercise 1}\textit{. Show that for all knots }$K$\textit{ in
}$\mathbb{K}^{(\infty)}=%
{\displaystyle\bigcup\limits_{\ell=0}^{\infty}}
\mathbb{K}^{(\ell)}$\textit{, the wiggle and wag variational gradients of the
functional }$Length[K]$ \textit{vanish at each vertex of }$K$\textit{.}

\bigskip

\noindent\textbf{Exercise 2.}\textit{ Show that for all knots }$K$\textit{ in
}$\mathbb{K}^{(\infty)}\cap\mathcal{LINK}_{\mathbf{2}}$\textit{, the tug,
wiggle, and wag variational gradients of the functional }$Link[K]$
\textit{vanish at each vertex of }$K$\textit{.\ }

\bigskip

\section{Infinitesimal Wiggles, Wags, and Tugs, differential forms, and
integrals}

\bigskip

There are many consequences to the research developments discussed in the
previous section of this paper. \ 

\bigskip

For example, the approach found in the previous section lead to the
construction of infinitesimal wiggles, wags, and tugs, such as for example the
infinitesimal wiggle%
\[%
\raisebox{-0.0406in}{\includegraphics[
height=0.1436in,
width=0.1436in
]%
{icon21.ps}%
}%
\left(  x\right)  ^{\frac{\partial}{\partial x_{1}}\otimes\frac{\partial
}{\partial x_{2}}}=\lim_{\ell\rightarrow\infty}%
\raisebox{-0.0406in}{\includegraphics[
height=0.1436in,
width=0.1436in
]%
{icon21.ps}%
}%
^{(\ell)}\left(  \left\lfloor x\right\rfloor _{\ell},3\right)  \text{ .}%
\]
Moreover, there is also the corresponding differential form%
\[%
\raisebox{-0.0406in}{\includegraphics[
height=0.1436in,
width=0.1436in
]%
{icon21.ps}%
}%
\left(  x\right)  ^{dx_{1}dx_{2}}\text{ ,}%
\]
and its multiplicative integrals, such as for example,
\[
\underset{x_{1}=0..1}{%
\raisebox{-0.1003in}{\includegraphics[
height=0.3061in,
width=0.198in
]%
{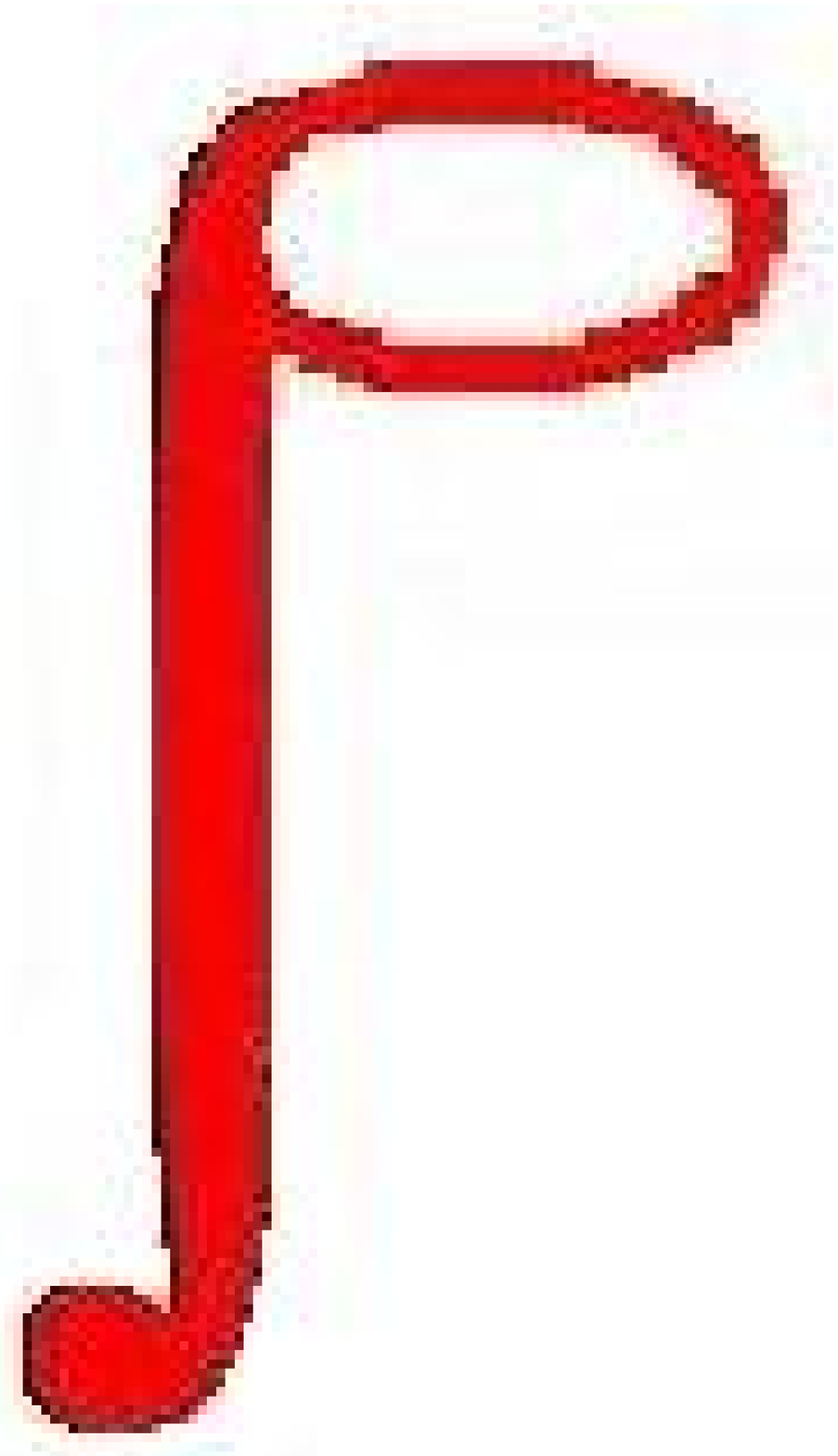}%
}%
}\ \underset{x_{2}=0..1}{%
\raisebox{-0.1003in}{\includegraphics[
height=0.3061in,
width=0.198in
]%
{Integral.ps}%
}%
}\
\raisebox{-0.0406in}{\includegraphics[
height=0.1436in,
width=0.1436in
]%
{icon21.ps}%
}%
\left(  x_{1},x_{2},0\right)  ^{dx_{1}dx_{2}}\text{ ,}%
\]
where, for example,
\[
\left(
\raisebox{-0.0406in}{\includegraphics[
height=0.1436in,
width=0.1436in
]%
{icon21.ps}%
}%
^{(0)}\left(  \overset{\mathstrut}{\left(  0,0,0\right)  },3\right)  \right)
\left(  K\right)  =\left(  \underset{x_{1}=0..1}{%
\raisebox{-0.1003in}{\includegraphics[
height=0.3061in,
width=0.198in
]%
{Integral.ps}%
}%
}\ \underset{x_{2}=0..1}{%
\raisebox{-0.1003in}{\includegraphics[
height=0.3061in,
width=0.198in
]%
{Integral.ps}%
}%
}\
\raisebox{-0.0406in}{\includegraphics[
height=0.1436in,
width=0.1436in
]%
{icon21.ps}%
}%
\left(  x_{1},x_{2},0\right)  ^{dx_{1}dx_{2}}\right)  \left(  K\right)  \text{
,}%
\]
for all $0$-th order lattice knots $K$ such that%
\[
K\cap F_{3}^{(0)}\left(  \overset{\mathstrut}{\left(  0,0,0\right)  }\right)
=%
\raisebox{-0.2508in}{\includegraphics[
height=0.5838in,
width=0.5838in
]%
{wigl0.ps}%
}%
\text{ \ \ or \ \ }%
\raisebox{-0.2508in}{\includegraphics[
height=0.5838in,
width=0.5838in
]%
{wigr0.ps}%
}%
\text{ .}%
\]

\bigskip

It is because of these developments (which were motivated by \cite{Lomonaco2})
that we have suggested that the moves wiggle, wag, and tug are more
"\textbf{physics friendly}" than the Reidemeister moves. \ Unfortunately,
because of the scope of this current paper, this section is of necessity
sketchy and abbreviated. \ Readers interested in a more in depth discussion of
this material are referred to the upcoming paper \cite{Lomonaco8}.

\bigskip

\section{$\mathbf{n}$-bounded lattice knots}

\bigskip

As a preliminary step to defining quantum knots and quantum knot systems, we
will now put bounds on the mathematical constructs given in previous sections
of this paper.

\bigskip

We define the $\mathbf{n}$\textbf{-bounded lattice }$\mathcal{L}_{\ell,n}%
$\textbf{ }as the sublattice of $\mathcal{L}_{\ell}$ given by%
\[
\mathcal{L}_{\ell,n}=\left\{  a\in\mathcal{L}_{\ell}:0\leq a_{j}\leq n\text{,
for }j=1,2,3\right\}  \text{ .}%
\]
Let $\mathcal{C}_{\ell,n}$ denote the corresponding $n$-\textbf{bounded cell
complex,} and $\mathcal{C}_{\ell,n}^{j}$ its $n$-\textbf{bounded}
$j$-\textbf{skeleton}. \ We also define a \textbf{bounded lattice knot }%
$K$\textbf{ of order }$(\ell,n)$ as a closed 2-valent subgraph of the
$n$-bounded $1$-skeleton $\mathcal{C}_{\ell,n}^{1}$, and let $\mathbb{K}%
^{(\ell,n)}$ denote the set of all \textbf{bounded lattice knots of order
}$(\ell,n)$.

\bigskip

\begin{remark}
For the sake of simplicity, we have defined $n$-bounded lattices that only lie
in the first octant. \ From the perspective of this paper there is no need to
extend the definition to all of $\mathbb{R}^{3}$. Such an extension would only
lead to unnecessary additional complexity.
\end{remark}

\bigskip

\begin{definition}
The \textbf{bounded (lattice) ambient group }$\Lambda_{\ell,n}$\textbf{\ (of
order }$(\ell,n)$ \textbf{) }is the finite subgroup of the (lattice) ambient
$\Lambda_{\ell}$ generated by all tugs $L_{1`}^{(\ell)}\left(  a,p,q\right)  $
and wiggles $L_{2`}^{(\ell)}\left(  a,p,q\right)  $ such that%
\[
0\leq a_{p}\leq n\text{, \ \ }0\leq a_{\left\lfloor p\right.  }\leq n-1\text{,
\ \ }0\leq a_{\left.  p\right\rceil }\leq n-1
\]
and by all wags $L_{3`}^{(\ell)}\left(  a,p,q\right)  $ such that $0\leq
a_{j}\leq n-1$, for $j=1,2,3$.
\end{definition}

\bigskip

We will also have need of the \textbf{lattice knot injection} defined by
\[%
\begin{array}
[c]{ccc}%
\iota:\mathbb{K}^{(\ell,n)} & \longrightarrow & \mathbb{K}^{(\ell,n+1)}\\
K^{(\ell,n)} & \longmapsto & K^{(\ell,n+1)}%
\end{array}
\text{ ,}%
\]
where $K^{(\ell,n+1)}=\iota\left(  K^{(\ell,n)}\right)  $ is the lattice knot
of order $(\ell,n+1)$ consisting of all the edges in $K^{(\ell,n)}$.
\ Moreover, we define in the obvious way the (lattice) knot group
monomorphism
\[
\iota:\Lambda_{\ell,n}\longrightarrow\Lambda_{\ell,n+1}\text{ .}%
\]

\bigskip

\begin{definition}
Let $\mathbb{K}^{[\ell]}$ denote the \textbf{directed system of bounded
lattice knots} $\left\{  \mathbb{K}^{(\ell,n)}\longrightarrow\mathbb{K}%
^{(\ell,n+1)}:n=1,2,3,\ldots\right\}  $ and $\Lambda_{\lbrack\ell]}$ the
\textbf{directed system of bounded ambient groups} $\left\{  \Lambda_{\ell
,n}\longrightarrow\Lambda_{\ell,n+1}:n=1,2,3,\ldots\right\}  $. \ Finally,let
$\left(  \mathbb{K}^{[\ell]},\Lambda_{\lbrack\ell]}\right)  $ denote the
directed graded system
\[
\left(  \mathbb{K}^{[\ell]},\Lambda_{\lbrack\ell]}\right)  =\left(
\mathbb{K}^{(\ell,1)},\Lambda_{\ell,1}\right)  \longrightarrow\left(
\mathbb{K(}^{\ell,2)},\Lambda_{\ell,2}\right)  \longrightarrow\cdots
\longrightarrow\left(  \mathbb{K}^{(\ell,n)},\Lambda_{\ell,n}\right)
\longrightarrow\cdots
\]

\end{definition}

\bigskip

\begin{definition}
Two lattice knots $K_{1}$ and $K_{2}$ of order $\left(  \ell,n\right)  $ are
said to be of the \textbf{same bounded knot }$(\ell,n)$-\textbf{type}, written%
\[
K_{1}\underset{\ell}{\overset{n}{\thicksim}}K_{2}\text{ \ ,}%
\]
provided there exists an element $g$ of the bounded ambient group
$\Lambda_{\ell,n}$ that transforms one into the other. \ \ They are said to be
of the \textbf{same knot }$\left(  \ell,\infty\right)  $-\textbf{type},
written%
\[
K_{1}\overset{\infty}{\underset{\ell}{\sim}}K_{2}\text{ \ ,}%
\]
provided there exists a non-negative integer $n^{\prime}$ such that
\[
\iota^{n^{\prime}}K_{1}\underset{\ell}{\overset{n+n^{\prime}}{\thicksim}}%
\iota^{n^{\prime}}K_{2}%
\]
There are of the \textbf{same knot }$\left(  \infty,n\right)  $-\textbf{type},
written%
\[
K_{1}\overset{n}{\underset{\infty}{\sim}}K_{2}\text{ \ ,}%
\]
provided there exists a non-negative integer $\ell^{\prime}$ such that
\[%
\raisebox{-0.0303in}{\includegraphics[
height=0.1332in,
width=0.1193in
]%
{red-refinement.ps}%
}%
^{\ell^{\prime}}K_{1}\underset{\ell+\ell^{\prime}}{\overset{n}{\thicksim}}%
\raisebox{-0.0303in}{\includegraphics[
height=0.1332in,
width=0.1193in
]%
{red-refinement.ps}%
}%
^{\ell^{\prime}}K_{2}%
\]

\end{definition}

\bigskip

It immediately follows that:

\bigskip

\begin{proposition}
$K_{1}\overset{\infty}{\underset{\ell}{\sim}}K_{2}\Longleftrightarrow
K_{1}\overset{n}{\underset{\infty}{\sim}}K_{2}\Longleftrightarrow K_{1}\sim
K_{2}$ .
\end{proposition}

\bigskip

In like manner, we can define \textbf{inextensible knot }$\left(
\ell,n\right)  $\textbf{-type} and \textbf{inextensible knot type}.\bigskip

\begin{definition}
Two lattice knots $K_{1}$ and $K_{2}$ of order $\left(  \ell,n\right)  $ are
said to be of the \textbf{same inextensible knot }$(\ell,n)$-\textbf{type},
written%
\[
K_{1}\underset{\ell}{\overset{n}{\approx}}K_{2}\text{ \ ,}%
\]
provided there exists an element $g$ of the bounded inextensible ambient group
$\widetilde{\Lambda}_{\ell,n}$ that transforms one into the other. \ \ They
are said to be of the \textbf{same inextensible knot }$\left(  \ell
,\infty\right)  $-\textbf{type}, written%
\[
K_{1}\overset{\infty}{\underset{\ell}{\approx}}K_{2}\text{ \ ,}%
\]
provided there exists a non-negative integer $n^{\prime}$ such that
\[
\iota^{n^{\prime}}K_{1}\underset{\ell}{\overset{n+n^{\prime}}{\approx}}%
\iota^{n^{\prime}}K_{2}%
\]
They are of the \textbf{same inextensible knot }$\left(  \infty,n\right)
$-\textbf{type}, written%
\[
K_{1}\overset{n}{\underset{\infty}{\approx}}K_{2}\text{ \ ,}%
\]
provided there exists a non-negative integer $\ell^{\prime}$ such that
\[%
\raisebox{-0.0303in}{\includegraphics[
height=0.1332in,
width=0.1193in
]%
{red-refinement.ps}%
}%
^{\ell^{\prime}}K_{1}\underset{\ell+\ell^{\prime}}{\overset{n}{\approx}}%
\raisebox{-0.0303in}{\includegraphics[
height=0.1332in,
width=0.1193in
]%
{red-refinement.ps}%
}%
^{\ell^{\prime}}K_{2}%
\]
\bigskip
\end{definition}

\begin{proposition}
$K_{1}\overset{n}{\underset{\infty}{\approx}}K_{2}\Longleftrightarrow
K_{1}\approx K_{2}$ . \ But $K_{1}\overset{\infty}{\underset{\ell}{\approx}%
}K_{2}\nRightarrow K_{1}\approx K_{2}$ .
\end{proposition}

\bigskip

\bigskip

\section{Part 2. \ Quantum Knots}

\bigskip

We are finally ready to define what is meant by a quantum knot system and a
quantum knot.

\bigskip

\section{The Definition of a Quantum Knot}

\bigskip

We will now create a Hilbert space by associating a qubit with each edge of
the cell complex $\mathcal{C}_{\ell,n}$. \ 

\bigskip

Let `$<$' denote the \textbf{lexicographic (lex) ordering} of the lattice
points $\mathcal{L}_{\ell,n}$ induced by the standard linear ordering of the
rationals. \ Extend this ordering in the obvious way to a lex ordering of
$\mathcal{L}_{\ell,n}\times\left\{  1,2,3\right\}  $, also denoted by `$<$'.
\ Finally, define a linear ordering, again denoted by `$<$' on the \textbf{set
}$\mathcal{E}_{\ell,n}$\textbf{ of edges of the cell complex} $\mathcal{C}%
_{\ell,n}$ given by
\[
E_{p}\left(  a\right)  <E_{p^{\prime}}(a^{\prime})\text{ if and only if
}(a,p)<\left(  a^{\prime},p^{\prime}\right)  \text{ .}%
\]

\bigskip

Let $\mathcal{H}$ be the two dimensional Hilbert space (called the
\textbf{edge state space}) with orthonormal basis%
\[
\left\{  \left\vert 0\right\rangle =\left\vert
\raisebox{-0.0303in}{\includegraphics[
height=0.1781in,
width=0.6287in
]%
{dotted-edge.ps}%
}%
\right\rangle ,\ \left\vert 1\right\rangle =\left\vert
\raisebox{-0.0303in}{\includegraphics[
height=0.1781in,
width=0.6287in
]%
{red-edge.ps}%
}%
\right\rangle \right\}  \text{ .}%
\]
The \textbf{Hilbert space }$\mathcal{G}_{\ell,n}$\textbf{ of lattice graphs}
\textbf{of order }$\left(  \ell,n\right)  $ is defined as the tensor product%
\[
\mathcal{G}_{\ell,n}=%
{\displaystyle\bigotimes\limits_{E\in\mathcal{E}_{\ell,n}}}
\mathcal{H}\text{ ,}%
\]
where the tensor product is taken with respect to the above defined linear
ordering `$<$'. \ Thus, as orthonormal basis for the Hilbert space
$\mathcal{G}_{\ell,n}$, we have%
\[
\left\{  \ \underset{E\in\mathcal{E}_{\ell,n}}{\otimes}\left\vert c\left(
E\right)  \right\rangle :c\in Map\left(  \mathcal{E}_{\ell,n}\ ,\ \left\{
0,1\right\}  \right)  \ \right\}  \text{ ,}%
\]
where\textbf{ }$Map\left(  \mathcal{E}_{\ell,n}\ ,\ \left\{  0,1\right\}
\right)  $ is the set of all maps $c:\mathcal{E}_{\ell,n}\longrightarrow
\left\{  0,1\right\}  $ from the set $\mathcal{E}_{\ell,n}$ of edges to the
set $\left\{  0,1\right\}  $.

\bigskip

We identify in the obvious way each basis element
\[
\underset{E\in\mathcal{E}_{\ell,n}}{\otimes}\left\vert c\left(  E\right)
\right\rangle
\]
with a corresponding lattice graph $G$. \ Under this identification, the space
$\mathcal{G}_{\ell,n}$ becomes the Hilbert space with orthonormal basis%
\[
\left\{  \left\vert G\right\rangle :G\text{ a lattice graph in }%
\mathcal{L}_{\ell,n}\right\}  \text{ ,}%
\]
called the \textbf{standard basis}. \ Finally, the \textbf{Hilbert space
}$\mathcal{K}^{(\ell,n)}$\textbf{ of lattice knots} \textbf{of order }$\left(
\ell,n\right)  $ is defined as the sub-Hilbert space of $\mathcal{G}_{\ell,n}$
with orthonormal basis
\[
\left\{  \left\vert K\right\rangle :K\in\mathbb{K}^{(\ell,n)}\right\}
\]

\bigskip

Our next step is to identify each element $g$ of the ambient group
$\Lambda_{\ell,n}$ with the corresponding linear transformation defined by
\[%
\begin{array}
[c]{ccc}%
\mathcal{K}^{(\ell,n)} & \overset{g}{\longrightarrow} & \mathcal{K}^{(\ell
,n)}\\
\left\vert K\right\rangle  & \longmapsto & \left\vert gK\right\rangle
\end{array}
\]
This is a unitary transformation, since each element $g$ simply permutes the
basis elements of $\mathcal{K}^{(\ell,n)}$. \ In this way, the ambient group
$\Lambda_{\ell,n}$ is identified with the discrete unitary subgroup (also
denoted by $\Lambda_{\ell,n}$) of the group $U\left(  \mathcal{K}^{(\ell
,n)}\right)  $, where $U\left(  \mathcal{K}^{(\ell,n)}\right)  $ denotes the
\textit{group of all unitary transformations on the Hilbert space}
$\mathcal{K}^{(\ell,n)}$. \ We also call the unitary group $\Lambda_{\ell,n}$
the \textbf{(lattice) ambient group of order }$\left(  \ell,n\right)
$\textbf{.%
\begin{center}
\includegraphics[
height=2.3722in,
width=3.5276in
]%
{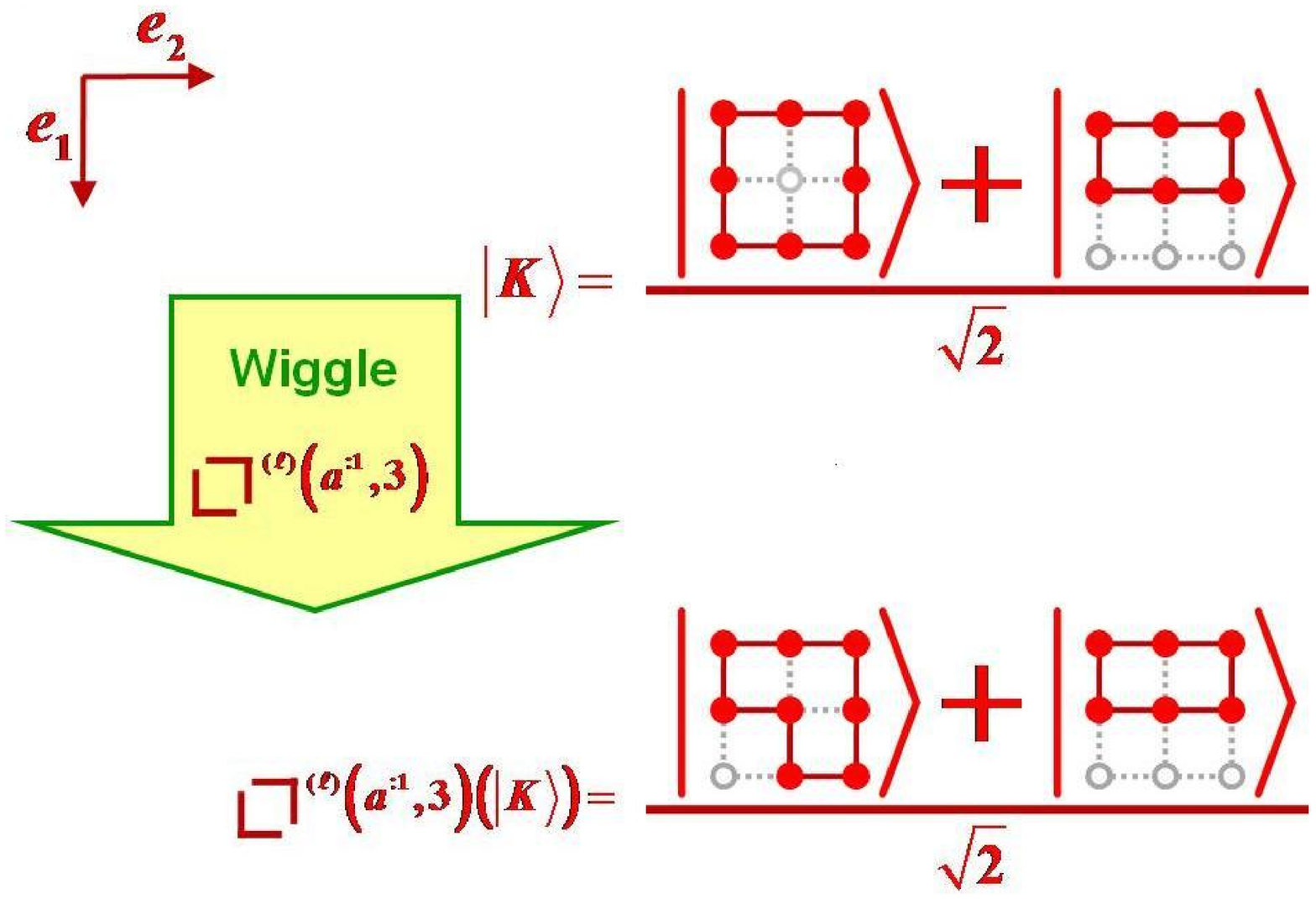}%
\\
\textbf{An example of the action of the ambient group }$\Lambda_{\ell,n}$
\textbf{on a quantum knot in }$\mathcal{K}^{(\ell,n)}$\textbf{.}%
\end{center}
}

\bigskip

We leave, as an exercise for the reader, the definition of the
\textbf{(lattice) inextensible ambient group} $\widetilde{\Lambda}_{\ell,n}%
$\textbf{ of order }$\left(  \ell,n\right)  $, which is defined in like manner.

\bigskip

\bigskip

Finally, everything comes together with the following definition.

\bigskip

\begin{definition}
Let $\ell\geq0$ and $n\geq1$ be a integers. \ A \textbf{quantum knot system}
$Q\left(  \mathcal{K}^{(\ell,n)},\Lambda_{\ell,n}\right)  $\textbf{ of order
}$\left(  \ell,n\right)  $ is a quantum system with the Hilbert space
$\mathcal{K}^{(\ell,n)}$ of $\left(  \ell,n\right)  $-th order lattice knots
as its state space, having the ambient group $\Lambda_{\ell,n}$ as an
accessible unitary control group. \ The states of the quantum system $Q\left(
\mathcal{K}^{(\ell,n)},\Lambda_{\ell,n}\right)  $ are called \textbf{quantum
knots of order} $\left(  \ell,n\right)  $, and the elements of the ambient
group $\Lambda_{\ell,n}$ are called \textbf{unitary knot moves}. \ Moreover,
the quantum knot system $Q\left(  \mathcal{K}^{(\ell,n)},\Lambda_{\ell
,n}\right)  $ is a subsystem of the quantum knot system $Q\left(
\mathcal{K}^{(\ell,n+1)},\Lambda_{\ell,n+1}\right)  $. \ Thus, the quantum
knot systems $Q\left(  \mathcal{K}^{(\ell,n)},\Lambda_{\ell,n}\right)  $
collectively become a \textbf{nested sequence of quantum knot systems}
\[
Q\left(  \mathcal{K}^{(\ell)},\Lambda_{\ell}\right)  =Q\left(  \mathcal{K}%
^{(\ell,1)},\Lambda_{\ell,1}\right)  \longrightarrow\cdots\longrightarrow
Q\left(  \mathcal{K}^{(\ell,n)},\Lambda_{\ell,n}\right)  \longrightarrow
Q\left(  \mathcal{K}^{(\ell,n+1)},\Lambda_{\ell,n}\right)  \longrightarrow
\cdots
\]
which we will denote simply by $Q\left(  \mathcal{K}^{(\ell)},\Lambda_{\ell
}\right)  $. \ We leave, as an exercise for the reader, the definition of the
\textbf{inextensible} \textbf{quantum knot system} $Q\left(  \mathcal{K}%
^{(\ell,n)},\widetilde{\Lambda}_{\ell,n}\right)  $\textbf{ of order }$\left(
\ell,n\right)  $, which is defined in like manner.
\end{definition}

\bigskip

\begin{remark}
The nested quantum knot systems $Q\left(  \mathcal{K}^{(\ell)},\Lambda_{\ell
}\right)  $ and $Q\left(  \mathcal{K}^{(\ell)},\widetilde{\Lambda}_{\ell
}\right)  $ are probably not physically realizable systems. \ However, each
quantum knot system $Q\left(  \mathcal{K}^{(\ell,n)},\Lambda_{\ell,n}\right)
$ of order $\left(  \ell,n\right)  $ (as well as each inextensible quantum
knot system $Q\left(  \mathcal{K}^{(\ell,n)},\widetilde{\Lambda}_{\ell
,n}\right)  $ of order $\left(  \ell,n\right)  $ ) is physically realizable.
\ By this we mean that such quantum knot systems are physically realizable in
the same sense as a quantum system implementing Shor's quantum factoring
algorithm is physically realizable.\footnote{It should be mentioned that,
although the quantum knot systems $Q\left(  \mathcal{K}^{(\ell,n)}%
,\Lambda_{\ell,n}\right)  $ and $Q\left(  \mathcal{K}^{(\ell,n)}%
,\widetilde{\Lambda}_{\ell,n}\right)  $ are physically realizable, they may
not be implementable with today's existing technology.}
\end{remark}

\bigskip

\section{Quantum knot type}

\bigskip

When are two quantum knots the same?

\bigskip

\begin{definition}
Let $\left\vert \psi_{1}\right\rangle $ and $\left\vert \psi_{2}\right\rangle
$ be two quantum knots of a quantum system $Q\left(  \mathcal{K}^{(\ell
,n)},\Lambda_{\ell,n}\right)  $ (of an inextensible $Q\left(  \mathcal{K}%
^{(\ell,n)},\widetilde{\Lambda}_{\ell,n}\right)  $). \ Then $\left\vert
\psi_{1}\right\rangle $ and $\left\vert \psi_{2}\right\rangle $ are said to be
of the \textbf{same quantum knot }$\left(  \ell,n\right)  $\textbf{-type} (of
the \textbf{same inextensible quantum knot }$\left(  \ell,n\right)
$\textbf{-type}), written%
\[
\left\vert \psi_{1}\right\rangle \underset{\ell}{\overset{n}{\thicksim}%
}\left\vert \psi_{2}\right\rangle \text{ ,}\left(  \left\vert \psi
_{1}\right\rangle \underset{\ell}{\overset{n}{\approx}}\left\vert \psi
_{2}\right\rangle \text{ ,}\right)
\]
provided there exists a unitary transformation $g$ in the ambient group
$\Lambda_{\ell,n}$ (in the inextensible ambient group $\widetilde{\Lambda
}_{\ell,n}$\ ) which transforms $\left\vert \psi_{1}\right\rangle $ into
$\left\vert \psi_{2}\right\rangle $, i.e., such that%
\[
g\left\vert \psi_{1}\right\rangle =\left\vert \psi_{2}\right\rangle \text{ .}%
\]
They are said to be of the \textbf{same quantum knot type (}of the
\textbf{same inextensible quantum knot type}), written
\[
\left\vert \psi_{1}\right\rangle \thicksim\left\vert \psi_{2}\right\rangle
\text{ , }\left(  \left\vert \psi_{1}\right\rangle \overset{\mathstrut
}{\approx}\left\vert \psi_{2}\right\rangle \text{ ,}\right)
\]
provided that for some non-negative integer $\ell^{\prime}$,
\[%
\raisebox{-0.0303in}{\includegraphics[
height=0.1332in,
width=0.1193in
]%
{red-refinement.ps}%
}%
^{\ell^{\prime}}\left\vert \psi_{1}\right\rangle \underset{\ell+\ell^{\prime}%
}{\overset{n}{\thicksim}}%
\raisebox{-0.0303in}{\includegraphics[
height=0.1332in,
width=0.1193in
]%
{red-refinement.ps}%
}%
^{\ell^{\prime}}\left\vert \psi_{2}\right\rangle \text{ , }\left(
\raisebox{-0.0303in}{\includegraphics[
height=0.1332in,
width=0.1193in
]%
{red-refinement.ps}%
}%
^{\ell^{\prime}}\left\vert \psi_{1}\right\rangle \underset{\ell+\ell^{\prime}%
}{\overset{n}{\approx}}%
\raisebox{-0.0303in}{\includegraphics[
height=0.1332in,
width=0.1193in
]%
{red-refinement.ps}%
}%
^{\ell^{\prime}}\left\vert \psi_{2}\right\rangle \text{ ,}\right)
\]
where $%
\raisebox{-0.0303in}{\includegraphics[
height=0.1332in,
width=0.1193in
]%
{red-refinement.ps}%
}%
:\mathcal{K}^{(\ell,m)}\longrightarrow\mathcal{K}^{(\ell+1,m)}$ is the Hilbert
space monomorphism induced by the previously defined refinement injection $%
\raisebox{-0.0303in}{\includegraphics[
height=0.1332in,
width=0.1193in
]%
{red-refinement.ps}%
}%
:\mathbb{K}^{(\ell,m)}\longrightarrow\mathbb{K}^{(\ell+1,m)}$.
\end{definition}

\bigskip

\begin{proposition}
Let $\left\vert \psi_{1}\right\rangle $ and $\left\vert \psi_{2}\right\rangle
$ be two quantum knots of a quantum system $Q\left(  \mathcal{K}^{(\ell
,n)},\Lambda_{\ell,n}\right)  $. \ Then $\left\vert \psi_{1}\right\rangle $
and $\left\vert \psi_{2}\right\rangle $ are of the \textbf{same quantum knot
type}%
\[
\left\vert \psi_{1}\right\rangle \thicksim\left\vert \psi_{2}\right\rangle
\]
if and only if there exists a non-negative integer $n^{\prime}$ such that
\[
\iota^{n^{\prime}}\left\vert \psi_{1}\right\rangle \underset{\ell}%
{\overset{n+n^{\prime}}{\thicksim}}\iota^{n^{\prime}}\left\vert \psi
_{2}\right\rangle \text{ ,}%
\]
where $\iota:\mathcal{K}^{(\ell,m)}\longrightarrow\mathcal{K}^{(\ell,m+1)}$ is
the monomorphism induced by the previously defined injection $\iota
:\mathbb{K}^{(\ell,m)}\longrightarrow\mathbb{K}^{(\ell,m+1)}$. \ The analogous
statement for inextensible quantum knot type is false.
\end{proposition}

\bigskip

Thus, the two quantum knots found in the last example of the previous section
are of the same quantum knot type. \ Surprisingly, the following two quantum
knots $\left\vert \psi_{1}\right\rangle $ and $\left\vert \psi_{2}%
\right\rangle $ are neither of the same quantum knot $\left(  1,1\right)
$-type nor quantum knot type:
\[
\left\vert \psi_{1}\right\rangle =\left\vert
\raisebox{-0.1609in}{\includegraphics[
height=0.4186in,
width=0.4186in
]%
{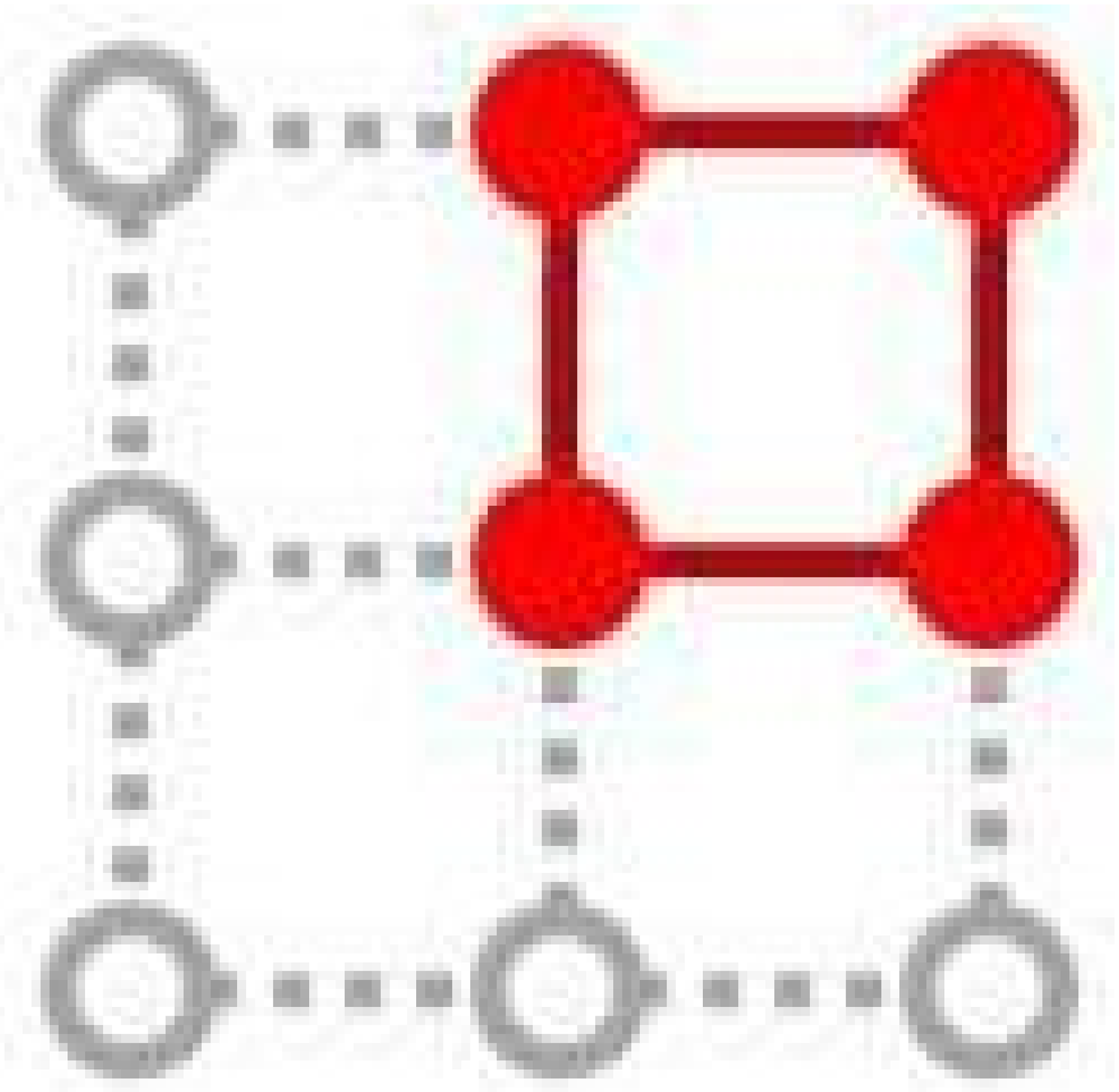}%
}%
\right\rangle \text{ \ \ and \ \ }\left\vert \psi_{2}\right\rangle =\frac
{1}{\sqrt{2}}\left(  \ \overset{\mathstrut}{\left\vert
\raisebox{-0.1609in}{\includegraphics[
height=0.4186in,
width=0.4186in
]%
{qknot-type2b.ps}%
}%
\right\rangle }+\left\vert
\raisebox{-0.1609in}{\includegraphics[
height=0.4186in,
width=0.4186in
]%
{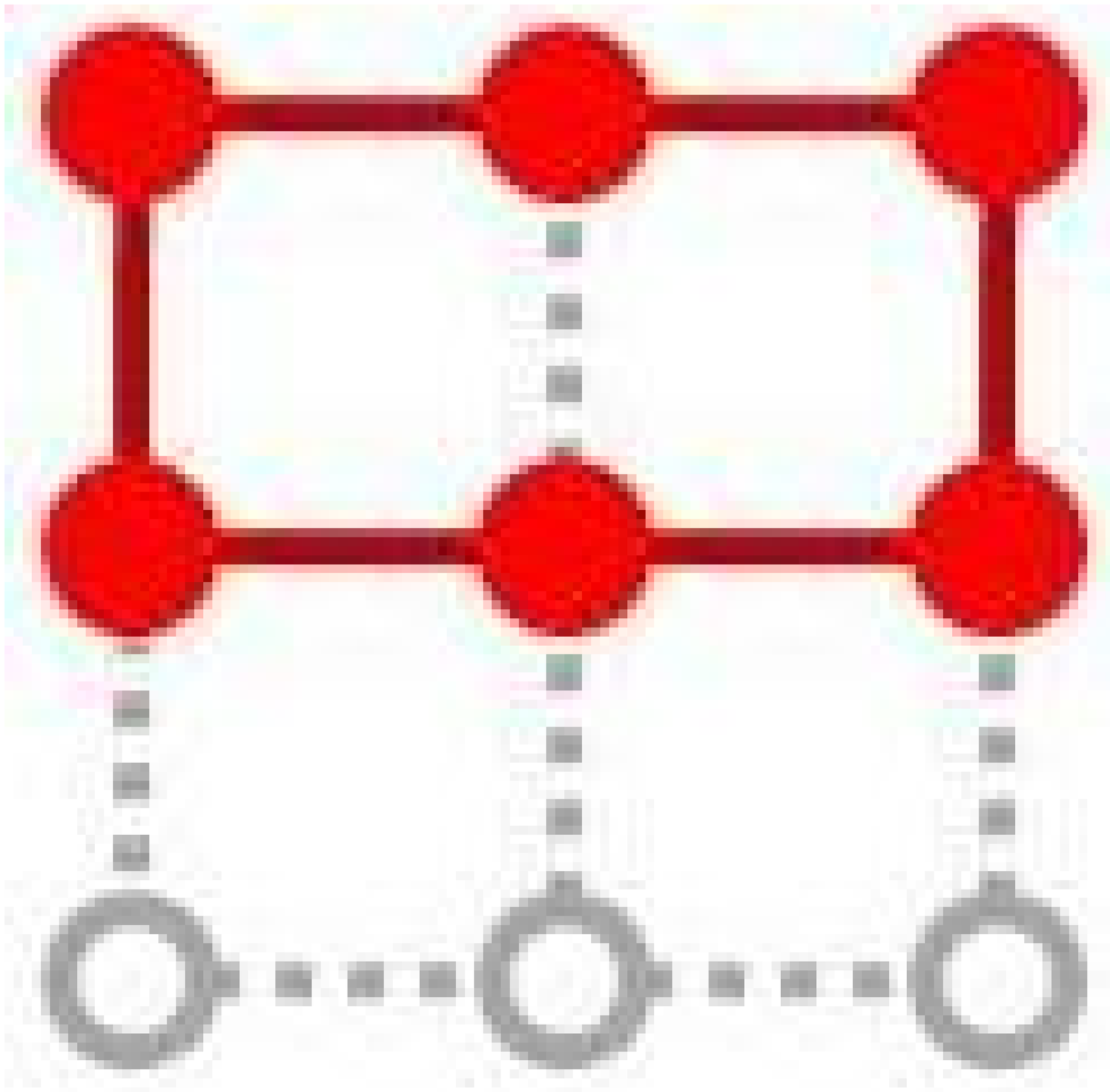}%
}%
\right\rangle \ \right)
\]

\bigskip

This follows from the fact that the ambient group $\Lambda_{\ell,n}$ is
generated by a finite set of involutions. \ 

\bigskip

\section{Hamiltonians of the generators of the ambient group $\Lambda$.}

\bigskip

\bigskip

In this section, we show how to find a Hamiltonians associated with the
generators of the ambient group $\Lambda_{\ell,n}$, i.e., Hamiltonians
associated with the wiggle, wag, and tug moves\footnote{Please keep in mind
that the Hamiltonian construction found below is far from unique. \ There are
many other ways of constructing a Hamiltonian corresponding to the generators
of the ambient group $\Lambda_{\ell,n}$.}. \ 

\bigskip

Let $g$ be an arbitrary wiggle, wag, or tug move in the ambient group
$\Lambda_{\ell,n}$. \ From proposition 1, we know that $g$, as a permutation,
is the product of disjoint transpositions of knot $(\ell,n)$-th order lattice
knots , i.e., of the form%
\[
g=\left(  K_{\alpha_{1}},K_{\beta_{1}}\right)  \left(  K_{\alpha_{2}}%
,K_{\beta_{2}}\right)  \cdots\left(  K_{\alpha_{s}},K_{\beta_{r}}\right)
\]
Without loss of generality, we can assume that $\alpha_{j}<\beta_{j}$ for
$1\leq j\leq r$, and that $\alpha_{j}<\alpha_{j+1}$ for $1\leq j<r$, where
`$<$' is the lexicographic (lex) ordering on\ the set of $\left(
\ell,n\right)  $-lattice graphs $\mathcal{G}^{(\ell,n)}$ induced by the
previously defined linear ordering `$<$' of the edges in the lattice
$\mathcal{L}_{\ell,n}$. \ For each permutation $\eta$ of $\mathcal{K}%
^{(\ell,n)}$, let `$<_{\eta}$' denote the new linear ordering created by the
application of the permutation $\eta$. \ 

\bigskip

Choose a permutation $\eta$ such that
\[
K_{\alpha_{1}}<_{\eta}K_{\beta_{1}}<_{\eta}K_{\alpha_{2}}<_{\eta}K_{\beta_{2}%
}<_{\eta}\cdots<_{\eta}K_{\alpha_{r}}<_{\eta}K_{\beta_{r}}%
\]
with $K_{\beta_{r}}$ $<_{\eta}$(i.e., $\eta$-less than) all other lattice
knots of order $(\ell,n)$, and let $\sigma_{0}$ and $\sigma_{1}$ denote
respectively the identity matrix and the first Pauli spin matrix given below%
\[
\sigma_{0}=\left(
\begin{array}
[c]{cc}%
1 & 0\\
0 & 1
\end{array}
\right)  \text{ \ \ and \ \ }\sigma_{1}=\left(
\begin{array}
[c]{cc}%
0 & 1\\
1 & 0
\end{array}
\right)  \text{ .}%
\]
Let $d\left(  \ell,n\right)  $ denote the dimension of the Hilbert space
$\mathcal{K}^{(\ell,n)}$. \ Then in the $\eta$-reordered basis of the Hilbert
space $\mathcal{K}^{(\ell,n)}$, the element $g$, as a unitary transformation,
is of the form%
\[
\eta^{-1}g\eta=\left(
\begin{array}
[c]{ccccc}%
\sigma_{1} & O & \ldots & O & O\\
O & \sigma_{1} & \ldots & O & O\\
\vdots & \vdots & \ddots & \vdots & \vdots\\
O & O & \ldots & \sigma_{1} & O\\
O & O & \ldots & O & I_{d\left(  \ell,n\right)  -2r}%
\end{array}
\right)  =\left(  I_{r}\otimes\sigma_{1}\right)  \oplus I_{d\left(
\ell,n\right)  -2r}\text{ ,}%
\]
where `$O$' denotes an all zero matrix of appropriate size, where $I_{d\left(
\ell,n\right)  -2r}$ denotes the $\left(  d\left(  \ell,n\right)  -2r\right)
\times\left(  d\left(  \ell,n\right)  -2r\right)  $ identity matrix, and where
`$\oplus$' denotes the direct sum of matrices, i.e., $A\oplus B=\left(
\begin{array}
[c]{cc}%
A & O\\
O & B
\end{array}
\right)  $.\bigskip

The natural log of $\sigma_{1}$ is%
\[
\ln\sigma_{1}=\frac{i\pi}{2}\left(  2s+1\right)  \left(  \sigma_{0}-\sigma
_{1}\right)
\]
where $s$ denotes an arbitrary integer. \ Hence, the natural log, $\ln\left(
\eta^{-1}g\eta\right)  $, of the unitary transformation $\eta^{-1}g\eta$ is

\bigskip

$\hspace{-0.85in}\frac{i\pi}{2}\left(
\begin{array}
[c]{cc}%
\left(
\begin{array}
[c]{cccc}%
\left(  2_{s_{1}}+1\right)  \left(  \sigma_{0}-\sigma_{1}\right)  & O & \ldots
& O\\
O & \left(  2s_{2}+1\right)  \left(  \sigma_{0}-\sigma_{1}\right)  & \ldots &
O\\
\vdots & \vdots & \ddots & \vdots\\
O & O & \ldots & \left(  2s_{r}+1\right)  \left(  \sigma_{0}-\sigma
_{1}\right)
\end{array}
\right)  & O\\
O & O_{\left(  d\left(  \ell,n\right)  -2r\right)  \times\left(  d\left(
\ell,n\right)  -2r\right)  }%
\end{array}
\right)  $

\bigskip

\noindent where $s_{1},s_{2},\ldots,s_{r}$ are arbitrary
integers.\footnote{Let $U$ be an arbitrary finite $r\times r$ unitary matrix,
and let $W$ be a unitary matrix that diagonalizes $U$, i.e., a unitary matrix
$W$ such that $WUW^{-1}=\Delta\left(  \ \lambda(1),\lambda(2)\ldots
,\lambda(r)\ \right)  $. \ Then the natural log of $A$ is $\ln A=W^{-1}%
\Delta\left(  \ \ln\lambda(1),\ln\lambda(2)\ldots,\ln\lambda(r)\ \right)  W$.}

\bigskip

Since we are interested only in the simplest Hamiltonian, we choose the
\textbf{principal branch} $\ln_{P}$ of the natural log, i.e., the branch for
which $s_{1}=s_{2}=\cdots s_{r}=0$, and obtain for our Hamiltonian%
\[
H_{g}=-i\eta\left[  \ln_{P}\left(  \eta^{-1}g\eta\right)  \right]  \eta
^{-1}=\frac{\pi}{2}\eta\left(
\begin{array}
[c]{cc}%
I_{r}\otimes\left(  \sigma_{0}-\sigma_{1}\right)  & O\\
O & O_{\left(  d\left(  \ell,n\right)  -2r\right)  \times\left(  d\left(
\ell,n\right)  -2r\right)  }%
\end{array}
\right)  \eta^{-1}%
\]

\bigskip

Thus, if the initial quantum knot is%
\[
\left\vert
\raisebox{-0.1609in}{\includegraphics[
height=0.4186in,
width=0.4186in
]%
{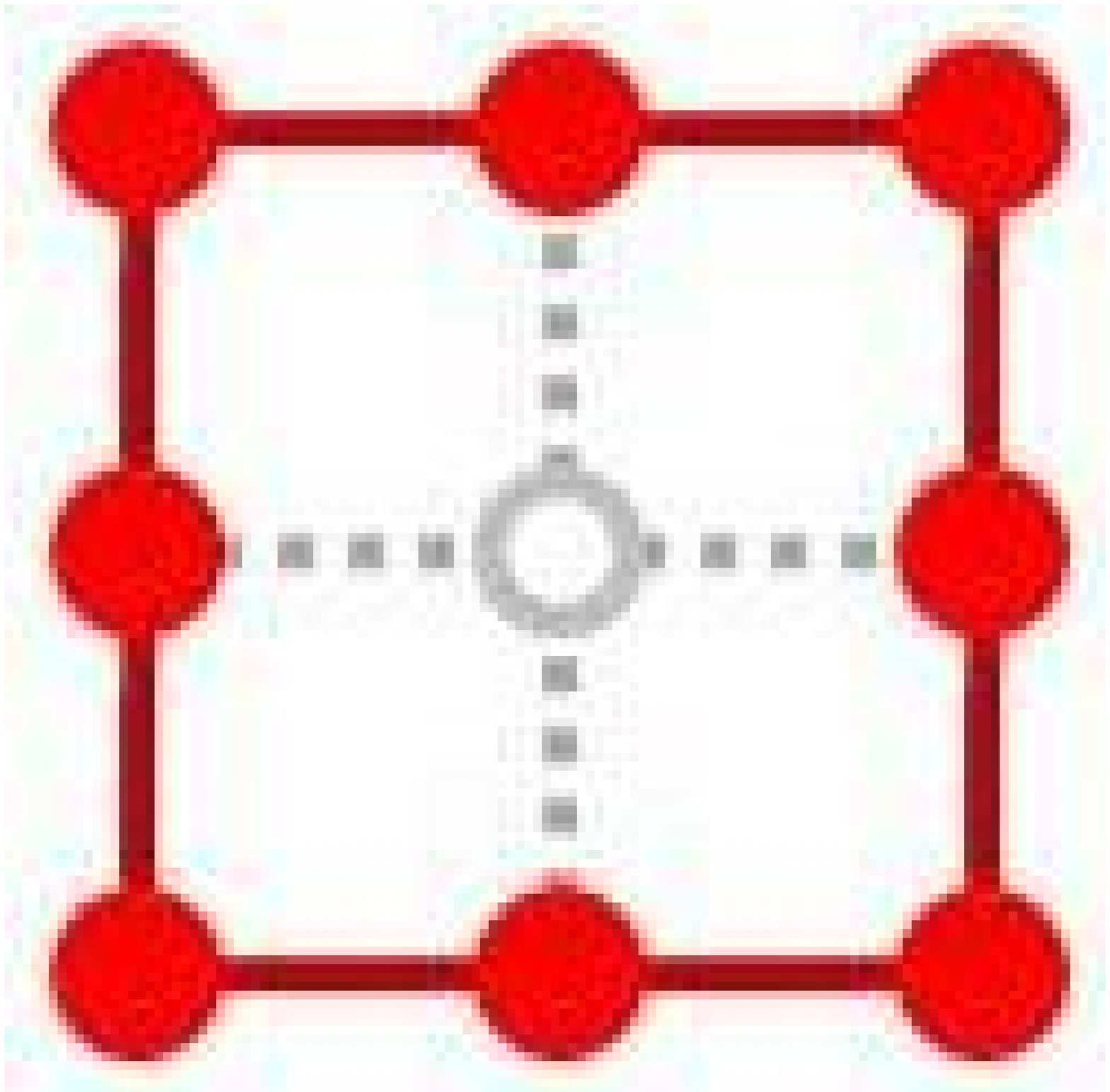}%
}%
\right\rangle \text{ ,}%
\]
and if we use the Hamiltonian $H_{g}$ for the wiggle move $%
\raisebox{-0.0406in}{\includegraphics[
height=0.1436in,
width=0.1436in
]%
{icon20.ps}%
}%
^{(0)}\left(  a^{:1},3\right)  $, then the solution to Schroedinger's equation
is
\[
e^{i\pi t/\left(  2\hslash\right)  }\left(  \cos\left(  \frac{\pi t}{2\hslash
}\right)  \left\vert
\raisebox{-0.1609in}{\includegraphics[
height=0.4186in,
width=0.4186in
]%
{qknot3a.ps}%
}%
\right\rangle -i\sin\left(  \frac{\pi t}{2\hslash}\right)  \left\vert
\raisebox{-0.1609in}{\includegraphics[
height=0.4186in,
width=0.4186in
]%
{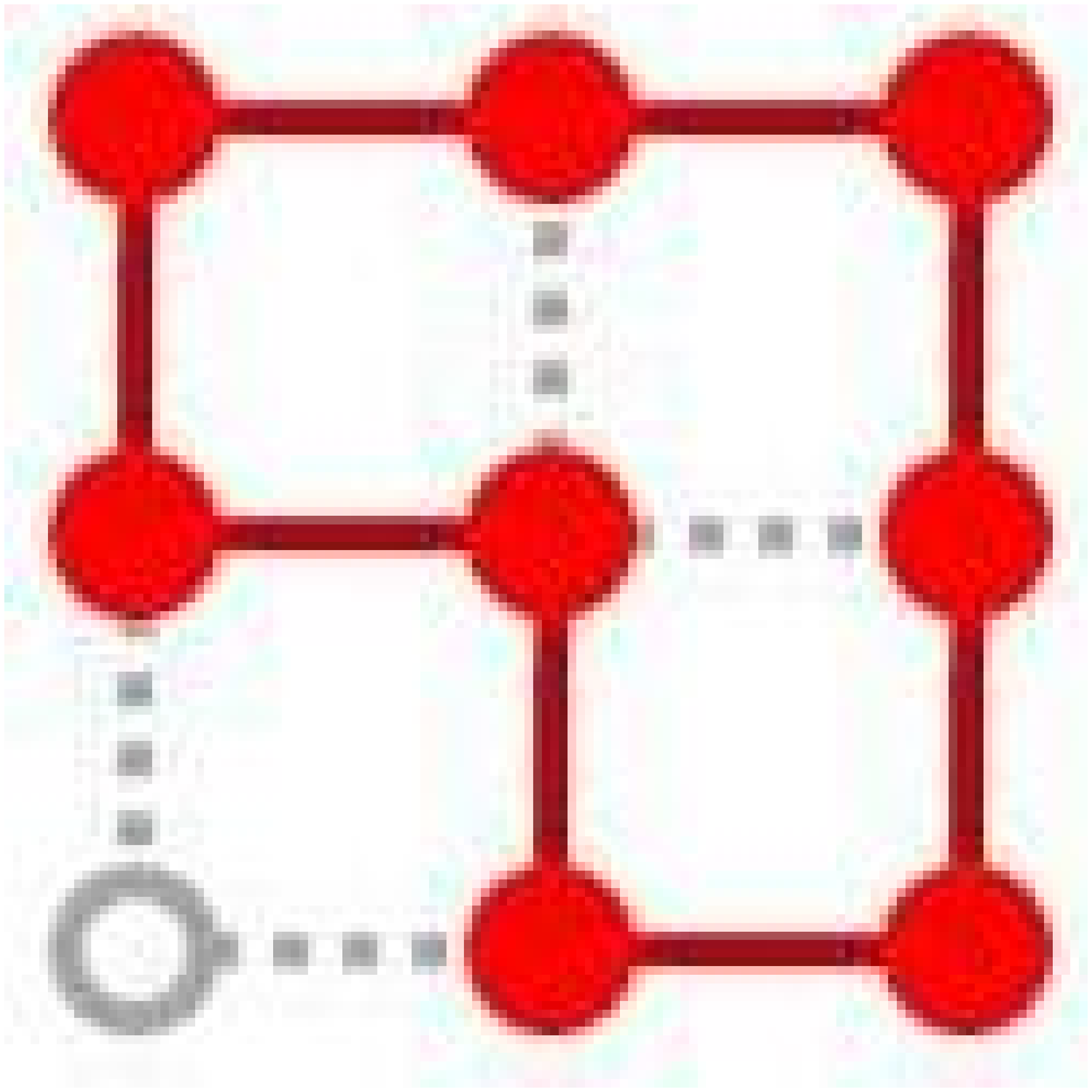}%
}%
\right\rangle \right)
\]
where $t$ denotes time, and where $\hslash$ denotes Planck's constant divided
by $2\pi$.

\bigskip

\section{Quantum observables as invariants of quantum knots}

\bigskip

\bigskip

We now consider the following question:

\bigskip

\noindent\textbf{Question.} \ \textit{What is a quantum knot invariant? \ How
do we define it?}

\bigskip

The objective of the first half of this section is to give a discursive
argument that justifies a definition which will be found to be equivalent to
the following:

\bigskip

\textit{A }\textbf{quantum knot }$\left(  \ell,n\right)  $\textbf{-invariant}%
\textit{ for a quantum system }$Q\left(  \mathcal{K}^{(\ell,n)},\Lambda
_{\ell,n}\right)  $\textit{ is an observable }$\Omega$\textit{ on the Hilbert
space }$\mathcal{K}^{(\ell,n)}$\textit{ of quantum knots which is invariant
under the action of the ambient group }$\Lambda_{\ell,n}$\textit{, i.e., such
that }$U\Omega U^{-1}=\Omega$\textit{ for all }$U$ in $\Lambda_{\ell,n}%
$\textit{. \ }

\bigskip

\noindent\textbf{Caveat:} \ \textit{We emphasize to the reader that the above
definition of quantum knot invariants is not the one currently used in quantum
topology. Quantum topology uses analogies with quantum mechanics to create
significant mathematical structures that do not necessarily correspond
directly to quantum mechanical observables. The invariants of quantum topology
have been investigated for their relevance to quantum computing and they can
be regarded, in our context, as possible secondary calculations made on the
basis of an observable. Here we are concerned with observables that are
themselves topological invariants.}

\bigskip

To justify our new use of the term `quantum knot invariant,' we will use the
following \textsc{yardstick}:

\bigskip

\noindent\textsc{Yardstick:} \textit{Quantum knot invariants are to be
physically meaningful invariants of quantum knot type. \ By "physically
meaningful," we mean that the quantum knot invariants can be directly obtained
from experimental data produced by an implementable physical
experiment.\footnote{Once again we remind the reader that, although the
quantum knot system $Q\left(  \mathcal{K}^{(\ell,n)},\Lambda_{\ell,n}\right)
$ is physically implementable, it may or may not be implentable within today's
existing technology.}}

\bigskip

Let $Q\left(  \mathcal{K}^{(\ell,n)},\Lambda_{\ell,n}\right)  $ be a quantum
knot system, where $\mathcal{K}^{(\ell,n)}$ is the Hilbert space of quantum
knots, and where $\Lambda_{\ell,n}$ is the underlying ambient group on
$\mathcal{K}^{(\ell,n)}$. \ Moreover, let $\mathcal{P}^{(\ell,n)}$ denote some
yet-to-be-chosen mathematical domain. \ By an $\left(  \ell,n\right)
$\textbf{-invariant} $I^{(\ell,n)}$ of quantum knots, we mean a map%
\[
I^{(\ell,n)}:\mathcal{K}^{(\ell,n)}\longrightarrow\mathcal{P}^{(\ell,n)}\text{
,}%
\]
such that, when two quantum knots $\left\vert \psi_{1}\right\rangle $ and
$\left\vert \psi_{2}\right\rangle $ are of the same knot $\left(
\ell,n\right)  $-type, i.e., when%
\[
\left\vert \psi_{1}\right\rangle \underset{\ell}{\overset{n}{\sim}}\left\vert
\psi_{2}\right\rangle \text{ ,}%
\]
then their respective invariants must be equal, i.e.,
\[
I^{(\ell,n)}\left(  \left\vert \psi_{1}\right\rangle \right)  =I^{(\ell
,n)}\left(  \left\vert \psi_{2}\right\rangle \right)
\]
In other words, $I^{(\ell,n)}:\mathcal{K}^{(\ell,n)}\longrightarrow
\mathcal{P}^{(\ell,n)}$ is a map which is invariant under the action of the
ambient group $\Lambda_{\ell,n}$, i.e.,%
\[
I^{(\ell,n)}\left(  \left\vert \psi\right\rangle \right)  =I^{(\ell,n)}\left(
g\left\vert \psi\right\rangle \right)
\]
for all elements $g$ in $\Lambda_{\ell,n}$.

\bigskip

\noindent\textbf{Question:} \textit{But which such invariants are physically
meaningful?}\ 

\bigskip

We begin to try to answer this question by noting that the only way to extract
information from a quantum system is through quantum measurement. \ Thus, if
we wish to extract information about quantum knot type from a quantum knot
system $Q\left(  \mathcal{K}^{(\ell,n)},\Lambda_{\ell,n}\right)  $\textit{,
}we of necessity must make a measurement with respect to some observable.
\ But what kind of observable? \ 

\bigskip

With this in mind, we will now describe quantum measurement from a different,
but nonetheless equivalent perspective, than that which is usually given in
standard texts on quantum mechanics.\footnote{In this paper, we will focus
only on von Neumann quantum measurement. \ We will discuss more general POVM
approach to quantum knot invariants in a later paper.} \ For knot theorists
who might not be familiar with standard quantum measurement, we have included
in the figure below a brief summary of quantum measurement.\footnote{For
readers unfamiliar with quantum measurement, there are many references, for
example, \cite{Helstrom1}, \cite{Lomonaco3} \cite{Nielsen1}, \cite{Peres1},
\cite{Sakuri1}, and \cite{Shankar1}.} \ 

\bigskip%

\begin{center}
\fbox{\includegraphics[
height=3.0277in,
width=4.0274in
]%
{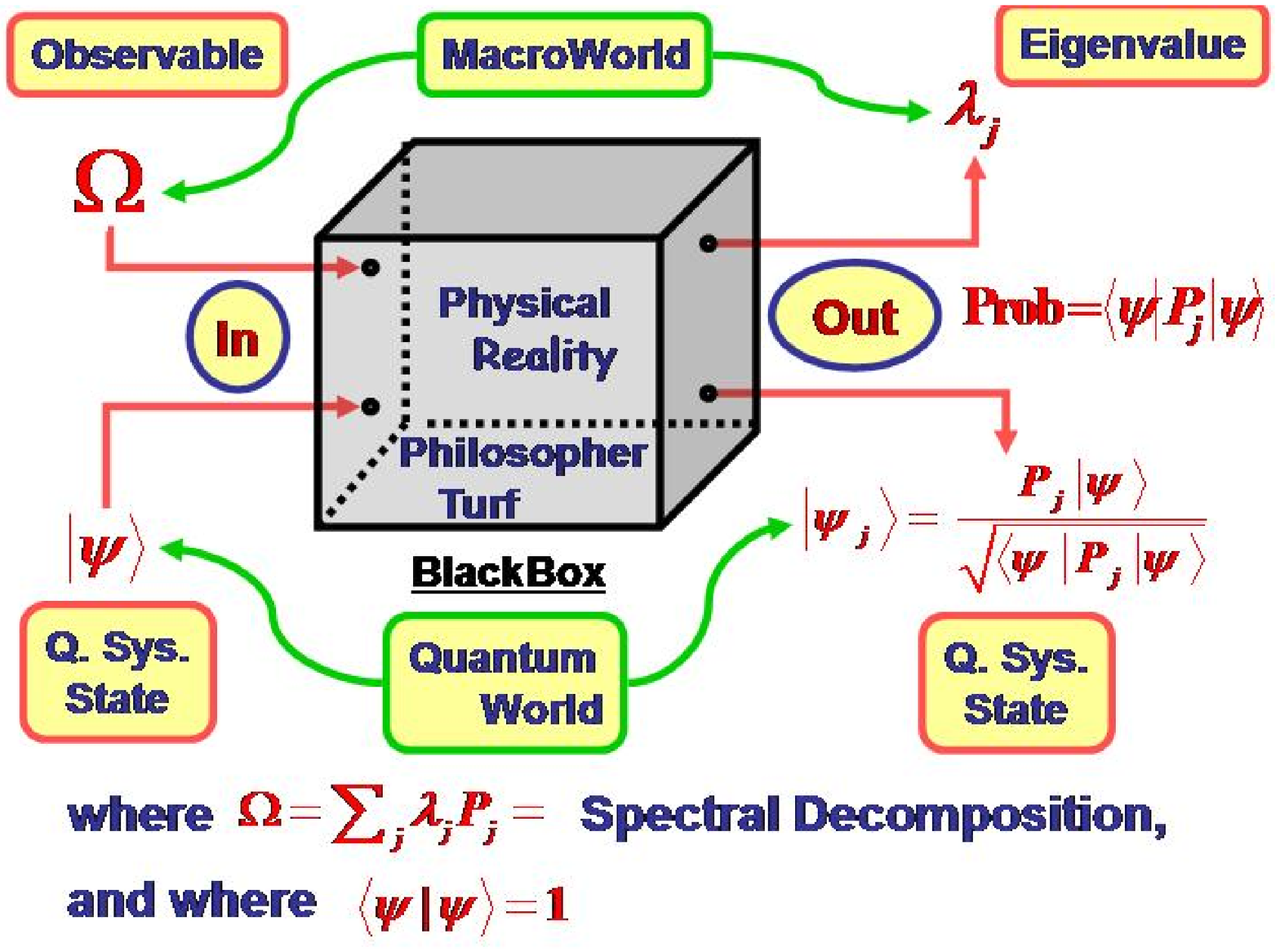}%
}\\
\textbf{Von Neumann measurement.}%
\end{center}

\bigskip

Let $\Omega$ be an observable for a quantum system $Q\left(  \mathcal{K}%
^{(\ell,n)},\Lambda_{\ell,n}\right)  $, i.e., a Hermitian (self-adjoint)
linear operator on the Hilbert space $\mathcal{K}^{(\ell,n)}$. \ Moreover,
let
\[
\Omega=\sum_{j=1}^{m}\lambda_{j}P_{j}%
\]
be the spectral decomposition of the observable $\Omega$, where $\lambda_{j}$
is the $j$-th eigenvalue of $\Omega$, and where $P_{j}$ is the corresponding
projection operator for the associated eigenspace $V_{j}$.

\bigskip

Let $\mathcal{P}_{\Omega}^{(\ell,n)}$ denote the \textbf{set of all
probability distributions on the spectrum} $\left\{  \lambda_{1},\lambda
_{2},\ldots,\lambda_{m}\right\}  $ of $\Omega$, i.e.,%
\[
\mathcal{P}_{\Omega}^{(\ell,n)}=\left\{  p:\left\{  \lambda_{1},\lambda
_{2},\ldots,\lambda_{m}\right\}  \longrightarrow\left[  0,1\right]
:\sum_{j=1}^{m}p\left(  \lambda_{j}\right)  =1\right\}
\]

We will call the probability distributions $p$ of $\mathcal{P}_{\Omega}%
^{(\ell,n)}$ \textbf{stochastic sources}.

\bigskip\bigskip

Then each observable $\Omega$ uniquely determines a map%
\[%
\begin{array}
[c]{cccc}%
\widehat{\Omega}: & \mathcal{K}^{(\ell,n)} & \longrightarrow & \mathcal{P}%
_{\Omega}^{(\ell,n)}\\
& \left\vert \psi\right\rangle  & \longmapsto & p
\end{array}
\]
from quantum knots to stochastic sources on the spectrum of $\Omega$ given by%
\[
p_{j}\left(  \left\vert \psi\right\rangle \right)  =\frac{\left\langle
\psi\left\vert P_{j}\right\vert \psi\right\rangle }{\sqrt{\left\langle
\psi|\psi\right\rangle }}\text{.}%
\]

\bigskip

Thus, what is seen, when a quantum system $Q$ in state $\left\vert
\psi\right\rangle $ is measured with respect to an observable $\Omega$, is a
random sample from the stochastic source $\widehat{\Omega}\left(  \left\vert
\psi\right\rangle \right)  $. \ But under what circumstances is such a random
sample a quantum knot invariant?

\bigskip

Our answer to this question is that quantum knots $\left\vert \psi
_{1}\right\rangle $ and $\left\vert \psi_{2}\right\rangle $ of the same knot
$\left(  \ell,n\right)  $-type must produce random samples from the same
stochastic source when measured with respect to the observable $\Omega$.
\ This answer is captured by the following definition:\bigskip

\begin{definition}
Let $Q\left(  \mathcal{K}^{(\ell,n)},\Lambda_{\ell,n}\right)  $ be a quantum
knot system, and let $\Omega$ be an observable on $\mathcal{K}^{(\ell,n)}$
with spectral decomposition
\[
\Omega=\sum_{j=1}^{m}\lambda_{j}P_{j}\text{ .}%
\]
Then the observable $\Omega$ is said to be a \textbf{quantum knot }$\left(
\ell,n\right)  $\textbf{-invariant} provided%
\[
\left\langle \psi\left\vert UP_{j}U^{-1}\right\vert \psi\right\rangle
=\left\langle \psi\left\vert P_{j}\right\vert \psi\right\rangle
\]
for all quantum knots $\left\vert \psi\right\rangle \in\mathcal{K}^{(\ell,n)}%
$, for all $U\in\Lambda_{\ell,n}$, and for all projectors $P_{j}$.
\end{definition}

\bigskip

\begin{theorem}
Let $Q\left(  \mathcal{K}^{(\ell,n)},\Lambda_{\ell,n}\right)  $ and $\Omega$
be as given in the above definition. Then the following statements are equivalent:

\begin{itemize}
\item[\textbf{1)}] The observable $\Omega$ is a quantum knot $\left(
\ell,n\right)  $-invariant

\item[\textbf{2)}] $\left[  U,P_{j}\right]  =0$ for all $U\in\Lambda_{\ell,n}$
and for all $P_{j}$.

\item[\textbf{3)}] $\left[  U,\Omega\right]  =0$ for all $U\in\Lambda_{\ell
,n}$,
\end{itemize}

\noindent where $\left[  A,B\right]  $ denotes the commutator $AB-BA$ of
operators $A$ and $B$.
\end{theorem}

\bigskip

The remaining half of this section is devoted to finding an answer to the
following question:

\bigskip

\noindent\textbf{Question:} \ \textit{How do we find observables which are
quantum knot invariants?}

\bigskip

One answer to this question is the following theorem, which is an almost
immediate consequence of the definition of a minimum invariant subspace of
$\mathcal{K}^{(\ell,n)}$:

\bigskip

\begin{theorem}
Let $Q\left(  \mathcal{K}^{(\ell,n)},\Lambda_{\ell,n}\right)  $ be a quantum
knot system, and let
\[
\mathcal{K}^{(\ell,n)}=\bigoplus_{s}W_{s}%
\]
be a decomposition of the representation%
\[
\Lambda_{\ell,n}\times\mathcal{K}^{(\ell,n)}\longrightarrow\mathcal{K}%
^{(\ell,n)}%
\]
into irreducible representations of the ambient group $\Lambda_{\ell,n}$.
\ Then, for each $s$, the projection operator $P_{s}$ for the subspace $W_{s}$
is an observable which is a quantum knot $\left(  \ell,n\right)  $-invariant.
\end{theorem}

\bigskip

Yet another way of finding quantum knot invariants is given by the following
theorem:\bigskip

\begin{theorem}
Let $Q\left(  \mathcal{K}^{(\ell,n)},\Lambda_{\ell,n}\right)  $ be a quantum
knot system, and let $\Omega$ be an observable on the Hilbert space
$\mathcal{K}^{(\ell,n)}$. \ Let $St\left(  \Omega\right)  $ be the stabilizer
subgroup for $\Omega$, i.e.,
\[
St\left(  \Omega\right)  =\left\{  U\in\Lambda_{\ell,n}:U\Omega U^{-1}%
=\Omega\right\}  \text{ .}%
\]
Then the observable%
\[
\sum_{U\in\Lambda_{\ell,n}/St\left(  \Omega\right)  }U\Omega U^{-1}%
\]
is a quantum knot $n$-invariant, where $\sum\limits_{U\in\Lambda_{\ell
,n}/St\left(  \Omega\right)  }U\Omega U^{-1}$ denotes a sum over a complete
set of coset representatives for the stabilizer subgroup $St\left(
\Omega\right)  $ of the ambient group $\Lambda_{\ell,n}$.
\end{theorem}

\begin{proof}
The observable $\sum_{g\in\Lambda_{\ell,n}}g\Omega g^{-1}$is obviously an
quantum knot $n$-invariant, since $g^{\prime}\left(  \sum_{g\in\Lambda
_{\ell,n})}g\Omega g^{-1}\right)  g^{\prime-1}=\sum_{g\in\Lambda_{\ell,n}%
}g\Omega g^{-1}$ for all $g^{\prime}\in\Lambda_{\ell,n}$. \ If we let
$\left\vert St\left(  \Omega\right)  \right\vert $ denote the order of
$\left\vert St\left(  \Omega\right)  \right\vert $, and if we let $c_{1}%
,c_{2},\ldots,c_{p}$ denote a complete set of coset representatives of the
stabilizer subgroup $St\left(  \Omega\right)  $, then $\sum_{j=1}^{p}%
c_{j}\Omega c_{j}^{-1}=\frac{1}{\left\vert St\left(  \Omega\right)
\right\vert }\sum_{g\in\Lambda_{\ell,n}}g\Omega g^{-1}$ is also a quantum knot
invariant. \ 
\end{proof}

\bigskip

Here are two other ways of creating observables which are quantum $\left(
\ell,n\right)  $-knot invariants:\bigskip

\begin{itemize}
\item For each lattice knot $K\in\mathbb{K}^{\left(  \ell,n\right)  }$, we
define the observable $P_{\Lambda_{\ell,n}K}$, called the \textbf{orbit
projector} of $K$ as%
\[
P_{\Lambda_{\ell,n}K}=\sum_{K^{\prime}\in\Lambda_{\ell,n}K}\left\vert
K^{\prime}\right\rangle \left\langle K^{\prime}\right\vert
\]

\item For each knot invariant $I:\mathbb{K}\longrightarrow\mathbb{C}$, we
define the observable%
\[
I^{\left(  \ell,n\right)  }=\sum_{K\in\mathbb{K}^{(\ell,n)}}I\left(  K\right)
\left\vert K\right\rangle \left\langle K\right\vert
\]

\end{itemize}

\bigskip

\begin{remark}
We leave, as an exercise for the reader, the task of verifying that the
subspace of $\mathcal{K}^{\left(  \ell,n\right)  }$ associated with the
projection operator $P_{\Lambda_{\ell,n}K}$ is not an irreducible
representation of the lattice ambient group $\Lambda_{\ell,n}$.
\end{remark}

\bigskip

We end this section with an example of a quantum knot invariant:

\bigskip

\begin{example}
Let $K$ be the $\left(  0,3\right)  $-lattice knot $K=%
\raisebox{-0.1609in}{\includegraphics[
height=0.4186in,
width=0.4186in
]%
{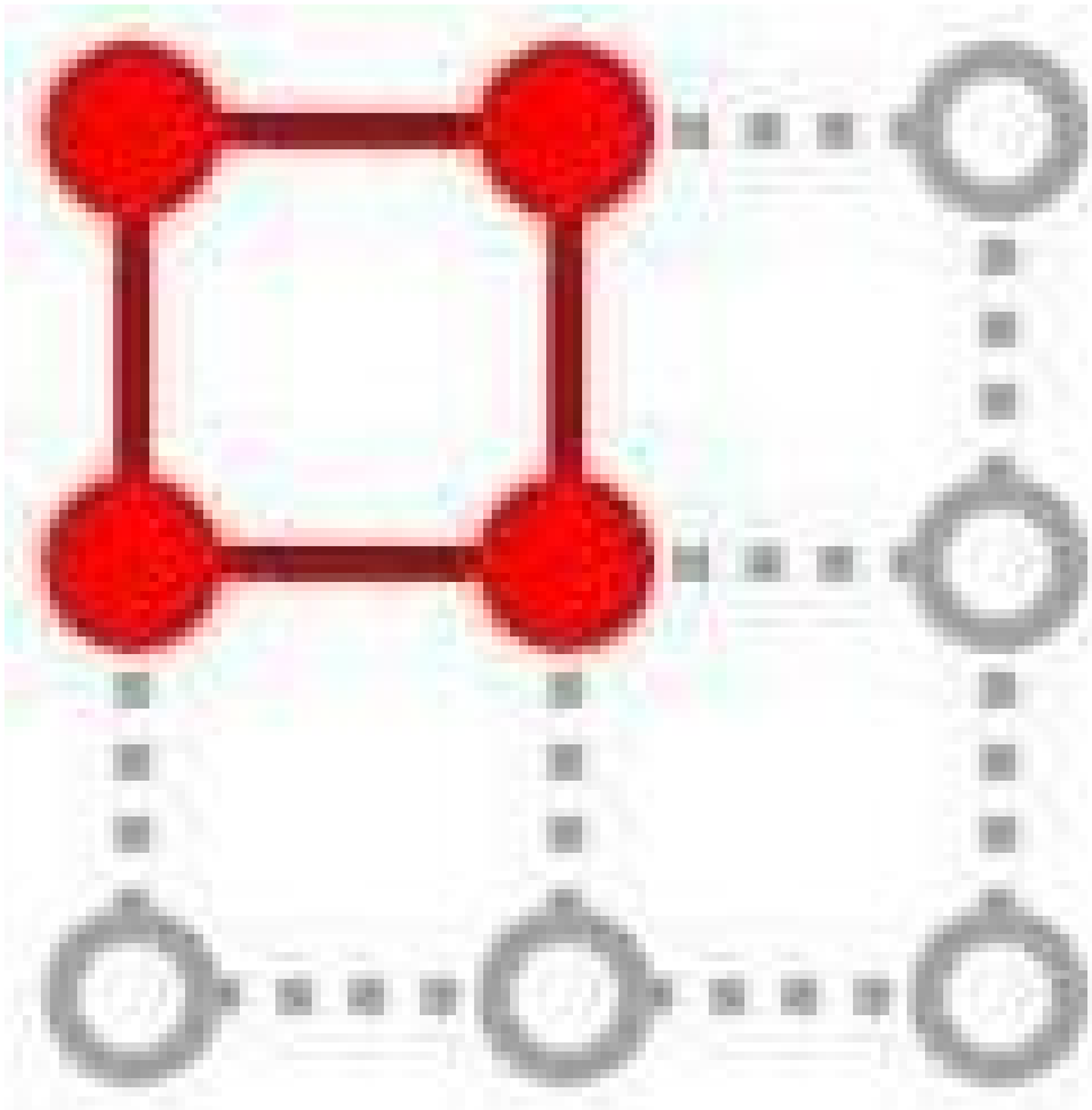}%
}%
$. \ Then the following observable $\Omega=P_{\Lambda_{0,3}K}$ is an example
of a quantum knot $\left(  0,3\right)  $-invariant for the quantum knot system
$Q\left(  K^{(0,3)},\Lambda_{0,3}\right)  $ :%
\[
\Omega=%
{\displaystyle\sum\limits_{g\in\Lambda_{0,3}}}
\left\vert g%
\raisebox{-0.1609in}{\includegraphics[
height=0.4186in,
width=0.4186in
]%
{qknot4.ps}%
}%
\right\rangle \left\langle g%
\raisebox{-0.1609in}{\includegraphics[
height=0.4186in,
width=0.4186in
]%
{qknot4.ps}%
}%
\right\vert
\]

\bigskip
\end{example}

\begin{remark}
For yet another approach to quantum knot measurement, we refer the reader to
the brief discussion on quantum knot tomography found in item 11) in the
conclusion of this paper. \ 
\end{remark}

\bigskip

\section{Conclusion: Open questions and future directions}

\bigskip

\bigskip

There are many possible open questions and future directions for research.
\ We mention only a few.

\bigskip

\begin{itemize}
\item[\textbf{1)}] What is the exact structure of each of the lattice ambient
groups $\Lambda_{\ell}$, $\widetilde{\Lambda}_{\ell}$, $\Lambda_{\ell,n}$,
$\widetilde{\Lambda}_{\ell,n}$, and their direct limits? \ Can one write down
an explicit presentation for for each of theses groups? for their direct
limits? \ The fact that each of the lattice ambient groups is generated by
involutions suggests that each may be a Coxeter group. \ Are the lattice
ambient groups Coxeter groups? \ How are the lattice ambient groups related to
the mosaic ambient groups $\mathbb{A}\left(  n\right)  $ found in
\cite{Lomonaco1}. \ Are they in some since the same?\bigskip

\item[\textbf{2)}] Unlike classical knots, quantum knots can exhibit the
non-classical behavior of quantum superposition and quantum entanglement.
\ Are topological entanglement and quantum entanglement related to one
another? \ If so, how? \ \bigskip

\item[\textbf{3)}] What other ways are there to distinguish quantum knots from
classical knots?\bigskip

\item[\textbf{4)}] Let $V_{K}\left(  t\right)  $ denote the Jones polynomial
of lattice $\left(  \ell,n\right)  $-knots $K\in\mathbb{K}^{(\ell,n)}$. \ Then
from Section 22 of this paper, we know that
\[
V^{\left(  \ell,n\right)  }\left(  t\right)  =\sum_{K\in\mathbb{K}^{(\ell,n)}%
}V_{K}\left(  t\right)  \left\vert K\right\rangle \left\langle K\right\vert
\]
is a family of quantum $\left(  \ell,n\right)  $-knot invariant observables
parameterized by the parameter $t$. \ We will refer to this parameterized
family $V^{\left(  \ell,n\right)  }\left(  t\right)  $ of observables simply
as the \textbf{Jones observable}. \ Can we use the Jones observable to create
an algorithmic improvement to the quantum algorithm given by Aharonov, Jones,
and Landau in \cite{Aharonov1}? \ (See also \cite{Lomonaco3}, \cite{Shor1}%
.)\bigskip

\item[\textbf{5)}] How does one create quantum knot observables that represent
other knot invariants such as, for example, the Vassiliev invariants?\bigskip

\item[\textbf{6)}] What is gained by extending the definition of quantum knot
observables to POVMs?\bigskip

\item[\textbf{7)}] What is gained by extending the definition of quantum knots
to mixed ensembles?\bigskip

\item[\textbf{8)}] Define the \textbf{lattice number} of a knot $k$ as the
smallest integer $n$ for which $k$ is representable as a lattice knot in the
$n$-bounded lattice $\mathcal{L}_{0,n}$. \ \ In general, how does one compute
the lattice number? \ Is the lattice number related to the crossing number of
a knot? \ How does one find an observable for the lattice number?\bigskip

\item[\textbf{9)}] Let $d\left(  \ell,n\right)  $ denote the dimension of the
Hilbert space $\mathcal{K}^{(\ell,n)}$ of $\left(  \ell,n\right)  $-th order
quantum knot. \ Find $d\left(  \ell,n\right)  $ for various values of $\ell$
an $n$. \ \ \bigskip

\item[\textbf{10)}] Consider the following alternate stronger definitions of
order $\left(  \ell,n\right)  $ quantum knot type and quantum knot type:
\bigskip\newline Let $Q\left(  \mathcal{K}^{(n)},\Lambda_{\ell,n}\right)  $ be
a quantum knot system, and let $\mathcal{U}\left(  \mathcal{K}^{(\ell
,n)}\right)  $ denote the Lie group of all unitary transformations on the
Hilbert space $\mathcal{K}^{(\ell,n)}$. \ Define the \textbf{continuous
ambient group} $\Lambda_{\ell,n}^{cont}$ as the smallest connected Lie
subgroup of $\mathcal{U}\left(  \mathcal{K}^{(\ell,n)}\right)  $ containing
the discrete ambient group $\Lambda_{\ell,n}$.\bigskip

\begin{proposition}
Let $\mathbb{S}$ denote the set of wiggle, wag, and tug generators of the
discrete ambient group $\Lambda_{\ell,n}$, and let $\lambda_{\ell,n}$ be the
Lie algebra generated by the elements of the set%
\[
\left\{  \ln_{P}\left(  g\right)  :g\in\mathbb{S}\right\}  \text{ ,}%
\]
where $\ln_{P}$ denotes the principal branch of the natural log on
$\mathcal{U}\left(  \mathcal{K}^{(\ell,n)}\right)  $. \ Then the continuous
ambient group is given by%
\[
\Lambda_{\ell,n}^{cont}=\exp\left(  \lambda_{\ell,n}\right)  \text{ .}%
\]

\end{proposition}

We define two quantum $\left(  \ell,n\right)  $-order quantum knots
$\left\vert \psi_{1}\right\rangle $ and $\left\vert \psi_{2}\right\rangle $ to
be of the \textbf{same continuous knot }$\left(  \ell,n\right)  $%
-\textbf{type}, written
\[
\left\vert \psi_{2}\right\rangle \underset{\ell}{\overset{n}{\asymp}%
}\left\vert \psi_{2}\right\rangle \text{ \ ,}%
\]
provided there exists an element $g$ of the continuous ambient group
$\Lambda_{\ell,n}^{cont}$ which transforms $\left\vert \psi_{1}\right\rangle $
into $\left\vert \psi_{2}\right\rangle $, i.e., such that $g\left\vert
\psi_{1}\right\rangle =\left\vert \psi_{2}\right\rangle $. \ They are of the
\textbf{same continuous knot type}, written $\left\vert \psi_{1}\right\rangle
\asymp\left\vert \psi_{2}\right\rangle $, if there exist non-negative integers
$\ell^{\prime}$ and $n^{\prime}$ such that $\iota^{\ell^{\prime}}%
\raisebox{-0.0303in}{\includegraphics[
height=0.1332in,
width=0.1193in
]%
{red-refinement.ps}%
}%
^{n^{\prime}}\left\vert \psi_{1}\right\rangle \underset{\ell+\ell^{\prime}%
}{\overset{n+n^{\prime}}{\asymp}}\iota^{\ell^{\prime}}%
\raisebox{-0.0303in}{\includegraphics[
height=0.1332in,
width=0.1193in
]%
{red-refinement.ps}%
}%
^{n^{\prime}}\left\vert \psi_{2}\right\rangle $.\bigskip

\begin{conjecture}
Let $K_{1}$ and $K_{2}$ denote two $\left(  \ell,n\right)  $-order lattice
knots, and let $\left\vert K_{1}\right\rangle $ and $\left\vert K_{2}%
\right\rangle $ denote the corresponding $\left(  \ell,n\right)  $-order
lattice quantum knots. \ Then
\[
\left\vert K_{1}\right\rangle \underset{\ell}{\overset{n}{\asymp}}\left\vert
K_{2}\right\rangle \Longleftrightarrow K_{1}\underset{\ell}{\overset{n}{\sim}%
}K_{2}\text{ \ \ and \ \ }\left\vert K_{1}\right\rangle \asymp\left\vert
K_{2}\right\rangle \Longleftrightarrow K_{1}\sim K_{2}%
\]

\end{conjecture}

Thus, if this conjecture is true, these two stronger definitions of quantum
knot $\left(  \ell,n\right)  $-type and quantum knot type fully capture all of
classical tame knot theory. \ Moreover, these two stronger definitions have a
number of advantages over the weaker definitions, two of which are the
following: \ \bigskip

\item Under the Hamiltonians associated with the generators $\mathbb{S}$, the
Schroedinger equation determines a connected continuous path in $\mathcal{K}%
^{(\ell,n)}$ consisting of $\left(  \ell,n\right)  $-th order quantum knots,
all of the same quantum continuous knot $\left(  \ell,n\right)  $-type.
\ \bigskip

\item Although the following two quantum knots are not of the same discrete
knot $\left(  \ell,n\right)  $-type, they are however of the same continuous
knot $\left(  \ell,n\right)  $-type%
\[
\left\vert \psi_{1}\right\rangle =\left\vert
\raisebox{-0.1609in}{\includegraphics[
height=0.4186in,
width=0.4186in
]%
{qknot-type2b.ps}%
}%
\right\rangle \text{ \ \ and \ \ }\left\vert \psi_{2}\right\rangle =\frac
{1}{\sqrt{2}}\left(  \ \overset{\mathstrut}{\left\vert
\raisebox{-0.1609in}{\includegraphics[
height=0.4186in,
width=0.4186in
]%
{qknot-type2b.ps}%
}%
\right\rangle }+\left\vert
\raisebox{-0.1609in}{\includegraphics[
height=0.4186in,
width=0.4186in
]%
{qknot-type2a.ps}%
}%
\right\rangle \ \right)
\]
\bigskip Is the above conjecture true? \ If so, what is the structure of the
continuous ambient group $\Lambda_{\ell,n}^{cont}$ ? \ What are its
irreducible representations?\bigskip

\item[\textbf{11)}] Can one create a more continuous definition of a quantum
knot by quantizing the classical electromagnetic knots found in
\cite{Lomonaco2}? \ This question was the original motivation for this paper
and for the paper \cite{Lomonaco1}.\bigskip

\item[\textbf{12)}] Quantum knot tomography: Given repeated copies of an
$\left(  \ell,n\right)  $-th order quantum knot $\left\vert \psi\right\rangle
$, how does one employ the method of quantum state tomography\cite{Leonhardt1}
to determine $\left\vert \psi\right\rangle $? \ Most importantly, how can this
be done with the greatest efficiency? \ For example, given repeated copies of
the unknown quantum knot basis state
\[
\left\vert \psi\right\rangle =\left\vert
\raisebox{-0.8026in}{\parbox[b]{2.2771in}{\begin{center}
\includegraphics[
height=1.6786in,
width=2.2771in
]%
{hopf-link.ps}%
\\
{}%
\end{center}}}%
\right\rangle \text{ \ ,}%
\]
how does find a universal set of observables that is best for determining the
quantum knot state in the sense of greatest efficiency for a given threshold
$\epsilon$? \ \bigskip

\item[\textbf{13)}] Quantum Braids: One can also use lattices to define
quantum braids. \ How is this related to the work found in \cite{Jacak1},
\cite{Kitaev1}, \cite{Sarma1}?\bigskip

\item[\textbf{14)}] Can quantum knot systems be used to model and to predict
the behavior of

\begin{itemize}
\item[\textbf{i)}] Quantum vortices in liquid Helium II? (See \cite{Rasetti1}.)

\item[\textbf{ii)}] Quantum vortices in the Bose-Einstein condensate?

\item[\textbf{iii)}] Fractional charge quantification that is manifest in the
fractional quantum Hall effect? (See \cite{Jacak1} and \cite{Wen1}.)
\end{itemize}
\end{itemize}

\bigskip

In closing this section, we should finally also say that, in the open
literature, the phrase "quantum knot" has many different meanings, and is
sometimes a phrase that is used loosely. \ We mention only two examples. \ In
\cite{Kauffman1}, a quantum knot is essentially defined as an element of the
Hilbert space with orthonormal basis in one-one-correspondence with knot
types, rather than knot representatives. \ Within the context of the mosaic
construction, a quantum knot in \cite{Kauffman1} corresponds to an element of
the orbit Hilbert space $\mathcal{K}^{(n)}/\mathbb{A}(n)$. \ \ In
\cite{Collins1} and in \cite{Sarma1} the phrase "quantum knot" refers not to
knots, but to the use of representations of the braid group to model the
dynamic behavior of certain quantum systems. \ In this context, braids are
used as a tool to model topological obstructions to quantum decoherence that
are conjectured to exist within certain quantum systems. \ 

\bigskip

\section{Appendix A: \ A quick review of knot theory}

\bigskip

In its most general form, knot theory is the study of the fundamental problem
of placement:

\bigskip

\noindent\textbf{The Placement Problem.} \ \textit{When are two placements of
a space }$X$\textit{ in a space }$Y$\textit{ the same or different?}

%

\begin{center}
\includegraphics[
height=2.2779in,
width=3.0277in
]%
{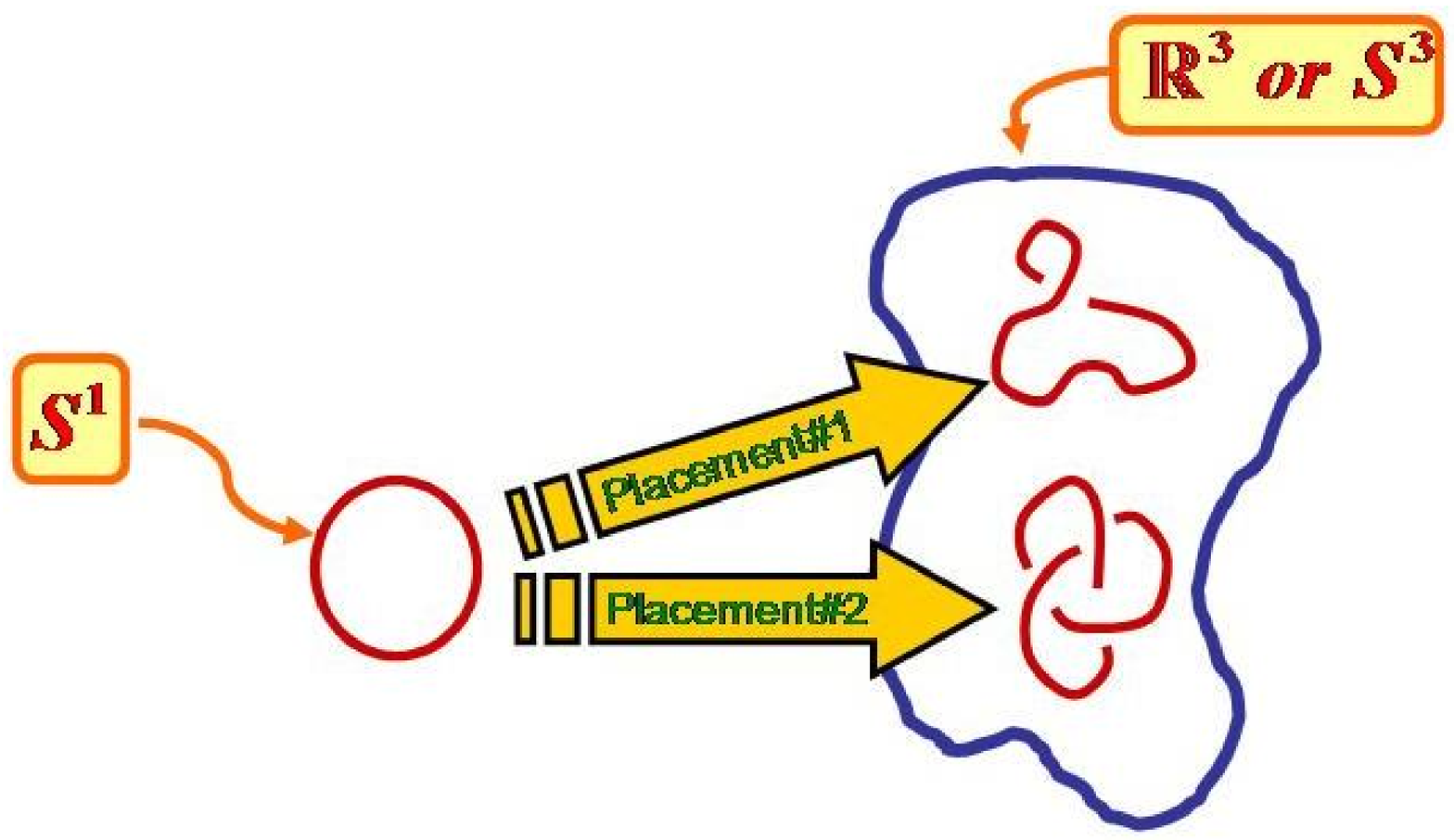}%
\\
\textbf{The placement problem.}%
\end{center}

In its most renowned form, knot theory is the study of the placement of a
1-sphere\footnote{By 1-sphere we mean a circle.} $S^{1}$(or a disjoint union
of 1-spheres) in 3-space $\mathbb{R}^{3}$ (or the 3-sphere $S^{3}$), called
the \textbf{ambient space}. \ In this case, "placement" usually means a smooth
(or piecewise linear) embedding, i.e., a smooth homeomorphism into the ambient
space. \ Such a placement is called a \textbf{knot} if a single 1-sphere is
embedded ( or a \textbf{link}, if a disjoint union of many 1-spheres is embedded.)

\bigskip

Two knots (or links) are said to be the same, i.e., of the \textbf{same knot
type}, if there exists an orientation preserving autohomeomorphism\footnote{A
provenly equivalent definition is that two knots are of the same knot type if
and only if there exists an isotopy of the ambient space that carries one knot
onto the other.} of the ambient space carrying one knot into the other.
Otherwise, they are said to be different, i.e., of \textbf{different knot
type}. \ Such knots are frequently represented by a knot diagram, i.e., a
planar 4-valent graph with vertices appropriately labelled as
undercrossings/overcrossings, as shown in figure 2.

\bigskip%

\begin{center}
\includegraphics[
height=2.2779in,
width=3.0277in
]%
{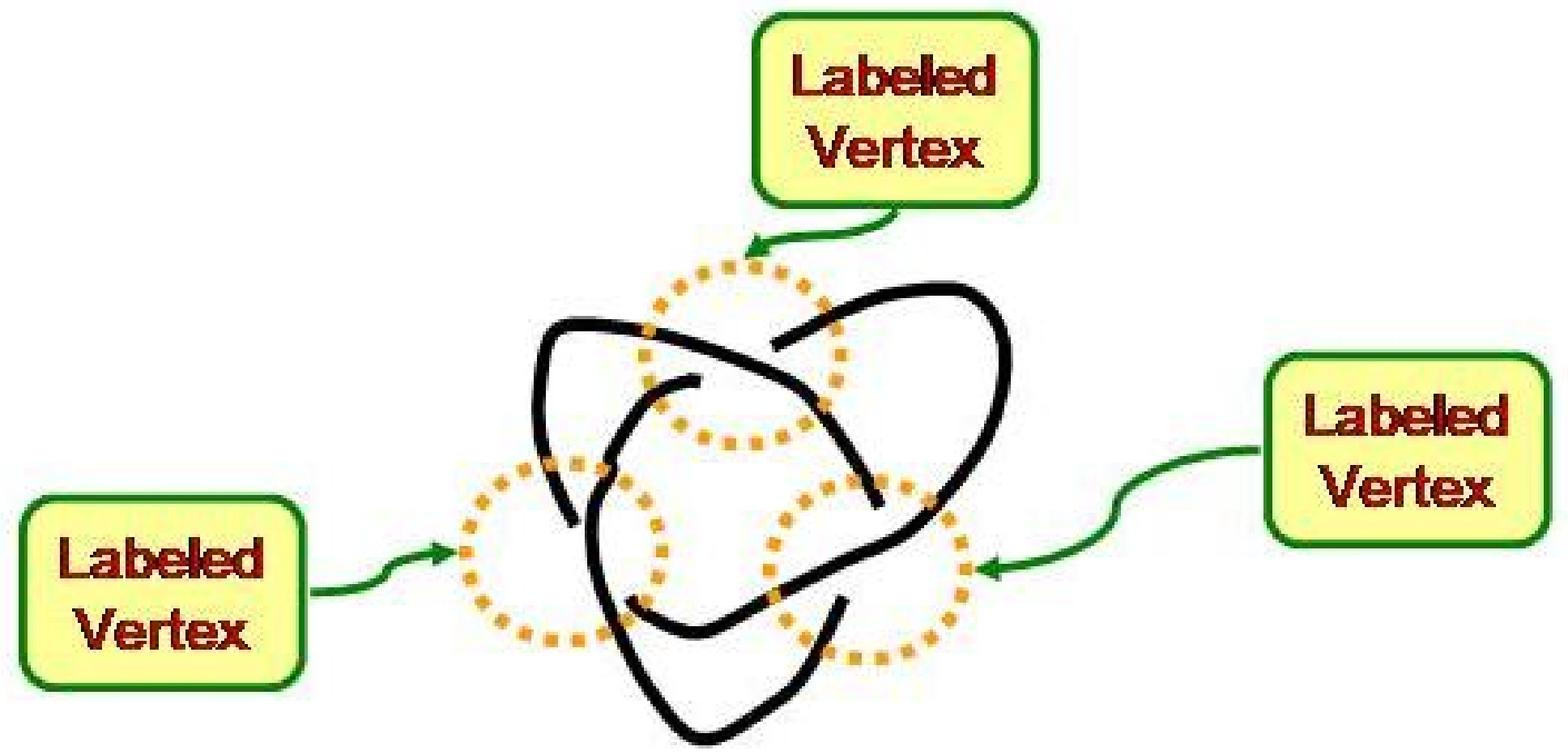}%
\\
\textbf{A knot diagram.}%
\end{center}

\bigskip

The fundamental problem of knot theory can now be stated as:

\bigskip

\noindent\textbf{The Fundamental Problem of Knot Theory.} \ \textit{When are
two knots of the same or of different knot type?}

\bigskip

A useful knot theoretic research tool is Reidemeister's theorem, which makes
use of the Reidemeister moves as defined in figure 3.:

\bigskip

\begin{theorem}
[\textbf{Reidemeister}]Two knot diagrams represent the same knot type if and
only if it is possible to transform one into the other by applying a finite
sequence of Reidemeister moves.
\end{theorem}

%

\begin{center}
\includegraphics[
height=2.2779in,
width=3.0277in
]%
{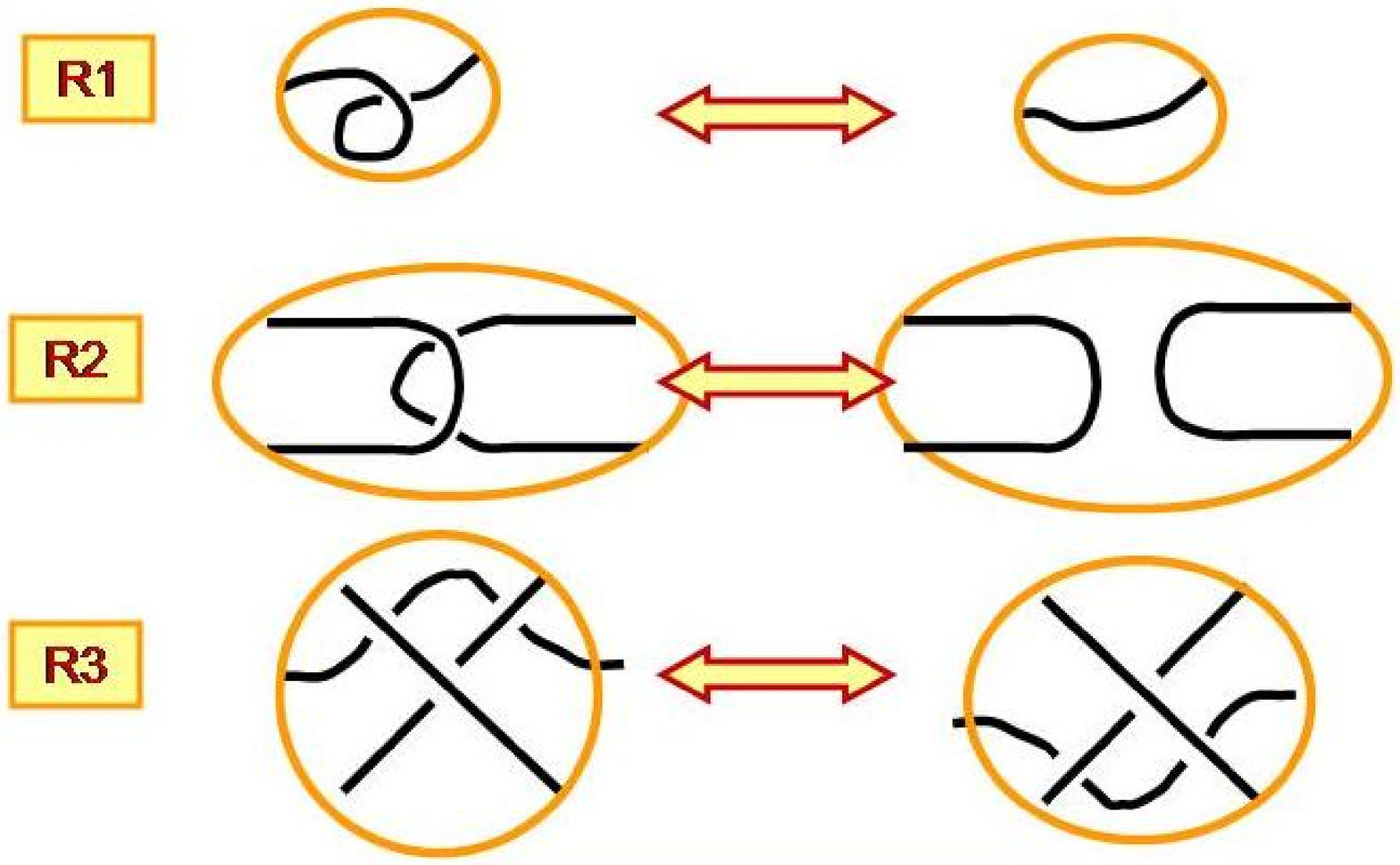}%
\\
\textbf{The Reidemeister moves. \ Th reader should note that these are local
moves, as indicated by the local enclosure.}%
\end{center}

\bigskip

The standard approach to attacking the fundamental problem of knot theory is
to create knot invariants for distinquishing knots. \ By a \textbf{knot
invariant} $I$ we mean a map from knots to a specified mathematical domain
which maps knots (or links) of the same type to the same mathematical object.
\ \textit{Thus, if an invariant is found to be different on two knots (or
links), then the two knots (or links) cannot be of the same knot type}!

\bigskip

For further information on knot theory, we refer the reader to, for example,
\cite{Adams1, Burde1, Crowell1, Fox1, Kauffman4, Lickorish1, Murasugi1,
Przytycki1, Reidemeister1, Reidemeister2, Rolfsen1}.

\bigskip

\section{Appendix B: A Rosetta Stone for notation}

\bigskip

This appendix gives a telegraphic summary of most of the mathematical symbols
used in this paper.

\bigskip%

\begin{center}
\fbox{\includegraphics[
height=2.4267in,
width=3.7905in
]%
{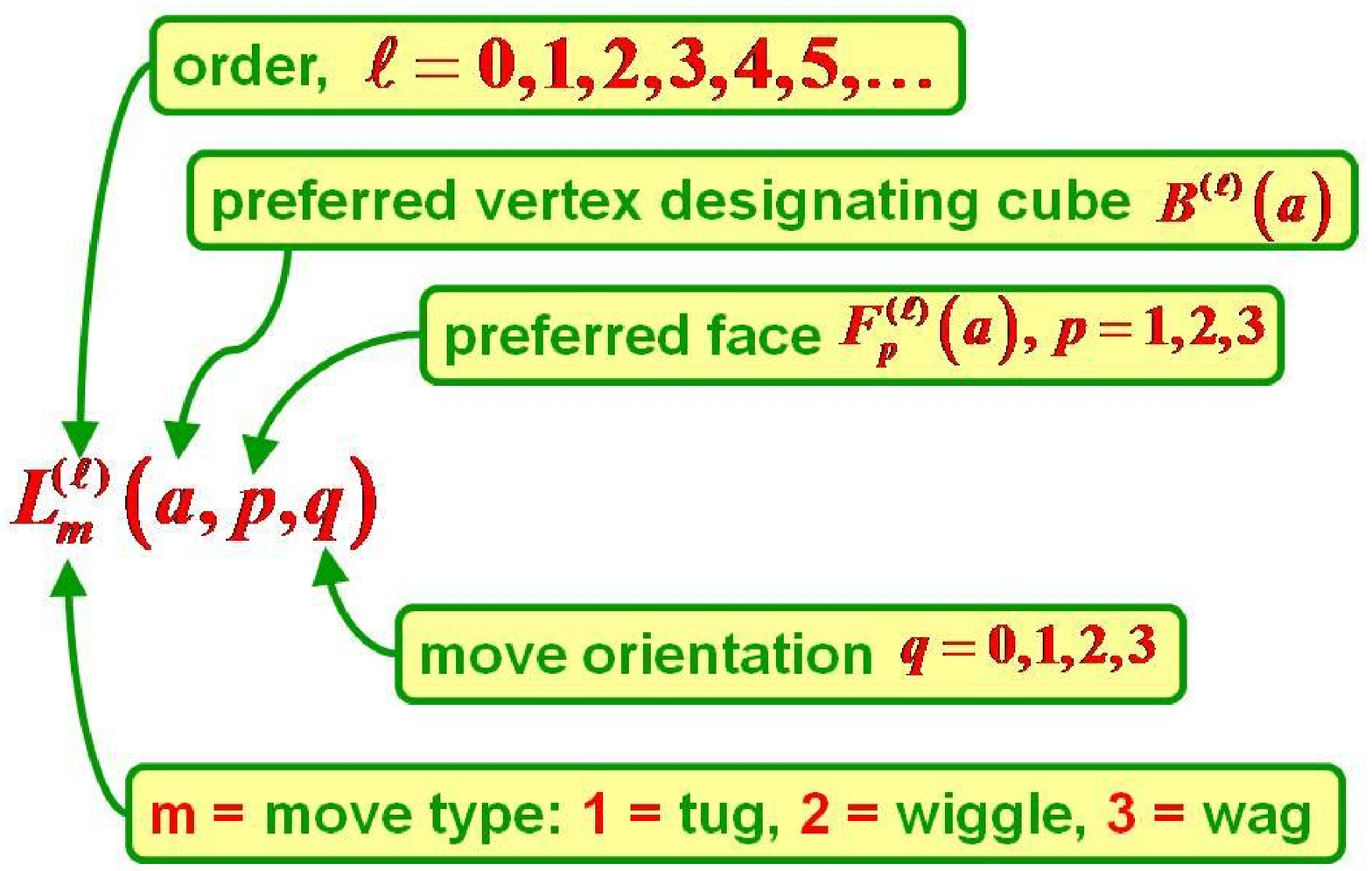}%
}\\
\textbf{Summary of notation used to designate moves.}%
\end{center}
\bigskip%

\[%
\begin{tabular}
[c]{c}%
\begin{tabular}
[c]{|c|c||c||c|c|c|c|}\hline
\quad & $\overset{\mathstrut}{\underset{\mathstrut}{L_{m}\left(  -,-,q\right)
}}$ & \quad & $L_{m}\left(  -,-,0\right)  $ & $L_{m}\left(  -,-,1\right)  $ &
$L_{m}\left(  -,-,2\right)  $ & $L_{m}\left(  -,-,3\right)  $\\\hline\hline
\quad &  & $m\backslash q$ & $0$ & $1$ & $2$ & $3$\\\hline\hline
\textbf{Tug} & $\underset{\mathstrut}{\overset{\mathstrut}{L_{1}\left(
-,-,q\right)  }}$ & $1$ & $%
\raisebox{-0.0406in}{\includegraphics[
height=0.1436in,
width=0.1436in
]%
{icon10.ps}%
}%
$ & $%
\raisebox{-0.0406in}{\includegraphics[
height=0.1436in,
width=0.1436in
]%
{icon11.ps}%
}%
$ & $%
\raisebox{-0.0406in}{\includegraphics[
height=0.1436in,
width=0.1436in
]%
{icon12.ps}%
}%
$ & $%
\raisebox{-0.0406in}{\includegraphics[
height=0.1436in,
width=0.1436in
]%
{icon13.ps}%
}%
$\\\hline
\textbf{Wiggle} & $\underset{\mathstrut}{\overset{\mathstrut}{L_{2}\left(
-,-,q\right)  }}$ & $2$ & $%
\raisebox{-0.0406in}{\includegraphics[
height=0.1436in,
width=0.1436in
]%
{icon21.ps}%
}%
$ & $%
\raisebox{-0.0406in}{\includegraphics[
height=0.1436in,
width=0.1436in
]%
{icon20.ps}%
}%
$ & $%
\raisebox{-0.0406in}{\includegraphics[
height=0.1436in,
width=0.1436in
]%
{icon21.ps}%
}%
$ & $%
\raisebox{-0.0406in}{\includegraphics[
height=0.1436in,
width=0.1436in
]%
{icon20.ps}%
}%
$\\\hline
\textbf{Wag} & $\underset{\mathstrut}{\overset{\mathstrut}{L_{3}\left(
-,-,q\right)  }}$ & $3$ & $%
\raisebox{-0.0406in}{\includegraphics[
height=0.1531in,
width=0.1531in
]%
{icon30.ps}%
}%
$ & $%
\raisebox{-0.0406in}{\includegraphics[
height=0.1531in,
width=0.1531in
]%
{icon31.ps}%
}%
$ & $%
\raisebox{-0.0406in}{\includegraphics[
height=0.1531in,
width=0.1531in
]%
{icon32.ps}%
}%
$ & $%
\raisebox{-0.0406in}{\includegraphics[
height=0.1531in,
width=0.1531in
]%
{icon33.ps}%
}%
$\\\hline
\end{tabular}
\\
\textbf{Summary of notation used to designate moves.}\\
\textbf{Please note that }$L_{2}\left(  -,-,0\right)  =L_{2}\left(
-,-,2\right)  $ and \textbf{ }$L_{2}\left(  -,-,1\right)  =$\textbf{ }%
$L_{2}\left(  -,-,3\right)  $.
\end{tabular}
\ \
\]

\bigskip%
\[
L_{m}\left(  x,q\right)  ^{dx_{\left\lfloor p\right.  }dx_{\left.
p\right\rceil }}=\lim_{\ell\rightarrow\infty}L_{m}^{(\ell)}\left(
\left\lfloor x\right\rfloor _{\ell},p,q\right)  \text{ ,}%
\]
where $dx_{\left\lfloor p\right.  }dx_{\left.  p\right\rceil }$ is the area
2-form in 3-space $\mathbb{R}^{3}$.

\bigskip

$a^{:1^{2}\overline{2}^{3}3}$ denotes the translate of the vertex $a$ given
by
\[
a^{:1^{2}\overline{2}^{3}3}=a+2\cdot2^{\ell}e_{1}-3\cdot2^{-\ell}%
e_{2}+2^{-\ell}e^{3}\text{ ,}%
\]
where $e_{1}$, $e_{2}$, $e_{3}$ is the preferred orthonormal frame.

\bigskip

\section{Appendix C: The refinement morphism conjecture}

\bigskip

In this appendix, we give a constructive rationale for conjectures 1A, !B, and
1C of section 12.

\bigskip

We begin by attempting to define $%
\raisebox{-0.0303in}{\includegraphics[
height=0.1332in,
width=0.1193in
]%
{red-refinement.ps}%
}%
(g)$ for each generator of the ambient group $\Lambda_{\ell}$, i.e., for each
tug, wiggle, and wag in $\Lambda_{\ell}$. \ A possible definition is suggested
by the figures given below:%

\[%
\begin{array}
[c]{l}%
\raisebox{-0.4021in}{\includegraphics[
height=0.9037in,
width=0.9037in
]%
{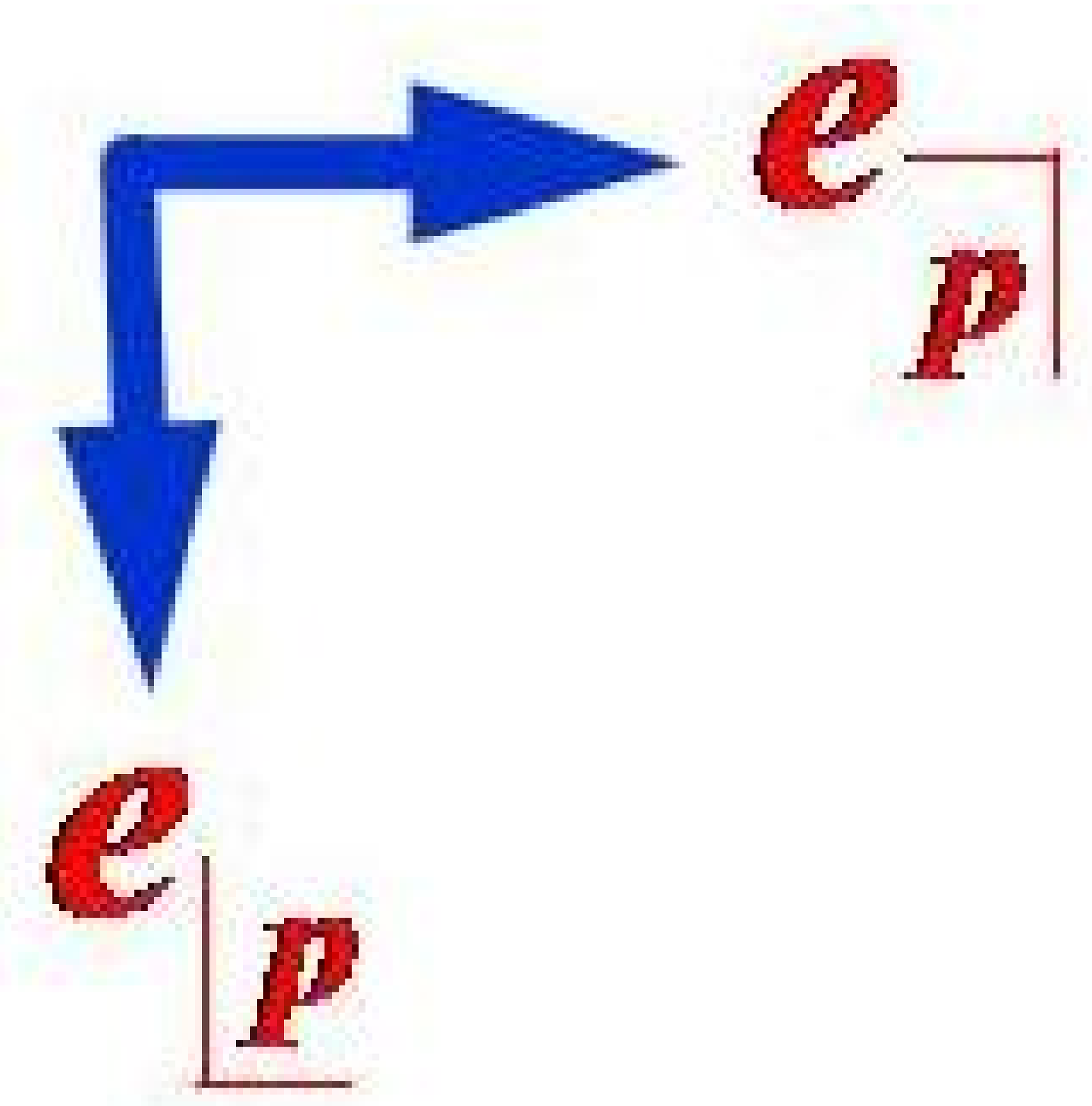}%
}%
\qquad%
\begin{array}
[c]{ccc}%
\raisebox{-0.3009in}{\includegraphics[
height=0.6789in,
width=0.6789in
]%
{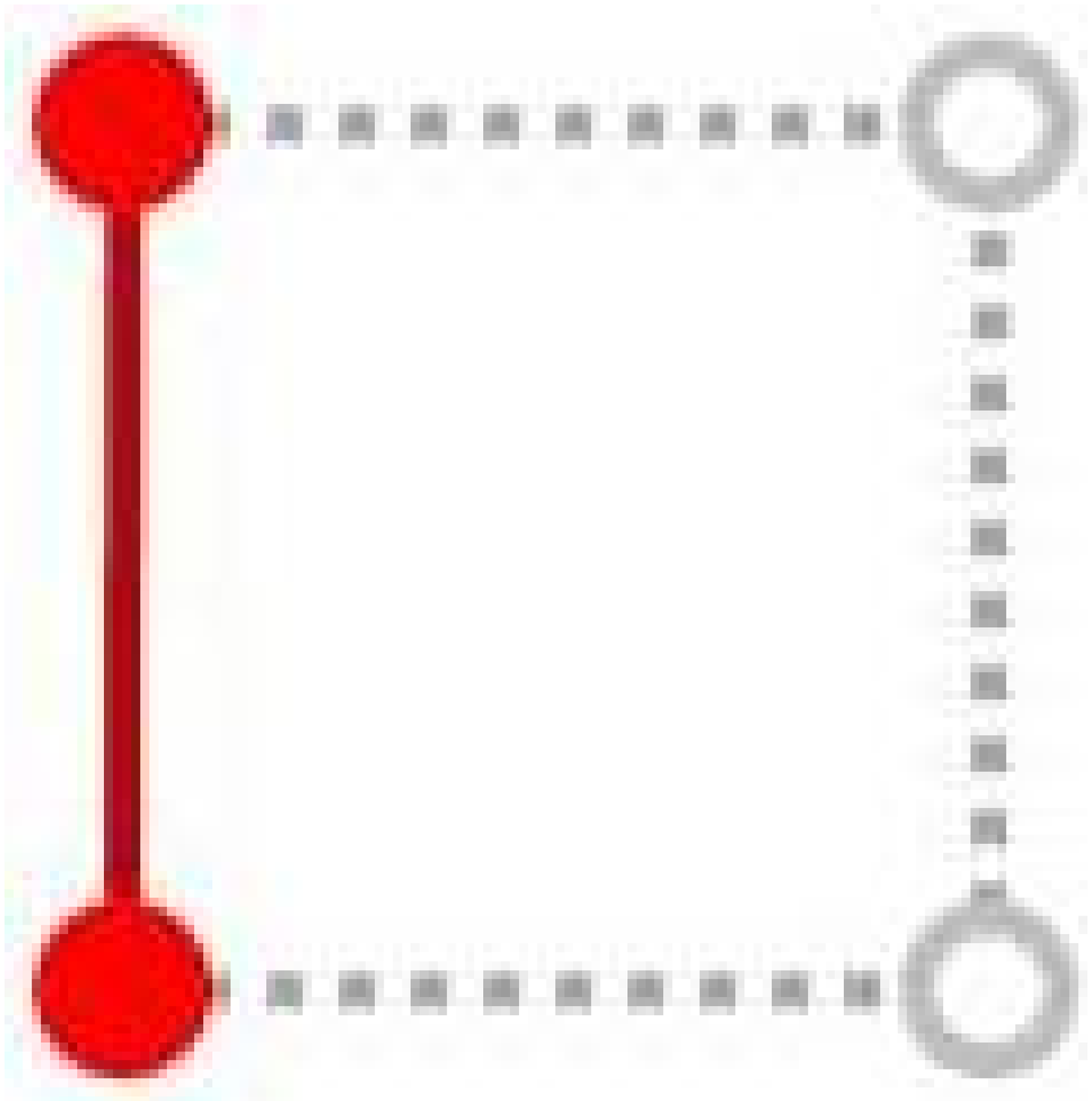}%
}%
&
\begin{array}
[c]{c}%
\raisebox{0.0069in}{\includegraphics[
height=0.2361in,
width=0.5967in
]%
{arrow-ya.ps}%
}%
\\
F_{p}^{(\ell)}(a)\\%
\raisebox{-0.0406in}{\includegraphics[
height=0.1436in,
width=0.1436in
]%
{icon10.ps}%
}%
^{(\ell)}\left(  a,p\right)
\end{array}
&
\raisebox{-0.3009in}{\includegraphics[
height=0.6789in,
width=0.6789in
]%
{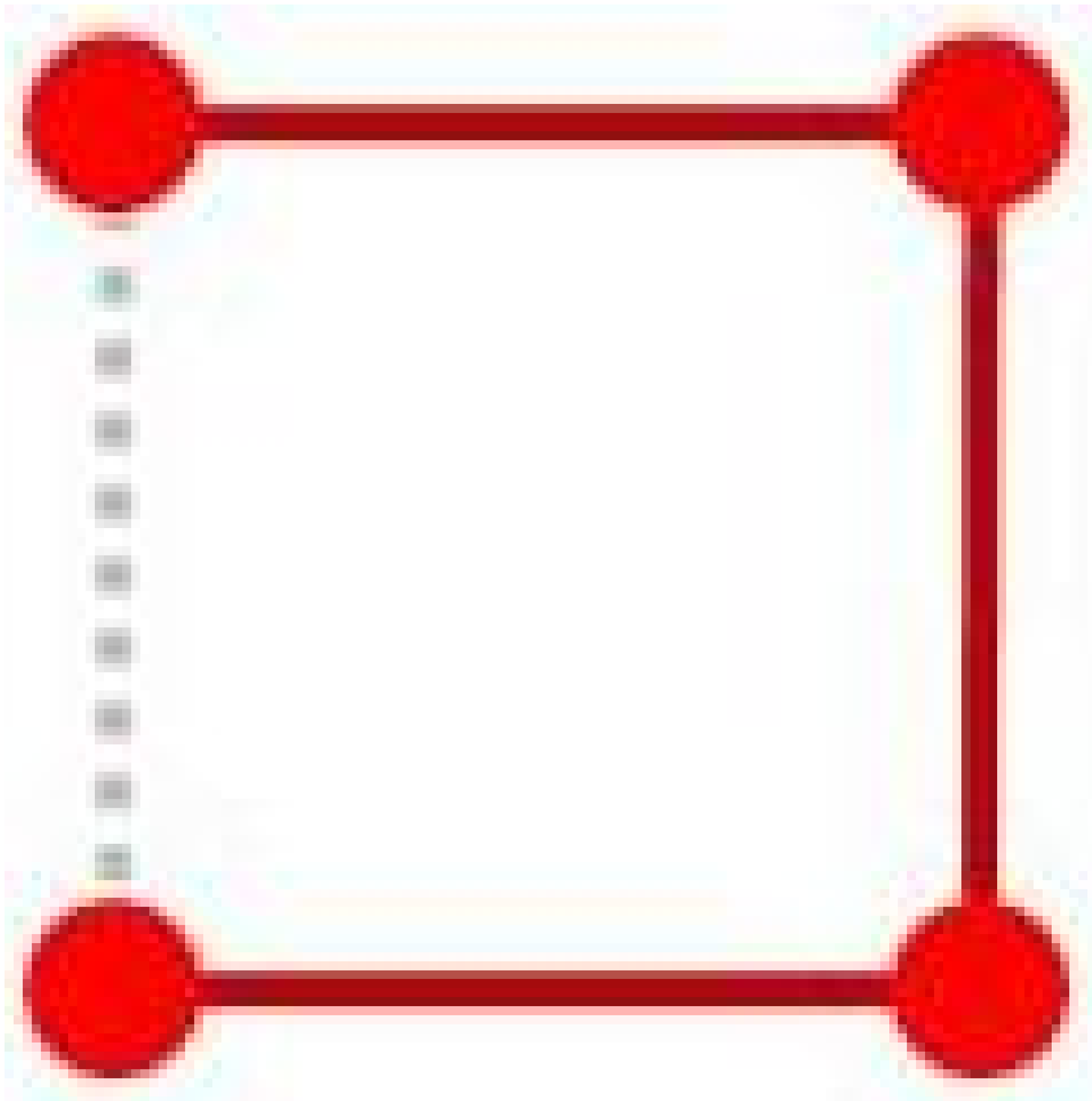}%
}%
\end{array}
\\
\multicolumn{1}{c}{}\\
\multicolumn{1}{c}{%
\begin{array}
[c]{ccccc}%
\raisebox{-0.3009in}{\includegraphics[
height=0.6789in,
width=0.6789in
]%
{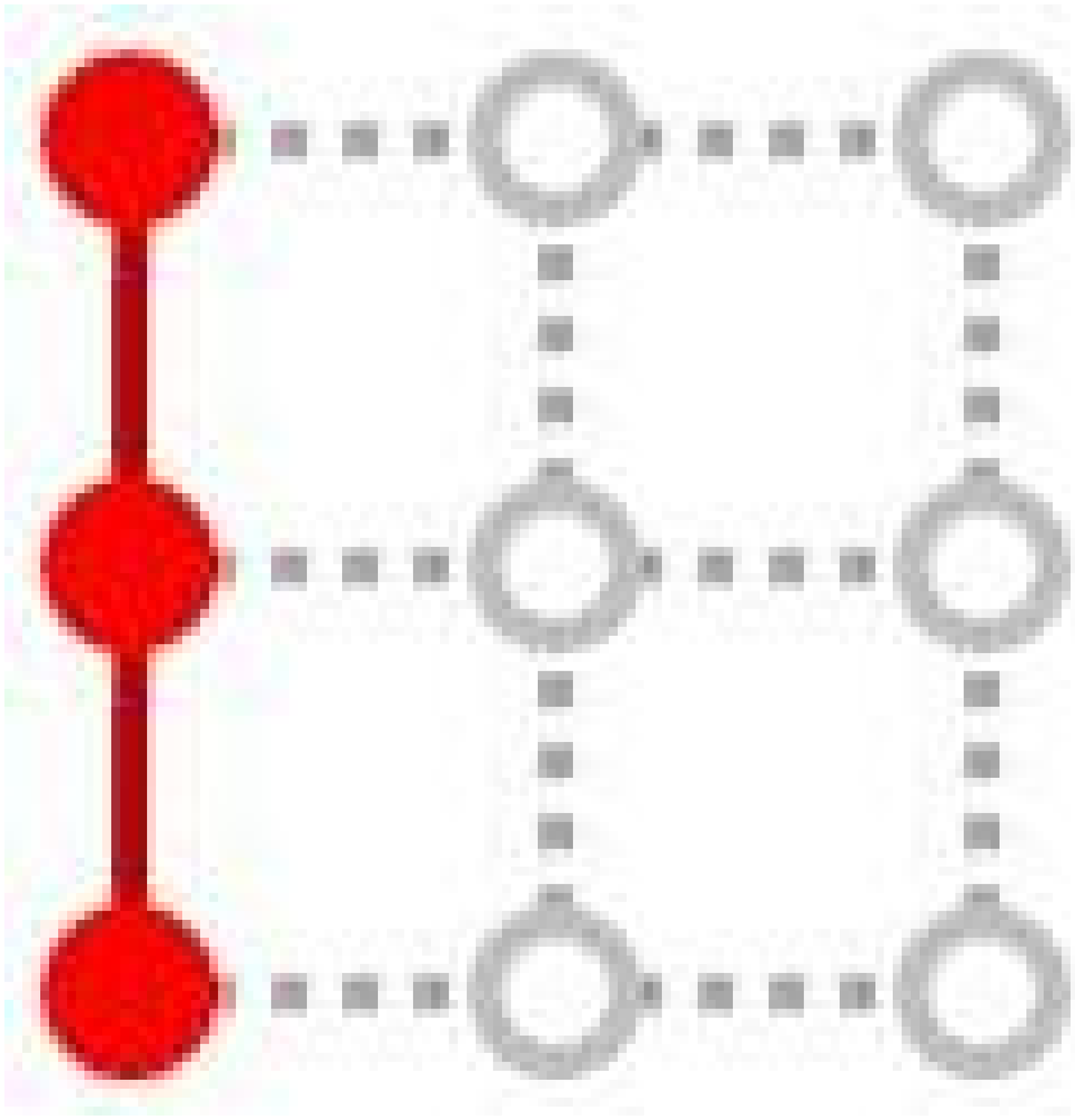}%
}%
&
\begin{array}
[c]{c}%
\raisebox{0.0069in}{\includegraphics[
height=0.2361in,
width=0.5967in
]%
{arrow-ya.ps}%
}%
\\
F_{p}^{(\ell+1)}(a)\\%
\raisebox{-0.0406in}{\includegraphics[
height=0.1436in,
width=0.1436in
]%
{icon10.ps}%
}%
^{(\ell+1)}\left(  a,p\right)
\end{array}
&
\raisebox{-0.3009in}{\includegraphics[
height=0.6789in,
width=0.6789in
]%
{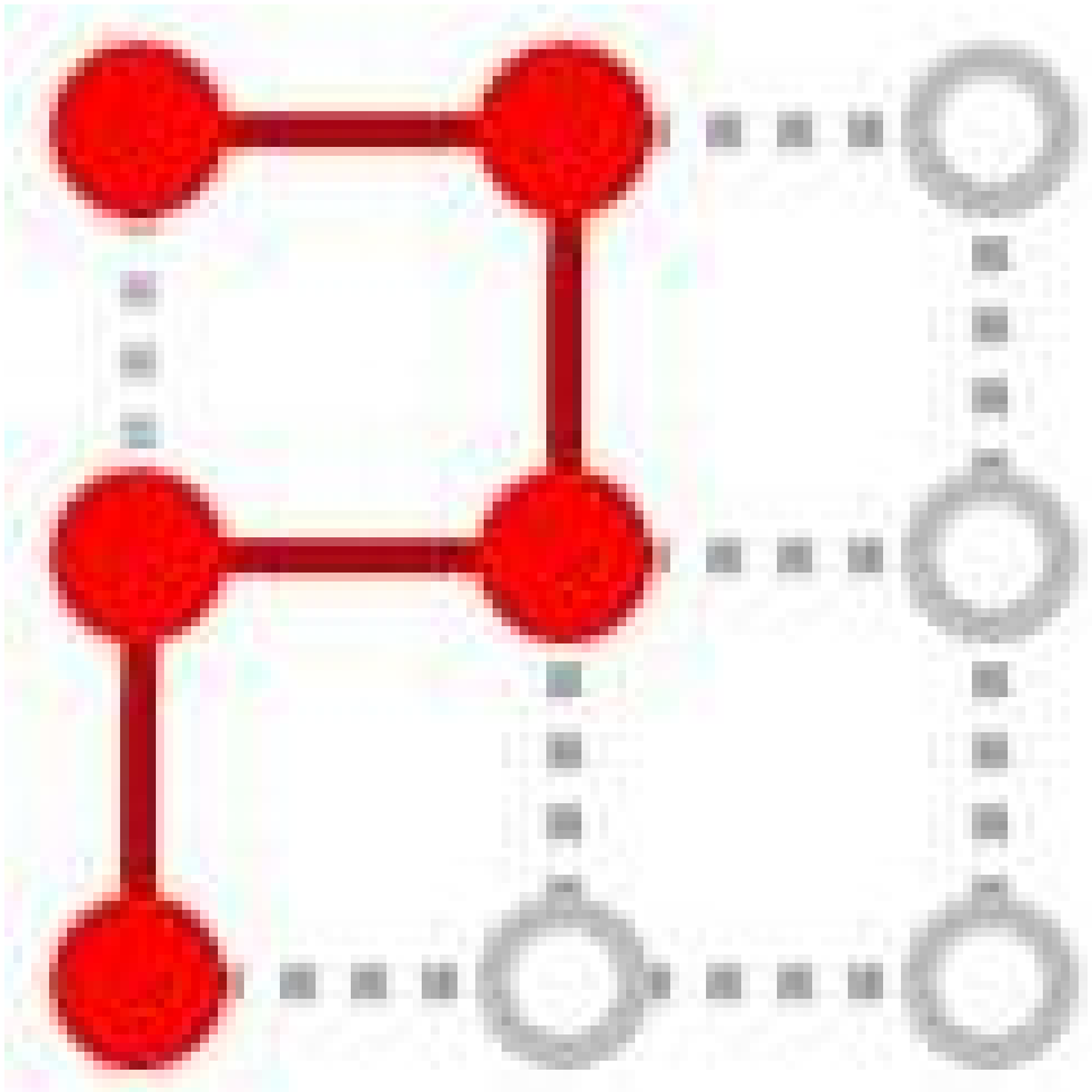}%
}%
&
\begin{array}
[c]{c}%
\raisebox{0.0069in}{\includegraphics[
height=0.2361in,
width=0.5967in
]%
{arrow-ya.ps}%
}%
\\
F_{p}^{(\ell+1)}(a^{:\left.  p\right\rceil })\\%
\raisebox{-0.0406in}{\includegraphics[
height=0.1436in,
width=0.1436in
]%
{icon10.ps}%
}%
^{(\ell+1)}\left(  a^{:\left.  p\right\rceil },p\right)
\end{array}
&
\raisebox{-0.3009in}{\includegraphics[
height=0.6789in,
width=0.6789in
]%
{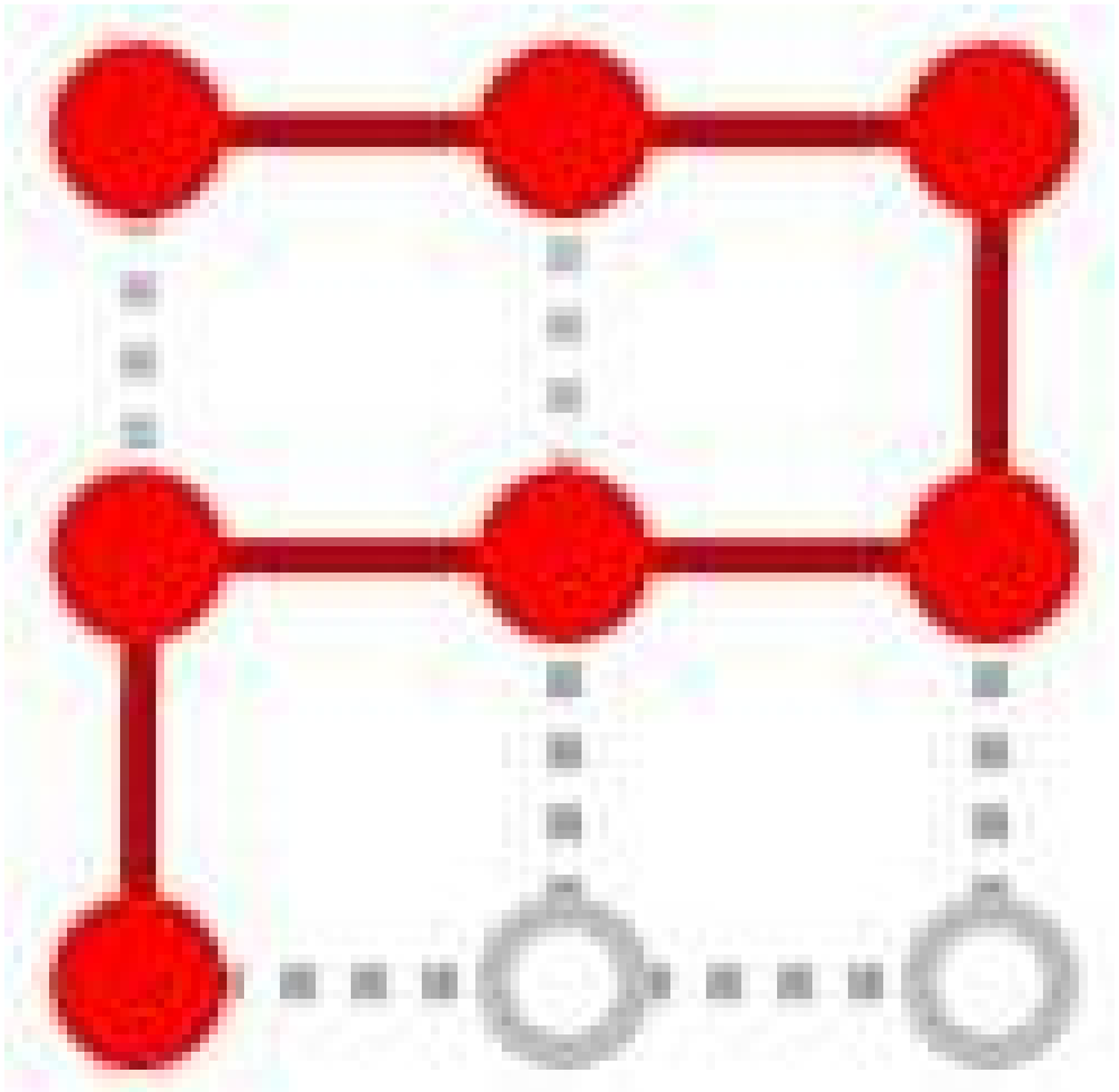}%
}%
\end{array}
}\\
\multicolumn{1}{c}{}\\
\multicolumn{1}{c}{%
\begin{array}
[c]{ccccc}%
\qquad\qquad\quad\ \  &
\begin{array}
[c]{c}%
\raisebox{0.0069in}{\includegraphics[
height=0.2361in,
width=0.5967in
]%
{arrow-ya.ps}%
}%
\\
F_{p}^{(\ell+1)}(a^{:_{\left\lfloor p\right.  }})\\%
\raisebox{-0.0406in}{\includegraphics[
height=0.1436in,
width=0.1436in
]%
{icon13.ps}%
}%
^{(\ell+1)}\left(  a^{:_{\left\lfloor p\right.  \left.  p\right\rceil }%
},p\right)
\end{array}
&
\raisebox{-0.3009in}{\includegraphics[
height=0.6789in,
width=0.6789in
]%
{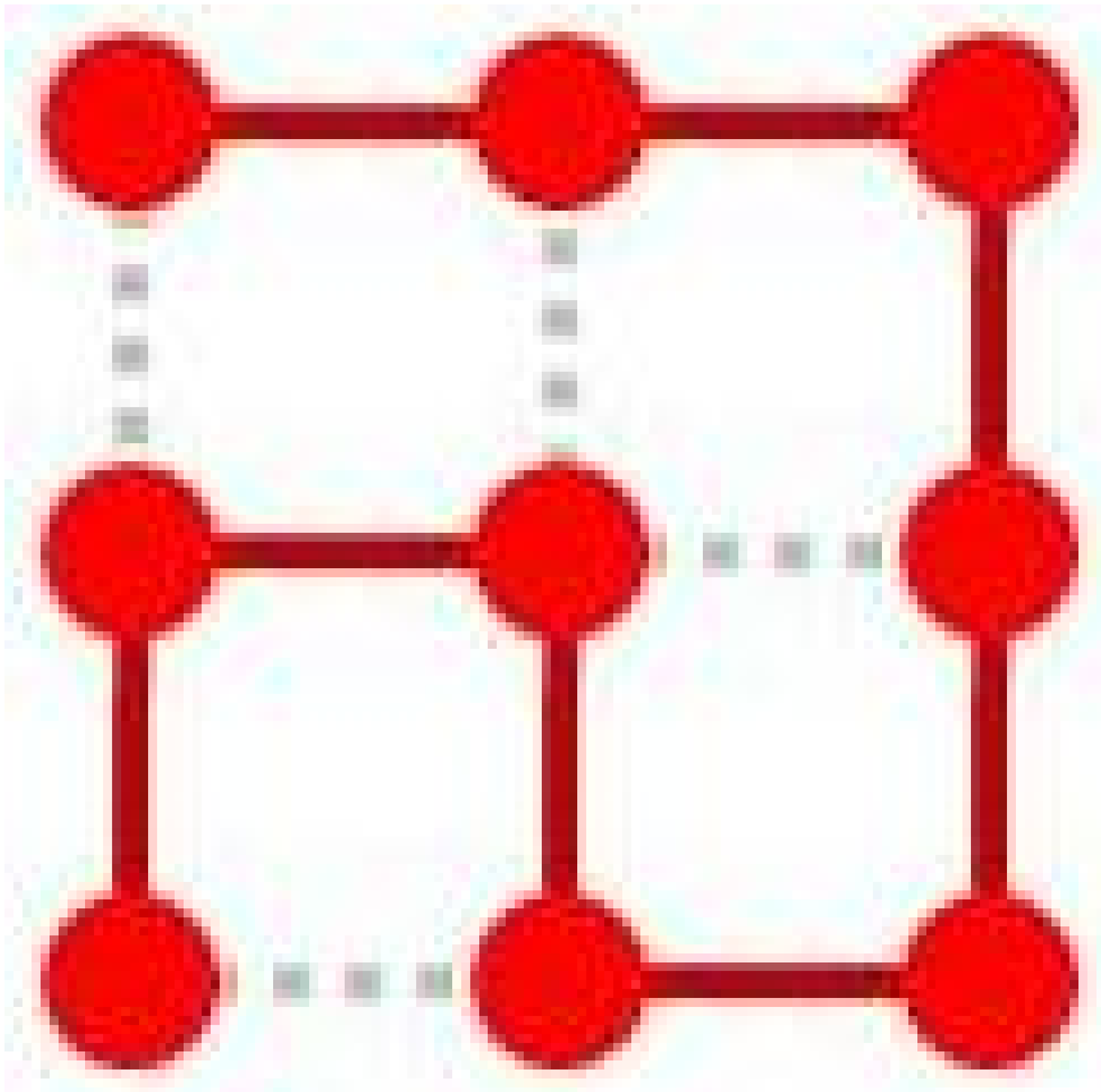}%
}%
&
\begin{array}
[c]{c}%
\raisebox{0.0069in}{\includegraphics[
height=0.2361in,
width=0.5967in
]%
{arrow-ya.ps}%
}%
\\
F_{p}^{(\ell+1)}(a^{:_{\left\lfloor p\right.  \left.  p\right\rceil }})\\%
\raisebox{-0.0406in}{\includegraphics[
height=0.1436in,
width=0.1436in
]%
{icon11.ps}%
}%
^{(\ell+1)}\left(  a^{:_{\left\lfloor p\right.  }},p\right)
\end{array}
&
\raisebox{-0.3009in}{\includegraphics[
height=0.6789in,
width=0.6789in
]%
{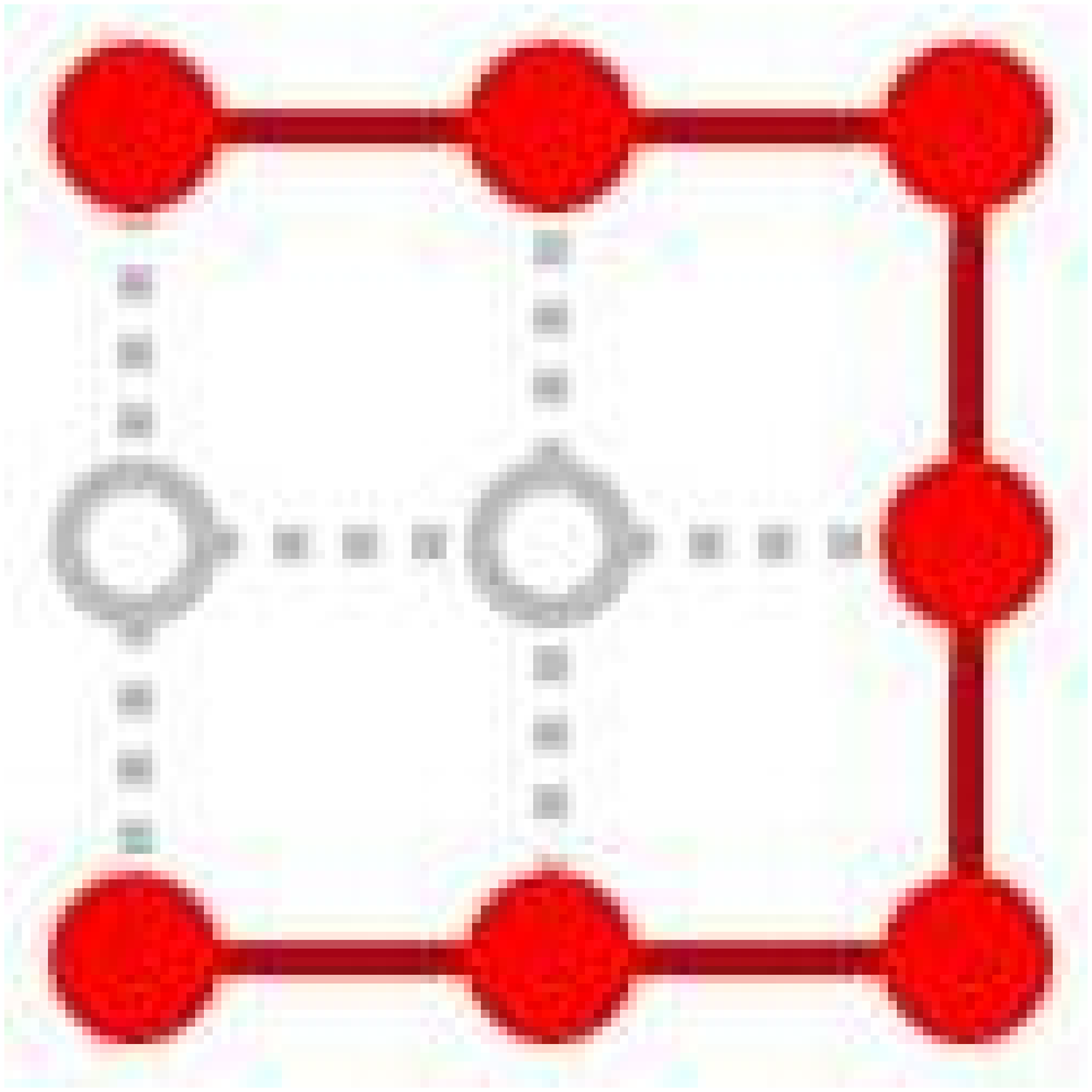}%
}%
\end{array}
}%
\end{array}
\]

\bigskip%

\[%
\begin{array}
[c]{l}%
\raisebox{-0.4021in}{\includegraphics[
height=0.9037in,
width=0.9037in
]%
{frame2.ps}%
}%
\qquad%
\begin{array}
[c]{ccc}%
\raisebox{-0.3009in}{\includegraphics[
height=0.6789in,
width=0.6789in
]%
{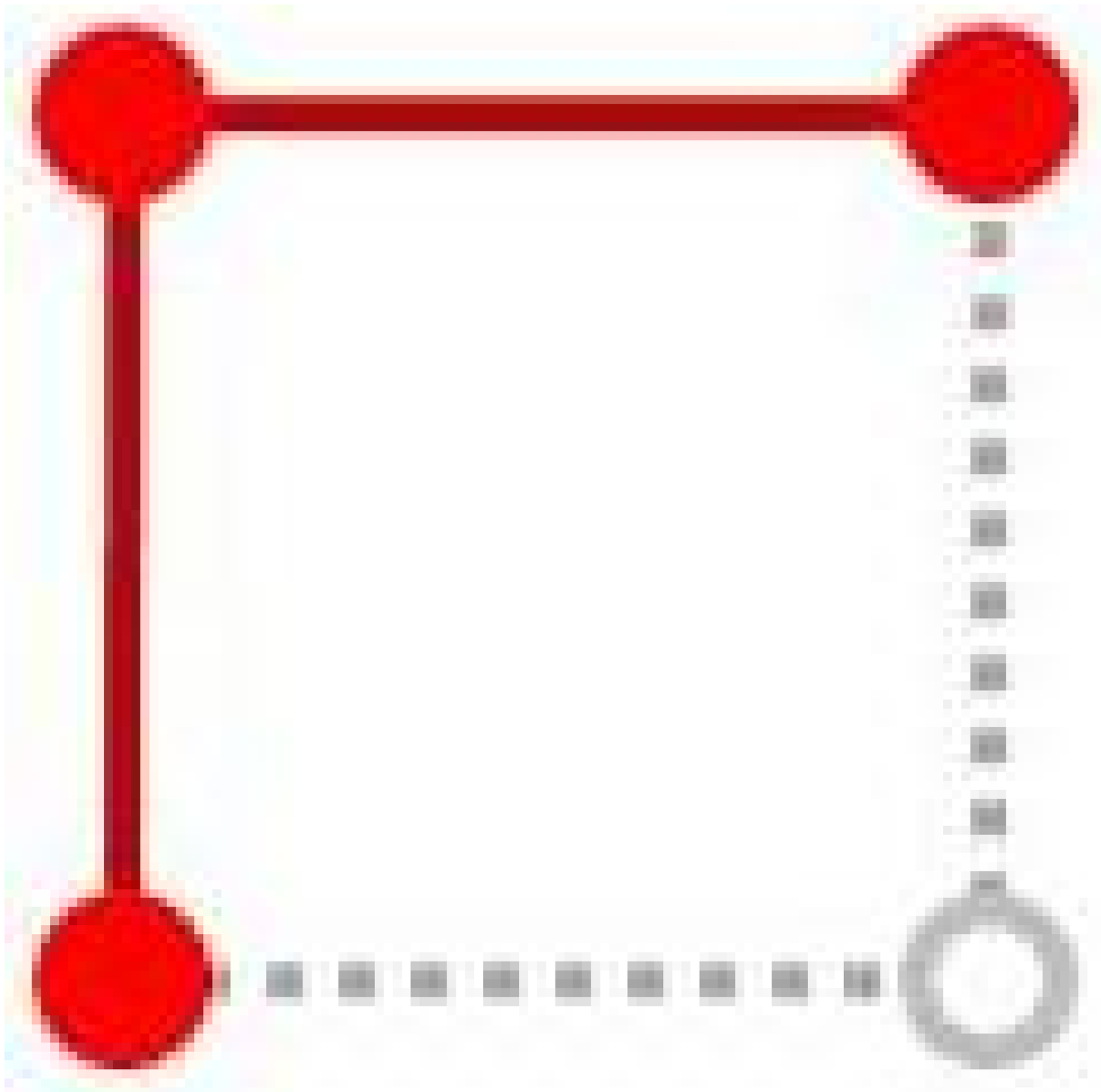}%
}%
&
\begin{array}
[c]{c}%
\raisebox{0.0069in}{\includegraphics[
height=0.2361in,
width=0.5967in
]%
{arrow-ya.ps}%
}%
\\
F_{p}^{(\ell)}(a)\\%
\raisebox{-0.0406in}{\includegraphics[
height=0.1436in,
width=0.1436in
]%
{icon21.ps}%
}%
^{(\ell)}\left(  a,p\right)
\end{array}
&
\raisebox{-0.3009in}{\includegraphics[
height=0.6789in,
width=0.6789in
]%
{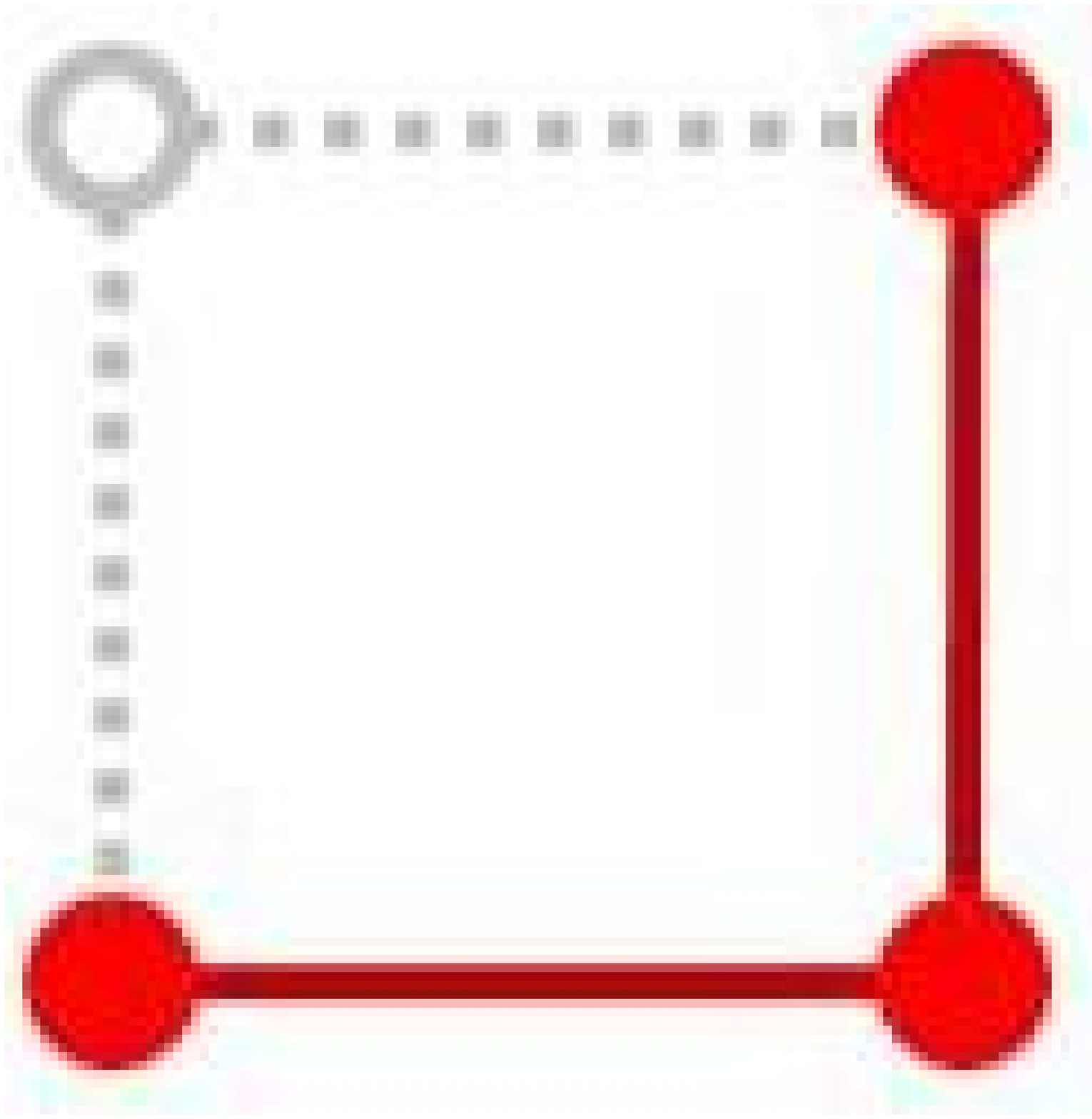}%
}%
\end{array}
\\
\multicolumn{1}{c}{}\\
\multicolumn{1}{c}{%
\begin{array}
[c]{ccccc}%
\raisebox{-0.3009in}{\includegraphics[
height=0.6789in,
width=0.6789in
]%
{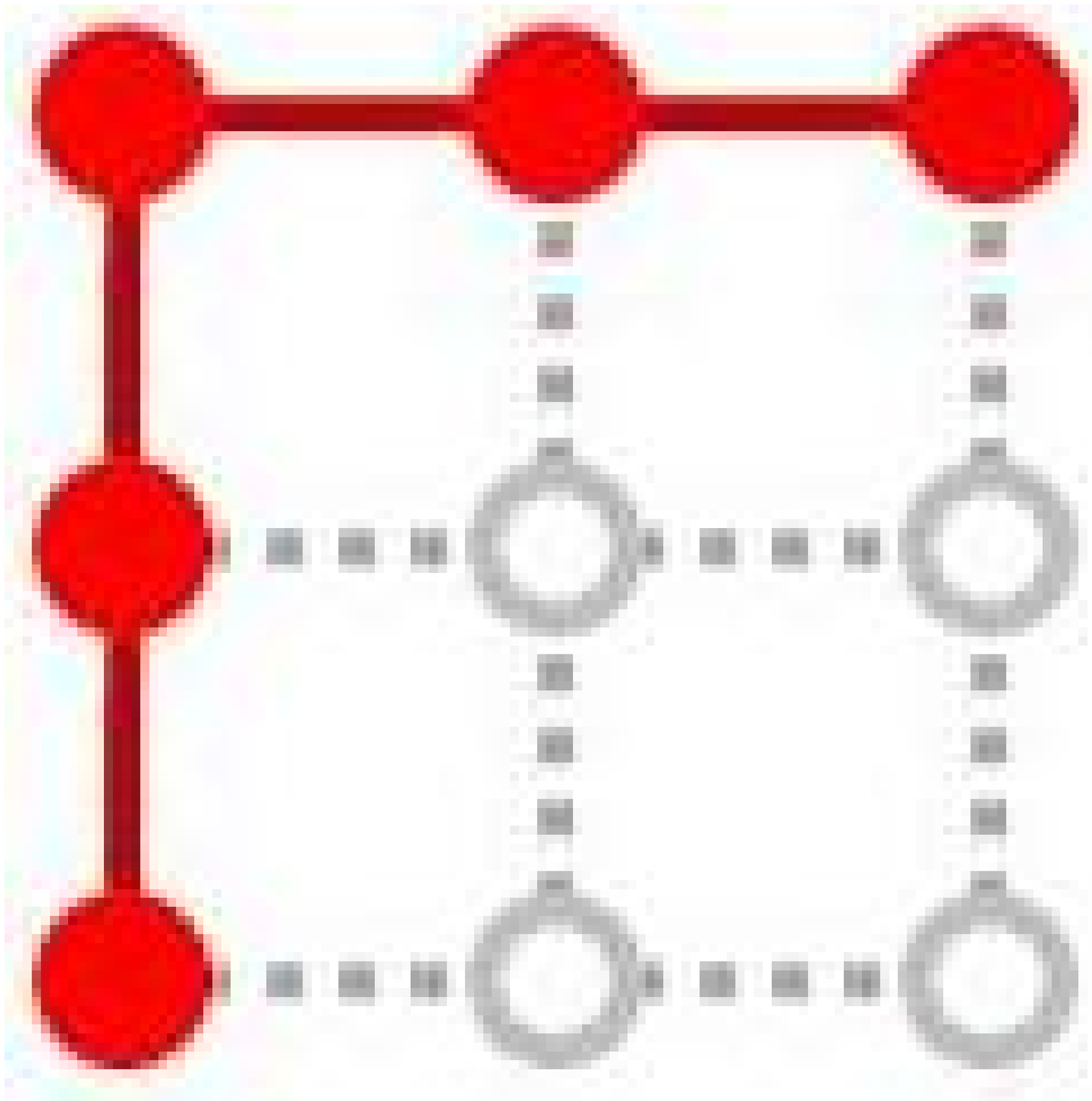}%
}%
&
\begin{array}
[c]{c}%
\raisebox{0.0069in}{\includegraphics[
height=0.2361in,
width=0.5967in
]%
{arrow-ya.ps}%
}%
\\
F_{p}^{(\ell+1)}(a)\\%
\raisebox{-0.0406in}{\includegraphics[
height=0.1436in,
width=0.1436in
]%
{icon21.ps}%
}%
^{(\ell+1)}\left(  a,p\right)
\end{array}
&
\raisebox{-0.3009in}{\includegraphics[
height=0.6789in,
width=0.6789in
]%
{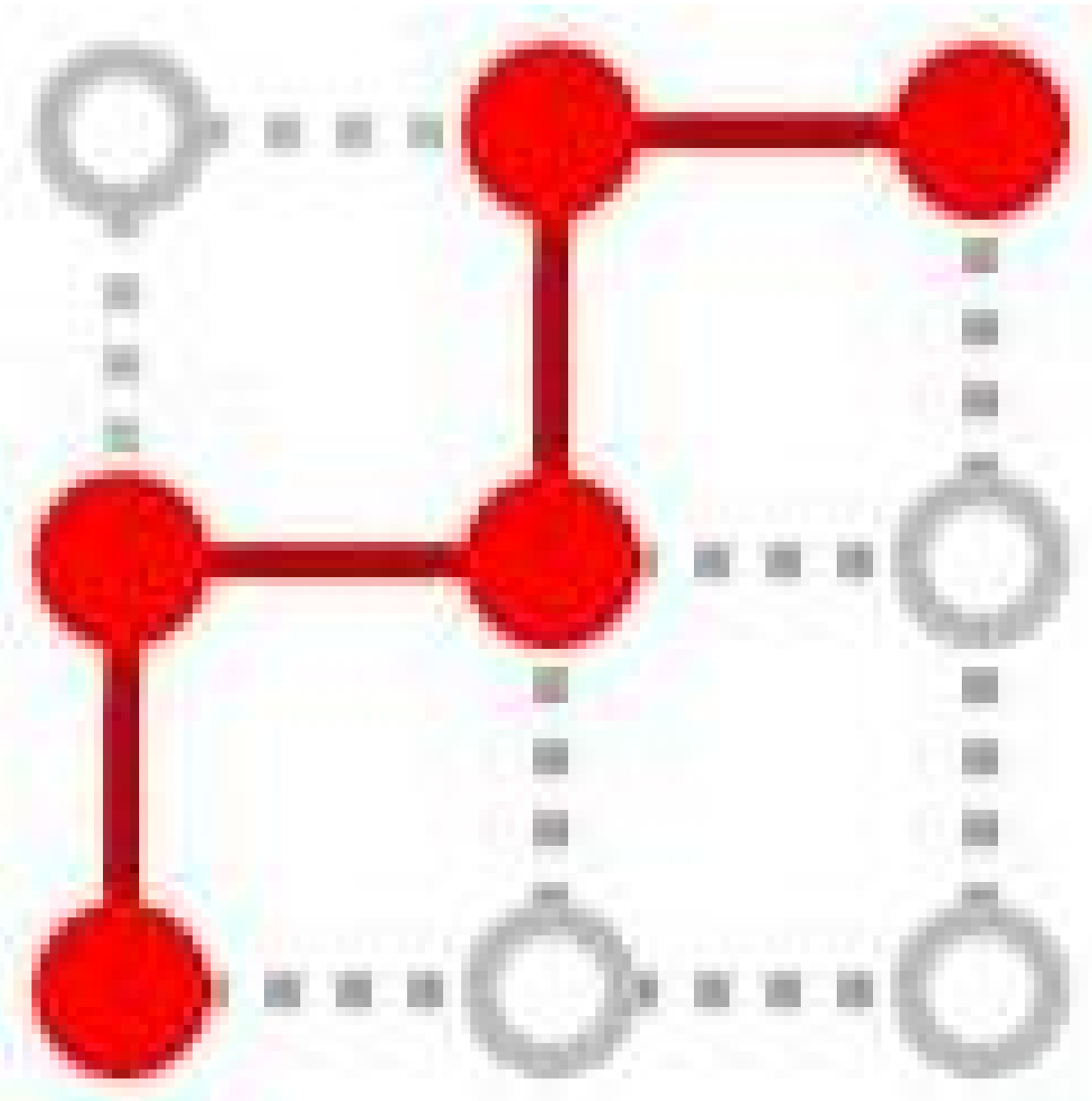}%
}%
&
\begin{array}
[c]{c}%
\raisebox{0.0069in}{\includegraphics[
height=0.2361in,
width=0.5967in
]%
{arrow-ya.ps}%
}%
\\
F_{p}^{(\ell+1)}(a^{:\left.  p\right\rceil })\\%
\raisebox{-0.0406in}{\includegraphics[
height=0.1436in,
width=0.1436in
]%
{icon21.ps}%
}%
^{(\ell+1)}\left(  a^{:\left.  p\right\rceil },p\right)
\end{array}
&
\raisebox{-0.3009in}{\includegraphics[
height=0.6789in,
width=0.6789in
]%
{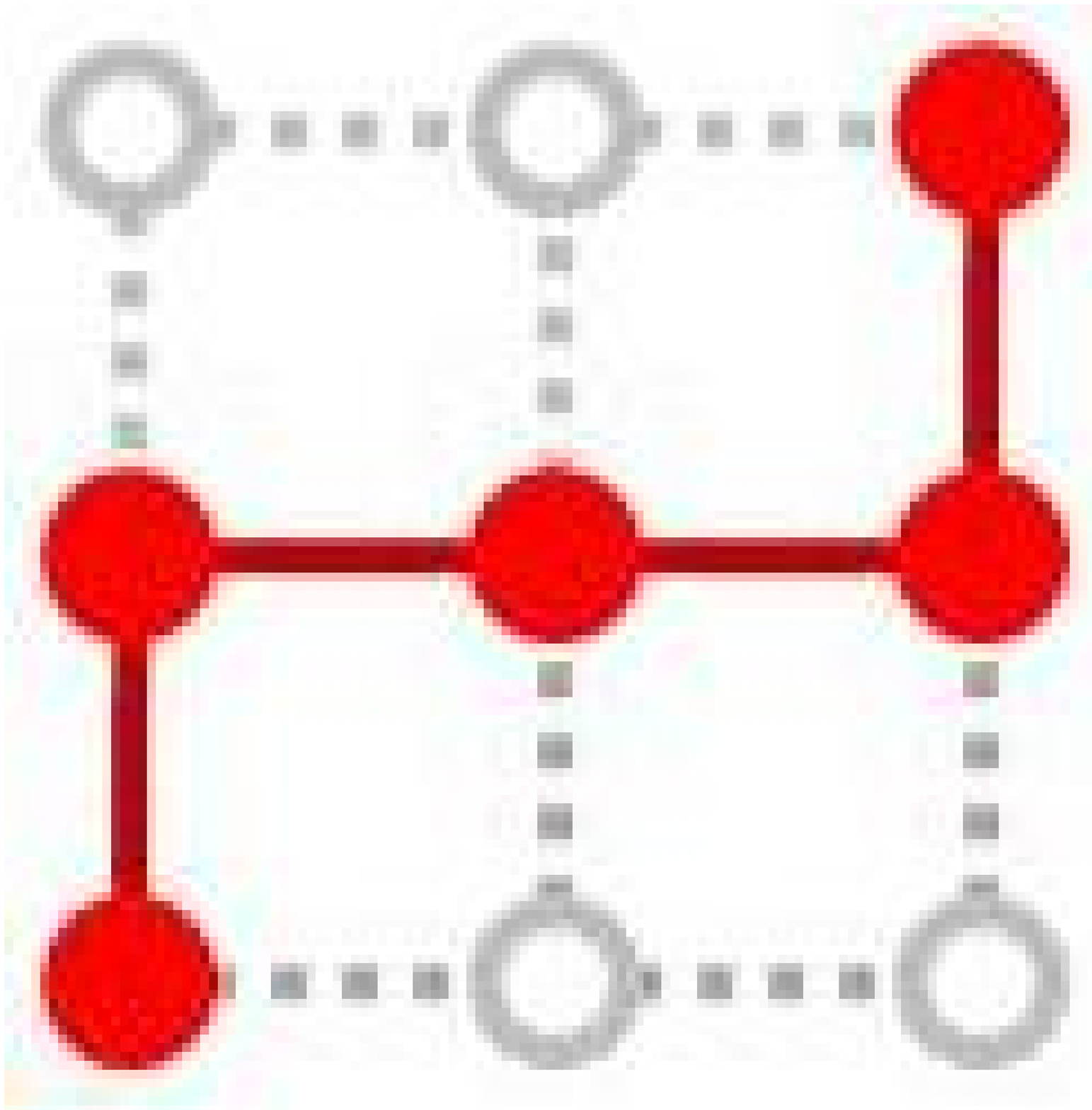}%
}%
\end{array}
}\\
\multicolumn{1}{c}{}\\
\multicolumn{1}{c}{%
\begin{array}
[c]{ccccc}%
\qquad\qquad\quad\ \  &
\begin{array}
[c]{c}%
\raisebox{0.0069in}{\includegraphics[
height=0.2361in,
width=0.5967in
]%
{arrow-ya.ps}%
}%
\\
F_{p}^{(\ell+1)}(a^{:\left\lfloor p\right.  })\\%
\raisebox{-0.0406in}{\includegraphics[
height=0.1436in,
width=0.1436in
]%
{icon21.ps}%
}%
^{(\ell+1)}\left(  a^{:\left\lfloor p\right.  },p\right)
\end{array}
&
\raisebox{-0.3009in}{\includegraphics[
height=0.6789in,
width=0.6789in
]%
{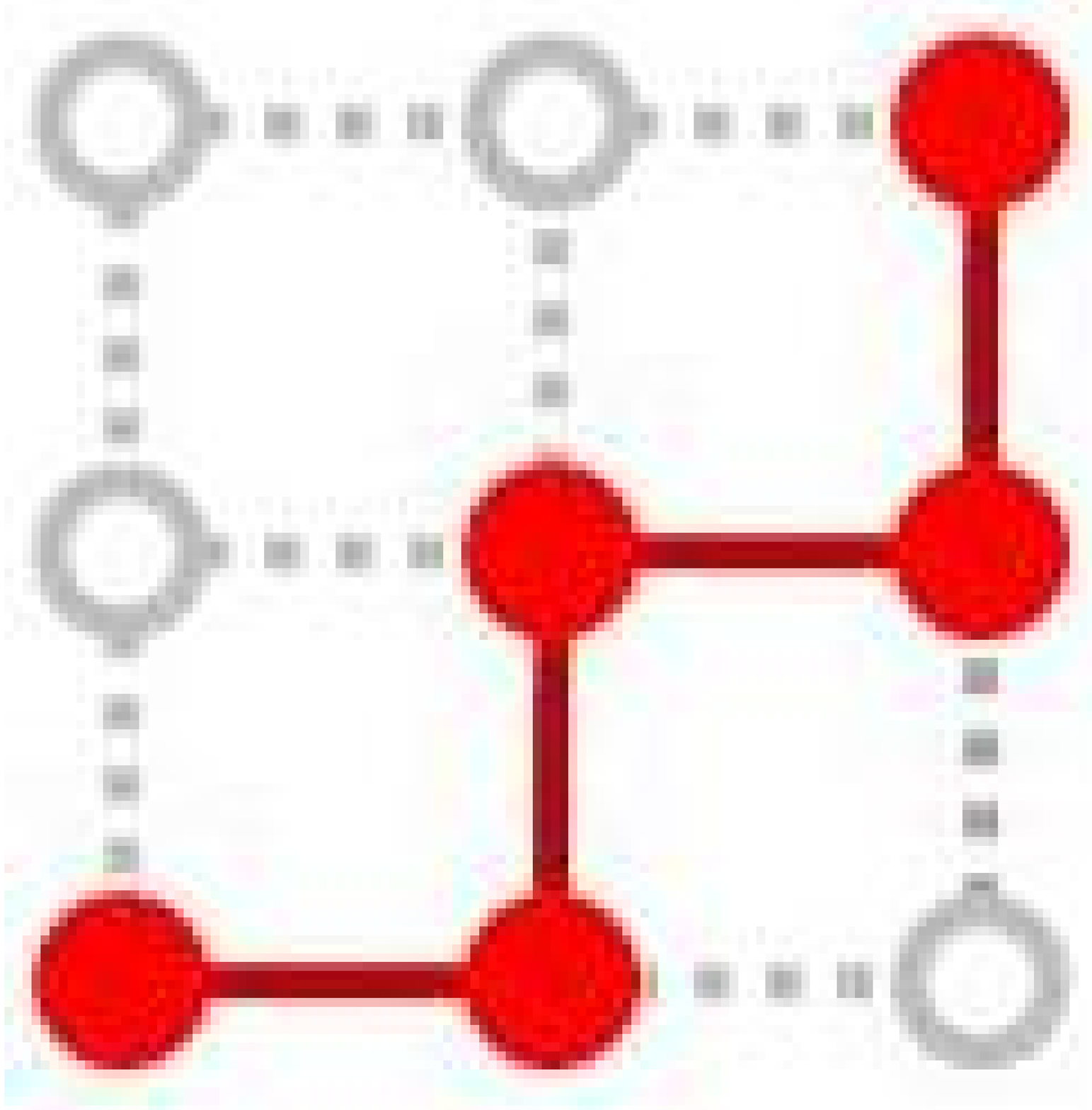}%
}%
&
\begin{array}
[c]{c}%
\raisebox{0.0069in}{\includegraphics[
height=0.2361in,
width=0.5967in
]%
{arrow-ya.ps}%
}%
\\
F_{p}^{(\ell+1)}(a^{:\left\lfloor p\right.  \left.  p\right\rceil })\\%
\raisebox{-0.0406in}{\includegraphics[
height=0.1436in,
width=0.1436in
]%
{icon21.ps}%
}%
^{(\ell+1)}\left(  a^{:\left\lfloor p\right.  \left.  p\right\rceil
},p\right)
\end{array}
&
\raisebox{-0.3009in}{\includegraphics[
height=0.6789in,
width=0.6789in
]%
{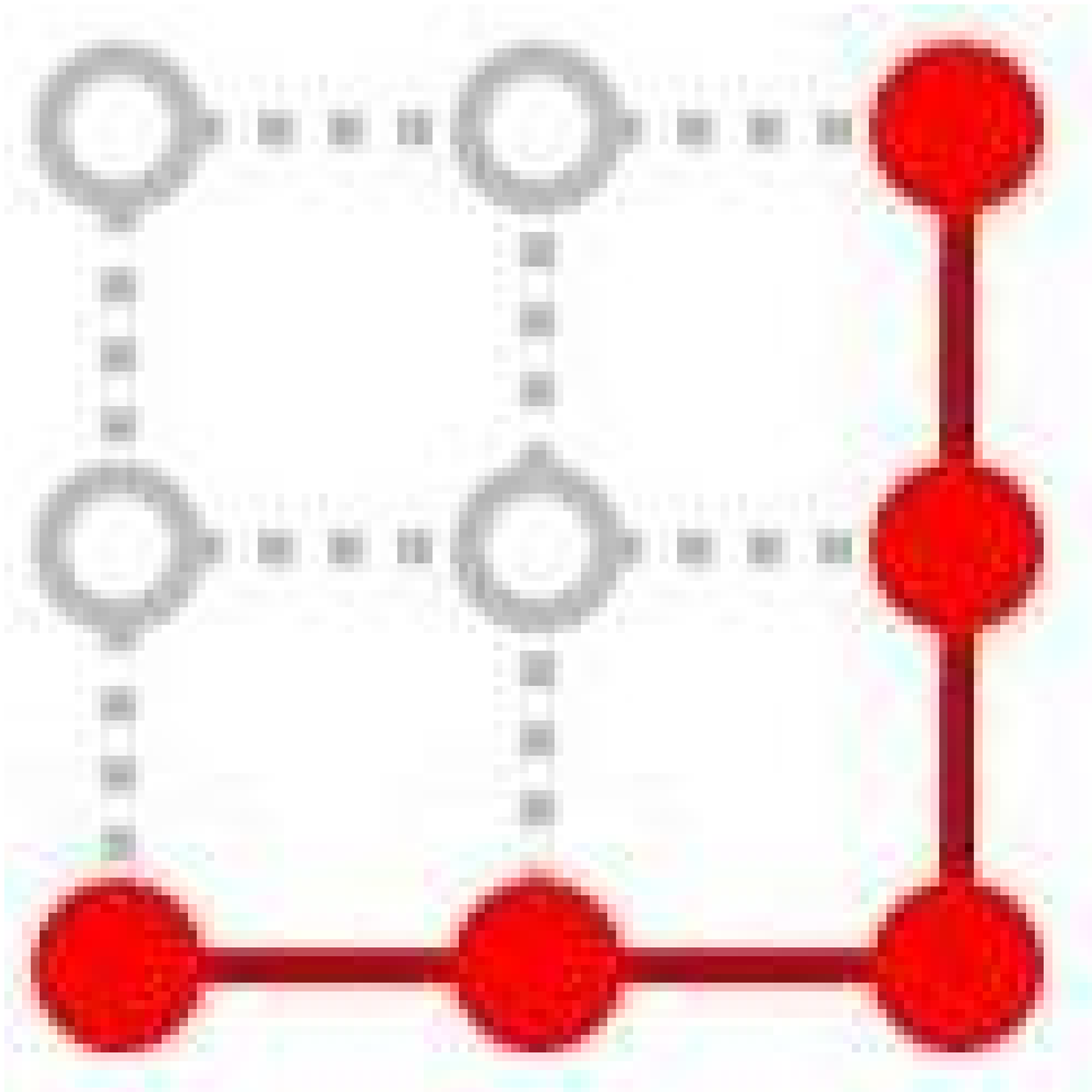}%
}%
\end{array}
}%
\end{array}
\]

\bigskip%

\[
\hspace{-0.5in}%
\begin{array}
[c]{c}%
\raisebox{-0.5518in}{\includegraphics[
height=1.1796in,
width=1.2254in
]%
{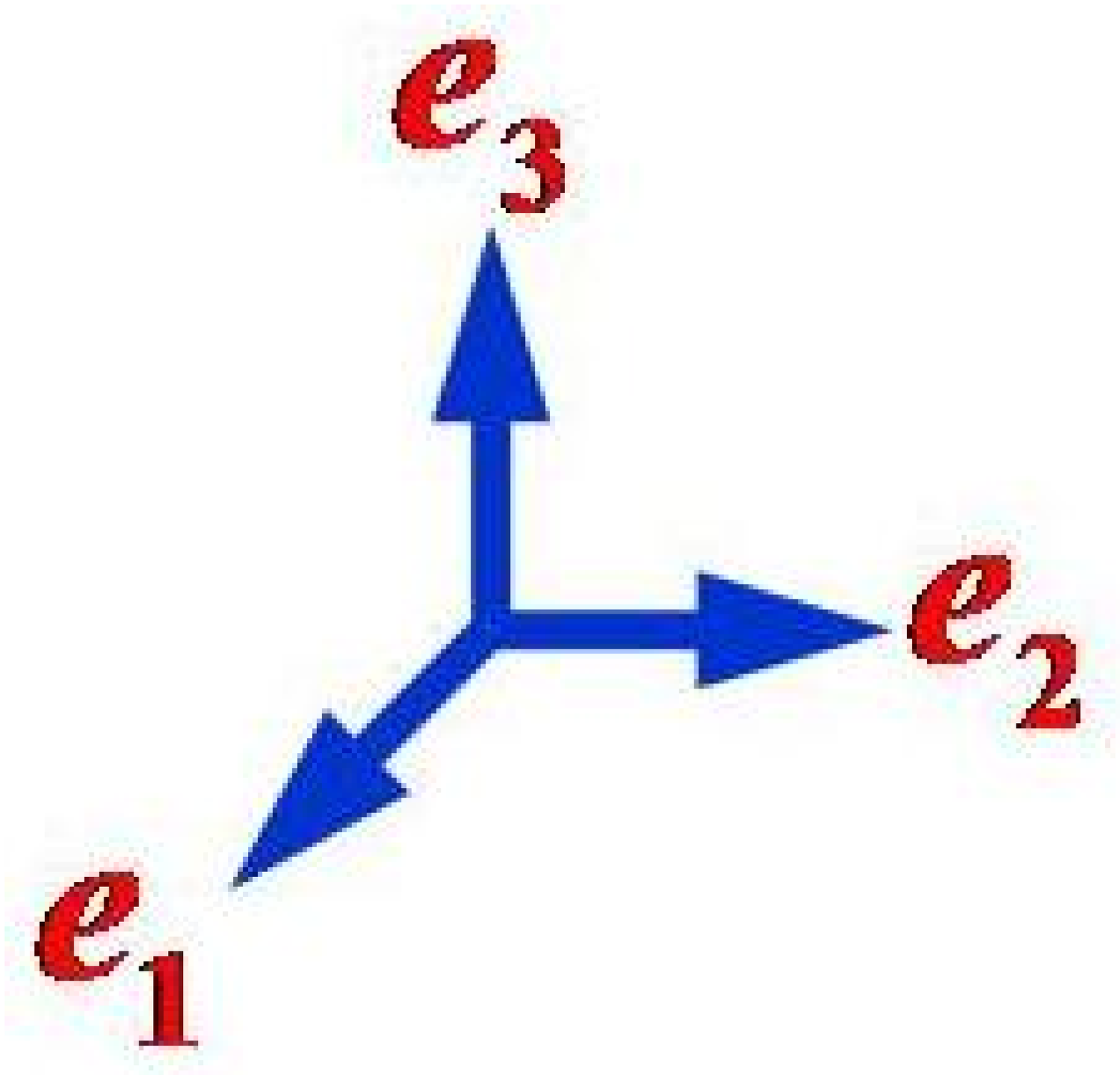}%
}%
\qquad%
\begin{array}
[c]{ccc}%
\raisebox{-0.5016in}{\includegraphics[
height=1.1372in,
width=1.1372in
]%
{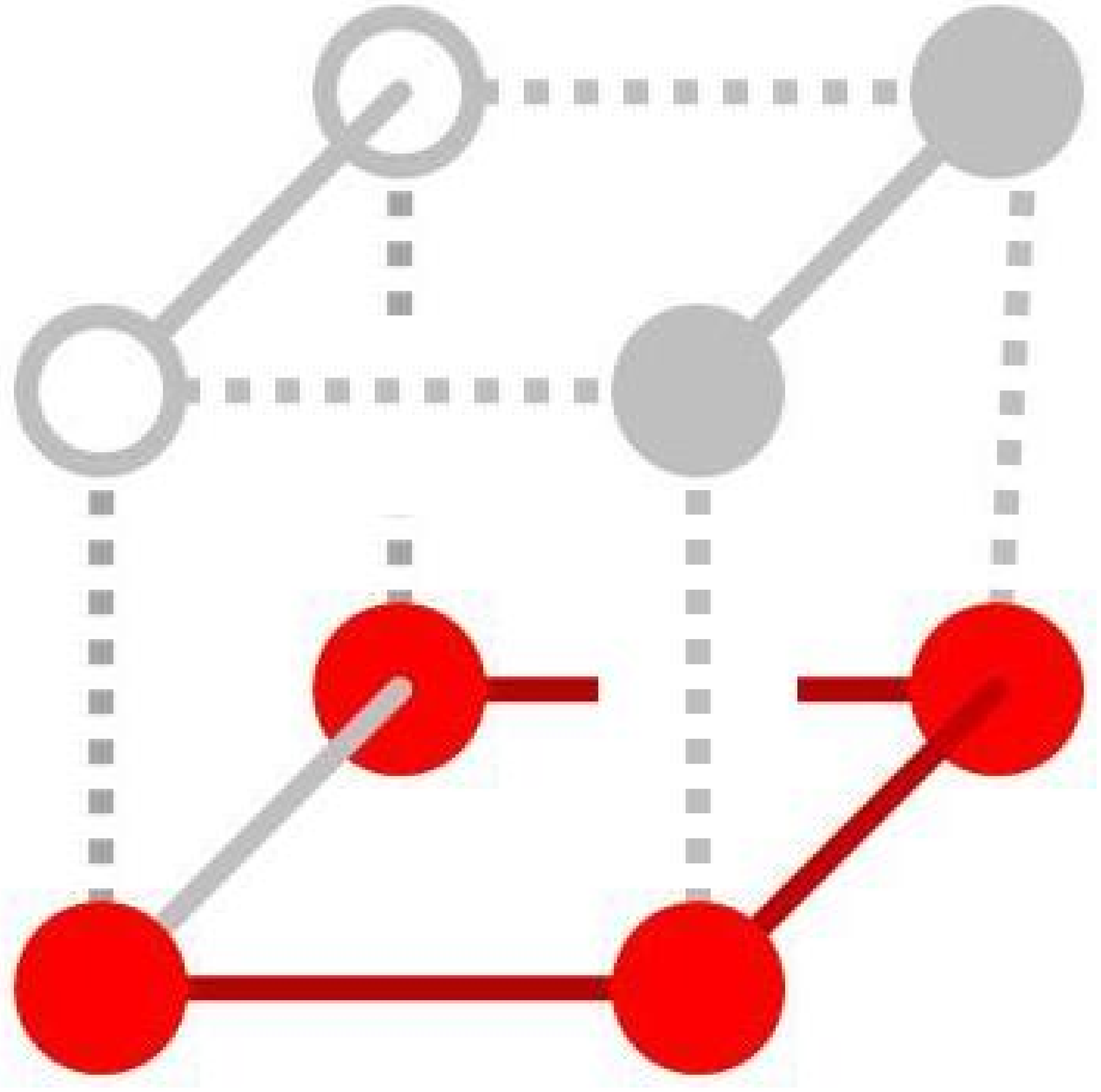}%
}%
&
\begin{array}
[c]{c}%
{\includegraphics[
height=0.237in,
width=0.5967in
]%
{arrow-ya.ps}%
}%
\\
F_{1}^{(\ell)}\left(  a\right) \\
\\%
\raisebox{-0.0406in}{\includegraphics[
height=0.1531in,
width=0.1531in
]%
{icon30.ps}%
}%
^{(\ell)}\left(  a,1\right)
\end{array}
&
\raisebox{-0.5016in}{\includegraphics[
height=1.1372in,
width=1.1372in
]%
{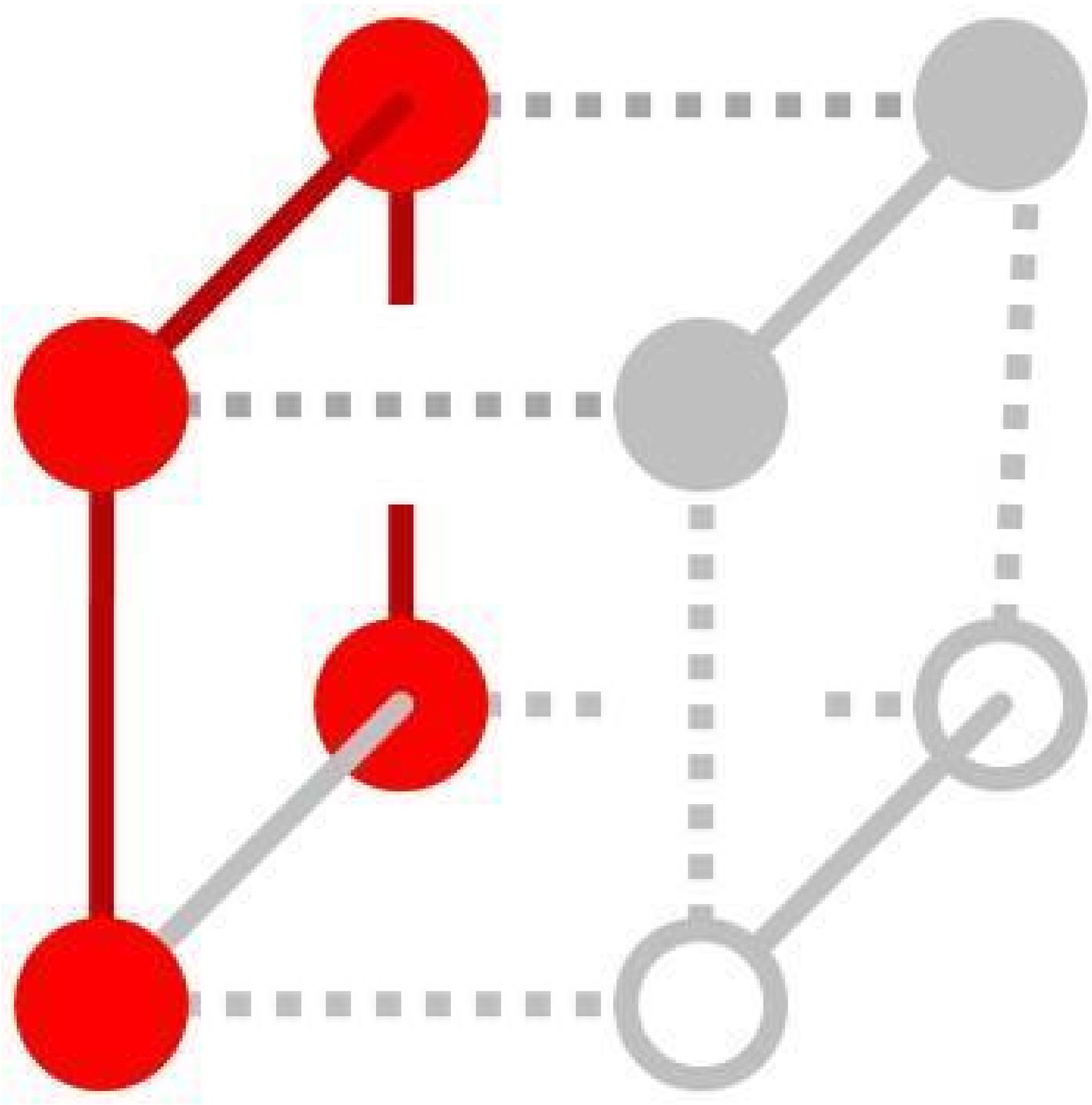}%
}%
\end{array}
\hspace{1.5in}\\
\\
\multicolumn{1}{l}{%
\begin{array}
[c]{cccc}%
\raisebox{-0.5016in}{\includegraphics[
height=1.2168in,
width=1.1147in
]%
{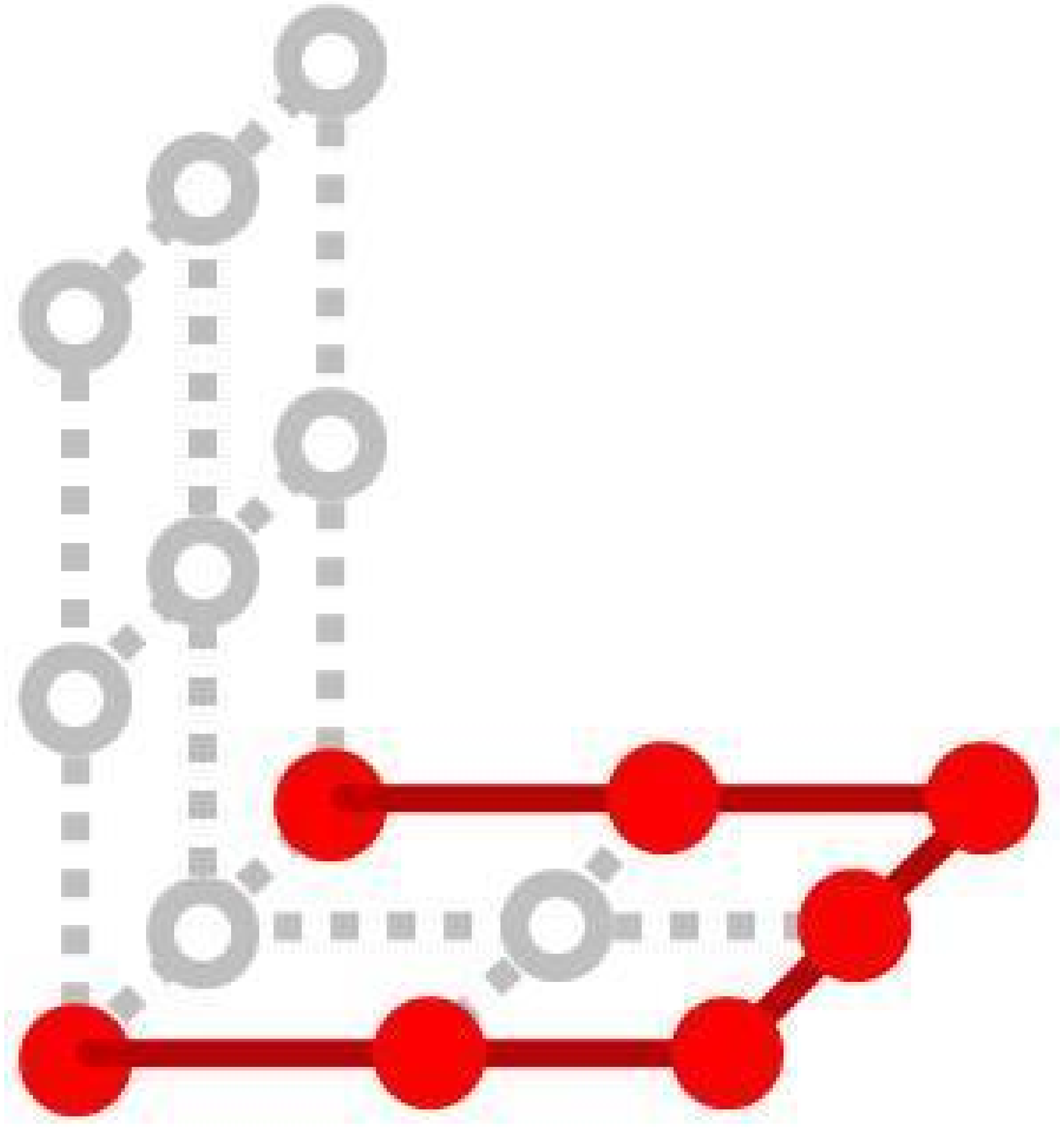}%
}%
&
\begin{array}
[c]{c}%
{\includegraphics[
height=0.237in,
width=0.5967in
]%
{arrow-ya.ps}%
}%
\\
F_{3}^{(\ell+1)}\left(  a^{:2}\right) \\
\\%
\raisebox{-0.0406in}{\includegraphics[
height=0.1436in,
width=0.1436in
]%
{icon20.ps}%
}%
^{(\ell+1)}\left(  a^{:2},3\right)
\end{array}
&
\raisebox{-0.5016in}{\includegraphics[
height=1.2159in,
width=1.0084in
]%
{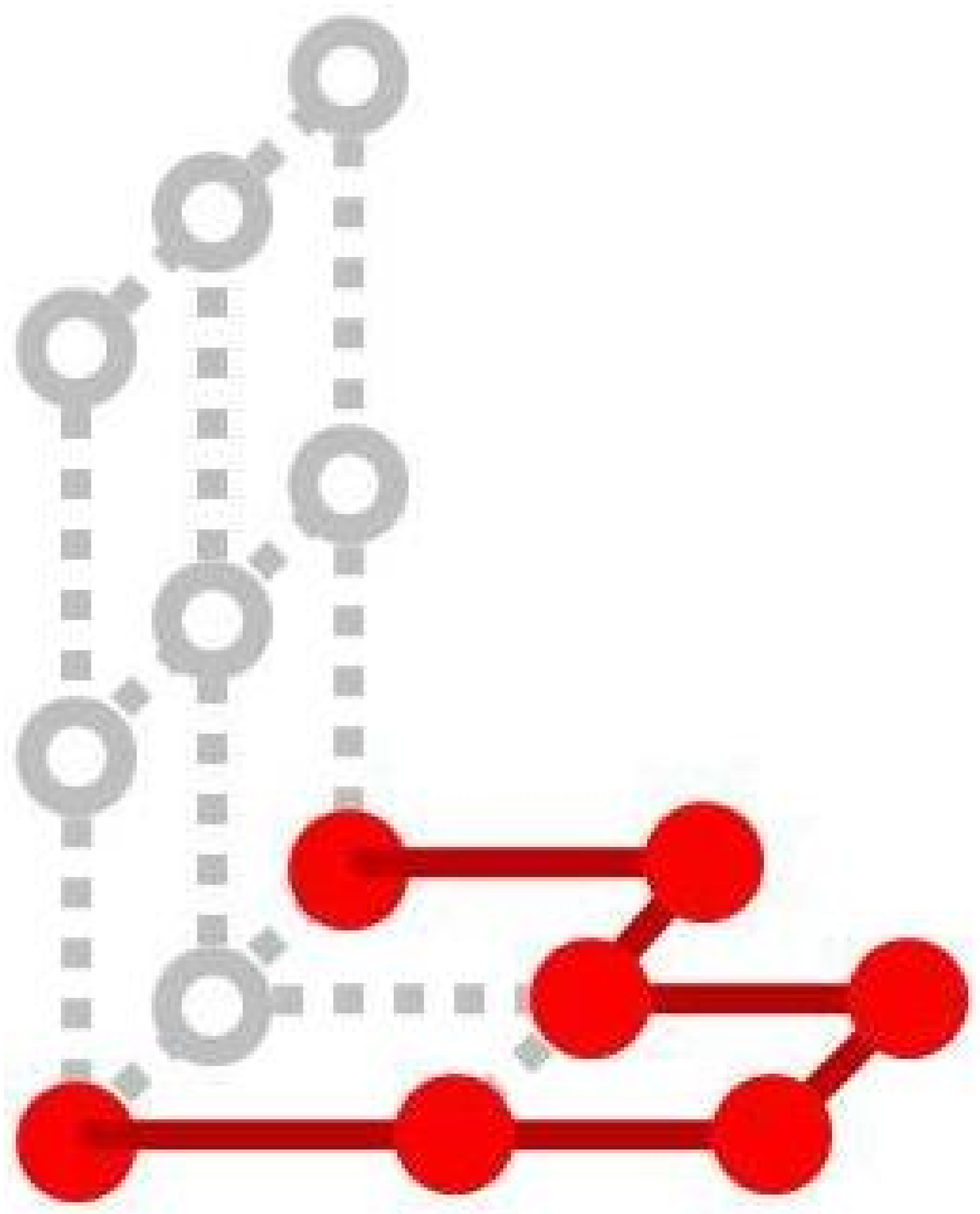}%
}%
&
\begin{array}
[c]{c}%
{\includegraphics[
height=0.237in,
width=0.5967in
]%
{arrow-ya.ps}%
}%
\\
F_{1}^{(\ell+1)}\left(  a^{:12}\right) \\
\\%
\raisebox{-0.0406in}{\includegraphics[
height=0.1531in,
width=0.1531in
]%
{icon30.ps}%
}%
^{(\ell+1)}\left(  a^{:12},1\right)
\end{array}
\\
&  &  & \\%
\raisebox{-0.5016in}{\includegraphics[
height=1.2168in,
width=0.7965in
]%
{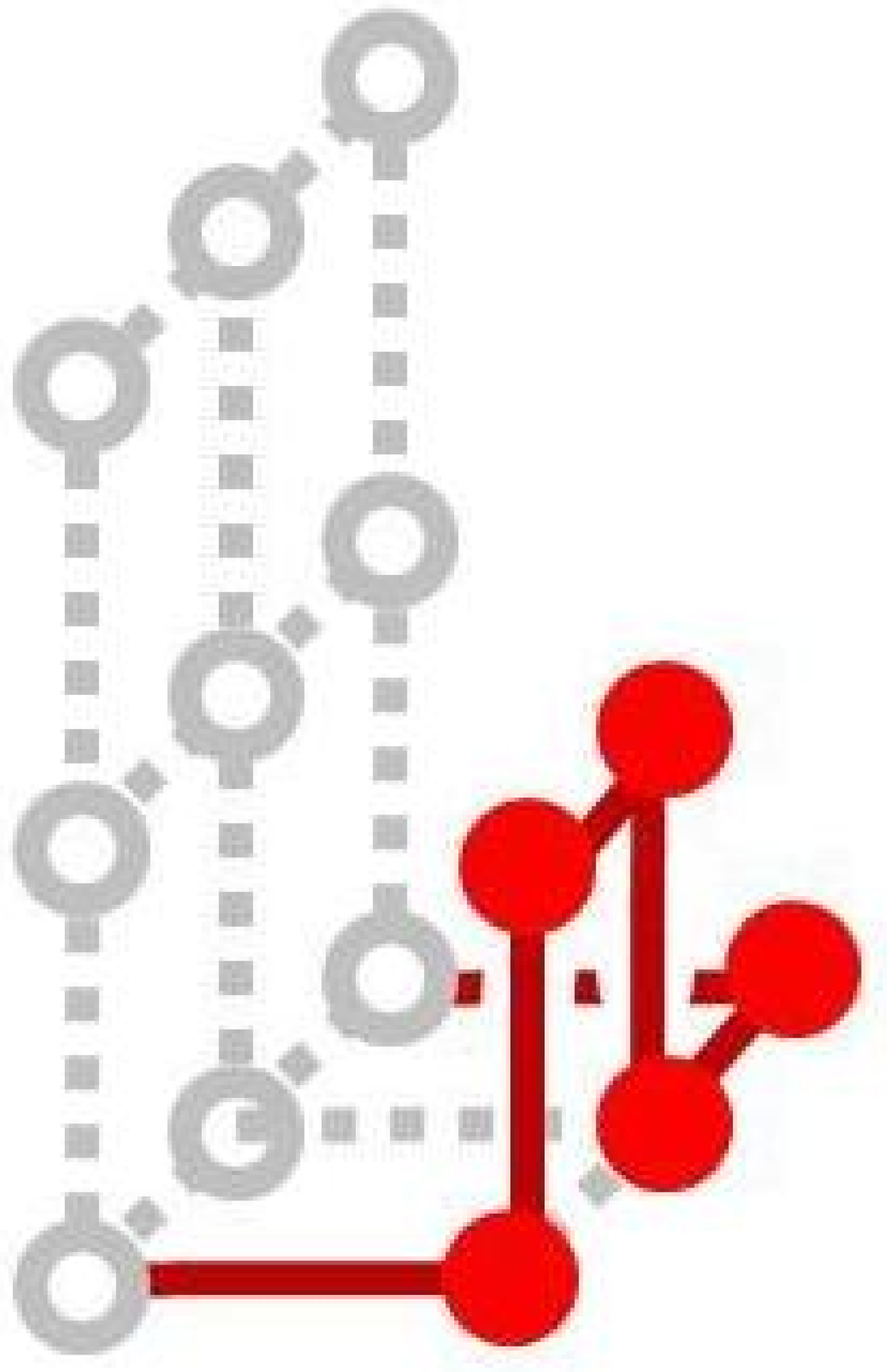}%
}%
&
\begin{array}
[c]{c}%
{\includegraphics[
height=0.237in,
width=0.5967in
]%
{arrow-ya.ps}%
}%
\\
F_{2}^{(\ell+1)}\left(  a^{:2}\right) \\
\\%
\raisebox{-0.0406in}{\includegraphics[
height=0.1436in,
width=0.1436in
]%
{icon20.ps}%
}%
^{(\ell+1)}\left(  a^{:2},2\right)
\end{array}
&
\raisebox{-0.5016in}{\includegraphics[
height=1.2159in,
width=0.902in
]%
{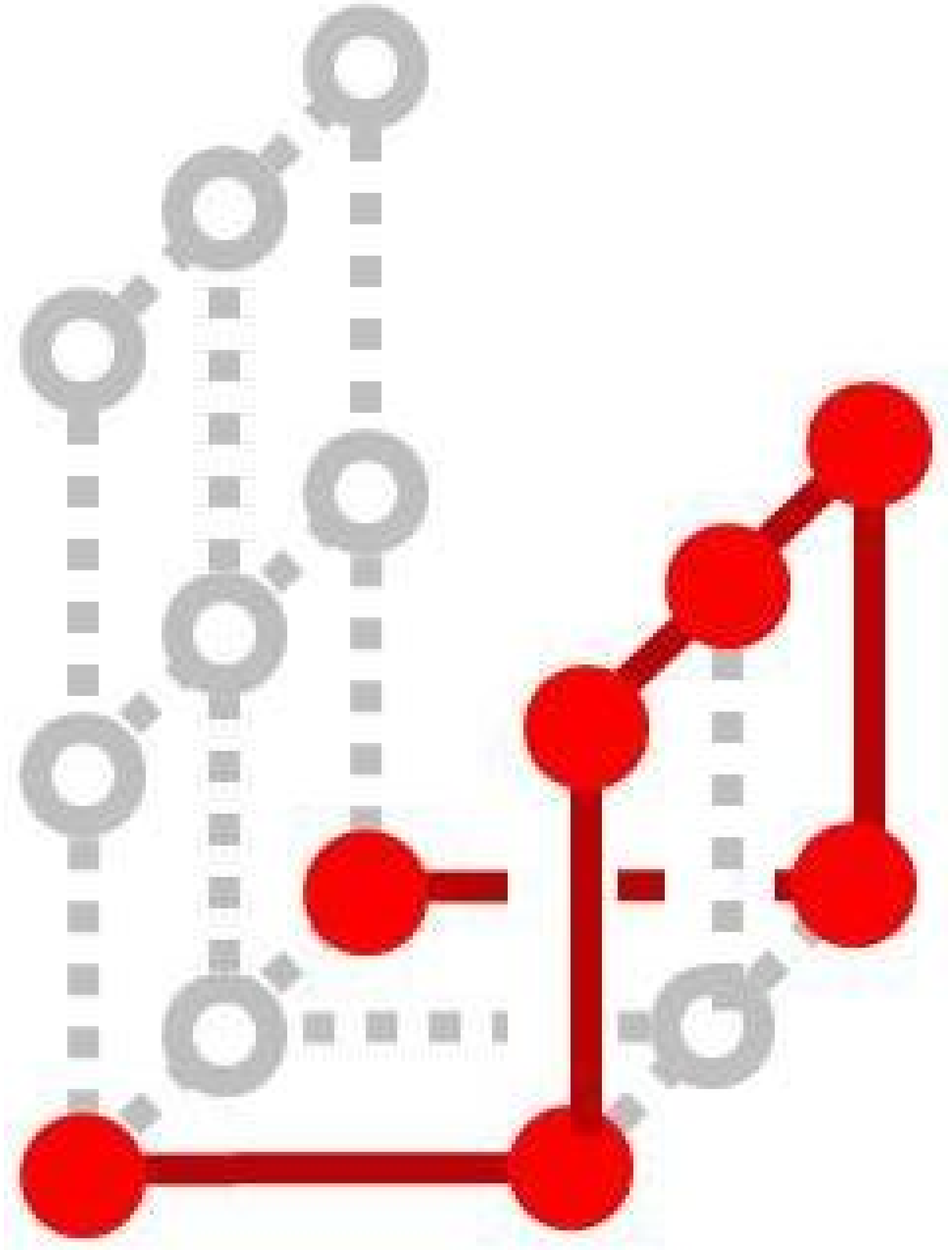}%
}%
&
\begin{array}
[c]{c}%
{\includegraphics[
height=0.237in,
width=0.5967in
]%
{arrow-ya.ps}%
}%
\\
F_{2}^{(\ell+1)}\left(  a^{:1^{2}}\right) \\
\\%
\raisebox{-0.0406in}{\includegraphics[
height=0.1436in,
width=0.1436in
]%
{icon20.ps}%
}%
^{(\ell+1)}\left(  a^{:1^{2}},1\right)
\end{array}
\\
&  &  & \\%
\raisebox{-0.5016in}{\includegraphics[
height=1.2159in,
width=0.902in
]%
{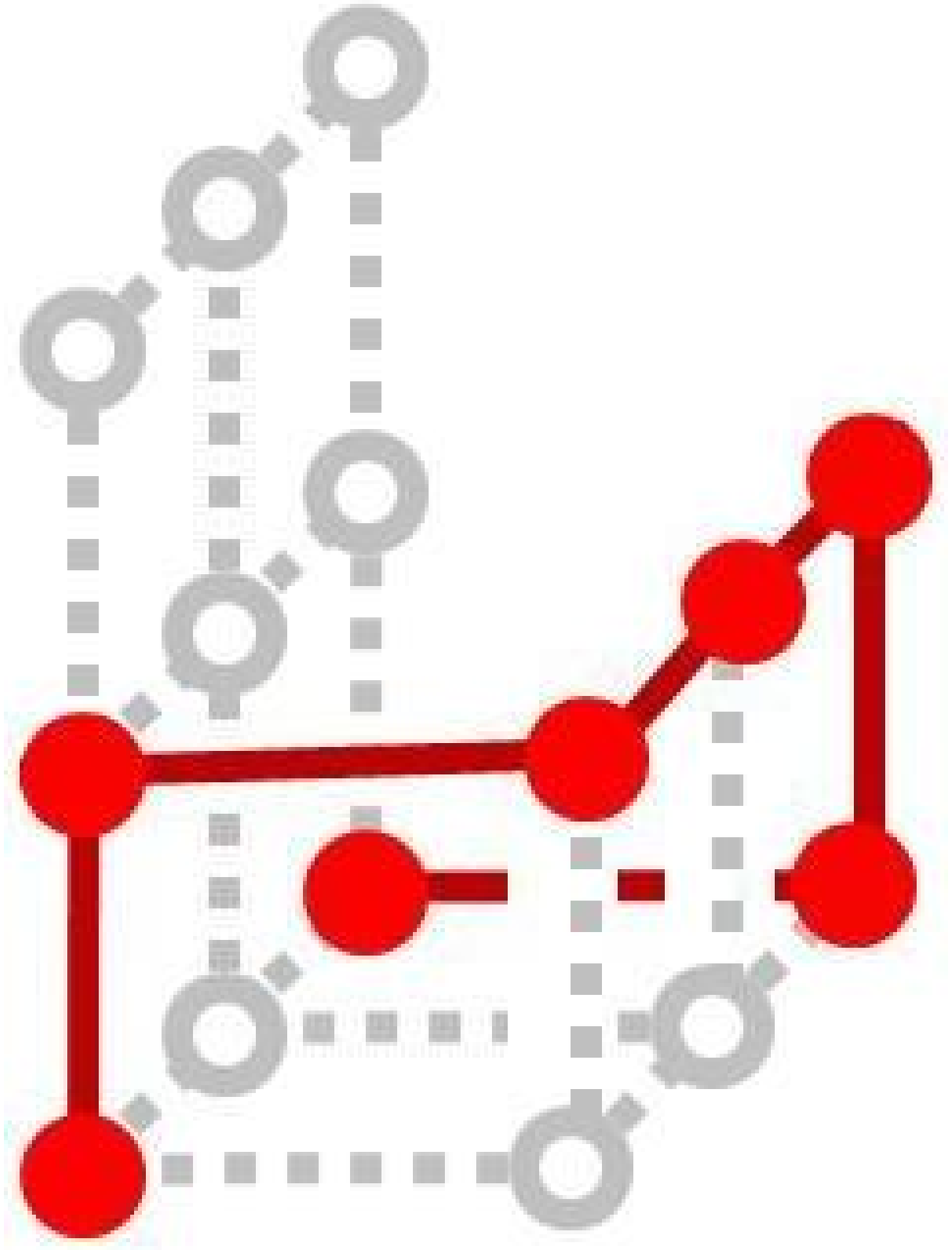}%
}%
&
\begin{array}
[c]{c}%
{\includegraphics[
height=0.237in,
width=0.5967in
]%
{arrow-ya.ps}%
}%
\\
F_{1}^{(\ell+1)}\left(  a\right) \\
\\%
\raisebox{-0.0406in}{\includegraphics[
height=0.1436in,
width=0.1436in
]%
{icon21.ps}%
}%
^{(\ell+1)}\left(  a,1\right)
\end{array}
&
\raisebox{-0.5016in}{\includegraphics[
height=1.2159in,
width=0.902in
]%
{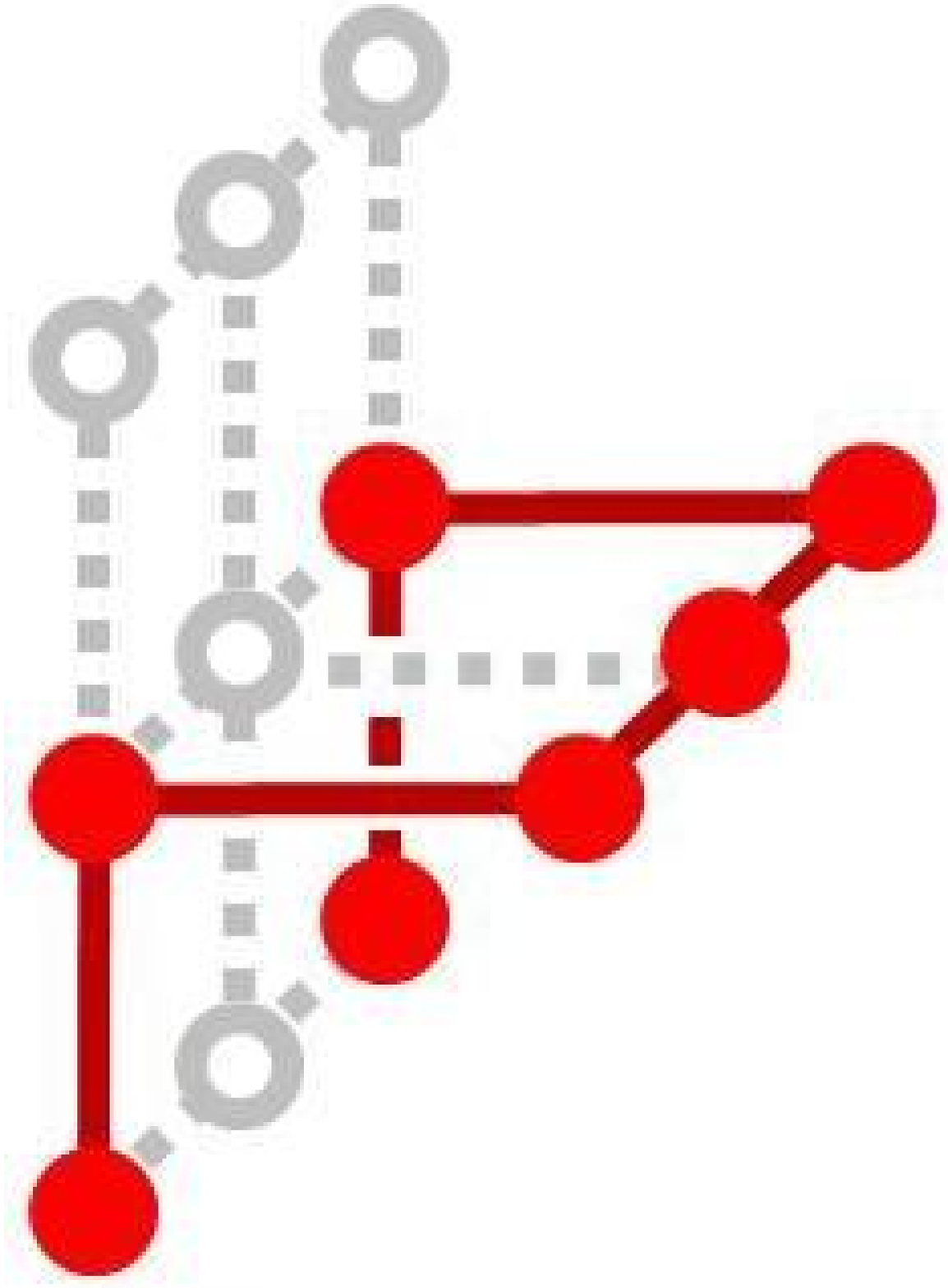}%
}%
&
\begin{array}
[c]{c}%
{\includegraphics[
height=0.237in,
width=0.5967in
]%
{arrow-ya.ps}%
}%
\\
F_{3}^{(\ell+1)}\left(  a^{:3}\right) \\
\\%
\raisebox{-0.0406in}{\includegraphics[
height=0.1436in,
width=0.1436in
]%
{icon21.ps}%
}%
^{(\ell+1)}\left(  a^{:3},3\right)
\end{array}
\\
&  &  & \\%
\raisebox{-0.5016in}{\includegraphics[
height=1.2159in,
width=0.902in
]%
{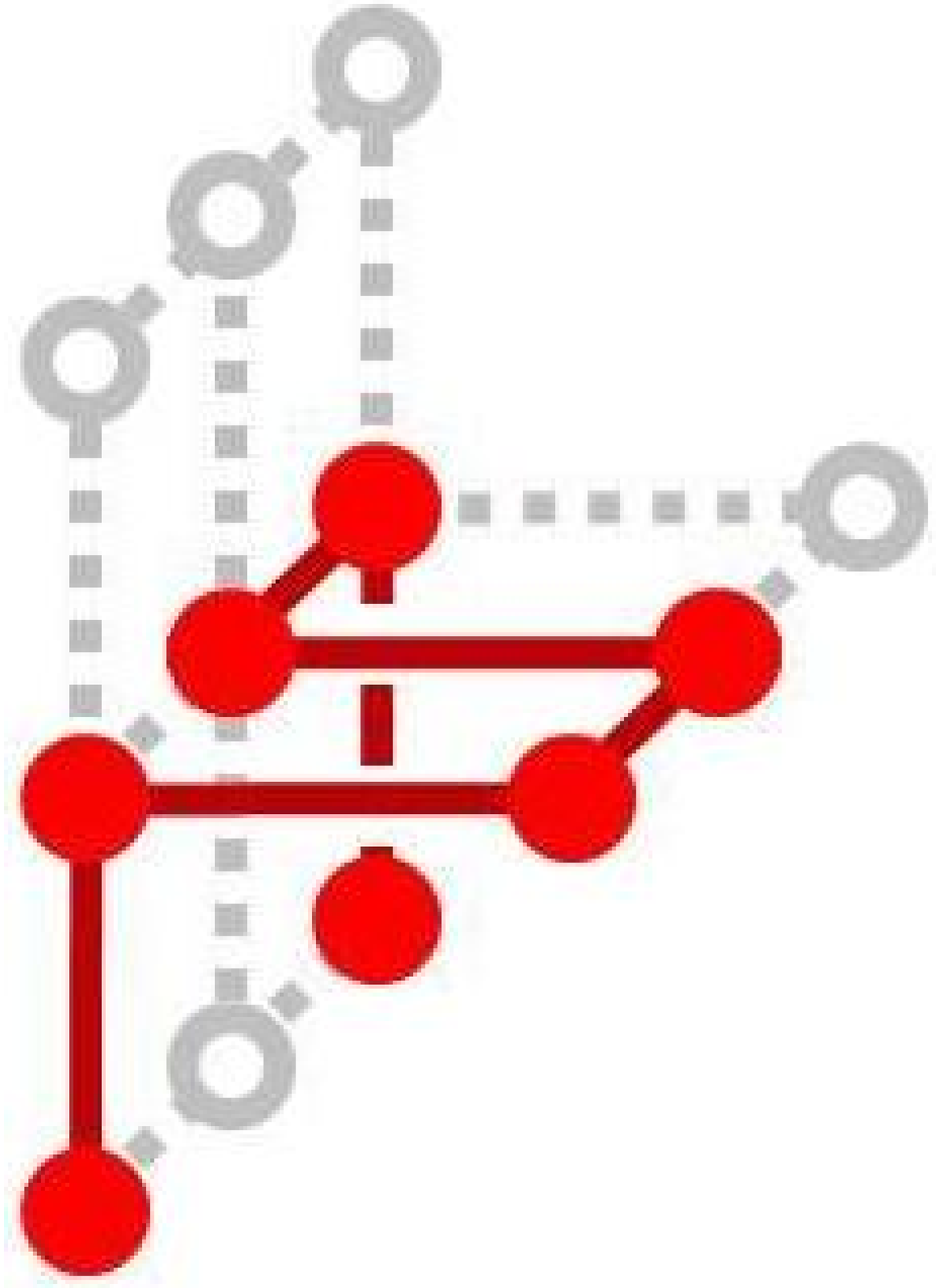}%
}%
&
\begin{array}
[c]{c}%
{\includegraphics[
height=0.237in,
width=0.5967in
]%
{arrow-ya.ps}%
}%
\\
F^{(\ell+1)}\left(  a^{:13}\right) \\
\\%
\raisebox{-0.0406in}{\includegraphics[
height=0.1531in,
width=0.1531in
]%
{icon30.ps}%
}%
^{(\ell+1)}\left(  a^{:13},1\right)
\end{array}
&
\raisebox{-0.5016in}{\includegraphics[
height=1.2168in,
width=0.448in
]%
{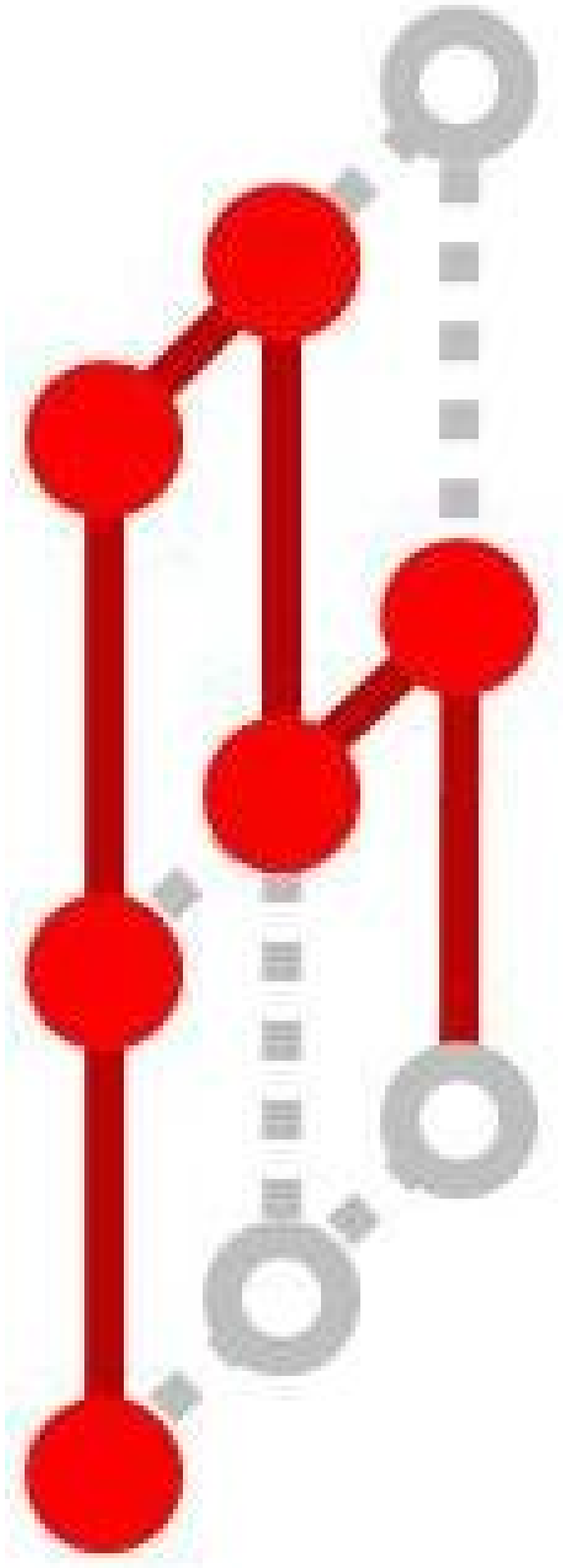}%
}%
\qquad &
\begin{array}
[c]{c}%
{\includegraphics[
height=0.237in,
width=0.5967in
]%
{arrow-ya.ps}%
}%
\\
F_{2}^{(\ell+1)}\left(  a^{:3}\right) \\
\\%
\raisebox{-0.0406in}{\includegraphics[
height=0.1436in,
width=0.1436in
]%
{icon20.ps}%
}%
^{(\ell+1)}\left(  a^{:3},2\right)
\end{array}
\\
&  &  & \\%
\raisebox{-0.5016in}{\includegraphics[
height=1.2168in,
width=0.448in
]%
{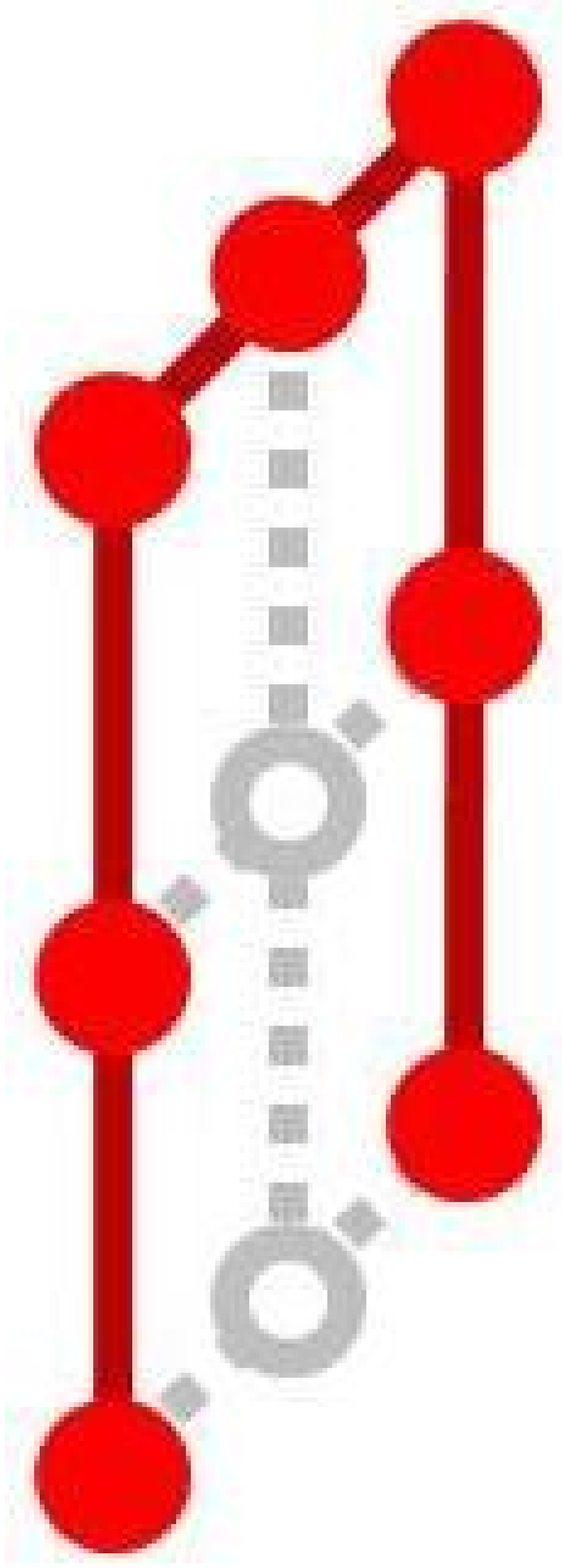}%
}%
\qquad &  &  &
\end{array}
}%
\end{array}
\]
\bigskip

The definition suggested by the above figures is:

\bigskip\bigskip

$\hspace{-0.75in}\left.
\raisebox{-0.0303in}{\includegraphics[
height=0.1332in,
width=0.1193in
]%
{red-refinement.ps}%
}%
\left(
\raisebox{-0.0406in}{\includegraphics[
height=0.1436in,
width=0.1436in
]%
{icon10.ps}%
}%
^{(\ell)}\left(  a,p\right)  \right)  \right\vert _{%
\raisebox{-0.0303in}{\includegraphics[
height=0.1332in,
width=0.1193in
]%
{red-refinement.ps}%
}%
\left(  \mathbb{K}^{(\ell)}\right)  }=\left.
\raisebox{-0.0406in}{\includegraphics[
height=0.1436in,
width=0.1436in
]%
{icon11.ps}%
}%
^{(\ell+1)}\left(  a^{:\left\lfloor p\right.  },p\right)  \cdot%
\raisebox{-0.0406in}{\includegraphics[
height=0.1436in,
width=0.1436in
]%
{icon13.ps}%
}%
^{(\ell+1)}\left(  a^{:\left.  p\right\rceil \left\lfloor p\right.
},p\right)  \cdot%
\raisebox{-0.0406in}{\includegraphics[
height=0.1436in,
width=0.1436in
]%
{icon10.ps}%
}%
^{(\ell+1)}\left(  a^{:\left.  p\right\rceil },p\right)  \cdot%
\raisebox{-0.0406in}{\includegraphics[
height=0.1436in,
width=0.1436in
]%
{icon10.ps}%
}%
^{(\ell+1)}\left(  a,p\right)  \right\vert _{%
\raisebox{-0.0303in}{\includegraphics[
height=0.1332in,
width=0.1193in
]%
{red-refinement.ps}%
}%
\left(  \mathbb{K}^{(\ell)}\right)  }$\bigskip\bigskip

$\hspace{-0.75in}\left.
\raisebox{-0.0303in}{\includegraphics[
height=0.1332in,
width=0.1193in
]%
{red-refinement.ps}%
}%
\left(
\raisebox{-0.0406in}{\includegraphics[
height=0.1436in,
width=0.1436in
]%
{icon21.ps}%
}%
^{(\ell)}\left(  a,p\right)  \right)  \right\vert _{%
\raisebox{-0.0303in}{\includegraphics[
height=0.1332in,
width=0.1193in
]%
{red-refinement.ps}%
}%
\left(  \mathbb{K}^{(\ell)}\right)  }=\left.
\raisebox{-0.0406in}{\includegraphics[
height=0.1436in,
width=0.1436in
]%
{icon21.ps}%
}%
^{(\ell+1)}\left(  a^{:\left.  p\right\rceil \left\lfloor p\right.
},p\right)  \cdot%
\raisebox{-0.0406in}{\includegraphics[
height=0.1436in,
width=0.1436in
]%
{icon21.ps}%
}%
^{(\ell+1)}\left(  a^{\left\lfloor p\right.  },p\right)  \cdot%
\raisebox{-0.0406in}{\includegraphics[
height=0.1436in,
width=0.1436in
]%
{icon21.ps}%
}%
^{(\ell+1)}\left(  a^{:\left.  p\right\rceil },p\right)  \cdot%
\raisebox{-0.0406in}{\includegraphics[
height=0.1436in,
width=0.1436in
]%
{icon21.ps}%
}%
^{(\ell+1)}\left(  a,p\right)  \right\vert _{%
\raisebox{-0.0303in}{\includegraphics[
height=0.1332in,
width=0.1193in
]%
{red-refinement.ps}%
}%
\left(  \mathbb{K}^{(\ell)}\right)  }$\bigskip\bigskip

$\hspace{-0.75in}\left.
\raisebox{-0.0303in}{\includegraphics[
height=0.1332in,
width=0.1193in
]%
{red-refinement.ps}%
}%
\left(
\raisebox{-0.0406in}{\includegraphics[
height=0.1531in,
width=0.1531in
]%
{icon30.ps}%
}%
^{(\ell)}\left(  a,1\right)  \right)  \right\vert _{%
\raisebox{-0.0303in}{\includegraphics[
height=0.1332in,
width=0.1193in
]%
{red-refinement.ps}%
}%
\left(  \mathbb{K}^{(\ell)}\right)  }=%
\raisebox{-0.0406in}{\includegraphics[
height=0.1436in,
width=0.1436in
]%
{icon20.ps}%
}%
^{(\ell+1)}\left(  a^{:3},2\right)  \cdot%
\raisebox{-0.0406in}{\includegraphics[
height=0.1531in,
width=0.1531in
]%
{icon30.ps}%
}%
^{(\ell+1)}\left(  a^{:13},1\right)  \cdot%
\raisebox{-0.0406in}{\includegraphics[
height=0.1436in,
width=0.1436in
]%
{icon21.ps}%
}%
^{(\ell+1)}\left(  a^{:3},3\right)  \cdot%
\raisebox{-0.0406in}{\includegraphics[
height=0.1436in,
width=0.1436in
]%
{icon21.ps}%
}%
^{(\ell+1)}\left(  a,1\right)  \bigskip$\bigskip

$\hspace{0.25in}\left.  \cdot%
\raisebox{-0.0406in}{\includegraphics[
height=0.1436in,
width=0.1436in
]%
{icon20.ps}%
}%
^{(\ell+1)}\left(  a^{:1^{2}},1\right)  \cdot%
\raisebox{-0.0406in}{\includegraphics[
height=0.1436in,
width=0.1436in
]%
{icon20.ps}%
}%
^{(\ell+1)}\left(  a^{:2},2\right)  \cdot%
\raisebox{-0.0406in}{\includegraphics[
height=0.1531in,
width=0.1531in
]%
{icon30.ps}%
}%
^{(\ell+1)}\left(  a^{:12},1\right)  \cdot%
\raisebox{-0.0406in}{\includegraphics[
height=0.1436in,
width=0.1436in
]%
{icon20.ps}%
}%
^{(\ell+1)}\left(  a^{:2},3\right)  \right\vert _{%
\raisebox{-0.0303in}{\includegraphics[
height=0.1332in,
width=0.1193in
]%
{red-refinement.ps}%
}%
\left(  \mathbb{K}^{(\ell)}\right)  }$

\bigskip\bigskip

\noindent where $\left.
\raisebox{-0.0303in}{\includegraphics[
height=0.1332in,
width=0.1193in
]%
{red-refinement.ps}%
}%
\left(  L_{\ast}^{(\ell)}\left(  \ast,\ast,\ast\right)  \right)  \right\vert
_{%
\raisebox{-0.0303in}{\includegraphics[
height=0.1332in,
width=0.1193in
]%
{red-refinement.ps}%
}%
\left(  \mathbb{K}^{(\ell)}\right)  }$ denotes a map $%
\raisebox{-0.0303in}{\includegraphics[
height=0.1332in,
width=0.1193in
]%
{red-refinement.ps}%
}%
\left(  \mathbb{K}^{(\ell)}\right)  \longrightarrow%
\raisebox{-0.0303in}{\includegraphics[
height=0.1332in,
width=0.1193in
]%
{red-refinement.ps}%
}%
\left(  \mathbb{K}^{(\ell)}\right)  $ from $%
\raisebox{-0.0303in}{\includegraphics[
height=0.1332in,
width=0.1193in
]%
{red-refinement.ps}%
}%
\left(  \mathbb{K}^{(\ell)}\right)  $ into itself.

\bigskip

We seek to construct a morphism $%
\raisebox{-0.0303in}{\includegraphics[
height=0.1332in,
width=0.1193in
]%
{red-refinement.ps}%
}%
:\mathbb{K}^{(\ell+1)}\longrightarrow\mathbb{K}^{(\ell+1)}$. \ Unfortunately,
if for example we make the most straight forward definition by extending
\[
\left.
\raisebox{-0.0303in}{\includegraphics[
height=0.1332in,
width=0.1193in
]%
{red-refinement.ps}%
}%
\left(
\raisebox{-0.0406in}{\includegraphics[
height=0.1436in,
width=0.1436in
]%
{icon10.ps}%
}%
^{(\ell)}\left(  a,p\right)  \right)  \right\vert _{%
\raisebox{-0.0303in}{\includegraphics[
height=0.1332in,
width=0.1193in
]%
{red-refinement.ps}%
}%
\left(  \mathbb{K}^{(\ell)}\right)  }:%
\raisebox{-0.0303in}{\includegraphics[
height=0.1332in,
width=0.1193in
]%
{red-refinement.ps}%
}%
\left(  \mathbb{K}^{(\ell)}\right)  \longrightarrow%
\raisebox{-0.0303in}{\includegraphics[
height=0.1332in,
width=0.1193in
]%
{red-refinement.ps}%
}%
\left(  \mathbb{K}^{(\ell)}\right)
\]
to%
\[%
\raisebox{-0.0303in}{\includegraphics[
height=0.1332in,
width=0.1193in
]%
{red-refinement.ps}%
}%
\left(  \left(
\raisebox{-0.0406in}{\includegraphics[
height=0.1436in,
width=0.1436in
]%
{icon10.ps}%
}%
^{(\ell)}\left(  a,p\right)  \right)  \right)  :\mathbb{K}^{(\ell
+1)}\longrightarrow\mathbb{K}^{(\ell+1)}%
\]
by defining%

\[%
\raisebox{-0.0303in}{\includegraphics[
height=0.1332in,
width=0.1193in
]%
{red-refinement.ps}%
}%
\left(
\raisebox{-0.0406in}{\includegraphics[
height=0.1436in,
width=0.1436in
]%
{icon10.ps}%
}%
^{(\ell)}\left(  a,p\right)  \right)  =%
\raisebox{-0.0406in}{\includegraphics[
height=0.1436in,
width=0.1436in
]%
{icon11.ps}%
}%
^{(\ell+1)}\left(  a^{:\left\lfloor p\right.  },p\right)  \cdot%
\raisebox{-0.0406in}{\includegraphics[
height=0.1436in,
width=0.1436in
]%
{icon13.ps}%
}%
^{(\ell+1)}\left(  a^{:\left.  p\right\rceil \left\lfloor p\right.
},p\right)  \cdot%
\raisebox{-0.0406in}{\includegraphics[
height=0.1436in,
width=0.1436in
]%
{icon10.ps}%
}%
^{(\ell+1)}\left(  a^{:\left.  p\right\rceil },p\right)  \cdot%
\raisebox{-0.0406in}{\includegraphics[
height=0.1436in,
width=0.1436in
]%
{icon10.ps}%
}%
^{(\ell+1)}\left(  a,p\right)  \text{ ,}%
\]
we produce an element of the ambient group $\Lambda_{\ell+1}$ which is not of
order $2$. \ \ But $%
\raisebox{-0.0406in}{\includegraphics[
height=0.1436in,
width=0.1436in
]%
{icon10.ps}%
}%
^{(\ell)}\left(  a,p\right)  $ is an element of $\Lambda_{\ell}$ of order $2$.
\ Hence, the morphism $%
\raisebox{-0.0303in}{\includegraphics[
height=0.1332in,
width=0.1193in
]%
{red-refinement.ps}%
}%
:\Lambda_{\ell}\longrightarrow\Lambda_{\ell+1}$ we seek cannot be defined in
this way!

\bigskip

Fortunately, there are other possible elements of $\Lambda_{\ell+1}$ which are
also extensions of
\[
\left.
\raisebox{-0.0303in}{\includegraphics[
height=0.1332in,
width=0.1193in
]%
{red-refinement.ps}%
}%
\left(
\raisebox{-0.0406in}{\includegraphics[
height=0.1436in,
width=0.1436in
]%
{icon10.ps}%
}%
^{(\ell)}\left(  a,p\right)  \right)  \right\vert _{%
\raisebox{-0.0303in}{\includegraphics[
height=0.1332in,
width=0.1193in
]%
{red-refinement.ps}%
}%
\left(  \mathbb{K}^{(\ell)}\right)  }:%
\raisebox{-0.0303in}{\includegraphics[
height=0.1332in,
width=0.1193in
]%
{red-refinement.ps}%
}%
\left(  \mathbb{K}^{(\ell)}\right)  \longrightarrow%
\raisebox{-0.0303in}{\includegraphics[
height=0.1332in,
width=0.1193in
]%
{red-refinement.ps}%
}%
\left(  \mathbb{K}^{(\ell)}\right)  .
\]
\ We know from a previous theorem that each generator of $\Lambda_{\ell}$ is a
product of disjoint transpositions. \ So a possible clue as to which extension
to select comes from the following well known formula for transpositions:%
\[
\left(  1n\right)  =\left(  \left(  (n-1)n\right)  \wedge\ldots\left(
\wedge(34)\wedge\left(  \overset{}{\underset{}{(23)\wedge(12)}}\right)
\right)  \right)  \text{ ,}%
\]
where $g\wedge h$ denotes $g\wedge h=g^{-1}hg$. \ This suggests the following
approach to defining a morphism $%
\raisebox{-0.0303in}{\includegraphics[
height=0.1332in,
width=0.1193in
]%
{red-refinement.ps}%
}%
:\Lambda_{\ell}\longrightarrow\Lambda_{\ell+1}$:

\bigskip

\begin{definition}
We define the \textbf{quandle product} $\wedge$ as%
\[%
\begin{array}
[c]{ccc}%
\Lambda_{\ell}\times\Lambda_{\ell} & \overset{\wedge}{\longrightarrow} &
\Lambda_{\ell}\\
\left(  g_{1},g_{2}\right)  & \longmapsto & g_{1}\symbol{94}g_{2}=g_{2}%
^{-1}g_{1}g_{2}%
\end{array}
\]

\end{definition}

\bigskip

Unfortunately, the quandle product `$\symbol{94}$' is not an associative
binary operation. \ To reduce the number and clutter of parentheses, we adopt
a right to left precedence rule for parentheses. \ For example,%
\[
g_{1}\wedge g_{2}\wedge g_{3}\text{ means }g_{1}\wedge\left(  g_{2}\wedge
g_{3}\right)  \text{ \ \ and \ \ }g_{1}\wedge g_{2}\wedge g_{3}\wedge
g_{4}\text{ \ \ means \ \ \ }g_{1}\wedge\left(  g_{2}\wedge\left(  g_{3}\wedge
g_{4}\right)  \right)  \text{ .}%
\]

\bigskip

Equipped with the quandle product, we are now prepared to define the
refinement map $%
\raisebox{-0.0303in}{\includegraphics[
height=0.1332in,
width=0.1193in
]%
{red-refinement.ps}%
}%
:\Lambda_{\ell}\longrightarrow\Lambda_{\ell+1}$ which we hope will be a
morphism. \ We define the image of the generators as:

\bigskip%

\begin{align*}%
\raisebox{-0.0303in}{\includegraphics[
height=0.1332in,
width=0.1193in
]%
{red-refinement.ps}%
}%
\left(
\raisebox{-0.0406in}{\includegraphics[
height=0.1436in,
width=0.1436in
]%
{icon10.ps}%
}%
^{(\ell)}\left(  a,p\right)  \right)   &  =%
\raisebox{-0.0406in}{\includegraphics[
height=0.1436in,
width=0.1436in
]%
{icon11.ps}%
}%
^{(\ell+1)}\left(  a^{:\left\lfloor p\right.  },p\right)  \wedge%
\raisebox{-0.0406in}{\includegraphics[
height=0.1436in,
width=0.1436in
]%
{icon13.ps}%
}%
^{(\ell+1)}\left(  a^{:\left.  p\right\rceil \left\lfloor p\right.
},p\right)  \wedge%
\raisebox{-0.0406in}{\includegraphics[
height=0.1436in,
width=0.1436in
]%
{icon10.ps}%
}%
^{(\ell+1)}\left(  a^{:\left.  p\right\rceil },p\right)  \wedge%
\raisebox{-0.0406in}{\includegraphics[
height=0.1436in,
width=0.1436in
]%
{icon10.ps}%
}%
^{(\ell+1)}\left(  a,p\right) \\
& \\%
\raisebox{-0.0303in}{\includegraphics[
height=0.1332in,
width=0.1193in
]%
{red-refinement.ps}%
}%
\left(
\raisebox{-0.0406in}{\includegraphics[
height=0.1436in,
width=0.1436in
]%
{icon21.ps}%
}%
^{(\ell)}\left(  a,p\right)  \right)   &  =%
\raisebox{-0.0406in}{\includegraphics[
height=0.1436in,
width=0.1436in
]%
{icon21.ps}%
}%
^{(\ell+1)}\left(  a^{:\left.  p\right\rceil \left\lfloor p\right.
},p\right)  \wedge%
\raisebox{-0.0406in}{\includegraphics[
height=0.1436in,
width=0.1436in
]%
{icon21.ps}%
}%
^{(\ell+1)}\left(  a^{\left\lfloor p\right.  },p\right)  \wedge%
\raisebox{-0.0406in}{\includegraphics[
height=0.1436in,
width=0.1436in
]%
{icon21.ps}%
}%
^{(\ell+1)}\left(  a^{:\left.  p\right\rceil },p\right)  \wedge%
\raisebox{-0.0406in}{\includegraphics[
height=0.1436in,
width=0.1436in
]%
{icon21.ps}%
}%
^{(\ell+1)}\left(  a,p\right) \\
& \\%
\raisebox{-0.0303in}{\includegraphics[
height=0.1332in,
width=0.1193in
]%
{red-refinement.ps}%
}%
\left(
\raisebox{-0.0406in}{\includegraphics[
height=0.1531in,
width=0.1531in
]%
{icon30.ps}%
}%
^{(\ell)}\left(  a,1\right)  \right)   &  =%
\raisebox{-0.0406in}{\includegraphics[
height=0.1436in,
width=0.1436in
]%
{icon20.ps}%
}%
^{(\ell+1)}\left(  a^{:3},2\right)  \wedge%
\raisebox{-0.0406in}{\includegraphics[
height=0.1531in,
width=0.1531in
]%
{icon30.ps}%
}%
^{(\ell+1)}\left(  a^{:13},1\right)  \wedge%
\raisebox{-0.0406in}{\includegraphics[
height=0.1436in,
width=0.1436in
]%
{icon21.ps}%
}%
^{(\ell+1)}\left(  a^{:3},3\right)  \wedge%
\raisebox{-0.0406in}{\includegraphics[
height=0.1436in,
width=0.1436in
]%
{icon21.ps}%
}%
^{(\ell+1)}\left(  a,1\right) \\
&  \wedge%
\raisebox{-0.0406in}{\includegraphics[
height=0.1436in,
width=0.1436in
]%
{icon20.ps}%
}%
^{(\ell+1)}\left(  a^{:1^{2}},1\right)  \wedge%
\raisebox{-0.0406in}{\includegraphics[
height=0.1436in,
width=0.1436in
]%
{icon20.ps}%
}%
^{(\ell+1)}\left(  a^{:2},2\right)  \wedge%
\raisebox{-0.0406in}{\includegraphics[
height=0.1531in,
width=0.1531in
]%
{icon30.ps}%
}%
^{(\ell+1)}\left(  a^{:12},1\right)  \wedge%
\raisebox{-0.0406in}{\includegraphics[
height=0.1436in,
width=0.1436in
]%
{icon20.ps}%
}%
^{(\ell+1)}\left(  a^{:2},3\right)
\end{align*}

\bigskip\bigskip

The key question is whether or not this extends to a group morphism $%
\raisebox{-0.0303in}{\includegraphics[
height=0.1332in,
width=0.1193in
]%
{red-refinement.ps}%
}%
:\Lambda_{\ell}\longrightarrow\Lambda_{\ell+1}$. \ For this to be true, each
relation among the generators of $\Lambda_{\ell}$ must map to a relation among
the generators of $\Lambda_{\ell+1}$. \ 

\bigskip

Conjectures 1A, 1B, and 1C of section 12 are essentially based on the above construction.

\bigskip

If the above conjecture is true, then we have created the following directed
(\textbf{lattice knot) ambient group system}%

\[
\left(  \mathbb{K},\Lambda\right)  =\left(  \mathbb{K}^{(0)},\Lambda
_{0}\right)  \overset{%
\raisebox{-0.0303in}{\includegraphics[
height=0.1332in,
width=0.1193in
]%
{red-refinement.ps}%
}%
}{\longrightarrow}\left(  \mathbb{K}^{(1)},\Lambda_{1}\right)  \overset{%
\raisebox{-0.0303in}{\includegraphics[
height=0.1332in,
width=0.1193in
]%
{red-refinement.ps}%
}%
}{\longrightarrow}\left(  \mathbb{K}^{(2)},\Lambda_{2}\right)  \overset{%
\raisebox{-0.0303in}{\includegraphics[
height=0.1332in,
width=0.1193in
]%
{red-refinement.ps}%
}%
}{\longrightarrow}\cdots\overset{%
\raisebox{-0.0303in}{\includegraphics[
height=0.1332in,
width=0.1193in
]%
{red-refinement.ps}%
}%
}{\longrightarrow}\left(  \mathbb{K}^{(\ell)},\Lambda_{\ell}\right)  \overset{%
\raisebox{-0.0303in}{\includegraphics[
height=0.1332in,
width=0.1193in
]%
{red-refinement.ps}%
}%
}{\longrightarrow}\left(  \mathbb{K}^{(\ell+1)},\Lambda_{\ell+1}\right)
\overset{%
\raisebox{-0.0303in}{\includegraphics[
height=0.1332in,
width=0.1193in
]%
{red-refinement.ps}%
}%
}{\longrightarrow}\cdots
\]
Modulo this conjecture, the direct limit
\[
\left(  \lim_{\longrightarrow}\mathbb{K},\lim_{\longrightarrow}\Lambda\right)
\]
of this system exists. \ Thus, we can think of each element of $\lim
\limits_{\longrightarrow}\Lambda$, when restricted to a lattice knot
$K\in\mathbb{K}$, as an \textbf{ambient isotopy} of $\mathbb{R}^{3}$ which
respects the lattice knot $K$.

\bigskip

If we omit the tug moves in the above definition, we can, in like manner,
define the \textbf{inextensible (lattice knot ambient) group system}%

\[
\left(  \mathbb{K},\widetilde{\Lambda}\right)  =\left(  \mathbb{K}%
^{(0)},\widetilde{\Lambda}_{0}\right)  \overset{%
\raisebox{-0.0303in}{\includegraphics[
height=0.1332in,
width=0.1193in
]%
{red-refinement.ps}%
}%
}{\longrightarrow}\left(  \mathbb{K}^{(1)},\widetilde{\Lambda}_{1}\right)
\overset{%
\raisebox{-0.0303in}{\includegraphics[
height=0.1332in,
width=0.1193in
]%
{red-refinement.ps}%
}%
}{\longrightarrow}\left(  \mathbb{K}^{(2)},\widetilde{\Lambda}_{2}\right)
\overset{%
\raisebox{-0.0303in}{\includegraphics[
height=0.1332in,
width=0.1193in
]%
{red-refinement.ps}%
}%
}{\longrightarrow}\cdots\overset{%
\raisebox{-0.0303in}{\includegraphics[
height=0.1332in,
width=0.1193in
]%
{red-refinement.ps}%
}%
}{\longrightarrow}\left(  \mathbb{K}^{(\ell)},\widetilde{\Lambda}_{\ell
}\right)  \overset{%
\raisebox{-0.0303in}{\includegraphics[
height=0.1332in,
width=0.1193in
]%
{red-refinement.ps}%
}%
}{\longrightarrow}\left(  \mathbb{K}^{(\ell+1)},\widetilde{\Lambda}_{\ell
+1}\right)  \overset{%
\raisebox{-0.0303in}{\includegraphics[
height=0.1332in,
width=0.1193in
]%
{red-refinement.ps}%
}%
}{\longrightarrow}\cdots
\]
In this case, each element of $\lim\limits_{\longrightarrow}\widetilde
{\Lambda}$, when restricted to a lattice knot $K\in\mathbb{K}$, becomes an
\textbf{inextensible ambient isotopy}.

\bigskip

\section{Appendix D: \ Oriented quantum knots}

\bigskip

We will now briefly outline how to define oriented quantum knots and oriented
quantum knot systems. \ This will be accomplished by associating a qutrit
(instead of a qubit) with each edge of the bounded cell complex $\mathcal{C}%
_{\ell,n}$. \ 

\bigskip

Let `$<$' denote the lex ordering of the set of edges $\mathcal{E}_{\ell,n}$
of the cell complex $\mathcal{C}_{\ell,n}$, as previously defined in Section
19 of this paper. \ Moreover, let $\mathcal{H}^{o}$ be the three dimensional
Hilbert space with orthonormal basis%
\[
\left\{  \left\vert 0\right\rangle =\left\vert
\raisebox{-0.0303in}{\includegraphics[
height=0.1781in,
width=0.6287in
]%
{dotted-edge.ps}%
}%
\right\rangle ,\ \left\vert 1\right\rangle =\left\vert
\raisebox{-0.0303in}{\includegraphics[
height=0.1781in,
width=0.6287in
]%
{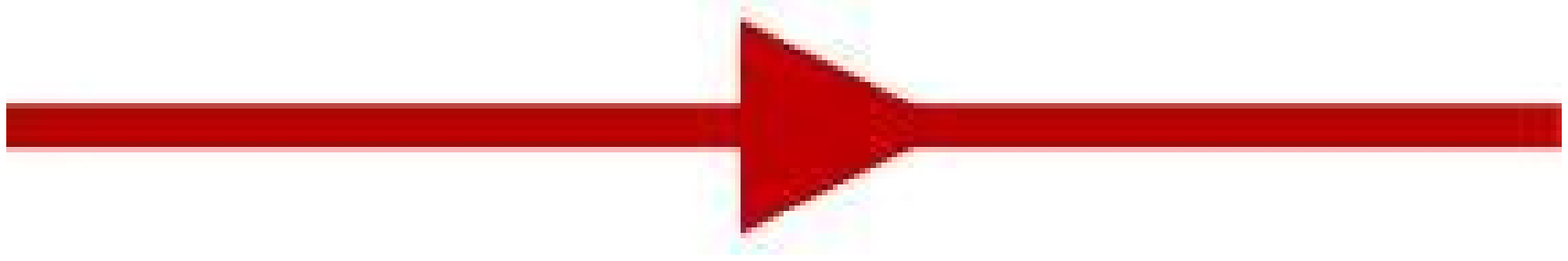}%
}%
\right\rangle ,\ \left\vert 2\right\rangle =\left\vert
\raisebox{-0.0303in}{\includegraphics[
height=0.1781in,
width=0.6287in
]%
{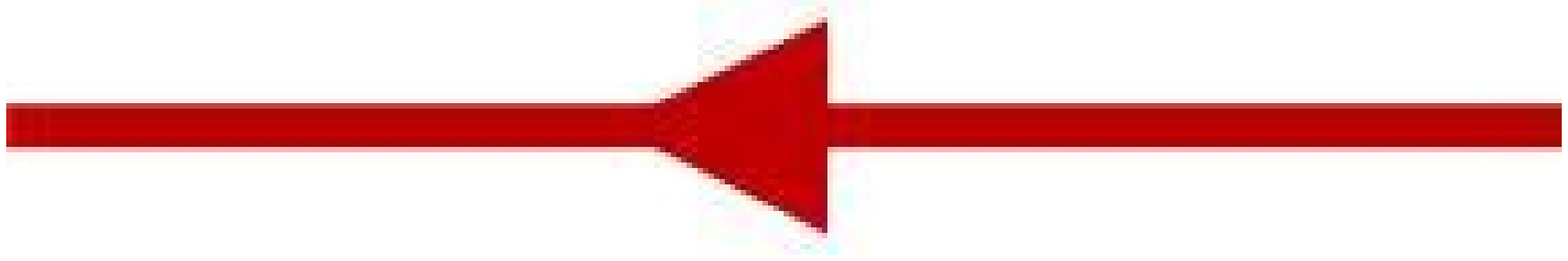}%
}%
\right\rangle \right\}  \text{ .}%
\]
Then the \textbf{Hilbert space }$\mathcal{G}_{\ell,n}^{o}$\textbf{ of oriented
lattice graphs} \textbf{of order }$\left(  \ell,n\right)  $ is defined as the
tensor product%
\[
\mathcal{G}_{\ell,n}^{o}=%
{\displaystyle\bigotimes\limits_{E\in\mathcal{E}_{\ell,n}}}
\mathcal{H}^{o}\text{ ,}%
\]
where the tensor product is taken with respect to the above defined linear
ordering `$<$'. \ Thus, as orthonormal basis for the Hilbert space
$\mathcal{G}_{\ell,n}^{o}$, we have%
\[
\left\{  \ \underset{E\in\mathcal{E}_{\ell,n}}{\otimes}\left\vert c\left(
E\right)  \right\rangle :c\in Map\left(  \mathcal{E}_{\ell,n}\ ,\ \left\{
0,1,2\right\}  \right)  \ \right\}  \text{ ,}%
\]
where\textbf{ }$Map\left(  \mathcal{E}_{\ell,n}\ ,\ \left\{  0,1,2\right\}
\right)  $ is the set of all maps $c:\mathcal{E}_{\ell,n}\longrightarrow
\left\{  0,1,2\right\}  $ from the set $\mathcal{E}_{\ell,n}$ of edges to the
set $\left\{  0,1,2\right\}  $.

\bigskip

We identify in the obvious way each basis element
\[
\underset{E\in\mathcal{E}_{\ell,n}}{\otimes}\left\vert c\left(  E\right)
\right\rangle
\]
with a corresponding oriented lattice graph $G$. \ Under this identification,
the space $\mathcal{G}_{\ell,n}^{o}$ becomes the Hilbert space with
the\textbf{ }orthonormal basis%
\[
\left\{  \left\vert G\right\rangle :G\text{ an oriented lattice graph in
}\mathcal{L}_{\ell,n}\right\}  \text{ .}%
\]

\bigskip

We define the \textbf{Hilbert space }$\mathcal{K}_{o}^{(\ell,n)}$\textbf{ of
oriented lattice knots} \textbf{of order }$\left(  \ell,n\right)  $ as the
sub-Hilbert space of $\mathcal{G}_{\ell,n}^{o}$ with orthonormal basis
\[
\left\{  \left\vert K\right\rangle :K\in\mathbb{K}_{o}^{(\ell,n)}\right\}
\text{ ,}%
\]
where $\mathbb{K}_{o}^{(\ell,n)}$ denotes the finite set of oriented lattice
knots of order $\left(  \ell,n\right)  $. \ We call each element of the
Hilbert space $\mathcal{K}_{o}^{(\ell,n)}$ an \textbf{oriented quantum knot}.

\bigskip

For oriented lattice knots, oriented wiggles, wags, and tugs, and the
corresponding oriented lattice graph ambient groups $\Lambda_{\ell,n}^{o}$ and
$\widetilde{\Lambda}_{\ell,n}^{o}$can be defined in the obvious manner. \ The
remaining definitions and constructions are straight forward, and left to the reader.

\bigskip

\begin{remark}
We should mention that, as abstract groups, the oriented lattice knot ambient
groups $\Lambda_{\ell,n}^{o}$ and $\widetilde{\Lambda}_{\ell,n}^{o}$ are
respectively isomorphic to the unoriented lattice knot ambient groups
$\Lambda_{\ell,n}$ and $\widetilde{\Lambda}_{\ell,n}$.
\end{remark}

\bigskip

\section{Appendix E: \ Quantum graphs}

\bigskip

We very briefly outline how to define quantum graphs and quantum graph
systems. \ 

\bigskip

Let $\mathcal{G}_{\ell,n}$ be the Hilbert space of lattice graphs with
orthonormal basis%
\[
\left\{  \ \left\vert G\right\rangle :G\text{ a lattice graph in }%
\mathcal{L}_{\ell,n}\ \right\}  \text{ ,}%
\]
as defined in Section 19 of this paper.

\bigskip

We call each element of the Hilbert space $\mathcal{G}_{\ell,n}$ a
\textbf{quantum graph}.

\bigskip

For lattice graphs, wiggles, wags, and tugs, and the corresponding\textbf{
lattice graph ambient groups} $\Lambda_{\ell,n}^{graph}$ and $\widetilde
{\Lambda}_{\ell,n}^{graph}$can be defined in the obvious manner. \ The
remaining definitions and constructions are straight forward, and left to the reader.

\bigskip

\noindent\textbf{Question.} \ \textit{Are the lattice graph lattice
groups}$\Lambda_{\ell,n}^{graph}$\textit{ and }$\widetilde{\Lambda}_{\ell
,n}^{graph}$\textit{ as abstract groups respectively isomorphic to the lattice
knot ambient groups }$\Lambda_{\ell,n}$\textit{ and }$\widetilde{\Lambda
}_{\ell,n}$\textit{?}


\bigskip

\begin{thebibliography}{99}                                                                                               %


\bibitem {Adams1}Adams, Colin C.,\textbf{ "The Knot Book,"} W.H. Freeman, ((1994).

\bibitem {Aharonov1}Aharonov, D., Jones, V., Landau, Z., \textbf{On the
quantum algorithm for approximating the jones polynomial}, http://arxiv.org/abs/quant-ph/0511096

\bibitem {Bottemal1}Bottema, O., and B. Roth, "Theoretical Knematics," Dover, (1990).

\bibitem {Burde1}Burde, Gerhard, and Heiner Zieschang, \textbf{"Knots,"}
Walter de Gruyter, (1985).

\bibitem {Chen1}Chen, G., Kauffman, L., Lomonaco, S.J. (eds.), \textbf{"The
Mathematics of Quantum Computation and Quantum Topology,"} Chapman \& Hall/CRC (2007)

\bibitem {Collins1}Collins, G.P., \textbf{Computing with quantum knots}, Sci.
Am. April, 56--63 (2006)

\bibitem {Crowell1}Crowell, R.H., Fox, R.H.: \textbf{"Introduction to Knot
Theory,"} Springer-Verlag, (1977)

\bibitem {Dehn1}Dehn, M., P. Heegaard, \textbf{Analysis Situs}, in Encykl.
Math. Wiss., vol. III, AB3, Leipzig, (1907), 153-220.

\bibitem {DoCarmo1}Do Carmo, Manfredo P., \textbf{"Differential Geometry of
Curves and Surfaces,"} Prentice-Hall, (1976).

\bibitem {Flanders1}Flanders, Harley, \textbf{"Differential Forms with
Applications to the Physical Sciences,"} Academic Press, (1963).

\bibitem {Fort1}Fort, Hugo, Rodolfo Gambini, and George Pullin,
\textbf{Lattice knot theory and quantum gravity in loop representation,
Vienna}, Preprint, ESI 368 (1996). (Republished by Dover in 2008.)

\bibitem {Fox1}Fox, Ralph H., \textbf{"A quick trip through knot theory,"} in
\textbf{"Topology of 3-Manifolds,"} (ed. by M.K. Fort, Jr.), Prentice-Hall,
(1962), 120-167.

\bibitem {Gauss1}Gauss, Karl Friedrich, \textbf{Gauss 1833 Werke,
K\"{o}niglichen Gesellschaft der Wissenschaften zu G\"{o}ttingen,} (1877),
605- .

\bibitem {Gelfand1}Gelfand, I.M., and S.V. Fomin, "\textbf{Calculus of
Variations,}" Prentice-Hall, ((1963).

\bibitem {Helstrom1}Helstrom, C.W., \textbf{Quantum Detection and Estimation
Theory}, Academic Press, (1976)

\bibitem {Hunt1}Hunt, K.H., "Kinematic Geometry of Mechanisms," Oxford Press, (1990).

\bibitem {Jacak1}Jacak, L., Sitko, P.,Wieczorek, K.,W\'{o}js, A.,
\textbf{"Quantum Hall Systems: Braid groups, Composite Fermions, and
Fractional Charge,"} Oxford University Press, (2003)

\bibitem {Kauffman1}Kauffman, L.H., Lomonaco, S.J., \textbf{Quantum knots},
Proc. SPIE, (2004). http://arxiv.org/abs/quant-ph/0403228

\bibitem {Kauffman2}Kauffman, L.H., Lomonaco, S.J., Jr., \textbf{q-Deformed
spin networks, knot polynomials, and anyonic topological quantum computation},
J. Knot Theory 16(3), (2007), 267-332. \ http://xxx.lanl.gov/abs/quant-ph/0606114

\bibitem {Kauffman3}Kauffman, L.H., Lomonaco, S.J., Jr., \textbf{Spin networks
and anyonic topological computing}, Proc. SPIE, 6244 (2006). http,//xxx.lanl.gov/abs/quant-ph/0603131

\bibitem {Kauffman4}Kauffman, L.H., \textbf{"Knots and Physics,"} 3rd edn.,
World Scientific, (2001)

\bibitem {Kitaev1}Kitaev, A.Y., \textbf{Fault-tolerant quantum computation by
anyons}. http://arxiv.org/abs/quant-ph/9707021

\bibitem {Leonhardt1}Leonhardt, U., \textbf{Measuring the Quantum State of
Light}, Cambridge University Press. (1997)

\bibitem {Lickorish1}Lickorish, W.B.R., \textbf{"An Introduction to Knot
Theory,"} Springer, (1997)

\bibitem {Lomonaco1}Lomonaco, Samuel J., and Louis H. Kauffman,
\textbf{Quantum knots and mosaics}, Journal of Quantum Information Processing,
vol. 7, nos. 2-3, (2008), 85-115. \ http://arxiv.org/abs/0805.0339

\bibitem {Lomonaco2}Lomonaco, S.J., \textbf{"The modern legacies of Thomson's
atomic vortex theory in classical electrodynamics,"} AMS PSAPM/51, Providence,
RI, (1996), 145--166

\bibitem {Lomonaco3}Lomonaco, S.J., Jr., and L.H. Kauffman,
\textbf{Topological quantum computing and the Jones polynomial}, Proc. SPIE
6244 (2006). http://xxx.lanl.gov/abs/quant-ph/0605004

\bibitem {Lomonaco4}Lomonaco, S.J., Jr., \textbf{A Rosetta stone for quantum
mechanics with an introduction to quantum computation}, PSAPM, 58, AMS,
Providence, RI, pp. 3--65 (2002)

\bibitem {Lomonaco5}Lomonaco, S.J., Jr. (ed.), \textbf{"Quantum computation,"}
PSAPM, 58, American Mathematical Society, Providence, Rhode Island (2002)

\bibitem {Lomonaco6}Lomonaco, S.J., Jr., \textbf{Five dimensional knot
theory}, AMS CONM 20, (1984), 249--270. \ http://www.cs.umbc.edu/\symbol{126}lomonaco/5knots/5knots.pdf

\bibitem {Lomonaco7}Lomonaco, S.J., Jr., \textbf{The homotopy groups of knots
I. How to compute the algebraic 3-type}, Pacific J. Math. 95(2), (1981), 349--390

\bibitem {Lomonaco8}Lomonaco, S.J., Jr., \textbf{Inextensible knot theory},
(in preparation).

\bibitem {Madison1}Madison, K.W., Chevy, F., Wohlleben, W., Dalibard, J.,
\textbf{Vortex formation in a stirred Bose-Einstein condensate}, Phys. Rev.
Lett. 84, (2000), 806--809

\bibitem {McCarthy1}McCarthy, J.M., \textbf{"An Introduction to Theoretical
Kinematics,"} MIT Press, (1990).

\bibitem {Murasugi1}Murasugi, K., \textbf{"Knot Theory and Its Applications,"}
Birkhauser, (1996)

\bibitem {O'Neill1}O'Neill, Barrett, \textbf{"Elementary Differential
Geometry,"} (2nd edition), Academic Press, (1997).

\bibitem {Nielsen1}Nielsen, M.A., Chuang, I.L., \textbf{"Quantum Computation
and Quantum Information,"} Cambridge University Press, (2000)

\bibitem {Paul1}Paul, Richard P., \textbf{"Robot Manipulators: Mathematics
Programming ,and Control,"} M.I.T. Press, (1992).

\bibitem {Peres1}Peres, A., \textbf{"Quantum Theory: Concepts and Methods,"}
Kluwer (1993)

\bibitem {Przytycki1}Przytycki, J\'{o}zef H., \textbf{"Knots: From
Combinatorics of Knot Diagrams to Combinatorial Topology Based on Knots,"}
Cambridge University Press, (2010).

\bibitem {Rasetti1}Rasetti, M., and T. Regge, \textbf{Vortices in He II,
current algebras and quantum knots}, Physica 80 A, North-Holland, (1975), 217--233

\bibitem {Reidemeister1}Reidemeister, K., \textbf{"Knotentheorie,"} Chelsea, (1948).

\bibitem {Reidemeister2}Reidemeister, K., \textbf{"Knot Theory,"} BCS
Associates, Moscow, Idaho, ((1983). (Translation)

\bibitem {Rolfsen1}Rolfsen, Dale, \textbf{"Knots and Links,"} Publish or
Perish, (1976).

\bibitem {Sakuri1}Sakuri J.J., \textbf{"Modern Quantum Mechanics,"} revised
edn., Addison-Wesley, (1994)

\bibitem {Sarma1}Sarma, S.D., Freedman, M., Nayak, C., \textbf{Topologically
protected qubits from a possible non-Abelian fractional quantum hall state,}
Phys. Rev. Lett. 94, 166802-1--168802-4 (2005)

\bibitem {Shankar1}Shankar, R., \textbf{"Principles of Quantum Mechanics,"}
2nd edn., Plenum, (1994)

\bibitem {Shor1}Shor, P.W., Jordan, S.P., \textbf{Estimating Jones polynomials
is a complete problem for one clean qubit}. \ http://arxiv.org/abs/0707.2831

\bibitem {Spivak1}Spivak, Michael, \textbf{"A Comprehensive Introduction to
Differential Geometry,"} vols. 1-5, (2nd edition), Publish or Perish, (1979).

\bibitem {Wen1}Wen, X.-G., \textbf{"Quantum Field Theory of Many-Body
Systems,"} Oxford Press, (2004)

\bibitem {Zawitz1}Zawitz, Richard E., United States Patent Number 4,509,929,
April 4, 1985.

\bibitem {Zee1}Zee, A., \textbf{"Quantum Field Theory in a Nutshell,"}
Princeton University Press, (2003)
\end{thebibliography}
\end{document}